\documentclass[12pt]{book}
\usepackage{amsmath,amssymb,amsfonts}
\usepackage[active]{srcltx}
\usepackage[all]{xy}
\input diagxy
\usepackage[pagebackref]{hyperref}
\usepackage{colortbl}

 1 %(caracteres goticos)
\newtheorem{teor}{Theorem}[chapter]
\newtheorem{prop}[teor]{Proposition}
\newtheorem{corol}[teor]{Corollary}
\newtheorem{lem}[teor]{Lemma}
\newtheorem{definition}[teor]{Definition}
\newtheorem{remark}[teor]{Remark}
\newtheorem{assump}[teor]{Assumption}

\label{key}

\newtheorem{pos}[teor]{Postulate}

\def\beq{\begin{equation}}
\def\eeq{\end{equation}}
\def\bea{\begin{eqnarray}}
\def\eea{\end{eqnarray}}
\def\beann{\begin{eqnarray*}}
\def\eeann{\end{eqnarray*}}
\def\beasn{\begin{sneqnarray}}
\def\eeasn{\end{sneqnarray}}
\def\ben{\begin{enumerate}}
\def\een{\end{enumerate}}
\def\bit{\begin{itemize}}
\def\eit{\end{itemize}}
\def\dst{\displaystyle}

\def\derpar#1#2{\frac{\partial{#1}}{\partial{#2}}}

\def\mapping#1{\mathrel{\mathop{\longrightarrow}\limits^{#1}}}

\def\coor#1#2#3{{#1}^{#2}, \ldots, {#1}^{#3}}
\def\moment#1#2#3{{#1}_{#2}, \ldots, {#1}_{#3}}
\def\qed{\ifvmode\removelastskip\fi
{\unskip\nobreak\hfil\penalty50\hbox{}\nobreak\hfil
\hbox{\vrule height1.2ex width1.2ex}\parfillskip=0pt
\finalhyphendemerits=0 \par\smallskip}}
\def\buit{\hbox{\rm\O}}
\def\vf{\mathfrak X}
\def\df{{\mit\Omega}}
\def\Lag{{\cal L}}
\def\Diff{{\rm Diff}}
\def\d{{\rm d}}
\def\h{{\rm h}}

\def\Real{\mathbb{R}}
\def\Complex{\mathbb{C}}
\def\Nat{\mathbb{N}}

\def\bmeta{\mbox{\boldmath $\eta$}}
\def\evo{\mbox{\boldmath $\varepsilon$}}
\def\X{{\rm X}}
\def\Tan{{\rm T}}
\def\Lie{\mathop{\rm L}\nolimits}
\def\inn{\mathop{i}\nolimits}
\def\Cinfty{{\rm C}^\infty}
\renewcommand{\neq}{=\hspace{-3.5mm}/\hspace{2mm}}

\def\Leg{{\cal F}\Lag}
\newcommand{\Reeb}{\mathcal{R}}
\title{GEOMETRY OF MECHANICS}
\author{\sc Miguel C. Mu\~noz-Lecanda
\thanks{%{\bf e}-{\it mail}: miguel.carlos.munoz@upc.edu /
ORCID: 0000-0002-7037-0248.},
   \\
\sc Narciso Rom\'an-Roy
\thanks{{\bf e}-{\it mail}: narciso.roman@upc.edu / 
ORCID: 0000-0003-3663-9861.}
   \\
  {\it Department of Mathematics.}\\
{\it Technical University of Catalonia.}\\
{\it  E-08034 Barcelona, Spain.}}
\begin{document}
\maketitle
\thispagestyle{empty}

%\begin{abstract}
\noindent {\bf Abstract}

\noindent  The aim of this work is to study
the geometry underlying mechanics and its application
to describe autonomous and nonautonomous conservative dynamical systems of different types; as well as dissipative dynamical systems.
We use different geometric descriptions to study
the main properties and characteristics of these systems;
such as their Lagrangian, Hamiltonian and unified formalisms, 
their symmetries, the variational principles, and others.
The study is done mainly for the regular case,
although some comments and explanations about singular systems are also included.
%\end{abstract}

\bigskip\bigskip\bigskip\bigskip\bigskip\bigskip\bigskip\bigskip\bigskip\bigskip

\noindent\textbf{Keywords:}

\noindent
Symplectic manifolds, cosymplectic manifolds, Hamiltonian systems,  Riemannian manifolds, Newtonian systems, contact manifolds, dissipative systems, Lagrangian formalism, Hamiltonian formalism, symmetries, conservation laws,  actions of Lie groups, fiber bundles, variational principles.

\medskip\medskip\medskip\medskip

\noindent\textbf{MSC\,2020 codes:}

\noindent Primary: 53B20, 53D05, 53D10, 70H03, 70H05, 70H33.

\noindent Secondary: 53B21, 53D20, 53D22,  53Z05, 58A30, 70A05, 70F20, 70F25, 70G10, 70G45, 70G65, 70G75, 70H15, 70H20, 70H25.

%51B20: Minkowski geometries
%53A45: Vector and tensor analysis
%53B21:  Methods of Riemannian geometry
%53D05: Symplectic manifolds, general
%53D10; Contact manifolds, general
%53Z05: Applications to physics
%55R10: Fiber bundles
%58A30: Vector distributions (subbundles of the tangent bundles)
%70A05: Axiomatics, foundations (MECHANICS OF PARTICLES AND SYSTEMS)
%70F20 Holonomic systems
%70F25 Nonholonomic systems
%70G10: Generalized coordinates; event, impulse-energy, configuration, state, or phase space
%70G45: Differential-geometric methods (tensors, connections, symplectic, Poisson, contact, Riemannian, nonholonomic, etc.)
%70G75: Variational methods
%70H03 LagrangeâÃÃ´s equations
%70H05 HamiltonâÃÃ´s equations
%70H15  Canonical and symplectic transformations
%70H25 HamiltonâÃÃ´s principle
%70H30 Other variational principles
%70H33 Symmetries and conservation laws, reverse symmetries, invariant

%%%%%%%%%%%%%%%%%%%%%%%%%%%%%%%%%%%%%

\chapter*{}

\pagenumbering{roman}
\setcounter{tocdepth}{2}

\section*{Acknowledgments}

Our special thanks to Prof. Xavier Gr\`acia-Sabat\'e for the many clarifying discussions held on many of the contents of this work,
and to Prof. Xavier Rivas-Guijarro for his careful reading of the entire manuscript and his very helpful comments and suggestions.

We also thank Prof. Manuel de Le\'on-Rodr\' iguez for his comments and his expert advice on the final editing and publication process of this manuscript.

Finally, we want to thank all our colleagues from the 
{\sl Geometry, Dynamics, and Field Theory Network} for their enriching collaboration over all these years.

We acknowledge the financial support from the 
Spanish Ministry of Science and Innovation, grants  PID2021-125515NB-C21, and RED2022-134301-T of AEI, 
and Ministry of Research and Universities of
the Catalan Government, project 2021 SGR 00603 \textsl{Geometry of Manifolds and Applications, GEOMVAP}.

%%%%%%%%%%%%%%%%%%%%%%%%%%%%%%%%%%%%%%%%%%%%%%%%%%%%%%%%%%%%%%%%%%%%%%%%%%%%%%%%%%%%%%%%%%%%%%%%%%%%%%%%%%%%%%%%%%%%%%%%%%%

{
\def\baselinestretch{0.9}
\small
% hack per a eliminar l'espai vertical 1em
\def\addvspace#1{\vskip 1pt}
\parskip 0pt plus 0.1mm
\tableofcontents
}

%%%%%%%%%%%%%%%%%%%%%%%%%%%%%%%%%%%%%%%%%%%%%%%%%%%%%%%%%%%%%%%%%%%%%%%%%%%%%%%%%%%%%%%%%%%%%%%%%%%%%%

\chapter*{Glossary of notation and terminology}

\begin{center}
{\small
\begin{tabular}{|c|c|}
\hline
 & \\
$\Cinfty(M)$ & Smooth functions (of class $\Cinfty$) on a manifold $M$. \\
$\df^k(M)$ & Differentiable $k$-forms on a manifold $M$. \\
$Z^k(M)$ & Differentiable closed $k$-forms on a manifold $M$. \\
$\vf(M)$ & Vector fields on a manifold $M$. \\
${\cal T}^m_k(M)$ & Tensor fields of type $(k,m)$ on a manifold $M$. \\
${\rm Diff}\,(M)$ & Diffeomorphisms of the manifold $M$. \\
$\vf^{{\rm V}(\pi)}(M)$ & $\pi$-vertical vector fields on a bundle $\pi\colon M\to N$. \\
$\vf(M,\pi)$ & Vector fields along a map $\pi$ on a manifold $M$. \\
${\mit\Gamma}(\pi)$ & Sections of the projection $\pi$ on a bundle $\pi\colon M\to N$. \\
$\d$ & Exterior differential. \\
$\beta(X)\equiv\langle X\vert\beta\rangle$ & Natural pairing between a differential $1$-form $\beta$ and a vector field $X$. \\
$\inn(X)\beta$ & Contraction of a differential $k$-form $\beta$ with a vector field $X$. \\
$\Lie(X)\beta$ & Lie derivative of a differential form $\beta$ by a vector field $X$. \\
$X(f)\equiv \Lie(X)f$ & Action of a vector field $X$ (Lie derivative) on a function $f$. \\
$\Tan Q$ , $\tau_Q\colon\Tan Q\to Q$ & Tangent bundle of a manifold $Q$ , canonical projection. \\
$\Tan^*Q$ , $\pi_Q\colon\Tan^*Q\to Q$ & Cotangent bundle of a manifold $Q$ , canonical projection. \\
$(q^i,v^i)\ ,\  (q^i,p_i)$ & Natural coordinates on $\Tan Q$ and $\Tan^*Q$. \\
$\Omega\in\df^2(M)$ , $(M,\Omega)$ & (Pre)symplectic form , (pre)symplectic manifold. \\
$\Theta\in\df^1(M)$ & (Pre)symplectic potential ($\Omega=\d\Theta$). \\
$\flat_\Omega\colon\Tan M\to\Tan^*M$ & \\
$\sharp_\Omega=\flat_\Omega^{-1}\colon\Tan^* M\to\Tan M$ & Canonical isomorphisms on a symplectic manifold $(M,\Omega)$. \\
$(M,\Omega,\alpha)$\ , \ $(M,\Omega,h)$ & Hamiltonian system on a symplectic manifold $(M,\Omega)$. \\
$\{ f,g\}$ & Poisson bracket between the functions $f$ and $g$. \\
$\vf_{lh} (M)$ , $\vf_H (M)$ & Local/global Hamiltonian vector fields on a symplectic manifold $(M,\Omega)$. \\
$X_h\in\vf(M)$ &
Hamiltonian vector field of $f\in\Cinfty(M)$
on a symplectic manifold $(M,\Omega)$. \\
$G$ , ${\bf g}$ & Lie group, Lie algebra. \\
$\xi_X\in\vf(M)$ & Fundamental vector field of $\Phi\colon G\to{\rm Diff}(M)$ associated to $X\in{\bf g}$. \\
${\rm j}^*\colon{\bf g}\to\Cinfty(M)$ & Comomentum map associated to an action $\Phi\colon G\to{\rm Diff}(M)$.  \\
${\rm J}\colon M\to{\bf g}^*$ & Momentum map associated to an action $\Phi\colon G\to{\rm Diff}(M)$.  \\
$\Delta\in\vf(M)$ & Liouville (dilation) vector field on a vector bundle $\pi\colon M\to N$. \\
$\Tan\phi\colon\Tan M\to\Tan N$ & Canonical lift of a map $\phi\colon M\to N$ to the tangent bundles.  \\
$Z^V\in\vf(\Tan Q)$ & Vertical lift of $Z\in\vf(Q)$ to the tangent bundle $\Tan Q$.  \\
$Z^C\in\vf(\Tan Q)$ & Complete canonical lift of $Z\in\vf(Q)$ to the tangent bundle $\Tan Q$.  \\
$\widetilde\gamma\colon\Real\to\Tan Q$ & Canonical lift of a curve $\gamma\colon\Real\to Q$ to the tangent bundle $\Tan Q$.  \\
$\Tan^*\phi\colon\Tan^*N\to\Tan^*M$ & Canonical lift of a map $\phi\colon M\to N$ to the cotangent bundles.  \\
 & \\
\hline
\end{tabular}
}
\end{center}

\begin{center}
{\small
\begin{tabular}{|c|c|}
\hline
 & \\
  $Z^*\in\vf(\Tan^*Q)$ & Canonical lift of $Z\in\vf(Q)$ to the cotangent bundle $\Tan^*Q$.  \\
 ${\cal FL}\colon\Tan Q\to\Tan^*Q$ &
 Legendre map defined by $\Lag\in\Cinfty(\Tan Q)$. \\
 $J\in{\cal T}^1_1(\Tan Q)$ & Canonical endomorphism on $\Tan Q$. \\
$\Lag\in\Cinfty(\Tan Q)$ , $(\Tan Q,\Lag)$ & Lagrangian function, Lagrangian system. \\
%$(\Tan Q,\Lag)$ & Lagrangian system. \\
$\Theta_\Lag\in\df^1(\Tan Q)$ & Lagrangian (Cartan) $1$-form associated with $\Lag\in\Cinfty(\Tan Q)$. \\
$\Omega_\Lag=-\d\Theta_\Lag\in\df^2(\Tan Q)$ & 
Lagrangian (Cartan) $2$-form associated with $\Lag\in\Cinfty(\Tan Q)$. \\
$X_\Lag\in\vf(\Tan Q)$ \ , \ $\Gamma_\Lag\in\vf(\Tan Q)$ &
Lagrangian/Euler-Lagrange vector field
of a Lagrangian system $(\Tan Q,\Lag)$. \\
$\Phi\colon Q\times\Real^n\longrightarrow\Tan^*Q$ &
Complete solution to the Hamilton-Jacobi equation. \\
${\cal S}\colon Q\times\Real^n\longrightarrow\Real$ &
Generating function of a complete solution to the H-J eq. \\
${\rm h}\in\Cinfty(\Tan^*Q)$ , $(\Tan^*Q,\Omega,{\rm h})$ & Canonical Hamiltonian function , Canonical Hamiltonian system. \\
${\cal W}=\Tan Q\times_Q\Tan^*Q$ &
Unified (Pontryagin) bundle. \\ 
$\varrho_1\colon{\cal W}\to\Tan Q$ & \\
$\varrho_2\colon{\cal W}\to\Tan^*Q$ & \\
$\varrho_0\colon{\cal W}\to Q$ &
Canonical projections of the unified  bundle ${\cal W}$. \\ 
$\Theta_{\cal W}\in\df^1({\cal W})$ &
Canonical $1$-form on the unified bundle ${\cal W}$. \\ 
$\Omega_{\cal W}=-\d\Theta_{\cal W}\in\df^2({\cal W})$ &
Canonical $2$-form on the unified bundle ${\cal W}$. \\ 
${\cal C}\colon{\cal W}\to\Real$ &
Coupling function on the unified bundle ${\cal W}$. \\
${\cal H}\colon{\cal W}\to\Real$ &
Hamiltonian function on the unified bundle ${\cal W}$. \\
$X_{\cal H}\in\vf({\cal W})$ &
Dynamical vector field of the dynamical system $({\cal W},\Omega_{\cal W},{\cal H})$. \\
$\jmath_0\colon\mathcal{W}_0\hookrightarrow\mathcal{W}$ &
Compatibility submanifold of the dynamical system $({\cal W},\Omega_{\cal W},{\cal H})$. \\
$(\eta,\omega)$ , $(M,\eta,\omega)$ & (Pre)cosymplectic structure , (pre)cosymplectic manifold. \\
$R\in\vf(M)$ & Reeb vector field on a cosymplectic manifold $(M,\eta,\omega)$. \\
$\flat_{(\eta,\omega)}\colon\Tan M\to\Tan^*M$ & \\
$\sharp_{(\eta,\omega)}=\flat_{(\eta,\omega)}^{-1}\colon\Tan^* M\to\Tan M$ & Canonical isomorphisms on a cosymplectic manifold $(M,\eta,\omega)$. \\
$X_f\in\vf(M)$ & Hamiltonian vector field on a cosymplectic manifold $(M,\eta,\omega)$. \\
${\rm grad}\,f\in\vf(M)$ & Gradient vector field on a cosymplectic manifold $(M,\eta,\omega)$. \\
${\cal E}_f\in\vf(M)$ & Evolution vector field on a cosymplectic manifold $(M,\eta,\omega)$. \\
$\rho\colon\Real\times Q\to\Real$ & Canonical projection. \\
$\tau_1\colon\Real \times\Tan Q \rightarrow \Real$ & \\
$\tau_2\colon\Real \times \Tan Q \rightarrow \Tan Q$ & \\
$\tau_0\colon\Real \times \Tan Q \rightarrow Q$ & \\
$\tau_{1,0}\colon\Real \times \Tan Q \rightarrow\Real\times Q$ & 
Canonical projections of the bundle $\Real\times\Tan Q$. \\
$\pi_1\colon\Real \times\Tan^*Q \rightarrow \Real$ & \\
$\pi_2\colon\Real \times \Tan^*Q \rightarrow\Tan^*Q$ & \\
$\pi_0\colon\Real \times \Tan^*Q \rightarrow Q$ & \\
$\pi_{1,0}\colon\Real \times \Tan^*Q \rightarrow\Real\times Q$ & 
Canonical projections of the bundle $\Real\times\Tan^*Q$. \\
$\mbox{\boldmath$\gamma$}\colon\Real \to\Real\times Q$ & 
Canonical lift of a curve $\gamma\colon\Real \to  Q$ to $\Real\times Q$.  \\
$\widehat{\bf c} \colon\Real  \to\Real\times \Tan Q$ &
Canonical lift of a curve ${\bf c}\colon\Real \to\Real\times Q$ to $\Real\times\Tan Q$.  \\
${\cal J}\in{\cal T}^1_1(\Real\times\Tan Q)$ & Extension of the canonical endomorphism to $\Real\times\Tan Q$. \\
 ${\rm F}\Lag\colon\Real\times\Tan Q\to\Real\times\Tan^*Q$ &
 Legendre map defined by $\Lag\in\Cinfty(\Real\times\Tan Q)$. \\
$\vartheta_\Lag\in\df^1(\Real\times\Tan Q)$ &
Cartan Lagrangian $1$-form associated with $\Lag\in\Cinfty(\Real\times\Tan Q)$. \\
$\omega_\Lag=-\d\vartheta_\Lag\in\df^1(\Real\times\Tan Q)$ &
Cartan Lagrangian $2$-form associated with $\Lag\in\Cinfty(\Real\times\Tan Q)$. \\
$R_\Lag\in\vf(\Real\times\Tan Q)$ &
Reeb vector field for the cosymplectic manifold 
$(\Real\times\Tan Q,\d t,\eta_\Lag)$. \\
$\mbox{\boldmath $\Theta$}_{\cal L}\in\df^1(\Real\times\Tan Q)$ &
Poincar\'e--Cartan $1$-form associated with $\Lag\in\Cinfty(\Real\times\Tan Q)$. \\
$\mbox{\boldmath $\Omega$}_{\cal L}=-\d\mbox{\boldmath $\Theta$}_{\cal L}\in\df^1(\Real\times\Tan Q)$ &
Poincar\'e--Cartan $2$-form associated with $\Lag\in\Cinfty(\Real\times\Tan Q)$. \\
 & \\
\hline
\end{tabular}
}
\end{center}

\begin{center}
{\small
\begin{tabular}{|c|c|}
\hline
 & \\
 $\mbox{\boldmath $\Theta$}_{\rm h}\in\df^2(\Real\times\Tan Q$ &
Hamilton--Cartan $1$-form associated with ${\rm h}\in\Cinfty(\Real\times\Tan^*Q)$. \\
$\mbox{\boldmath $\Omega$}_{\rm h}=-\d\mbox{\boldmath $\Theta$}_{\rm h}\in\df^2(\Real\times\Tan^*Q)$ &
Hamilton--Cartan $2$-form associated with ${\rm h}\in\Cinfty(\Real\times\Tan^*Q$. \\
$\nu\colon\Tan(\Real \times Q)\rightarrow \Real\times\Tan Q$ & \\
$\varpi\colon\Tan(\Real \times Q)\rightarrow \Real$ &
Canonical projections of the bundle $\Tan(\Real \times Q)$. \\
$pr_1\colon\Tan^*(\Real \times Q)\rightarrow \Tan^*\Real$ & \\
$pr_2\colon\Tan^*(\Real \times Q)\rightarrow\Tan^*Q$ & \\
$\mu\colon\Tan^*(\Real \times Q)\rightarrow\Real\times\Tan Q$ & \\
$u\colon\Tan^*(\Real \times Q)\rightarrow\Real$ & 
Canonical projections of the bundle $\Tan^*(\Real\times Q)$. \\
 ${\mathbf \Theta}_{{\cal L}_{ext}}\in\df^1(\Tan (\Real\times Q))$ &
Lagrangian $1$-form associated with $\Lag_{ext}\in\Cinfty(\Tan(\Real\times Q))$. \\
${\mathbf \Omega}_{{\cal L}_{ext}}=-\d{\mathbf \Theta}_{{\cal L}_{ext}}\in\df^2(\Tan (\Real\times Q))$ &
Lagrangian $2$-form associated with $\Lag_{ext}\in\Cinfty(\Tan(\Real\times Q))$. \\
 ${\mathbf \Theta}_{ext}\in\df^1(\Tan^*(\Real\times Q))$ &
Canonical $1$-form on $\Tan^*(Q\times\Real)$. \\
$\mbox{\boldmath $\Omega$}_{ext}=-\d\mbox{\boldmath $\Theta$}_{ext}\in\df^2(\Tan^*(\Real\times Q))$  &
Canonical $2$-form on $\Tan^*(Q\times\Real)$. \\
 ${\cal M}=\Real\times\Tan Q\times_Q\Tan^*Q$ &
Extended unified (Pontryagin) bundle. \\ 
$\kappa_1\colon{\cal M}\to\Real\times\Tan Q$ & \\
$\kappa_2\colon{\cal M}\to\Real\times\Tan^*Q$ & \\
$\kappa_0\colon{\cal M}\to\Real\times Q$ &
Canonical projections of the extended unified  bundle ${\cal M}$. \\ 
 $\Theta_{\cal M}\in\df^1({\cal M})$ &
Canonical $1$-form on the extended unified bundle ${\cal M}$. \\ 
$\Omega_{\cal M}=-\d\Theta_{\cal M}\in\df^2({\cal M})$ &
Canonical $2$-form on the extended unified bundle ${\cal M}$. \\ 
${\rm C}\colon{\cal M}\to\Real$ &
Coupling function on the unified bundle ${\cal M}$. \\
${\rm H}\colon{\cal M}\to\Real$ &
Hamiltonian function on the unified bundle ${\cal M}$. \\
$X_{\rm H}\in\vf({\cal M})$ &
Dynamical vector field of the dynamical system $({\cal M},\Omega_{\cal M},{\rm H})$. \\
$\jmath_0\colon\mathcal{M}_0\hookrightarrow\mathcal{M}$ &
Compatibility submanifold of the system $({\cal M},\Omega_{\cal M},{\rm H})$. \\
${\tt g}\in{\cal T}_2(Q)$ & Riemannian metric on a manifold $Q$. \\
${\rm F}\in\vf (Q)$ & Vector force field. \\
$\omega\in\df^1 (Q)$ & Work differential form. \\
$(Q,{\tt g})$ & Riemannian manifold. \\
$(Q,{\tt g},{\rm F})$ , $(Q,{\tt g},\omega)$ & Newtonian system. \\
$\nabla$ & Connection on a manifold (Levi-Civita). \\
$\nabla_X$ & Covariant derivative by a vector field $X\in\vf(M)$. \\
$\nabla_t$ & Covariant derivative along a curve. \\
${\mathtt T}\in{\mathcal T}_2^1(M)$ & Torsion tensor of a connection. \\
${\mathit R}\in{\mathcal T}_3^1(M)$ & Curvature tensor of a connection. \\
${\rm Rie}\in{\mathcal T}_4(M)$ & Riemann's curvature tensor of a Riemannian manifold. \\
${\rm Ric}\in{\mathcal T}_2(M)$ & Ricci tensor of a Riemannian manifold. \\
${\rm S}\in\Cinfty(M)$ & Scalar curvature of a Riemannian manifold. \\
$\bmeta$\ ,\ $(M,\bmeta)$ & (Pre)contact form, (pre)contact manifold. \\
$\mathcal{D}^{\rm C}$\ ,\ $\mathcal{D}^{\rm R}$ &
Contact and Reeb distributions on a contact manifold $(M,\bmeta)$.\\
$\Reeb\in\vf(M)$ & Reeb vector field on a contact manifold $(M,\bmeta)$. \\
$\flat_{\bmeta}\colon\Tan M\to\Tan^*M$ & \\
$\sharp_{\bmeta}=\flat_{\bmeta}^{-1}\colon\Tan^* M\to\Tan M$ & Canonical isomorphisms on a contact manifold $(M,\bmeta)$. \\
$\X_f\in\vf(M)$ & Hamiltonian vector field on a contact manifold $(M,\bmeta)$. \\
${\bf grad}\,f\in\vf(M)$ & Gradient vector field on a contact manifold $(M,\bmeta)$. \\
$\evo_f\in\vf(M)$ & Evolution vector field on a contact manifold $(M,\bmeta)$. \\
$\tau_1\colon\Tan Q\times\Real\rightarrow\Tan Q$ & 
Canonical projection of the bundle $\Tan Q\times\Real$. \\
 $\mathfrak{F}\Lag\colon\Tan Q\times\Real\to\Tan^*Q\times\Real$ &
 Legendre map defined by $\Lag\in\Cinfty(\Tan Q\times\Real)$. \\
 & \\
\hline
\end{tabular}
}
\end{center}

\begin{center}
{\small
\begin{tabular}{|c|c|}
\hline
 & \\
 $\theta_\Lag\in\df^1(\Tan Q\times\Real)$ &
Cartan Lagrangian $1$-form associated with $\Lag\in\Cinfty(\Tan Q\times\Real)$. \\
$\omega_\Lag=-\d\theta_\Lag\in\df^2(\Tan Q\times\Real)$ &
Cartan Lagrangian $2$-form associated with $\Lag\in\Cinfty(\Tan Q\times\Real)$. \\
 $\bmeta_\Lag=\d s-\theta_\Lag\in\df^1(\Tan Q\times\Real)$ &
(Pre)contact Lagrangian form associated with $\Lag\in\Cinfty(\Tan Q\times\Real)$. \\
$\Reeb_\Lag\in\vf(\Tan Q\times\Real)$ &
Reeb vector field for the (pre)contact manifold 
$(\Tan Q\times\Real,\bmeta_\Lag)$. \\
 $\mathfrak{M}=\Tan Q\times_Q\Tan^*Q\times\Real$ &
Extended precontact unified (or Pontryagin) bundle. \\ 
$\bmeta_{\mathfrak{M}}\in\df^1(\mathfrak{M})$ &
Canonical precontact form on the precontact unified bundle $\mathfrak{M}$. 
\\ 
$\mathfrak{C}\colon\mathfrak{M}\to\Real$ &
Coupling function on the unified bundle $\mathfrak{M}$. \\
${\bf H}\in\Cinfty(\mathfrak{M})$ &
Hamiltonian function on the precontact unified bundle $\mathfrak{M}$.\\
$\X_{\bf H}\in\vf(\mathfrak{M})$ &
Dynamical vector field of the dynamical system $(\mathfrak{M},\Omega_\mathfrak{M},{\bf H})$. \\
$\jmath_0\colon\mathfrak{M}_0\hookrightarrow\mathfrak{M}$ &
Compatibility submanifold of the dynamical system $(\mathfrak{M},\Omega_\mathfrak{M}
,{\bf H})$. \\
  & \\
\hline
\end{tabular}
}
\end{center}

%%%%%%%%%%%%%%%%%%%%%%%%%%%%%%%%%%%%%%%%%%%%%%%%%%%%%%%%%%%%%%%%%%%%%%%%%%%%%%%%%%%%%%%%%%%%%%%%%%%%%%%%%%%%%%%%%%%%%%%%%%%

\chapter{Introduction}

\pagenumbering{arabic}
\setcounter{page}{1}

The study of mechanics experienced a substantial advance during 
the 18th and 19th centuries with the emergence of what was called 
{\sl Analytical} or {\sl Rational Mechanics}.
Many relevant mathematicians contributed to its development, 
such as {\it L. Euler, W.R. Hamilton, C.G.J. Jacobi, J.L. Lagrange, A.M. Legendre, S.D. Poisson}, among others.
Their techniques are essentially based in the application of variational methods 
to obtain the equations of motion of dynamical systems, and their subsequent development;
and are widely collected in a lot of classical treaties, such as
\cite{FLS-2005,Ga-70,GPS-01,KB-2004,La-70,LL-76,Wh-37}
or, in a more modern perspective, \cite{Arovas2014,Ca-96,Co-99,De-2010,JS-98,Le-2014,Sch-2005,SM-74,Su-2013,Wo-2009},
among many others.

Although the dynamical equations of Analytical Mechanics had been well-stablished
previously by {\it J.L. Lagrange} and others,
the connection between the variational methods and Mechanics was done 
by the {\sl minimum action principles}
\cite{Bl-81,Bli,Els,GF-63,La-70,Le-2014}. Thus, the so-called
{\sl Hamilton principle} leads to the Euler--Lagrange equations
and  the so-called {\sl Lagrangian formalism} of mechanics.
Introducing the {\sl Legendre transformation} and the ``momentum coordinates''
the dynamical equations become the {\sl Hamilton equations},
originating in this way the {\sl canonical} or {\sl Hamiltonian formalism}
of Mechanics. As in the Lagrangian case, these Hamiltonian equations
can also be derived from the so-called {\sl Hamilton--Jacobi variational principle}.

In the second half of the 20th century, 
an extensive group of mathematicians and physicists used 
the techniques of differential geometry to formalize analytical mechanics
and physics in general,
and studied the properties of dynamical systems intrinsically,
giving rise to what is known today as {\sl geometric mechanics}.
Among the most relevant books and contributions, we can cite
\cite{AM-78,AMR-88,Ar-89,CIMM-2015,CDW-87,Go-69,GS-77,GS-90,Holm-2008,Ib-96,klein,Lewis,LM-sgam,Li-75,MS-98,MR-99,Ol-02,RTSST-2005,Sch-2005,So-ssd,We-77},
in addition to many others.

In the geometric formulation, 
differentiable manifolds that model the phase spaces of dynamical systems
are endowed with different kinds of geometric structures,
which are used to study the properties of these systems.
Thus, for {\sl conservative systems},
the {\sl symplectic manifolds} are used to describe
{\sl autonomous} or {\sl time-independent} dynamical systems
\cite{AM-78,AA-80,Ar-89,CDW-87,dLe89,Lewis,LM-sgam,Li-75,MR-99,MT-78,So-ssd,VK-75,VK-77,We-77};
meanwhile for the {\sl nonautonomous} or {\sl time-dependent} case
also other types of manifolds can be considered ,
such as {\sl jet bundles} and {\sl contact} and {\sl cosymplectic manifolds} 
\cite{AM-78,AA-80,CLL-92,CLM-91,Cr-95,CPT-hctd,EMR-gstds,MS-98,MV-03,Saunders89}.
In particular, for variational dynamical systems; 
that is, those of {\sl Lagrangian type},
their phase spaces have the geometric structure of a
tangent bundle \cite{Cr-83b,CP-adg,Ga-52,Go-69,klein,Tu-74,Tu-76a}, 
or a cotangent bundle in the dual canonical Hamiltonian formalism
\cite{AM-78,Ar-89,CDW-87,LM-sgam,Tu-74,Tu-76b}.
These formalisms can also be described in a single unified description
\cite{SR-83,SR-83b}.
Furthermore, for those dynamical systems which are called of {\sl mechanical type}, their configuration spaces are also endowed 
with a (semi)Riemannian metric  \cite{AM-78,Ar-89,CC-2005,CDW-87,GN-2014,Ol-02}
which is used to construct their Lagrangian functions in a canonical way.
Finally, symplectic mechanics has also been  used to describe  {\sl nonholonomic systems};
that is, systems subjected to constraints depending on positions and velocities,
\cite{Ar-89,CLMM-2002,LMM,LM-96,Ga-70,Lewis2020,Neimark-Fufaev,Ve};
and also {\sl vakonomic systems}
which consists in considering those systems but modifying the variational principle
using only the variations allowed by the constraints
\cite{Ar-89,Bli,Els,MCL-2000}.

In addition, there are other more general geometric structures
which are used to model some special types of dynamical systems.
These are {\sl Poisson and Jacobi manifolds} \cite{LM-sgam,Marle-83} or also 
{\sl Lie algebroids} \cite{CNS-2007b,LMM-2005,GG-2006,GUG-2006,Edu-2001,MMS-2002,We-96}.
Finally, {\sl dissipative} or {\sl nonconservative dynamical systems} are modelled using {\sl contact manifolds}
\cite{Bravetti2017,BCT-2017,CIAGLIA2018,LGMMR-2020,DeLeon2019,
DeLeon2016b,GGMRR-2019b,KA-2013,LL-2018}.
In all these cases, there are deep relations between the differential equations 
describing the dynamics and the underlying geometric structures.

Another very important fact that characterize the behaviour of dynamical systems is
the existence of symmetries. 
The interest of studying symmetries arises as a consequence of the well known fact
that their existence is closely related to that of conserved quantities,
and the main result comes from the work of {\it Emmy Noether} \cite{No-18}
(see \cite{KS-2011} for a review of her results).
On turn, the existence of conserved quantities
facilitates the integration of  the dynamical equations,
applying the suitable reduction procedures.
Concerning to this, it is interesting to mention the fundamental {\sl Arnold-Liouville Theorem}
for {\sl integrable systems} \cite{Ar-89}
(see also \cite{Ba-99,Mi-2014}), and also those by
{\it J.E. Marsden} and {\it A. Weinstein} on the problem of the {\sl (symplectic) reduction} by symmetries
\cite{MW-74} (see also \cite{MMOPR-2007,MR-99,MW,Or-2002,OR-2002,OR-2004} and the references therein).

The physical systems for which these geometric approaches were initially developed
had the characteristic of being {\sl regular}
\footnote{
For Lagrangian systems, this means, locally, that the Hessian matrix with respect to the generalized velocities is regular everywhere.}.
The dynamical systems in classical analytical mechanics are of these kind.
Nevertheless, with the advent of Relativity Theory and the 
development of field theories in modern physics, in general,
many systems appear for which that property does not hold; that is,
they are {\sl singular systems\/}.
For autonomous systems, this feature of being regular or not manifests in the fact that 
the underlying geometric structure is 
{\sl symplectic} or {\sl presymplectic}, respectively
\cite{Ca-90,CGIR-85,CLR-87,CLR-88,GN-79,GN-80,GNH-78,GP-92,MMT-97,Mu-89,MR-92,Sn-74}.

The aim of this work is to do a detailed exposition of
the geometry underlying dynamical systems and how it is applied
to describe the different kinds of them:  conservative autonomous and nonautonomous dynamical systems in general, and Newtonian systems in particular; as well as dissipative systems.
We use different geometric descriptions to study
the main properties and characteristics of these systems;
such as their Lagrangian, Hamiltonian and unified formalisms, 
their symmetries, the variational principles, and others.
The study is done mainly for the regular case,
although comments and explanations about singular systems 
are also introduced throughout the exposition and in some proposed problems.

The organization of the work is the following:

The first chapter of the exposition, Chapter \ref{ch2}, is devoted
to present the most general geometric framework for autonomous mechanics;
that is, {\sl symplectic manifolds} and their structures and characteristics
({\sl Poisson brackets, symplectomorphisms,} etc.),
as well as the fundamental notion of {\sl Hamiltonian dynamical system}.
A general description of {\sl symmetries, conserved quantities} and the 
{\sl Noether theorem} is also done for this general situation.

In Chapter \ref{sdl}, we continue the exposition about {\sl symplectic mechanics}
studying, in particular, autonomous mechanical systems of variational type
which are characterized by being described by means of time-independent Lagrangian functions.
The phase spaces of these kinds of systems are represented by tangent
and cotangent bundles of manifolds representing 
the configuration spaces of the systems; so the canonical geometric structures
of these bundles are introduced first.
Starting from the Lagrangian function and using the canonical structures of the tangent bundle,
we can define a (pre)symplectic form and the {\sl Legendre map}
which allow us to establish the Lagrangian formalism,
the associated canonical Hamiltonian formalism for these systems,
and the equivalence between both formalisms.
We also state the foundations of a new geometric setting of the {\sl Hamilton--Jacobi theory}
for Hamiltonian systems.
In addition, we also give a description of the so-called {Skinner-Rusk unified formalism}, 
which combines in a single framework the two previous formalisms.
The study of symmetries, conserved quantities and Noether's theorem is
also performed for Lagrangian systems in both formalisms,
which lead to introduce some new types of symmetries.
As these dynamical systems are of variational type,
a section is devoted to introduce the variational formulation
and derive the dynamical equations from the corresponding variational principle.
Finally, two of the most classic examples in mechanics,
the {\sl harmonic oscillator\/} and the {\sl Kepler problem},
are analyzed.

Chapter \ref{chap:cosym} is devoted to develop one of the most interesting
and generic geometric formulations for describing 
{\sl nonautonomous dynamical systems}, 
both the Lagrangian and Hamiltonian formalisms.
It is based in using {\sl cosymplectic manifolds},
which are presented in the first section of this chapter.
Then, the dynamics and symmetries for these time-dependent systems are studied
and the formalism is used to describe the classical examples introduced in the previous chapter
when external time-dependent forces act on the oscillator
or for variable mass systems in the case of the Kepler problem.
In this chapter we also give a brief presentation
of two other very common formulations of nonautonomous mechanics;
namely, the {\sl contact} and the {\sl extended symplectic formulations}, 
and showing their equivalence with the cosymplectic picture.

Next, Chapter \ref{sdn}
deals with the study of those dynamical systems whose configuration spaces are
endowed with additional geometric structures, such as a {\sl metric}.
First, we review the foundations of connections in manifolds and Riemannian geometry,
which are the geometric structures needed to develop what is called {\sl Newtonian mechanics}.
In particular, the existence of a
metric allows us to construct what are known as
 {\sl mechanical Lagrangians}. We analyze
different types of them; namely, {\sl conservative} and {\sl coupled systems}, and
systems with {\sl holonomic} and {\sl nonholonomic constraints};
stating their dynamics and variational formulation.
As a particular situation, the case of {\sl (nonautonomous) Newtonian systems} is
displayed along this exposition.

The last chapter, Chapter \ref{chap:contactmech}, is an introduction to the study of {\sl autonomous dissipative systems},
using the {\sl contact geometry}.
After reviewing the foundations on {contact manifolds},
we establish the generic concepts about {\sl contact Hamiltonian systems}. 
and develop the Lagrangian, Hamiltonian and unified Lagrangian-Hamiltonian formalisms for these kinds of systems.
Next, we study symmetries and the concept of {\sl dissipated quantities} for contact Hamiltonian and Lagrangian systems;
establishing the so-called ``dissipation theorems'',
and showing how to associate dissipated and conserved quantities to these symmetries.
As in the above chapters, the examples of the damped harmonic oscillator and the Kepler problem with friction
are analyzed in this context.

The work ends with an appendix where various contents on other
geometric structures that appear throughout the exposition are collected;
in particular, tangent and cotangent bundles, and Lie groups and Lie algebras.

This work is intended for readers who have completed at least 
the first courses of a degree on mathematics or physics, 
and who have a basic training on differential geometry of smooth manifolds and analytical mechanics.

Along the exposition, all manifolds are supposed to be real, second-countable and $\Cinfty$. 
All the maps and the structures are smooth.  
The summation criterion for repeated crossed indices is adopted.

%%%%%%%%%%%%%%%%%%%%%%%%%%%%%%%%%%%%%%%%%%%%%%%%%%%%%%%%%%%%%%%%%%%%%%%%%%%%%%%%%%%%%%%%%%%%%%%%%%%%%%%%%%%%%%%%%%

\chapter{Symplectic mechanics (I): Autonomous Hamiltonian dynamical systems}
\label{ch2}

The general geometric framework for describing autonomous mechanical systems (with a finite number of degrees of freedom)
uses some particular types of differentiable manifolds to model the phase spaces of these systems,
they are the {\sl symplectic} and {\sl presymplectic manifolds}.
This formulation is known as {\sl symplectic mechanics}. Taking these kinds of manifolds as phase spaces, 
we have a very general setting to study dynamical systems,
from which other descriptions such as the Lagrangian and the Hamiltonian formalisms of Lagrangian systems can be analyzed as particular situations
\footnote{
This general formulation allows us to describe also
dynamical systems which are not of Lagrangian type,
for instance, the system of a classical spin particle \cite{So-ssd}.}.

The symplectic description of (autonomous) Hamiltonian systems has been exposed in many works and books
(see, for instance, \cite{AM-78,Ar-89,dLe89,GS-90,art:Krupkova00,LM-sgam,Li-75,MR-99,So-ssd,We-77} and the references quoted therein).

In this chapter, after reviewing the fundamental concepts on symplectic geometry,
we present the symplectic description of the autonomous Hamiltonian systems
and we introduce several types of symmetries and their associated conserved quantities
from a geometric perspective.

\section{Notions on symplectic and presymplectic geometry}
\label{sympgeom}

In this section, we state the fundamental concepts and properties 
of symplectic (and presymplectic) manifolds
(see, for instance, \cite{AM-78,ABKLR-2012,CdS-2008,GS-77,KY-2019,LM-sgam,We-77}).

\subsection{Symplectic and presymplectic vector spaces}

Let $\textbf{E}$ be a finite dimensional real vector space.

\begin{definition}
A \textbf{symplectic inner product}, or a
\textbf{linear symplectic structure}, on $\textbf{E}$ is a non--degenerate skew symmetric bilinear function $\omega$ on $\textbf{E}$.

We say that the pair $(\textbf{E},\omega)$ is a \textbf{symplectic vector space}. 
\end{definition}

Non degeneracy means that, if $\omega(\bf{x},\bf{y})=0$, for every $\bf{y}\in \textbf{E}$, then $\mathbf{x}=0$. 
Let $\mathbf{e}_{1},\ldots,\mathbf{e}_{r}$ be a basis of $\textbf{E}$ and $\mbox{\boldmath $\alpha$}^{1},\ldots,\mbox{\boldmath $\alpha$}^{r}$ its dual basis. If $\omega_{ij}=\omega(\mathbf{e}_{i},\mathbf{e}_{j})$, then the expression of $\omega$ on these bases is
$$
\omega=\omega_{ij}\,\mbox{\boldmath $\alpha$}^{i}\otimes\mbox{\boldmath $\alpha$}^{j}\ ,
$$
being $(\omega)=(\omega_{ij})$ the matrix of $\omega$ in these bases.

The linear mapping $\omega^{\flat}:\textbf{E}\to \textbf{E}^{*}$ is defined by
$$
(\omega^{\flat}(\mathbf{x}))(\mathbf{y})=\langle\omega^{\flat}(\mathbf{x}),\mathbf{y}\rangle=\omega(\mathbf{x},\mathbf{y})\ ,
$$
and its matrix relative to these bases is  $(\omega_{ij})$,
since $\omega(\mathbf{e}_{i})=\omega_{ij}\mbox{\boldmath $\alpha$}^{j}$.
By the skew symmetry of the bilinear mapping $\omega$, we have that the matrix $(\omega_{ij})$ is skew symmetric. The non degeneracy gives us that the mapping $\omega^{\flat}$ is one--to--one, hence an isomorphism, and $(\omega_{ij})$ is regular. 
As a consequence, the dimension of $\textbf{E}$ is even, that is $r=2n$, because $(\omega)=-(\omega)^{t}$ (the transpose matrix),
and $\det(\omega)=\det(\omega)^{t}=(-1)^{r}\det(\omega)$.

A basis $\mathbf{a}_{1},\ldots,\mathbf{a}_{n}, \mathbf{a}_{n+1},\ldots,\mathbf{a}_{2n}$ of $\textbf{E}$ is called \textsl{symplectic} if
$$
\begin{matrix}
 \omega(\mathbf{a}_{i},\mathbf{a}_{j})=0 \ ;&  i,j=1,\ldots,n \ ; \\ 
\omega(\mathbf{a}_{i},\mathbf{a}_{j})=0 \ ;&  i,j=n+1,\ldots,2n\\
\omega(\mathbf{a}_{i},\mathbf{a}_{n+j})=\delta_{ij} \ ;& i,j=1,\ldots,n \ .
 \end{matrix}\ .
$$

The existence of symplectic basis is given by the following Lemma: 

\begin{lem}
Let $({\bf E},\omega)$ be a $2n$ dimensional symplectic vector space.
\begin{enumerate}
\item There exists a symplectic basis $({\bf e}_k)$ ($k=1,\ldots ,2n$) on ${\bf E}$.

\item If $(\mbox{\boldmath $\alpha$}^k)$ is the corresponding dual basis on ${\bf E}^*$,
then
\(\dst\omega=\sum_{i=1}^n\mbox{\boldmath $\alpha$}^i\wedge\mbox{\boldmath $\alpha$}^{i+n}\) or,
what is the same, the matrix of $\omega$ in this basis is
$$
\left(\begin{matrix} 0_n & I_n \\ -I_n & 0_n \end{matrix}\right)
$$
where $I_n$ denotes the identity matrix of order $n$ and $0_n$ the zero square matrix $n\times n$.
\end{enumerate}
\label{lemaaux}
\end{lem}
\begin{proof}
By induction on the dimension of ${\bf E}$. Let $\mathbf{e}_{1},\mathbf{e}_{n+1}\in{\bf E}$ with  $\omega(\mathbf{e}_{1},\mathbf{e}_{n+1})\neq 0$. We can choose two such vectors unless $\omega=0$. Dividing $\mathbf{e}_{1}$ by a scalar, we have that $\omega(\mathbf{e}_{1},\mathbf{e}_{n+1})=1$. Then, on the plane ${\bf P}_{1}$ spanned by $\mathbf{e}_{1},\mathbf{e}_{n+1}$ the matrix of $\omega$ is 
$$
\left(\begin{matrix} 0 & 1 \\ -1 & 0 \end{matrix}\right)\, .
$$
Let ${\bf E}_{1}$ the $\omega$--orthogonal complement of ${\bf P}_{1}$ in ${\bf E}$, that is:
$$
{\bf E}_{1}=\{{\bf a}\in{\bf E}\,|\, \omega({\bf a},{\bf a}_{1})=0, \forall{\bf a}_{1}\in{\bf P}_{1}\}\ .
$$
Observe that ${\bf E}_{1}\cap{\bf P}_{1}=\{0\}$ and ${\bf E}_{1}+{\bf P}_{1}={\bf E}$. If $\mathbf{a}\in{\bf E}$, we have
$$
\mathbf{a}-\omega({\bf a},{\bf e}_{n+1}){\bf e}_{1}+\omega({\bf a},{\bf e}_{1}){\bf e}_{n+1}\in{\bf E}_{1}\ .
$$
Then, ${\bf E}_{1}\oplus{\bf P}_{1}={\bf E}$ and we can repeat the process on ${\bf E}_{1}$, with dimension less that $2n$, and we have finished the proof.

The result on the expression of $\omega$ as \(\dst\omega=\sum_{i=1}^n\mbox{\boldmath $\alpha$}^i\wedge\mbox{\boldmath $\alpha$}^{i+n}\) is immediate.
\\ \qed  \end{proof}

\begin{remark}{\rm  
\begin{itemize}
\item
The non degeneracy condition is equivalent to say that $\omega^{n}=\omega\overbrace{\wedge\ldots\wedge}^{n}\omega\equiv\wedge^n\omega$ is a volume form on ${\bf E}$.
\item
If we take out the non degeneracy condition, then the mapping $\omega^{\flat}$ is not one--to--one. In this case, we need to add the basis $\mathbf{u}_{1},\ldots,\mathbf{u}_{h}$ of $\ker \omega^{\flat}$ to the given basis for the decomposition of ${\bf E}$. Then the corresponding matrix is,
$$
\left(\begin{matrix} 0_n & I_n&0_{n\times h} \\ -I_n & 0_n& 0_{n\times h}\\0_{h\times n} & 0_{h\times n}&0_{h}
\end{matrix}\right)
$$
being $2n+h$ the dimension of ${\bf E}$ and denoting by $0_{n\times h}$ the zero matrix of $n$ rows and $h$ columns. We say that $2n$ is the rank of $\omega$. 
\end{itemize}
}\end{remark}

\subsection{Subspaces of a symplectic linear space}

Associated to a linear symplectic structure we have a notion of orthogonality in parallel to the ideas related to the Euclidean scalar product, but the results are very different. In the next lines, we develop these ideas that are necessary to understand our later description of Lagrangian and Hamiltonian systems.

Let $({\bf E},\omega)$ be a $2n$-dimensional symplectic vector space and ${\bf F}\subset {\bf E}$ a linear subspace. The {\sl \textbf{$\omega$--orthogonal complement}} of ${\bf F}$ is defined as
$$
{\bf F}^\bot=\{{\bf u}\in{\bf E}\ |\ \omega({\bf u},{\bf u}')=0\, , \, \mathrm{for\,\, every}\,\, {\bf u}'\in{\bf F}\}.
$$
Observe that, in general, ${\bf F}\cap{\bf F}^\bot\neq 0$ (for example, if ${\bf F}=\mathrm{span}\,\{{\bf u}\}$, then ${\bf F}\subset{\bf F}^\bot$).

\begin{definition}
Let ${\bf F}$ be a subspace of a linear symplectic $2n$--dimensional vector space $({\bf E},\omega)$. 
%We say:
\begin{enumerate}
\item $ {\bf F}$ is \textbf{isotropic} if ${\bf F}\subset  {\bf F}^\bot$; that is, $\omega({\bf u},{\bf u}')=0$, for every ${\bf u},{\bf u}'\in{\bf F}$.
\item $ {\bf F}$ is \textbf{coisotropic} if ${\bf F}\supset  {\bf F}^\bot$; that is, if 
$\omega({\bf u},{\bf u}')=0$, for every ${\bf u}'\in{\bf F}$. Then ${\bf u}\in{\bf F}$.
\item $ {\bf F}$ is \textbf{Lagrangian} if ${\bf F}$ is isotropic and there exists an isotropic complement ${\bf F}'$; that is, an isotropic subspace ${\bf F}'\subset{\bf E}$ such that ${\bf F}\oplus{\bf F}'={\bf E}$.
\item $ {\bf F}$ is \textbf{symplectic} if $\omega$ restricted to ${\bf F}$ is non degenerated; that is, $({\bf F},\omega_{\bf F}=\left.\omega\right|_{\bf F})$ is a symplectic vector space.
\end{enumerate}
\end{definition}

These notions have several associated properties that we resume in the following two propositions, with some proofs and indications. 

\begin{prop}
\label{propos1}
Let $({\bf E},\omega)$ be a symplectic vector space and 
${\bf F},{\bf G}$ subspaces of ${\bf E}$.
\begin{enumerate}
\item If ${\bf F}\subset{\bf G}$, then ${\bf F}^\bot\supset{\bf G}^\bot$\,.
\item  ${\bf F}^\bot\cap{\bf G}^\bot=({\bf F}+{\bf G})^\bot$.
\item $\dim{\bf E}=\dim{\bf F}+\dim {\bf F}^\bot$.
\item ${\bf F}= {\bf F}^{\bot\bot}$.
\item $({\bf F}\cap{\bf G})^\bot={\bf F}^\bot+{\bf G}^\bot$.
\end{enumerate}
\end{prop}
\begin{proof} 
\begin{description}
\item[{\rm 3}.] 
Consider the linear map $\omega^{\flat}\colon\textbf{E}\to \textbf{E}^{*}$. 
Observe that, if ${\bf u}\in{\bf F}$, then $\omega^{\flat}({\bf u})\in\textbf{E}^{*}$ 
annihilates the subspace ${\bf F}^\bot$. Hence, the restricted linear map 
$\omega_{\bf F}^{\flat}\colon\textbf{F}\to \textbf{E}^*$ induces another one 
$\hat{\omega}_{\bf F}^{\flat}\colon\textbf{F}\to\left(\textbf{E}/{\bf F}^\bot\right)^{*}$. 
This last map is injective, but not necessarily onto. Thus, we have
$$
\dim{\bf F}\leq\dim\left(\textbf{E}/{\bf F}^\bot\right)^{*}=\dim\left(\textbf{E}/{\bf F}^\bot\right)=\dim\textbf{E}-\dim{\bf F} \ .
$$
Conversely, consider the map 
${\bf F}\stackrel{\omega^{\flat}}{\longrightarrow}\textbf{E}^{*}\stackrel{j}{\longrightarrow}{\bf F}^{*}$, 
where $j$ is the restriction from $\textbf{E}$ to ${\bf F}$. 
Let  $\bar{\omega}_{\bf F}^{\flat}=\omega^{\flat}\circ j$ and observe that 
$\ker\bar{\omega}_{\bf F}^{\flat}={\bf F}^\bot$. Hence,
$$
\dim{\bf F}=\dim{\bf F}^{*}\geq\dim(\mathrm{ img}\,\bar{\omega}_{\bf F}^{\flat})=
\dim\textbf{E}-\dim{\bf F}^{*}=\dim\textbf{E}-\dim{\bf F} \ .
$$
From both expressions, we obtain the result we wanted.
\item[{\rm 4}.] 
It is clear that ${\bf F}\subset{\bf F}^{\bot\bot}$. 
If we apply the previous result to ${\bf F}$ and to ${\bf F}^\bot$ we have
$$
\dim{\bf E}=\dim{\bf F}+\dim {\bf F}^\bot=\dim{\bf F}^\bot+\dim {\bf F}^{\bot\bot} \ ;
$$
that is, $\dim{\bf F}=\dim {\bf F}^{\bot\bot}$, hence ${\bf F}={\bf F}^{\bot\bot}$ as we wanted.
\item[{\rm 5}.] 
From 2. and 4. we have
$$
({\bf F}\cap{\bf G})^\bot\stackrel{4.}{=}({\bf F}^{\bot\bot}\cap{\bf G}^{\bot\bot})^\bot\stackrel{2.}{=}
({\bf F}^\bot\cap{\bf G}^\bot)^{\bot\bot}\stackrel{4.}{=}{\bf F}^\bot\cap{\bf G}^\bot  \ .
$$
\end{description}
\qed  \end{proof}

And the next proposition gives alternative definitions for a subspace to be Lagrangian.

\begin{prop}
Let $({\bf E},\omega)$ be a symplectic vector space $({\bf E},\omega)$ and ${\bf F}$ a subspace of ${\bf E}$. The following statements are equivalent:
\begin{enumerate}
\item ${\bf F}$ is Lagrangian.
\item ${\bf F}={\bf F}^\bot$.
\item ${\bf F}$ is isotropic and $\dim{\bf F} =\frac{1}{2} \dim{\bf E}$.
\end{enumerate}
\end{prop}
\begin{proof} 
(1. $\Longrightarrow$  2).
As ${\bf F}$ is  Lagrangian it is isotropic, hence ${\bf F}\subset{\bf F}^\bot$, and there exists another isotropic subspace ${\bf F}'$ with ${\bf F}\oplus{\bf F}'={\bf E}$.

Suppose now that ${\bf u}\in{\bf F}^\bot$ and put ${\bf u}={\bf a}+{\bf b}$ with ${\bf a}\in{\bf F}$, ${\bf b}\in{\bf F}'$. We prove that ${\bf b}=0$. In fact, ${\bf b}\in{\bf F}'^{\bot}$ since ${\bf F}'$ is isotropic. We also have that ${\bf b}={\bf u}-{\bf a}\in{\bf F}^\bot$ since  ${\bf u},{\bf a}\in{\bf F}$ and ${\bf F}$ is isotropic. Then 
$$
{\bf b}\in{\bf F}'^{\bot}\cap{\bf F}^\bot=({\bf F}'+{\bf F})^\bot={\bf E}^\bot=\{0\}\ ,
$$
because $\omega$ is non--degenerate. Then we have that ${\bf u}={\bf a}\in{\bf F}$ and 
${\bf F}^\bot\subset{\bf F}$.
%Thus from 1. we obtain 2.

(2. $\Longrightarrow$  3).
The hypothesis and the item 3 of the above proposition give the result.

(3. $\Longrightarrow$  2).
From 3 and the above proposition we obtain that $\dim{\bf F}=\dim{\bf F}^\bot$ and being ${\bf F}$ isotropic, we have that ${\bf F}={\bf F}^\bot$; that is,
the statement 2.

(2. $\Longrightarrow$  1).
We have ${\bf F}={\bf F}^\bot$, hence ${\bf F}$ is isotropic. We need to construct a subspace ${\bf F}'$.
Let ${\bf a}_{1}\notin{\bf F}$ and ${\bf F}_{1}=\mathrm{span}\{{\bf a}_{1}\}$. Then 
${\bf F}\cap{\bf F}_{1}=\{0\}$, hence 
${\bf F}^\bot+{\bf F}_{1}^\bot={\bf F}\oplus{\bf F}_{1}^\bot={\bf E}$, that is ${\bf F}\oplus{\bf F}_{1}^\bot={\bf E}$ by the above proposition.
Now let ${\bf a}_{2}\notin{\bf F}+{\bf F}_{1}$, ${\bf a}_{2}\in{\bf F}_{1}^\bot$. We have two alternatives:
\begin{enumerate}
\item This vector ${\bf a}_{2}$ does not exist. In this case ${\bf F}_{1}^\bot\subset{\bf F}+{\bf F}_{1}$, hence ${\bf F}+{\bf F}_{1}={\bf E}$ and
$$
{\bf F}_{1}^\bot=(\mathrm{span}\{{\bf a}_{1}\})^\bot\supset\mathrm{span}\{{\bf a}_{1}\}={\bf F}_{1}\ ,
$$
then ${\bf F}_{1}$ is isotropic and we can take ${\bf F}'={\bf F}_{1}$.
\item  
There exists such ${\bf a}_{2}$. Let 
 ${\bf F}_{2}={\bf F}_{1}+\mathrm{span}\{{\bf a}_{2}\}$, following the same procedure as above, we have that 
${\bf F}\cap{\bf F}_{2}=\{0\}$, hence 
${\bf F}^\bot+{\bf F}_{2}^\bot={\bf F}\oplus{\bf F}_{2}^\bot={\bf E}$, that is ${\bf F}\oplus{\bf F}_{2}^\bot={\bf E}$. But
$$
{\bf F}_{2}^\bot=(\mathrm{span}\{{\bf a}_{1},{\bf a}_{2}\})^\bot
\supset\mathrm{span}\{{\bf a}_{1},{\bf a}_{2}\}={\bf F}_{2} \ ,
$$
hence ${\bf F}_{2}$ is isotropic and we can take ${\bf F}'={\bf F}_{2}$.
\end{enumerate}
Inductively we can continue and, as $\dim{\bf E}$ is finite, we arrive to some ${\bf F}_{k}$ such that ${\bf F}'={\bf F}_{k}$ and we have finished the proof.
\\ \qed  \end{proof}

This last proposition says that a Lagrangian subspace is a \textsl{maximal isotropic subspace}.

\subsection{Symplectic and presymplectic manifolds}
\protect\label{vs}

Bearing in mind the results of the above section, first we define:

\begin{definition}
\label{vsm}
\begin{description}
\item[{\rm (a)}]
Let $M$ be a differentiable manifold.
A \textbf{symplectic form}
\footnote{
The word {\sl  symplectic} comes from the Greek word
``$\sigma\iota\mu\pi\lambda\epsilon\kappa\tau\iota\kappa o$''  which means ``what unites''.
It was introduced by H. Weyl, who substituted
the Latin root of the term ``complex'',
to refer to a structure of the group
${\bf Sp}(n,\Complex )$.
Although the structure of the symplectic manifolds
had been implicitly considered before, 
it is not until the decade of 1950 when the symplectic geometry appears
as a differentiated branch of Differential Geometry, being
A. Lichnerowicz the first in introducing the term
{\sl  symplectic manifold}.}
in $M$ is a differential $2$-form
$\Omega \in {\mit\Omega}^2(M)$ such that:
\begin{enumerate}
\item
It is closed: $\d\, \Omega = 0$, (we write $\Omega \in Z^2(M)$).
\item
It is non-degenerated at every point of $M$; that is, $\Omega_{{\rm p}}$ is a linear symplectic structure in $\Tan_{\rm p}M$ for every ${\rm p}\in M$.
\end{enumerate}
If the form $\Omega$ is closed but degenerated it is said to be a
\textbf{presymplectic form}, and if $\Omega$ is nondegenerated but not closed
it is an \textbf{almost-symplectic form}.
\item[{\rm (b)}]
A \textbf{symplectic} (resp. \textbf{presymplectic}) \textbf{manifold}
 is a pair $(M,\Omega )$ where $M$ is a differentiable manifold
and $\Omega$ is a symplectic (resp. presymplectic) form.

If the symplectic (resp. presymplectic)  form is exact;
that is, there exists
$\Theta \in {\mit\Omega}^1(M)$ such that $\d\Theta = \Omega$,
then $\Omega$ is an \textbf{exact symplectic} (resp. \textbf{exact presymplectic}) 
\textbf{form}, and $(M,\Omega )$ is an
\textbf{exact symplectic} (resp. \textbf{exact presymplectic}) \textbf{manifold}.
The form $\Theta$ is called a \textbf{symplectic} (resp. \textbf{presymplectic}) \textbf{potential}.
\end{description}
\end{definition}

\begin{remark}{\rm 
\begin{itemize}
\item
As a consequence of the second condition of the definition, a non-degenerated differential $2$-form
can only be defined in manifolds of even dimensi\'on,
thus we write $\dim M = 2n$.
\item
As, since {\sl Poincar\'e's Lemma}, every closed form is locally exact,
if $\Omega$ is a symplectic or a presymplectic form, for every point
${\rm p}\in M$, there is an open neighbourhood
$U \subset M$, ${\rm p}\in U$, and $\vartheta \in {\mit\Omega}^1(U)$ such that
$\Omega\mid_U=\d\vartheta$.
Every $1$-form $\vartheta$ satisfying this condition is called a
{\sl local symplectic ({\rm or} presymplectic) potential}.

Observe that if $\vartheta \in {\mit\Omega}^1(U)$
and $\vartheta ' \in {\mit\Omega}^1(U')$ are two different symplectic
(or presymplectic) potentials, then $\vartheta ' = \vartheta + \d f$
in $U\cap U'$, for some $f \in \Cinfty(U\cap U')$.
\item 
In the case of a presymplectic form, the dimension of $\ker \omega_{\rm p}^{\flat}$, for ${\rm p}\in M$, may depend on the chosen point ${\rm p}\in M$. Usually, we demand that this dimension does not depend on the point and say that the presymplectic form is {\sl regular}.
\end{itemize}
}\end{remark}

The following theorem describes the local structure of the symplectic manifolds \cite{Darboux}.

\begin{teor}
{\rm (Darboux)} \
Let $(M,\Omega )$ be a $2n$-dimensional symplectic manifold.
For every point ${\rm p}\in M$ there exists an open neighbourhood
$U \subset M$, ${\rm p}\in U$, 
which is the domain of a local chart $(U;x^i,y_i)_{i=1 \ldots n}$,
such that $\Omega$ has the expression
$$
\Omega \mid_U = \d x^i \wedge \d y_i \ .
$$
These local charts are called {\sl \textbf{symplectic charts}}
and their coordinates are the {\sl \textbf{canonical coordinates}} or
 {\sl \textbf{Darboux coordinates}} of the symplectic manifold in this chart.
\end{teor}
\begin{proof} 
(This proof is taken from \cite{AM-78}. For other different proofs, see for example
\cite{book:Bryantetal, CP-adg, dLGRR-2023}).
The proof is organized in several parts:
\begin{enumerate}
\item As this is a local result, we can suppose that $M$ is $\Real^{2n}$ and ${\rm p}={\bf 0}$.
\item Let $\Omega$ be a symplectic form in $\Real^{2n}$ and $\Omega_{0}=\Omega({\rm p})$ the constant symplectic form in $\Real^{2n}$ equal to  $\Omega$ at ${\rm p}$. 

It is enough to prove that there exists a neighbourhood $U$ of ${\rm p}$, and a diffeomorphism $\phi:U\to U$ such that $\phi^{*}\Omega_{0}=\Omega$ on $U$. This is true since, by the above Lemma, we can choose a global coordinate system in $\Real^{2n}$ such that the symplectic form $\Omega_{0}$ take the expression we need.

\item Consider the 2--form $\omega_{t}\equiv\Omega+t(\Omega_{0}-\Omega)$ in $\Real^{2n}$ for $t\in[0,1]$. We have that
\begin{enumerate}
\item $\omega_{0}=\Omega$, $\omega_{1}=\Omega_{0}$, $\d\,\omega_{t}=0$, for every $t\in[0,1]$.
\item $\omega_{t}({\rm p})=\Omega({\rm p})+t(\Omega_{0}({\rm p})-\Omega({\rm p}))=\Omega({\rm p})=\Omega_{0}({\rm p})$ is non--degenerated.
\item As the interval $[0,1]$ is compact, there exists an open set $U_{0}\subset U$  with ${\rm p}\in U_{0}$ such that $\left.\omega_{t}\right|_{U_{0}}$ is non--degenerate for every $t\in[0,1]$.
\item Being $\Omega_{0}-\Omega$ closed,  we can suppose that $U_{0}$ is a ball with center in ${\rm p}$ such that there exists $\alpha\in\Omega^{1}(U_{0})$ with $\Omega_{0}-\Omega=\d\,\alpha$, by Poincar\'e Lemma.
\end{enumerate}
\item Let $X_{t}\in\vf(U_{0})$ be defined by $\inn (X_t)\omega_t=-\alpha$. This time-dependent vector field is well-defined because $\omega_t$ is non--degenerated. Observe that $X_t({\rm p})=0$ since  $\alpha({\rm p})=0$.
\item Let $F_{t,s}$ be the time-dependent flux of $X_{t}$, satisfying $F_{t,s}\circ F_{s,r}=F_{t,r}$ and $F_{t,t}=I$, and we can consider it is defined in $U_{0}$, or reduce it if necessary. Let $F_{t}=F_{t,0}$ the associated diffeomorphic flux with $F_{t+h}=F_{t+h,t}\circ F_{t,0}$. Then, we have:
\begin{eqnarray*}
\frac{\d}{\d\, t}F_{t}^{*}\omega_{t} &=& 
\lim_{h\to 0}\frac{F_{t+h}^{*}\omega_{t}  -F_{t}^{*}\omega_{t}}{h}+F_{t}^{*}(\Omega_{0}-\Omega)\\
&=&F_{t}^{*}( \Lie(X_{t})\omega_{t} )+F_{t}^{*}(\Omega_{0}-\Omega)
=F_{t}^{*}(\d\, \inn (X_t)\omega_t +\Omega_{0}-\Omega)
=F_{t}^{*}(0)=0\,,
\end{eqnarray*}
that is $F_{t}^{*}\omega_{t}$ is constant, hence $F_{1}^{*}\omega_{1}=F_{0}^{*}\omega_{0}$. Then,
$$
F_{1}^{*}\Omega_{0}=F_{1}^{*}\omega_{1}=F_{0}^{*}\omega_{0}=\Omega\, ,
$$
and $F_{1}$ is the diffeomorphism transforming the constant symplectic form $\Omega_{0}$ into 
our given symplectic form 
$\Omega$.
\end{enumerate}
\qed  \end{proof}

\begin{remark}{\rm 
For presymplectic manifolds there is a similar result (see, for instance, \cite{CP-adg,dLGRR-2023}).
}\end{remark}

Finally, as a straightforward consequence of the definition, we have:

\begin{prop}
\label{Liouville}
Every symplectic manifold is an {\sl oriented manifold}.
\end{prop}
\begin{proof}
In fact, using the symplectic form we can define the volume form
\(\dst\Omega^n:=\bigwedge^n\Omega \in {\mit\Omega}^{2n}(M)\),
which is called the {\sl \textbf{Liouville volume form}} on $M$.
\\ \qed  \end{proof}

\bigskip
\noindent{\bf Some examples of symplectic manifolds}:
\begin{enumerate}
\item
The cotangent bundle $\Tan^*Q$ of a manifold $Q$ is an example of symplectic manifold which, moreover, is the canonical model of these kinds of manifolds 
(see Theorem \ref{ficotvs} in the Appendix \ref{sec:cotbun} for all the details). Here we do a short survey:
We have a canonical 1-form $\Theta \in\Omega^1(\Tan^*Q)$ which, in natural coordinates $(q^i,p_i)$ of $\Tan^*Q$, is 
$\Theta=p_i\d q^i$. Then $\Omega=-\d\Theta=\d q^i\wedge\d p_i$ is a symplectic form. 
Observe that the natural coordinates of $\Tan^*Q$ are Darboux coordinates for $\Omega$. 
\item 
$\Real^{2n}$ with the usual Cartesian coordinate system, $(x^{1},\ldots,x^{n},y_{1},\ldots,y_{n})$, and the 2-form $\Omega=\d x^{i}\wedge\d y_{i}$ is a symplectic manifold.
\item 
The 2-sphere $S^{2}\subset \Real^{3}$ as a Riemannian submanifold with the induced  area 2-form is a symplectic manifold.
\item 
Another relevant example which will be used later is given by the following:

\begin{prop}
\label{productsymp}
Let $(M_1,\Omega_1)$, $(M_2,\Omega_2)$ be symplectic manifolds
with $\dim M_1 = \dim M_2$,
and $\pi_j\colon M_1\times M_2\to M_j$, $j=1,2$ be the natural projections.
Then the product manifold $(M_1\times M_2,\Omega_{12}=\pi_1^*\Omega_1-\pi_2^*\Omega_2)$
is a symplectic manifold.
\end{prop}
\end{enumerate}

\subsection{Submanifolds of a symplectic manifold}

As in the case of linear symplectic structures, in a symplectic manifold there exist some interesting kinds of submanifolds. They are associated to the idea of orthogonality in the tangent space at every point with respect to the induced linear symplectic structure at this point. Lagrangian submanifolds play an important role in the study of dynamics of Hamiltonian systems.

\begin{definition}
Let $(M,\Omega)$ be a symplectic manifold and $j:L\to M$ an immersion.
\begin{enumerate}
\item 
$L$ is an \textbf{isotropic immersed submanifold} of $(M,\Omega)$ if
$\Tan_{\rm p}j(\Tan_{\rm p}L)\subset\Tan_{j(\rm p)}M$ is an isotropic subspace of 
$(\Tan_{j(\rm p)}M,\Omega_{j(\rm p)})$ as a linear symplectic space, for every ${\rm p}\in L$;
\item
$L$ is a \textbf{coisotropic immersed  submanifold} of $(M,\Omega)$ if
$\Tan_{\rm p}j(\Tan_{\rm p}L)\subset\Tan_{j(\rm p)}M$ is a coisotropic subspace of 
$(\Tan_{j(\rm p)}M,\Omega_{j(\rm p)})$ as a linear symplectic space, for every ${\rm p}\in L$;
\item
$L$ is a \textbf{symplectic immersed submanifold} of $(M,\Omega)$ if
$\Tan_{\rm p}j(\Tan_{\rm p}L)\subset\Tan_{j(\rm p)}M$ is a symplectic subspace of 
$(\Tan_{j(\rm p)}M,\Omega_{j(\rm p)})$ as a linear symplectic space, for every ${\rm p}\in L$;
\end{enumerate}
This same terminology is used for subbundles of $\Tan M$ over submanifolds of $M$.
\end{definition}

\begin{definition} 
Let $(M,\Omega)$ be a symplectic manifold and $L\subset M$ a submanifold. We say that $L$ is a \textbf{Lagrangian submanifold}  if it is isotropic and $\dim L=(1/2)\dim M$.\end{definition}

\begin{remark}{\rm 
\bit
\item 
A submanifold
$j\colon L\to M$ is isotropic if, and only if, $j^{*}\Omega=0$.
\item 
From the above study about linear Lagrangian subspaces, we have that, if $L\subset M$ is Lagrangian, then $\dim L=(1/2)\dim M$ and $(\Tan_{\rm p}L)^{\bot}=\Tan_{\rm p}L$.
\eit
}\end{remark}

\bigskip
\noindent {\bf Examples of Lagrangian submanifolds}:
\begin{enumerate}
\item We know that $\Real^{2n}$ with elements denoted by $(x,y)\in\Real^{2n}$ and coordinates $(x,y)=(x^{i} ,y_{j})$ has a natural symplectic form given by $\omega= \d x^{i}\wedge\d y_{i}$. Examples of Lagrangian submanifolds are the following:
$$
L_{1}=\{(x,y)|x=0\} \quad , \quad
L_{2}=\{(x,y)|y=0\} \quad , \quad
L_{3}=\{(x,y)|x=y\} \ .
$$
\item Taking the cotangent bundle of a manifold $Q$, we have a symplectic manifold, $(\Tan^*Q,\omega)$. Then Lagrangian submanifolds are the fibers of the bundle, that is $q=\mathrm{constant}$, or the section zero, that is the manifold $Q$ as a submanifold of $\Tan^*Q$.
\item In the symplectic manifold $(\Tan^*Q,\omega)$, let $\alpha:Q\to\Tan^*Q$ be a differential form and $N=\{(q,p)\in\Tan^*Q\ | \ p=\alpha(q)\}$ the graph of $\alpha$. Then $N$ is a Lagrangian submanifold of $\Tan^*Q$ if, and only if, $\alpha$ is a closed form. In fact, we have
$$
\d\alpha=\d(\alpha^{*}\theta)=\alpha^{*}\d\theta=-\alpha^{*}\d\omega\, ,
$$
by the properties defining the canonical forms in $\Tan^*Q$. You can see the corresponding section in the coming chapters for a detailed account of all this.
\item As a particular case, given a function $f:Q\to\Real$, the graph of $\d f$ is a Lagrangian submanifold of $\Tan^*Q$.
\end{enumerate}

\subsection{Canonical isomorphism. Hamiltonian vector fields}
\protect\label{icch}

The fact that a symplectic form is non-degenerated
has important consequences. One of the main ones is the following:
every differential form $\Omega\in\df^k(M)$
defines a linear map
$$
\begin{array}{ccccc}
\flat_\Omega & \colon & \Tan M & \longrightarrow &\wedge^{k-1} \Tan^*M
\\
 & & ({\rm p},X_{\rm p}) & \mapsto & ({\rm p},\inn(X_{\rm p})\Omega_{\rm p})
\end{array} \ ,
$$
and its natural extension (which we denote with the same notation)
$$
\begin{array}{ccccc}
\flat_\Omega & \colon & \vf (M) & \longrightarrow & {\mit\Omega}^{k-1}(M)
\\
 & & X & \mapsto & \inn(X)\Omega
\end{array} \ .
$$
The inverse of this isomorphism is denoted $\sharp_\Omega:=\flat_\Omega^{-1}$.

Given a differentiable manifold $M$
and a form $\Omega\in\df^2(M)$;
it is obvious that $\Omega$ is non-degenerated (symplectic) if, and only if,
$\flat_\Omega$ is an isomorphism between $\Tan M$ and $\Tan^{*} M$. Then:

\begin{definition}
If $(M,\Omega)$ is a symplectic manifold, the map
$\flat_\Omega$ is called the \textbf{canonical isomorphism} induced by
$\Omega$.
\end{definition}

Given a symplectic manifold $(M,\Omega )$,
every function $f \in \Cinfty (M)$ has associated a unique
vector field $X_f \in \vf (M)$ by means of the map
$$
\sharp_\Omega\circ\d\colon \Cinfty (M)\mapping{\d}
\df^1(M)\mapping{\sharp_\Omega}\vf (M) \ ;
$$
that is, defined as $X_f:=\sharp_\Omega(\d f)$
or, what is equivalent, given implicitly by
\beq
\inn(X_f)\Omega:=\d f \ .
\label{ch}
\eeq

\begin{remark}{\rm 
Observe that the map $\sharp_\Omega\circ\d$ is not surjective.
This means that, although the canonical isomorphism allows us to
associate to every vector field $X$ a differential $1$-form 
$\inn(X)\Omega$,
it is not always possible to associate a function since,
in order to do this, the form would be exact necessarily;
but that form is not even closed, in general.
Neither is the map $\sharp_\Omega\circ\d$ injective,
since functions differing in an additive constant have
the same vector field associated by this map.
}\end{remark}

Bearing in mind this comment, we define:

\begin{definition}
Let $(M,\Omega )$ be a symplectic manifold. A vector field $X \in \vf (M)$ is a
\textbf{ (global) Hamiltonian vector field} if
$\inn(X)\Omega$ is an exact form.
In this case, the function $f \in \Cinfty (M)$
such that $\inn(X)\Omega = \d f$
is called a {\sl \textbf{(global) Hamiltonian function}}
of the vector field $X$.

The set of global Hamiltonian vector fields in $M$ is denoted $\vf_H (M)$.
\end{definition}

Observe that, taking into account the comment before Equation
(\ref{ch}), every function $f\in\Cinfty (M)$
is a Hamiltonian function of a global Hamiltonian vector field $X_f$.

Nevertheless, the requirement in this definition is too restrictive
and,  for the physical interest, it is sufficient to demand that:

\begin{definition}
\label{locHvf}
Let $(M,\Omega )$ be a symplectic manifold. A vector field
$X \in \vf (M)$ is a \textbf{local Hamiltonian vector field} if
$\inn(X)\Omega$ is a closed form.

In this case, for every point ${\rm p} \in M$, Poincar\'e's Lemma assures 
the existence of a neighbourhood $U \subset M$, ${\rm p}\in U$,  and a function
$f \in \Cinfty (U)$ such that $\inn(X)\Omega = \d f$, in $U$.
This function is called a \textbf{local Hamiltonian function}
of the vector field $X$, in $U$.

The set of local Hamiltonian vector fields in $M$ is denoted $\vf_{lh} (M)$.
\end{definition}

\begin{remark}{\rm 
\bit
\item
Obviously $\vf_H (M) \subset \vf_{lh} (M)$.
Thus, all we state for local Hamiltonian vector fields
holds also for global Hamiltonian vector fields.
\item
The above definitions are also valid for presymplectic manifolds. 
The difference is that, in this case, the map $\flat_\Omega$
is not an isomorphism because it is not injective, thus not exhaustive, and, hence,
not every function in the manifold is associated to a Hamiltonian vector field.
\eit
}\end{remark}

Remember that a curve $c\colon I\subseteq\Real\to M$ is an {\sl integral curve}
of a vector field $X\in\vf(M)$ if $\dot c(t)=(X\circ c)(t)$, for $t\in I$;
where $\dot c(t)$ denotes de {\sl derivative} of $c$ at $t$
(i.e, the {\sl tangent vector} of the curve at $c(t)$).
Let $\widetilde c\colon I\subseteq\Real\to\Tan M$ the canonical lift of $c$
to the tangent bundle $\Tan M$; that is, $\widetilde c(t)=(c(t),\dot c(t))$,
for $t\in I$ (see Definition \ref{canlifcurv}).
Then, bearing in mind these definitions and equation  \eqref{Hvecf},
it is immediate to prove that:

\begin{teor}
\label{teo:Hcurv}
A vector field  $X\in\vf(M)$ in a symplectic manifold $(M,\Omega)$, is the (local) Hamiltonian vector field corresponding to the function $f \in \Cinfty (M)$, that is $X=X_{f}$, if, and only if,
the integral curves $c\colon I\subset\Real\to M$ of $X$ are the solutions to the equation
 \begin{equation}
\label{Hcurv}
\inn( \widetilde c)(\Omega\circ c)=\d f\circ c \ ;
\eeq
Equation \eqref{Hcurv} is the {\sl \textbf{Hamilton equation}} for the
integral curves of $X$.
\end{teor}

Remember that equation \eqref{Hcurv} is a straightforward consequence of equation \eqref{ch}, 
of the definitions of integral curve of a vector field, and of the contraction 
$\inn(\widetilde c)(\Omega\circ c)$
(see Remark \ref{intcurve}).

\noindent {\bf Local expressions}:
If $(U;x^i,y_i)$ is a symplectic chart, we have that
\begin{eqnarray*}
X_{f}\mid_U &=& A^i\derpar{}{x^i}+B_i\derpar{}{y_i} \ ,
\\
\d f\mid_U &=& \derpar{f}{x^i}\,\d x^i + \derpar{f}{y_i}\,\d y_i \ ,
\end{eqnarray*}
then, for $X_{f}$, the solution to equation (\ref{ch}), we have that
$$
0 = (\inn(X_{f})\Omega - \d f)\mid_U
=\left( -B_i - \derpar{f}{x^i}\right)\d x^i +
\left( A^i - \derpar{f}{y_i}\right)\d y_i \ ;
$$
that is,
\beq
A^i = \derpar{f}{y_i} \quad , \quad B_i= -\derpar{f}{x^i} \ ,
\label{heqs1}
\eeq
and then
\beq
\label{eq:locHvf}
X_{f}\mid_U = \derpar{f}{y_i}\derpar{}{x^i}-\derpar{f}{x^i}\derpar{}{y_i} \ ,
\eeq
 therefore, the integral curves $c(t)=(x^i(t),y_i(t))$ of $X_{f}$ are the solutions to the system of 
first-order differential equations
\beq
\frac{d y_i}{d t} = -\derpar{f}{x^i}(c(t))
\quad , \quad
\frac{d x^i}{d t} = \derpar{f}{y_i}(c(t)) \ .
\label{heqs2}
\eeq
Equations \eqref{heqs1} and  \eqref{heqs2} are the local expression of equations 
 \eqref{ch} and \eqref{Hcurv} respectively, and are called the {\sl Hamilton equations} 
of the (local) Hamiltonian vector field $X_{f}$ and its integral curves, respectively.

An important result which is later used is:

\begin{lem}
Let $(M,\Omega)$ be a symplectic manifold.
For every point ${\rm p}\in M$, there exist vector fields $X_j\in\vf_{lh}(M)$, $j=1,\ldots ,2n$,
such that $\{ X_j({\rm p})\}$ is a basis of $\Tan_{\rm p}M$.
\label{auxil}
\end{lem}
\begin{proof}
It is immediate using symplectic charts,
since the local coordinate vector fields 
$\displaystyle\derpar{}{x^i},\derpar{}{y_i}$
are locally Hamiltonian vector fields associated to the Hamiltonian functions
$y_i$ and $-x^i$, respectively.
\\ \qed  \end{proof}

This means that local Hamiltonian vector field expand locally the tangent bundle of $M$.

\subsection{Invariant forms}

The properties of Hamiltonian vector field are closely related with the properties
of the symplectic form. 
Originally, this relationship was established by studying the so-called {\sl integral invariants} of mechanics \cite{Ca-22,DG-80,Ga-70}.
Next, we explore this relation.
First, we introduce the following concept:

\begin{definition}
Let $M$ be a differentiable manifold and $X \in \vf (M)$.
A form $\beta \in \df^p(M)$ is an \textbf{absolute invariant form}
for $X$ if $\Lie (X)\beta = 0$.
\end{definition}

\begin{remark}{\rm  
Reminding the interpretation of the Lie derivative,
to be an absolute invariant form for $X$ means that
$\beta$ is invariant along the integral curves of $X$,
and, if $F_t$ denotes the flux of the vector field $X$,
this is equivalent to demand that $F_t^*\beta = \beta$ for every $t$.
}\end{remark}

Now we can state the following result, which is usually taken 
as an alternative definition of Hamiltonian vector field:

\begin{teor}
Let $(M,\Omega )$ be a symplectic (resp. presymplectic) manifold.
The vector field $X \in \vf (M)$ is a
local Hamiltonian vector field if, and only if,
$\Omega$ is an absolute invariant form for $X$.
\label{teoin}
\end{teor}
\begin{proof}
As $\Omega$ is a closed form, we have that
$$
\Lie (X)\Omega = \inn (X)\d \Omega + \d\inn (X)\Omega =
\d\inn (X)\Omega = 0
\ \Longleftrightarrow \  \inn (X)\Omega \in Z^1(M)
\ \Longleftrightarrow \ X \in \vf_{lh}(M) \ .
$$
\qed  \end{proof}

\begin{remark}{\rm 
This result relates the Hamiltonian vector fields with the fact that $\Omega$
is closed, although it is less precise than Definition \ref{locHvf},
since it does not allow us to distinguish the global Hamiltonian vector fields inside the set of the local Hamiltonian ones.
}\end{remark}

From this theorem we deduce:

\begin{teor} {\rm (Liouville)}:
Let $(M,\Omega )$ be a symplectic manifold
and $\Omega^n$ the Liouville volume form in $M$.
Then $\Lie (X)\Omega^n =0$, for every $X \in \vf_{lh}(M)$.
\end{teor}
\begin{proof}
It is immediate since, as the Lie derivative is a derivation,
\vspace{-3mm}
$$
\Lie (X)\Omega^n =n\,(\Lie (X)\Omega)\wedge
\overbrace{\Omega\wedge\ldots\wedge\Omega}^{(n-1\ times)}=0 \ .
$$
\qed  \end{proof}

\begin{prop}
Let $(M,\Omega )$ be a symplectic (resp. presymplectic) manifold.
The set $\vf_{lh}(M)$ is closed for the Lie bracket of vector fields
and it is a  real Lie algebra.
\label{allie}
\end{prop}
\begin{proof}
We have to prove that $[X,Y]\in\vf_{lh}(M)$, for every $X,Y\in\vf_{lh}(M)$. 
Taking into account the relation
$$
\inn([X,Y])\Omega =
\Lie (X)\inn(Y)\Omega - \inn(Y)\Lie (X)\Omega
$$
and Theorem \ref{teoin}, we have that, if \(\dst\inn (Y)\Omega\vert_U=\d f\), on an open set $U$, we have
$$
\inn([X,Y])\Omega =\Lie (X)\inn(Y)\Omega\vert_U=\Lie(X)\d f=\d\Lie(X)f \ ,
$$
that is, in $U$, the local Hamiltonian function for $[X,Y]$ is $\Lie(X)f$.

Linearity, skew symmetry and the Jacobi identity of the Lie bracket
complete the proof.
\\ \qed  \end{proof}

So, given a symplectic manifold
$(M,\Omega )$, we have proven that the symplectic and the Liouville volume forms are invariant
by the local Hamiltonian vector fields; that is,
by the groups of local diffeomorphisms generated by their fluxes. 
Actually, there is a more general property stating that some geometrical structures,
such as those defined by a volume or a symplectic
form on a differentiable manifold,
are determined by their automorphism groups
(i.e.; the groups of volume-preserving and symplectic diffeomorphisms), as it was shown by {\it A. Banyaga} \cite{Ba86,Ba88,Ba97}.

Now we can try to characterize all the elements in
${\mit\Omega}^k(M)$ with the above property.
The answer to this problem was established by
{\it Lee Hwa Chung} \cite{Hw-47}, who studied the uniqueness
of  the {\sl  integral invariant forms} by local transformations generated
by the flux of  the local Hamiltonian vector fields.
Next, we state the geometrical version of this theorem and prove a partial result of it
\footnote{
The original proof of this theorem is local, using Darboux coordinates.}.

\begin{teor} {\rm (Lee Hwa Chung)}.
Let $(M,\Omega )$ be a symplectic manifold
and $\alpha \in {\mit\Omega}^k(M)$ an absolute invariant form for 
every $X \in \vf_{lh}(M)$. Then:
\begin{enumerate}
\item
If $k$ is odd, that is $k=2r-1$ with $r\in\Nat$, then
$\alpha = 0$.
\item
If $k$ is even, that is $k=2r$ with $r\in\Nat$, then
$\displaystyle\alpha = c\,\overbrace{\Omega\wedge\ldots\wedge\Omega}^{(r\ times)}
\equiv c \,\bigwedge^r\Omega$,
where $c \in \Real$.
\end{enumerate}
\end{teor}
\begin{proof}
We prove the statement for the case $k\leq 2$, which is the only we need later
(for the proof of the general  case, see  \cite{LlR-88}).

As $\alpha$ is invariant under the action of  every $X\in\vf_{lh}(M)$, we have that
\beq
0=\Lie (X)\alpha =\d\inn (X)\alpha +\inn (X)\d\alpha
\quad \Longleftrightarrow \quad
\d\inn (X)\alpha =-\inn (X)\d\alpha \ .
\label{uno}
\eeq
Consider now $X,X'\in\vf_{lh}(M)$ and ${\rm p}\in M$; there exist
$U\subset M$, ${\rm p}\in U$, and $f,g\in\Cinfty (U)$ such that
$\inn (X)\Omega\vert_U=\d f$ and $\inn (X')\Omega\vert_U=\d g$
(from now on we write $X\vert_U\equiv X_f$ y
$X'\vert_U\equiv X_g$).
The vector field $X_{h}$ given in $U$ as $X_h\vert_U=fX_g+gX_f$, is locally Hamiltonian and, in $U$, its local Hamiltonian function is $h=fg\in\Cinfty (U)$, since:
$$
\inn (X_h)\Omega \vert_U=
\inn (fX_g+gX_f)\Omega =f\inn (X_g)\Omega +g\inn (X_f)\Omega =
f\d g+g\d f\equiv\d h \ .
$$
Thus we have
$$
\inn (X_h)\alpha\vert_U= f\inn (X_g)\alpha +g\inn (X_f)\alpha \ ,
$$
and taking the exterior differential
$$
\d\inn (X_h)\alpha\vert_U=\d f\wedge\inn (X_g)\alpha+f\d\inn (X_g)\alpha+
\d g\wedge\inn (X_f)\alpha+g\d\inn (X_f)\alpha \ .
$$
But, having in mind \eqref{uno},
$$
\d\inn (X_h)\alpha=-\inn(X_fh)\d\alpha\vert_U=
-f\inn (X_g)\d\alpha -g\inn (X_f)\d\alpha =
f\d\inn (X_g)\alpha+g\d\inn (X_f)\alpha \ ,
$$
and comparing the last two equations, we conclude that
\beq
(\d f\wedge\inn (X_g)\alpha+\d g\wedge\inn (X_f)\alpha)\vert_U=0 \ .
\label{dos}
\eeq
Putting now $X_f=X_g$ in this expression,
that is, $f=g$, we obtain that, for every $f\in\Cinfty (U)$,
$$
\d f\wedge\inn (X_f)\alpha\vert_U=0 \ .
$$
Now, we have two options:
\ben
\item
If $k=1$, then $\inn (X_f)\alpha\in\Cinfty (M)$,
and this last equality leads to $\inn (X_f)\alpha =0$,
for every $X_f\in\vf_{lh}(M)$.
Taking into account that, by Lemma \ref{auxil},
local Hamiltonian vector fields span locally $\Tan M$,
we obtain that $\inn (X)\alpha\vert_U=0$,
for every $X\in\vf (U)$, and this implies necessarily that
$\alpha\vert_U=0$ (for every $U$) and then $\alpha =0$.
\item
If $k=2$, we have to conclude:
\bit
\item
either $\inn (X_f)\alpha =0$,
\item
or $\inn (X_f)\alpha\vert_U=\eta_{X_f}\d f$,
where $\eta_{X_f}\in\Cinfty (U)$.
\eit
In the first case, reasoning as in the above item,
we conclude that $\alpha =0$.
In the second case, going to the expression (\ref{dos}),
we obtain
$$
(\d f\wedge\d g\,\eta_{X_g}+\d g\wedge\d f\,\eta_{X_f})\vert_U=0\ ;
$$
that is
$$
(\d f\wedge\d g)(\eta_{X_g}-\eta_{X_f})\vert_U=0 \ ,
$$
for every $f,g\in\Cinfty (U)$. Then it must be
$\eta_{X_f}=\eta_{X_g}\equiv\eta$;
that is, the function $\eta$ does not depend on the
local Hamiltonian vector field.

Therefore, for every $X_f\in\vf_{lh}(M)$, we have that
$\inn (X_f)\alpha\vert_U=\eta\,\d f$, with $\eta\in\Cinfty (M)$, then
$$
\inn (X_f)\alpha\vert_U=\eta\,\d f=\eta\inn (X_f)\Omega=
\inn (X_f)(\eta\Omega ) \ ;
$$
but, taking into account Lemma \ref{auxil},
this equality leads to
$$
\inn (X)(\alpha -\eta\Omega )\vert_U=0 \ ,
$$
for every $X\in\vf (M)$, and we conclude that
$\alpha =\eta\,\Omega$.
Finally, as $\alpha$ is invariant for every local
Hamiltonian vector field, we have that, for every $Y\in\vf_{lh}(M)$,
$$
0=\Lie (Y)\alpha =\Lie (Y)(\eta\Omega )=
(\Lie (Y)\eta)\Omega +\eta\Lie (Y)\Omega=(\Lie (Y)\eta)\Omega \ ,
$$
then $\Lie (Y)\eta =0$ and, by Lemma \ref{auxil}, this result holds for every $Y\in\vf (M)$;
therefore $\eta =c$ (constant) and thus
$\alpha =c\,\Omega$.
\een
\qed  \end{proof}

\begin{remark}{\rm 
\begin{itemize}
\item
So that, the absolute invariant forms for every
local Hamiltonian vector field are multiple
of exterior products of the symplectic form and hence
they are of degree even necessarily. The above proof is for 2-forms only.
\item
A similar result can be proved also for presymplectic manifolds \cite{GLR-84}
(See also \cite{EIMR-2012} for another interesting generalization of this theorem).
\end{itemize}
}\end{remark}

\subsection{Poisson brackets}

The symplectic form allows us to introduce
certain well known operations in Analytical Mechanics in a natural way.

\begin{definition}
Let $(M,\Omega )$ be a symplectic manifold.
The \textbf{Lagrange bracket} of two vector fields $X,Y \in \vf (M)$
is the bilinear map
$$
\begin{array}{ccccc}
(\ ,\ ) & \colon & \vf (M) \times \vf (M) & \longrightarrow & \Cinfty (M)
\\
& & X,Y & \mapsto & (X,Y)
\end{array}
$$
defined by
$$
(X,Y) := \Omega (X,Y) := \inn(Y)\inn(X)\Omega \ .
$$
\end{definition}

\begin{remark}{\rm 
\begin{itemize}
\item
This bracket is not an internal operation in  $\vf (M)$, as is the {\sl  Lie bracket},
whose result is another vector field.
\item
From the skew symmetry of $\Omega$ we deduce immediately that the Lagrange bracket is also skew symmetric; that is,
$(X,Y) = -(Y,X)$.
\end{itemize}
}\end{remark}

Taking into account (\ref{ch}), from this concept we obtain:
 
\begin{definition}
Let $(M,\Omega )$ be a symplectic manifold.
The \textbf{Poisson bracket} of two functions $f,g \in \Cinfty (M)$
is the Lagrange bracket of their associated Hamiltonian vector fields,
that is, the bilinear map
$$
\begin{array}{ccccc}
\{\ ,\ \} & \colon &
\Cinfty (M) \times \Cinfty (M) & \longrightarrow & \Cinfty (M)
\\
& & f,g & \mapsto & \{ f,g \}
\end{array}
$$
defined by
$$
\{ f,g \} := \Omega (X_f,X_g) := \inn(X_g)\inn(X_f)\Omega \ .
$$
\end{definition}

\noindent {\bf Local expressions}:
If $(U;x^i,y_i)$ is a symplectic chart and
$$
X\mid_U = A^i\derpar{}{x^i}+B_i\derpar{}{y_i} \quad ,
\quad
Y\mid_U = C^i\derpar{}{x^i}+D_i\derpar{}{y_i} \ ,
$$
then
$$
(X,Y)\mid_U = -B_iC^i+A^iD_i \ .
$$
Furthermore,
$$
\{ f,g \}\mid_U =
\derpar{f}{x^i}\derpar{g}{y_i} - \derpar{f}{y_i}\derpar{g}{x^i} \ .
$$
In particular, for  the canonical coordinates $x^i,y_i$ we have that
$$
\{ x^i,x^j \} = 0 \ , \
\{ y_i,y_j \} = 0 \ , \
\{ x^i,y_j \} =  \delta^i_j \ .
$$

The main properties of the Poisson bracket are collected in  the following:

\begin{prop}
Let $(M,\Omega )$ be a symplectic manifold and $\{\,,\,\}$ the associate Poisson bracket. Then
\vspace{-3mm}
\begin{enumerate}
\item
$\{ f,g \}=-\{ g,f \}$ ({\sl skew symmetry}).
\item
$\{ f,\{ g,h \}\} + \{ g,\{ h,f \}\} + \{ h,\{ f,g \}\} = 0$ ({\sl Jacobi identity}).
\item
$\{ f,g \}=\Lie (X_g)f =-\Lie (X_f)g$
\item
$X_{\{ f,g \}} = [X_g,X_f]$.
\end{enumerate}
\end{prop}
\begin{proof}
\begin{enumerate}
\item
Immediate from the definition.
\item
It is a consequence of $\Omega$ being closed.
\item
Bearing in mind the Cartan formula for the Lie derivative:
$$
\{ f,g \} = \inn(X_g)\inn(X_f)\Omega = \inn(X_g)\d f = \Lie (X_g)f \ .
$$
In an analogous way, we have that
$\{ f,g \} = -\Lie (X_f)g$.
\item
The statement is equivalent to $\inn([X_g,X_f])\Omega = \d \{ f,g \}$,
and remembering Proposition \ref{allie} we obtain
\beq
\inn([X_g,X_f])\Omega =\Lie (X_g)\inn(X_f)\Omega = \Lie(X_g)\d f=
\d\Lie(X_g)f=\d \{ f,g \} \ .
\label{expre}
\eeq
\end{enumerate}
\qed  \end{proof}

\begin{remark}{\rm 
\begin{itemize}
\item
The first two properties establish that
$\Cinfty (M)$  with the Poisson bracket is a real Lie algebra.
\item
The third property allows us to give a geometric interpretation
of the Poisson bracket between two functions: it measures the variation of one of them along
the integral curves of the Hamiltonian vector field associated to the other.
\item
The fourth property tells us that the map  $\Cinfty (M)\to\vf(M)$ given by $f\mapsto X_{f}$, is a Lie algebra
(anti)--homomorphism between $(\Cinfty (M),\{\,\, ,\,\})$ and $(\vf (M),[\, \, ,\,])$.
\end{itemize}
}\end{remark}

Using again the canonical isomorphism,
and taking into account (\ref{expre}),
we can establish the following generalization:

\begin{definition}
Let $(M,\Omega )$ be a symplectic manifold.
The {\sl \textbf{Poisson bracket}} of two $1$-forms 
$\alpha ,\beta \in {\mit\Omega}^1 (M)$
is the bilinear map
$$
\begin{array}{ccccc}
\{\ ,\ \} & \colon &
{\mit\Omega}^1 (M)\times{\mit\Omega}^1(M)&\longrightarrow&\df^1 (M)
\\
& & \alpha ,\beta & \mapsto & \{ \alpha ,\beta \}
\end{array} \ ,
$$
defined by
$$
\{ \alpha ,\beta \} := \inn([X_\alpha,X_\beta ])\Omega \ ,
$$
where $X_\alpha = \sharp_\Omega (\alpha )$
and $X_\beta = \sharp_\Omega (\beta )$.
\end{definition}

It is evident that:

\begin{prop}
Let $(M,\Omega )$ be a symplectic manifold. Then,
for every $f,g \in  \Cinfty (M)$,
$$
\d \{ f,g \} = -\{ \d f,\d g \} \ .
$$
\end{prop}

The properties of the Poisson bracket of $1$-forms
are obviously analogous to those of the Poisson bracket of functions.

\subsection{Canonical transformations and symplectomorphisms}
\protect\label{tcs}

In the previous section we have seen how the properties 
of the symplectic structure allow us to introduce the concept of
local Hamiltonian vector field and how the
integral curves of these fields are obtained 
as solutions to the Hamilton equations.
We will see also that this kind of vector fields is suitable
to describe dynamical systems.
This means that there is a deep relation between the dynamics
of physical systems and the geometric properties of 
their phase spaces.

In this way, from a dynamical perspective,
it is reasonable to suppose that the more relevant transformations
among dynamical systems are those preserving the dynamical equations
which, in our case, means geometrically 
that they transform Hamiltonian vector fields into Hamiltonian vector fields.
Consequently, we define:

\begin{definition}
Let $(M_1,\Omega_1)$ and $(M_2,\Omega_2)$ be
symplectic manifolds and $\Phi:M_{1}\to M_{2}$ a diffeomorphism.
%such that
%$\dim M_1=\dim M_2$, and a diffeomorphism
%$\Phi \in {\rm Diff}\,  (M_1,M_2)$. 
We say that
$\Phi$ is a \textbf{canonical transformation} 
%between these symplectic manifolds
if it maps local Hamiltonian vector fields into
local Hamiltonian vector fields biunivocally; that is,
$\Phi_*(\vf_{lh} (M_1)) = \vf_{lh}(M_2)$.
\end{definition}

%Observe that this implies that $\Phi_*(\vf_{h} (M_1)) = \vf_{h}(M_2)$.

Concerning the geometrical aspects,
the more interesting transformations between symplectic manifolds
are the following:

\begin{definition}
Let $(M_1,\Omega_1)$ and $(M_2,\Omega_2)$ be
symplectic manifolds and $\Phi:M_{1}\to M_{2}$ a diffeomorphism.
We say that $\Phi$ is a \textbf{symplectomorphism}
(or also a \textbf{symplectic  transformation\/})
if it  preserves their symplectic structures;
that is, $\Phi^*\Omega_2 = \Omega_1$.
\end{definition}

As Hamiltonian vector fields are defined using the symplectic form,
we can expect that there is some relation between both kinds of transformations.
In fact, Lee Hwa Chung's Theorem allows us to prove that
both concepts are essentially the same:

\begin{teor}
Let $(M_1,\Omega_1)$ and $(M_2,\Omega_2)$ be
symplectic manifolds 
and $\Phi \in {\rm Diff}\,(M_1,M_2)$.
The necessary and sufficient condition for $\Phi$
to be a canonical transformation is that
$\Phi^*\Omega_2 = c\Omega_1$, with $c \in {\bf R}$.
\label{teosim}
\end{teor}
\begin{proof}  Suppose that $\Phi$ is a canonical transformation. For every
$X_1 \in \vf_{lh} (M_1)$, we have that
$X_{2}=\Phi_*X_1\in \vf_{lh}(M_2)$ and,
as we know from Theorem \ref{teoin},
$\Lie (X_2)\Omega_2 = 0$, hence
$$
0 = \Phi^*(\Lie (X_2)\Omega_2) =
\Lie (\Phi_*^{-1}X_2)(\Phi^*\Omega_2) =
\Lie (X_1)(\Phi^*\Omega_2) \ .
$$
But this means that $\Phi^*\Omega_2$ is invariant by any element of $\vf_{lh} (M_1)$ and, from Lee Hwa Chung's Theorem,
we conclude that
$\Phi^*\Omega_2 = c\Omega_1$, $c \in {\bf R}$ and $\Phi$ is a symplectomorphism.

Conversely, suppose that $\Phi$ is a symplectomorphism, $\Phi^*\Omega_2 = c\Omega_1$. Given $X_1 \in \vf_{lh}(M_1)$ we have that $\Lie (X_1)\Omega_1 = 0$. Then
$$
0 = \Phi^{*^{-1}}(\Lie (X_1)\Omega_1) =
\Lie (\Phi_*X_1)(\Phi^{*^{-1}}\Omega_1) =
\frac{1}{c}\Lie (\Phi_*X_1)\Omega_2 \ ;
$$
hence, by Theorem \ref{teoin}, $\Phi_*X_1 \in \vf_{lh}(M_2)$, that is $\Omega_{2}$ is invariant by every element of  $\vf_{lh}(M_2)$,
and $\Phi$ is a canonical transformation.
\\ \qed  \end{proof}

\begin{remark}{\rm 
\begin{itemize}
\item
The constant $c$ that appears in this last theorem
is called  the {\sl \textbf{valence}} of the canonical transformation.
It is usual to consider transformations with $c=1$
(that is, symplectomorphisms),
and they are called {\sl  \textbf{univalent {\rm or} restricted canonical transformations}}.
Another terminology is also used, calling canonical transformations
to those with valence $c=1$, and then call 
{\sl  \textbf{generalized canonical transformations}} to the rest \cite{SC-71}.
\item
Observe that a diffeomorphism is an {\sl univalent canonical transformation}
if, and only if, it is a {\sl  symplectomorphism}.
\item
All these definitions and properties are also valid for
presymplectic manifolds. In this case we talk about
{\sl \textbf{presymplectomorphisms}}
(the study of this case is done in \cite{BGPR-86,CGIR-85,CGIR-87}).
\end{itemize}
}\end{remark}

Another fundamental result is:

\begin{prop}
Let $(M,\Omega )$ be a symplectic (resp. presymplectic) manifold.
A vector field $X \in\vf(M)$ is a local Hamiltonian vector field
if, and only if, its flux is a group of 
local symplectomorphisms (resp. local presymplectomorphisms).
\end{prop}
\begin{proof} Let $F_t$ be the flux of a vector field $X$. We know that $
\Lie (X)\Omega = 0$ is equivalent to $F_t^*\Omega = \Omega$, hence the result follows directly.
\\ \qed  \end{proof}

Finally, it is easy to prove that:

\begin{prop}
The set of canonical transformations  of a symplectic (resp. presymplectic) manifold
$(M,\Omega)$, with the operation of composition, is a group.
\end{prop}

The group of the symplectomorphisms
of a symplectic manifold is denoted by ${\bf Sp}(M,\Omega )$,
and it has a crucial relevance in the study of {\sl  symmetries} of dynamical systems.

\subsection{Characterization of canonical transformations}

The last theorem has some important corollaries, which give
alternative characterizations for a transformation to be canonical
(or a symplectomorphism).
The most important of them uses the Poisson bracket of functions.
Previously, we have to specify how the Hamiltonian functions associated to
Hamiltonian vector fields are transformed under 
these kinds of transformations.

\begin{prop}
Let $(M_1,\Omega_1)$ and $(M_2,\Omega_2)$ be
symplectic manifolds  
and $\Phi \in {\rm Diff}\,  (M_1,M_2)$
a canonical transformation of valence $c$.
If $X_1 \in \vf_{lh}(M_1)$, let $X_2 := \Phi_*X_1 \in \vf_{lh}(M_2)$,
and $h_1 \in \Cinfty (U_1)$ and $h_2 \in \Cinfty (U_2)$
local Hamiltonian functions of $X_1$ and $X_2$
on $U_1 \subset M_1$ and $U_2 :=\Phi (U_1) \subset M_2$,
respectively. Then
$$
ch_1 = \Phi^*h_2 + k \quad , \quad k \in {\bf R} \ .
$$
\end{prop}
\begin{proof}
Following the previous theorem, we have that

\begin{eqnarray*}
\d(\Phi^{*^{-1}}h_{1})&=&\Phi^{*^{-1}}\d h_{1} =\Phi^{*^{-1}}(\inn (X_1)\Omega_1)\mid_{U_1} \\&=&
\inn (\Phi_*X_1)(\Phi^{*^{-1}}\Omega_1)\mid_{\Phi (U_1)} =
\frac{1}{c}\inn (X_2)\Omega_2\mid_{U_2} =
\frac{1}{c}\d h_2 =\d\Big( \frac{1}{c}h_2\Big)  \,  ,
\end{eqnarray*}
that is
$\d(ch_{1})=\d(\Phi^{*} h_2)$, and the result follows.
\\ \qed  \end{proof}

%but, by hypothesis,
%$\inn (X_1)\Omega_1 \mid_{U_1} = \d h_1$, then
%$$
%\Phi^{*^{-1}}(\inn (X_1)\Omega_1)\mid_{U_1} =
%\Phi^{*^{-1}} \d h_1 = \d (\Phi^{*^{-1}} h_1) \ ,
%$$
%therefore, comparing both expressions we lead to the result.
%\qed  \end{proof}

Bearing this in mind, we state:

\begin{teor}
Let $(M_1,\Omega_1)$ and $(M_2,\Omega_2)$ be
symplectic manifolds. A diffeomorphism 
$\Phi :M_1\to M_2$ is a canonical transformation of valence $c$
if, and only if 
$$
\Phi^*\{ f_2,g_2 \} = \frac{1}{c}\{ \Phi^*f_2,\Phi^*g_2 \} \ .
$$
for every $f_2,g_2 \in \Cinfty (M_2)$,
\end{teor}
\begin{proof}
Let $f_{2},g_2 \in \Cinfty (M_2)$ and
$X_{f_{2}},X_{g_2} \in \vf_{h}(M_2)$ be the corresponding Hamiltonian vector fields.

If $\Phi$ is a canonical transformation, then $\Phi_*^{-1}X_{g_2} \in \vf_{h}(M_1)$
and, by the above proposition we know that
$\inn (\Phi_*^{-1}X_{g_2})\Omega_1 = \d (\frac{1}{c}\Phi^*g_2)$;
that is, $\Phi_*^{-1}X_{g_2} = X_{\frac{1}{c}\Phi^*g_2}$,
and then
$$
\Phi^*\{ f_2,g_2 \} =
\Phi^*(\Lie (X_{g_2})f_2) =
\Lie (\Phi_*^{-1}X_{g_2})\Phi^*f_2
=\Lie (X_{\frac{1}{c}\Phi^*g_2})\Phi^*f_2 =
\frac{1}{c}\{ \Phi^*f_2,\Phi^*g_2 \} \ .
$$

Conversely, if the condici\'on holds, first we have that,
given $f_2,g_2 \in \Cinfty (M_2)$,
$$
\Phi^*\{ f_2,g_2 \} =
\Lie (\Phi_*^{-1}X_{g_2})\Phi^*f_2 \ ,
$$
and, furthermore,
$$
\frac{1}{c}\{ \Phi^*f_2,\Phi^*g_2 \} =
\Lie (X_{\frac{1}{c}\Phi^*g_2})\Phi^*f_2 \ ;
$$
hence, combining these two last equalities,
$$
\Phi_*^{-1}X_{g_2} = X_{\frac{1}{c}\Phi^*g_2} \in \vf_{lh}(M_1)
\ ; \ \forall X_{g_2} \in \vf_{lh}(M_2) \ .
$$
Therefore, $\Phi$ is a canonical transformation
with valence $c$ as a consequence of Lee Hwa Chung's Theorem.
\\ \qed  \end{proof}

\begin{remark}{\rm 
This result states that a transformation is canonical
if, and only if, 
the Poisson bracket is invariant, up to a multiplicative constant, by its action.
Really, this is an alternative way to say that
the symplectic structure is invariant by the transformation.
}\end{remark}

From the local point of view, if $\Phi:M_{1}\to M_{2}$ is a canonical transformation with valence $c=1$, and we have Darboux coordinates
$(x^i,y_i)$ in $M_1$ and $(\tilde  x^i,\tilde  y_i)$ in $M_2$, we have

\hspace{40mm}\(\begin{array}{ccccccc}
\{\Phi^*\tilde  x^i,\Phi^*\tilde  x^j\}&=&\Phi^*\{\tilde  x^i,\tilde  x^j\}  \ , \\
\{\Phi^*\tilde  y_i,\Phi^*\tilde  y_j\}&=&\Phi^*\{\tilde  y_i,\tilde  y_j\}  \ ,\\
\{\Phi^*\tilde  y_i,\Phi^*\tilde  y_j\} &=&\Phi^*\{\tilde  y_i,\tilde  y_j\}   \,,
\end{array}\)

\noindent that is: $(\Phi^*\tilde  x^i,\Phi^*\tilde  y_i)$ is a symplectic coordinate system in $M_{1}$ if, and only if, $(\tilde  x^i,\tilde  y_i)$ is a symplectic coordinate system in $M_{2}$.

\subsection{Generating functions of canonical transformations}

Another characterization of canonical transformations
is by means of the so called  {\sl generating functions}.

\begin{prop}
Let $(M_1,\Omega_1)$ and $(M_2,\Omega_2)$ be
symplectic manifolds and $\Phi \in {\rm Diff}\,  (M_1,M_2)$.
Let $U_1 \subset M_1$ and $U_2 :=\Phi (U_1) \subset M_2$, and
$\Theta_i \in {\mit\Omega}^1(U_i)$ such that
$\Omega_i\mid_{U_i} = \d \Theta_i$, ($i=1,2$).
Then $\Phi$ is
a canonical transformation, with valence $c$, if, and only if,
there exists a function $F_1 \in \Cinfty (U_1)$ such that
$$
(\Phi^*\Theta_2 - c\Theta_1 - \d F_1)\mid_{U_1} = 0 \ ,
$$
or, equivalently, that there exists a function $F_2 \in \Cinfty (U_2)$
such that
$$
(\Phi^{*^{-1}}\Theta_1 - \frac{1}{c}\Theta_2 - \d F_2)\mid_{U_2} = 0 \ .
$$
These functions $F_1,F_2$ are called {\sl \textbf{(Poincar\'e) generating functions}}
of the canonical transformation, and the relation between them is 
$F_1 = c\Phi^*F_2 + k$, $k \in {\bf R}$.
\end{prop}
\begin{proof}
From Theorem \ref{teosim},
$\Phi$ is a canonical transformation if, and only if,
$$
0 = (\Phi^*\Omega_2 - c\Omega_1)\mid_{U_1} = \d (\Phi^*\Theta_2 - \Theta_1) \ ,
$$
then Poincar\'e's Lemma leads to the result.

The statement concerning $F_2$ is obtained in an analogous way,
and comparing both results we arrive to the relation between these functions.
\\ \qed  \end{proof}

In classical texts of Mechanics (see, for example, \cite{Ga-70, LL-76, GPS-01}) a more general concept
of generating function, which include the above ones, is studied. 
In the geometrical context, they are introduced as follows
 \cite{AM-78}
(without loss of generality, we will restrict ourselves to the case $c=1$;
that is, symplectomorphisms or univalent canonical transformations).

First, taking into account Proposition \ref{productsymp}, we have that:

\begin{prop}
Let $(M_1,\Omega_1)$, $(M_2,\Omega_2)$ be symplectic manifolds
with $\dim M_1 = \dim M_2$,
 $\Phi\colon M_1\to M_2$ a diffeomorphism, and
${\jmath}\colon graph\,\Phi\hookrightarrow M_1\times M_2$ the natural embedding.

The map $\Phi$ is a symplectomorphism if, and only if,
$graph\,\Phi$ is a Lagrangian submanifold of the symplectic manifold
$(M_1\times M_2,\Omega_{12}:=\pi_1^*\Omega_1-\pi_2^*\Omega_2)$.
\end{prop}
\begin{proof}
Remember that {\sl Lagrangian submanifolds} of a symplectic manifold 
$({\cal M},\Omega)$ are {\sl maximal isotropic submanifolds}
${\jmath}\colon S\hookrightarrow{\cal M}$ or, equivalently,
such that $S$ verify that ${\rm dim}\,S=\frac{1}{2}{\rm dim}\,{\cal M}$
and ${\jmath}^*\Omega=0$.
Obviously ${\rm dim}(graph\,\Phi)=\frac{1}{2}{\rm dim}\,(M_1\times M_2)$ ,
and as $\Phi^*\Omega_2=\Omega_1$,
the other condition holds trivially.
\\ \qed  \end{proof}

If $\Omega_j=-\d\theta_j$, $j=1,2$,
and $\Theta_j$ are local symplectic potentials for $\Omega_j$;
being $graph\,\Phi$ a Lagrangian submanifold we have
\beq
0=\jmath^*(\pi_1^*\Omega_1-\pi_2^*\Omega_2)=\d\jmath^*(\pi_2^*\Theta_2-\pi_1^*\Theta_1)
\eeq
which, given a point in  $graph\,\Phi$, is locally equivalent to
\beq
 \jmath^*(\pi_2^*\Theta_2-\pi_1^*\Theta_1)\vert_W=-\d{\cal S} \ .
\label{Wlocal}
\eeq
where ${\cal S}$ is a function defined in 
an open neighbourhood of the given point, $W\subset graph\,\Phi$.
This function depends on the choice of $\Theta_1$ and $\Theta_2$.

\begin{definition}
${\cal S}$ is called a
\textbf{Weinstein generating function} of the Lagrangian submanifold $graph\,\Phi$
and hence of the symplectomorphism $\Phi$.
\end{definition}

If  $(U_1; x^i,y_i)$, $(U_2; \tilde x^i,\tilde y_i)$ are Darboux charts such that
$W\subset U_1\times U_2$, local coordinates in $W$  can be chosen in several ways.
This leads to six different possible choices for ${\cal S}$.
Thus, for instance,  if $(W; x^i,\tilde x^i)$ is a chart,
then \eqref{Wlocal} gives the symplectomorphism  explicitly as
$$
\tilde y_i\d\tilde x^i-y_i\d x^i=-\d {\cal S}(x,\tilde x)
\quad \Longleftrightarrow\ \quad
\tilde y_i=-\derpar{{\cal S}}{\tilde x^i}(x,\tilde x)\ , \ y_i=\derpar{{\cal S}}{x^i}(x,\tilde x)\ .
$$
Of course, the {\sl Poincar\'e generating functions} are two particular choices
of {\sl Weinstein generating functions}.

\section{Hamiltonian dynamical systems}
\protect\label{shreh}

Once the foundations of symplectic geometry are established,
we are ready to use it for describing the autonomous Hamiltonian dynamical systems.
As we will see in the next chapter, this geometric formulation includes, as particular cases, 
the Lagrangian and the canonical Hamiltonian formalisms of variational dynamical systems;
that is, those which are described by Lagrangian functions.

\subsection{Hamiltonian systems}

Some of the authors who developed geometric mechanics chose an axiomatic manner in their exposition
(see, for instance, \cite{So-ssd}).
Following their ideas, first, we state the postulates
for the geometric study of autonomous Hamiltonian dynamical systems.

The first postulate concerns the physical states:
\begin{pos}
{\rm (First Postulate of Hamiltonian mechanics\/)}:
The state space, or phase space, of a dynamical system 
is a differentiable manifold $M$ endowed with a closed form
$\Omega\in Z^2(M)$ such that:
\bit
\item
If $\Omega$ is nondegenerate; that is {\sl symplectic},
the system is {\sl regular} and the dimension of $M$
is twice the number of degrees of freedom of the system.
In this case, every point of this manifold represents a physical state of the system.
\item
If $\Omega$ is degenerate, that is {\sl presymplectic},
the system is {\sl singular} (and the manifold $M$ 
is not even-dimensional necessarily).
\eit
\label{ax1}
\end{pos}

The second Postulate refers to the {\sl  observables},
that is,  the {\sl physical magnitudes}:

\begin{pos}
{\rm (Second Postulate of Hamiltonian mechanics\/)}:
The observables or physical magnitudes of a  dynamical  system
are functions of $\Cinfty (M)$.

The result of the measure of an observable
is the value that the function which represents it takes 
at a point in the phase space $M$
(that is, in a given state, according to the first postulate).
\label{ax2}
\end{pos}

If, according to the Postulate \ref{ax1},
a (symplectic or presymplectic) manifold $(M,\Omega )$
constitutes the phase space of a dynamical system,
there is a very natural way of introducing the dynamics.
Thus, we state:

\begin{pos}
{\rm (Third Postulate of Hamiltonian mechanics\/)}:
The dynamics of a dynamical  system
is given by a closed $1$-form $\alpha \in Z^1(M)$
which is called the {\sl \textbf{Hamiltonian $1$-form}} of the system
\footnote{
In some cases, the 2-form $\Omega$ can also contain dynamical information
as it happens, for instance, in the Lagrangian formalism of  the Lagrangian systems
(see Chapter \ref{sdl}).}.
\end{pos}

And, finally, the dynamical equations are stated in:

\begin{pos}
{\rm (Fourth Postulate of Hamiltonian mechanics\/)}:
The dynamical trajectories of the system are the integral curves
of a vector field $X_{\alpha} \in \vf(M)$, if it exists, associated with the form $\alpha$
by the map $\flat_\Omega$; that is,
of the vector field solution to the equation
\beq
\label{Hvecf}
\inn(X_{\alpha})\Omega = \alpha \ .
\eeq
Then, the integral curves $c\colon I\subseteq\Real\to M$ of $X_\alpha$ are solutions to the equation
\beq
\label{Hvecf2}
\inn(\widetilde c)(\Omega\circ c)=\alpha\circ c \ . 
\eeq
The vector field $X_\alpha$ is called a (local or global) {\sl \textbf{Hamiltonian vector field}}
of the Hamiltonian system, and equations \eqref{Hvecf} and \eqref{Hvecf2} are the {\sl \textbf{Hamilton equation}} for $X_\alpha$
and its integral curves.

These equations can be obtained from a variational principle called
the {\sl  minimal action Principle of Hamilton--Jacobi}
\footnote{
As we will see in Section \ref{hamvariational}.}.
\label{ax3}
\end{pos}

Then, we define:

\begin{definition}
\begin{enumerate}
\item
A \textbf{regular} or \textbf{symplectic Hamiltonian dynamical system}
is a triple $(M,\Omega ,\alpha )$,
where $(M,\Omega )$ is a symplectic manifold 
and $\alpha \in Z^1 (M)$ is the Hamiltonian $1$-form of the system.
If $(M,\Omega )$ is a presymplectic manifold, then $(M,\Omega ,\alpha )$
is said to be a \textbf{singular} or \textbf{presymplectic Hamiltonian dynamical system}.
\item
By Poincar\'e's Lemma, for every $m \in M$, 
there exists $U \subset M$, with $m\in U$,
and $h \in \Cinfty (U)$, such that $\alpha\mid_U = \d h$,
which is called a \textbf{local Hamiltonian function} of the system,
and the above mentioned triple is said to be a
\textbf{local Hamiltonian system}.

If $\alpha$ is an exact form, then there exist 
$h \in \Cinfty (M)$ such that $\alpha = \d h$,
which is called a \textbf{global Hamiltonian function} of the system,
and the triple $(M,\Omega ,h)$ sis said to be a
\textbf{global Hamiltonian system}.
\end{enumerate}
\label{sdhr}
\end{definition}

\subsection{Hamilton equations}

Given a Hamiltonian dynamical system $(M,\Omega,\alpha)$, 
the {\sl \textbf{Hamiltonian problem}} posed by  the system
consists in finding a vector field $X_\alpha \in \vf (M)$
verifying  equations (\ref{Hvecf}).

For the case of regular systems, we have

\begin{prop}
\label{teo-hameqs}
If $(M,\Omega,\alpha)$ is a regular Hamiltonian system, 
then there exists a unique vector field $X_\alpha\in\vf(M)$ which is a solution to equation  \eqref{Hvecf}
\end{prop}
\begin{proof}
If the system is regular, $\flat_\Omega$ is the canonical  isomorphism of
the symplectic manifold $(M,\Omega)$ and
the existence and  uniqueness of $X_\alpha$ is assured.
\\ \qed  \end{proof}

\bigskip
\begin{remark}{\rm 
If $(M,\Omega,\alpha)$ is a presymplectic Hamiltonian system,
the Hamilton equations are not necessarily compatible everywhere on $M$ and we need to look for the maximal subset of $M$ where there exists solution. 
This subset is obtained by a procedure called the {\sl constraint algorithm} which, in the most favourable cases, 
leads to find  a {\sl final constraint submanifold} $P_f\hookrightarrow M$, where there are Hamiltonian vector fields $X_\alpha\in\vf(M)$,
tangent to $P_f$, which are
solutions to the Hamilton equations on $P_f$; although 
these vector fields solution are not necessarily unique. 
Functions vanishing on the final constraint submanifold are called {\sl constraints}.

Singular (presymplectic) Hamiltonian systems and the corresponding constraint algorithms
have been widely studied in the literature. The interested reader can 
see, for instance, \cite{Ca-90,CGIR-85,Dir-64,GNH-78,GP-92,HRT-76,ILDM-99,MMT-97,Mu-89,Sn-74}
and others papers cited therein. Apart from some examples in the problems, 
in this work we only work with regular Hamiltonian systems. 
}\end{remark}

\bigskip
\noindent {\bf Local expressions}:
If $(M,\Omega, \alpha)$ is a regular Hamiltonian system, in a symplectic chart $(U;x^i,y_i)$ of $M$ with $\alpha=\d h$,
if $\displaystyle X\vert_U=f^i\derpar{}{x^i}+g_i\derpar{}{y_i}\in\vf (M)$
then, according to \eqref{eq:locHvf}, we have that
$$
X_{\alpha}\mid_U = \derpar{h}{y_i}\derpar{}{x^i}-\derpar{h}{x^i}\derpar{}{y_i} \ ,
$$
and the integral curves $c(t)=(x^i(t),y_i(t))$ of $X_{\alpha}$ are the solution to the Hamilton equations \eqref{Hcurv} whose local expression is
$$
\frac{d x^i}{d t} = \derpar{h}{y_i}(c(t))
\quad , \quad
\frac{d y_i}{d t} = -\derpar{h}{x^i}(c(t)) \ .
$$

\begin{remark}{\rm 
\begin{itemize}
\item
In a Hamiltonian dynamical  system,
the physical observable represented by the Hamiltonian function 
is associated with the energy of the system.
\item
This presentation of dynamical systems
is called {\sl Hamiltonian formalism of Mechanics} 
because the dynamical trajectories
are given by  the integral curves of a Hamiltonian vector field.
\end{itemize}
}\end{remark}

\section{Symmetries of regular Hamiltonian systems}
\label{secsym}

In this section we introduce the ideas of {\sl constants of motion}, or {\sl conserved quantities}, and the {\sl symmetries} of a dynamical system and the relationship between both notions. Noether's theorem will be proved and actions of Lie group on Hamiltonian systems together with reduction theory will be also introduced and developed.

\subsection{Preliminaries}

The adequate way of doing this study
is by means of the theory of {\sl actions of Lie groups} on
symplectic and presymplectic manifolds
(see, for instance, \cite{AM-78,CP-adg,EMR-99,Gi-hss,GGK-2002,LM-sgam,MSSV-85,Ok-87}).
Within this theory, the associated concept of {\sl momentum map} plays a crucial role 
in order to introduce conserved quantities and the subsequent 
{\sl Marsden--Weinstein reduction theorem}.
Although this theorem was initially stated for regular autonomous Hamiltonian systems,
the Marsden--Weinstein technique has been applied and generalized to many different situations.
For instance,  for time-dependent regular Hamiltonian systems 
(with regular values of the momentum map)
in the framework of cosymplectic manifolds \cite{albert} and
for autonomous and nonautonomous Hamiltonian systems 
with singular values of the momentum map \cite{ACG-91,LS-93,SL-91},
also for singular systems in the autonomous (presymplectic)
case \cite{CCCI-86,EMR-99} and in the nonautonomous case \cite{LMR-92,GIMP-2025,IM-92}.
Further generalizations are reduction on Poisson and Jacobi manifolds \cite{CCM-2007,GLMV-94,LOT-2009,MR-86,NP-2002,NP-2005},
on Lie algebroids \cite{CNS-2005,CNS-2007},
on cotangent bundles of Lie groups \cite{MRW-84},
Lagrangian reduction \cite{CMR-2001,MS-93}, Euler-Poincar\'e reduction \cite{CGR-2001,CRS-2000},
Routh reduction for regular and singular Lagrangians \cite{CL-2010,LCV-2010},
reduction of nonholonomic systems \cite{BS-93,BKMM-1996,CLMM-98,Ma-95},
reduction of optimal control systems \cite{BC-99,Ma-2004,Su-95,Sh-87},
and in the context of Dirac structures and implicit Hamiltonian systems
\cite{BVS-2000,MCL-01}.
Finally, although symmetries and conservation laws in field theories
have been studied in several geometric frameworks (see, for instance,
\cite{LMS-2004,Gimmsy,GR-2023,GR-2024,MRSV-2010,RSV-2007}
and the references quoted therein),
the problem of reduction by symmetries of classical field theories
has been solved only for particular situations such as Lagrangian and Poisson
reduction \cite{CM-2003}, Euler-Poincar\'e reduction in principal fiber bundles 
\cite{CGR-2007,CR-2003},
reduction in multisymplectic and poly(co)symplectic manifolds \cite{LGRRV-2022,LRVZ-2023,EMR-2018,MRSV-2015},
and for discrete field theories \cite{Va-2007}.
Of course, this list of references is far from being complete.

To develop our study we follow approximately the historical order: 
we begin by introducing the basic notions about conserved quantities 
and symmetries, paying particular attention to the case of 
{\sl Noether symmetries} and {\sl Noether's theorem}. The primitive ideas were the constants of motion, the linear momentum and the angular momentum, and their relation with symmetries as global transformation of the space of positions of a system, and after that the infinitesimal symmetries and the Noether theorem. The modern approach comes from the action of Lie groups on the space of states of the system and the reduction theory, hence we review this theory, the notions of comomentum and momentum maps and the subsequent {\sl Marsden--Weinstein reduction theorem}.
We consider only the case of regular dynamical systems 
(although some results can be generalized to the singular case). 

Throughout this section $(M,\Omega,\alpha)$ will be
a regular Hamiltonian dynamical system and $\alpha=\d h$, where the Hamiltonian function $h$ is locally or globally defined.
Usually we will write $(M,\Omega,h)$.
Then, $X_h\in\vf_{lh}(M)$ denotes the dynamical vector field
solution to the system.

\subsection{Conserved quantities (constants of motion)}
\label{conquan}

Let $(M,\Omega,h)$ be a regular Hamiltonian system and $X_h\in\vf_{lh}(M)$ the dynamical vector field. 

The dynamical evolution of an observable which is represented
by a function $f \in \Cinfty (M)$,
is the variation of this function along the integral curves 
$c(t)=(q^i(t),y_i(t))$ of the vector field $X_h$; that is, is given by
$$
\frac{d (f\circ c)(t)}{d t} = \left(\left(\Lie (X_h)f\right)\circ c\right)(t)=\left(X_h(f)\circ c\right)(t) \ ,
$$
and, bearing in mind the definition of Poisson bracket, it can be written as
$$
\frac{d (f\circ c)}{d t} = \{ f,h \}\circ c \ .
$$
%\noindent {\bf Local expression}:
Or, in a symplectic chart $(U;x^i,y_i)$ of $M$, 
%the evolution of an observable $f$ is
$$
\Lie (X_h)f)\mid_U=X_h(f)\mid_U = \derpar{h}{y_i}\derpar{f}{x^i}-\derpar{h}{x^i}\derpar{f}{y_i} \ .
$$

Then, we define:

\begin{definition}
A function $f \in \Cinfty (M)$ is a \textbf{conserved quantity}
or a \textbf{constant of motion} if
$$
\Lie (X_h) f = 0 \ ;
$$
that is, it is invariant by the dynamical vector field.
\end{definition}

To say that $f\in\Cinfty (M)$ is a conserved quantity of the
system $(M,\Omega,h)$ means the following:
if ${\rm p} \in M$ and $X_h$ is the Hamiltonian vector field of the system,
let $c\colon(-\epsilon ,\epsilon )\subset\Real\to M$ be the integral curve of $X_h$ 
with initial condition $c(0)={\rm p}$.
Then, if $f$ is a conserved quantity, we have that
$f(c(t))=f({\rm p})$, for every $t\in(-\epsilon ,\epsilon )$;
that is, the image of $c$ by $f$ is contained in $S_{\rm p}=\{ {\rm q}\in M \ \mid \ f({\rm q})=f({\rm p})\}$,
the level surface of $f$ passing through ${\rm p}$.

One of the fundamental properties of the autonomous Hamiltonian dynamical systems
is the {\sl  conservation of the Hamiltonian}. Remember that for the physical systems, the Hamiltonian is the energy of the system.
Now, we can state this result in a geometric way:

\begin{prop}
{\rm (Conservation of energy)}:
Let $(M,\Omega,h)$ be a Hamiltonian system.
The (local or global) Hamiltonian function $h$ is a conserved quantity.
\end{prop}
\begin{proof}
It is immediate since, as $\Omega$ is skew symmetric, we have
$$
\Lie (X_h) h = \inn (X_h) \d h = \Omega (X_h,X_h) = 0 \ .
$$
\qed  \end{proof}

\begin{remark}{\rm 
We have seen how the fundamental properties
of the symplectic form are essential to describe
geometrically physical systems:
the {\sl  non degeneracy} allows us to assure the existence (and uniqueness)
of the dynamical Hamiltonian vector field, and hence, to determine the dynamical evolution
of the system, and from skew symmetry  and the fact to be closed
we obtain the conservation of the energy.
}\end{remark}

\subsection{Dynamical symmetries}
\label{dynsym}

When talking about symmetries, it is usual to refer to the idea that
``a symmetry of a dynamical system lets invariant the solutions
to the differential equations describing the dynamics of the system''.
In this way we define:

 \begin{definition}
A \textbf{dynamical symmetry} of the Hamiltonian dynamical system $(M,\Omega,h)$
 is a diffeomorphism $\Phi\colon M \to M$ satisfying that,
$$
\Phi_*X_h=X_h \ ;
$$
that is, the dynamical vector field $X_h$ is invariant by $\Phi$. 
\label{gsdef}
 \end{definition} 

Observe that if $\Phi$ is a dynamical symmetry, then so is $\Phi^{-1}$.

 Let $c\colon(-\epsilon ,\epsilon )\subset\Real\to M$ be an integral curve of $X_h$, that is $\dot{c}=X_h\circ c$, then  $\dot{c}=X_h\circ c=\Phi_*X_h\circ c$. If $X_h$ is invariant by $\Phi$ we have:
 $$
 X_h\circ(\Phi\circ c)=\Phi_*X_h\circ\Phi\circ c=\Tan\Phi\circ X_h\circ\Phi^{-1}\circ\Phi\circ c=\Tan\Phi\circ X_h\circ\circ c=\Tan\Phi\circ\dot{c}=\dot{\overline{\Phi\circ c}}\ ;
 $$
 hence $\Phi$ transforms integral curves of $X_h$ into integral curves of $X_h$. The converse is also true; that is: if $\Phi$ transforms integral curves of $X_h$ into integral curves of $X_h$ then $\Phi$ lets invariant the vector field $X_h$. This is another equivalent version of a dynamical symmetry.

If the diffeomorphism $\Phi$ is 
locally generated by a vector field, by means of 
the local group of diffeomorphisms generated by its flux, 
then, the above definition leads to the infinitesimal version of symmetries:

 \begin{definition}
\ An \textbf{infinitesimal dynamical symmetry} of a Hamiltonian system 
$(M,\Omega,h)$ is a vector field $Y\in\vf(M)$ such that  
the local diffeomorphisms  generated by its flux
are dynamical symmetries of the system; that is,
\beq
\Lie(Y)X_h=[Y,X_h]=0 \ .
\label{si}
\eeq
\label{gsdefi}
 \end{definition}

%\vspace{-10mm}
Recall that this definition is equivalent to say that the flux of $Y$ is made of dynamical symmetries of $X_h$. Thus, if $\Phi_t$ is the flux of $Y$ and $c$ is an integral curve of $X_h$, the $\Phi_t\circ c$ is also an integral curve of $X_h$.

\begin{remark}{\rm 
In the above definition, sometimes the condition (\ref{si}) is relaxed by setting
\beq
[Y,X_h]=gX_h \ , \ g\in\Cinfty(M) \ ,
\label{si2}
\eeq
(and the original condition is recovered taking $g=0$).
This allows us to consider also as symmetries
the reparametrizations of the integral curves of the
 dynamical vector field.
 }\end{remark}
 
 The infinitesimal dynamical symmetries of $(M,\Omega,h)$ have a natural structure of real Lie algebra. Clearly they are closed by real linear combinations and we have the following proposition:

\begin{prop}
\label{prop:comsym}
 If $Y_1,Y_2\in\vf (M)$ are infinitesimal dynamical symmetries,
then $[Y_1,Y_2]$ is an infinitesimal dynamical symmetry.
\end{prop}
\begin{proof}
Using the Jacobi identity, we have
$$
[[Y_1,Y_2],X_h]= [Y_2,[X_h,Y_1]]+[Y_1,[Y_2,X_h]]=0 \ .
$$
\qed  \end{proof}

A first result relating symmetries with conserved quantities is:

 \begin{prop}
\label{prop:consym}
\ben
\item
Let $\Phi\colon M\to M$ be a dynamical symmetry of $(M,\Omega,h)$. If $f\in\Cinfty(M)$
is a constant of motion of the system, 
then $\Phi^*f$ is also a conserved quantity.
\item
Let $Y\in\vf(M)$ be an infinitesimal dynamical symmetry of $(M,\Omega,h)$. If $f\in\Cinfty(M)$ is a constant of motion of the system, 
then $\Lie(Y)f$ is also a conserved quantity.
\een
\label{generador}
 \end{prop}
 \begin{proof}
\ben
\item
Directly we have
$$
\Lie(X_h)(\Phi^*f)=\Phi^*(\Lie(\Phi_*X_h)f)=0 \ .
$$
\item
The same proof holds taking the flux $\Phi_t$ of $Y$ as a dynamical symmetry. 
Alternatively,
\beann
\Lie(X_h)\Lie(Y)f&=&
\inn(X_h)\d\Lie(Y)f=\inn(X_h)\Lie(Y)\d f
\\ &=&\Lie(Y)\inn(X_h)\d f-\inn([Y,X_h])\d f
=\Lie(Y)\Lie(X_h)f=0 \ .
\eeann
\een
 \qed  \end{proof}

\subsection{Noether symmetries. Noether Theorem}

In the above study of symmetries of the Hamiltonian system $(M,\Omega,\alpha=\d h)$ we have not used the specific structure of the system. In fact, the only element we have considered is the dynamical vector field $X_h\in\vf(M)$. No mention have been made to the symplectic structure of $M$, 
the Hamiltonian function $h$, or the Hamiltonian form $\alpha$. We have studied symmetries of the dynamical vector field independently of its origin.

If we consider the other elements defining the system, we arrive to different types of dynamical symmetries  depending on whether they leave invariant the geometric structure;
that is, the symplectic form or the dynamical elements
(i.e., the Hamiltonian function). 
For a more complete analysis of all these kinds of symmetries
and how they generate conserved quantities, see, for instance, \cite{Ch-2003,Ch-2005,NRR-2018,SC-81,SC-81b}.

The following kind of symmetries has a special relevance as
generators of constants of motion:

\begin{definition}
A diffeomorphism $\Phi\in{\rm Diff}\, (M)$ is a  \textbf{Noether symmetry}
of the Hamiltonian system $(M,\Omega,h)$ if:
\ben
\item
$\Phi$ is a symplectomorphism in $M$: that is, $\Phi^*\Omega =\Omega$.
\item
$\Phi$ lets the dynamics invariant ; that is, $\Phi^*h =h$.
\een
\label{sime1}
\end{definition}

\begin{remark}{\rm 
The second condition can also be expressed in the form
$\Phi^*h =h+c$ (with $c\in\Real$).
This means that it transforms a Hamiltonian function into another 
Hamiltonian function of the same dynamical vector field $X_h$.
More generically, if $\alpha=\d h$ is the Hamiltonian $1$-form, this is equivalent to say that $\Phi^*\alpha=\alpha$.
}\end{remark}

\begin{definition}
A vector field $Y\in\vf (M)$ is an \textbf{infinitesimal Noether symmetry}
of the system if  the local diffeomorphisms  generated by the flux of $Y$
are dynamical symmetries of the system; that is,
\ben
\item
$\Lie (Y)\Omega =0$; that is, $Y\in\vf_{lh}(M)$.
\item
$\Lie (Y)h =0$.
\een
\label{simeinf1}
\end{definition}

\begin{remark}{\rm 
In the infinitesimal case, as every infinitesimal Noether symmetry
$Y\in\vf(M)$ is a locally Hamiltonian vector field, by the first item of the definition,
we have that, for every ${\rm p}\in M$, there exists an open set $U_{\rm p}\ni p$ 
and $f_Y\in\Cinfty (U_{\rm p})$ such that $\inn(Y)\Omega=\d f_Y$, in $U_{\rm p}$. The local Hamiltonian $f$ is unique, up to the sum of constant functions. 
}\end{remark}

A first relevant result is:

\begin{prop}
\ben
\item
Every Noether symmetry of $(M,\Omega,h)$ is a dynamical symmetry.
%If $\Phi\in{\rm Diff}\, (M)$ is a Noether symmetry then it is a
%dynamical symmetry.
\item
Every infinitesimal Noether symmetry of $(M,\Omega,h)$ is an infinitesimal dynamical symmetry.
\een
\label{simdin}
\end{prop}
\begin{proof}
\ben
\item
Let $\Phi\in{\rm Diff}\, (M)$ be a Noether symmetry. If $X_h\in\vf (M)$ is a solution to the dynamical equations, then
$0=\inn(X_h)\Omega-\d h$, and as 
$\Phi^*\Omega=\Omega$ and $\Phi^*h=h$, we have
$$
0=\Phi^*(\inn(X_h)\Omega-\d h)=\inn(\Phi_*^{-1}X_h)\Phi^*\Omega-\Phi^*\d h=
\inn(\Phi_*^{-1}X_h)\Omega-\d h=\inn(\Phi_*^{-1}X_h)\Omega-\d h \ ;
$$
but as the system regular, the vector field solution is unique,
then $\Phi_*^{-1}X_h=X_h$, and the result holds.
\item
If $Y\in\vf(M)$ is an infinitesimal Noether symmetry, then
$Y\in\vf_{lh}(M)$ and $\Lie(Y)h~=~0$, hence
$$
\inn([Y,X_h])\Omega=\Lie(Y)\inn(X_h)\Omega-\inn(X_h)
\Lie(Y)\Omega=\Lie(Y)\d h=\d\Lie(Y)h=0 \ ,
$$
but being $\Omega$ nondegenerated, this implies that 
$[Y,X_h]=0$, and therefore $Y$ is an infinitesimal dynamical symmetry.
\een
\qed  \end{proof}

Clearly, the set of infinitesimal Noether symmetries is a real vector space. Moreover, it is easy to prove that it is real Lie algebra:

\begin{prop}
 If $Y_1,Y_2\in\vf (M)$ are  infinitesimal Noether symmetries,
then $[Y_1,Y_2]$ is also an infinitesimal Noether symmetry.
\end{prop}
\begin{proof}
First we have that
$\Lie([Y_1,Y_2])\Omega=0$, since $[Y_1,Y_2]\in\vf_{lh}(M)$, because $Y_1,Y_2\in\vf_{lh}(M)$.
Furthermore,
\beann
\Lie([Y_1,Y_2])h &=& \inn([Y_1,Y_2])\d h=
\Lie(Y_1)\inn(Y_2)\d h-\inn(Y_2)\Lie(Y_1)\d h \\ &=&
\Lie(Y_1)\Lie(Y_2)h-\inn(Y_2)\d\Lie(Y_1) h =0 \ .
\eeann
\qed  \end{proof}

In addition, we get:

\begin{prop}
Let $Y\in\vf(M)$ be an infinitesimal Noether symmetry of $(M,\Omega,h)$.
For every ${\rm p}\in M$, there exists an open set $U_{\rm p}\ni {\rm p}$ such that,
if $\vartheta\in\df^1(U_{\rm p})$ is a symplectic potential of $\Omega$
 in $U_{\rm p}$, that is $\d\vartheta=\Omega$, then:
\ben
\item
There exists $\zeta_Y\in\Cinfty(U_{\rm p})$ verifying that 
$\Lie(Y)\vartheta=\d\zeta_Y$.
\item
If $f_Y\in\Cinfty(U_{\rm p})$ is a local Hamiltonian function of $Y$, then
in $U_{\rm p}$, we have that,
\beq
f_Y=\zeta_Y-\inn(Y)\vartheta
\qquad \mbox{\it (up to the sum of a constant function)} \ .
\label{fdos}
\eeq
\een
\label{structure}
\end{prop}
\begin{proof}
\ben
\item
Observe that $\Lie(Y)\vartheta$ is a closed form in $U_{\rm p}$ since
$$
\d(\Lie(Y)\vartheta)=\Lie(Y)\d\vartheta=\Lie(Y)\Omega=0 \ ,
$$
Then, by Poincar\'e' s Lemma, 
there exists $\zeta_Y\in\Cinfty(U_{\rm p})$
such that $\Lie(Y)\vartheta=\d\zeta_Y$, in an open subset of $U_{\rm p}$ and, reducing the initial domain if necessary, we can suppose it is in $U_{\rm p}$.
\item
If $\inn(Y)\Omega=\d f_Y$ in $U_{\rm p}$, we obtain that
$$
\d\zeta_Y=\Lie(Y)\vartheta=\d\inn(Y)\vartheta+\inn(Y)\d\vartheta=
\d\inn(Y)\vartheta+\inn(Y)\Omega=\d \{\inn(Y)\vartheta+f_Y\} \ ,
 $$
and the result holds.
\een
\qed  \end{proof}

 Finally, the fundamental result related with infinitesimal Noether-type symmetries 
is the classical {\sl Noether's Theorem}, whose geometric Hamiltonian version is the following:

\begin{teor}
{\rm (Noether)}.
If $Y\in\vf (M)$ is an infinitesimal Noether symmetry of a Hamiltonian system $(M,\Omega,h)$, then
its Hamiltonian function $f_Y$
is a conserved quantity; that is, $$\Lie (X_h)f_Y=0 \ .$$
\label{Nth}
\end{teor}
\begin{proof}
In fact, as $\inn(Y)\Omega=\d f_Y$ we have
$$
\Lie(X_h)f_Y =
\inn(X_h)\d f_Y = \inn(X_h)\inn(Y)\Omega=
-\inn(Y)\inn(X_h)\Omega = -\inn(Y)\d h= -\Lie(Y) h=0 \ .
$$
\qed  \end{proof}

Noethers's Theorem is very relevant, since it gives a way to associate
a constant of motion to every Noether symmetry.
The converse statement of Noether's theorem also holds, and it
allows associating a (Noether) symmetry to every conserved quantity:

\begin{teor} {\rm (Inverse Noether)}:
\label{invNo}
For every conserved quantity $f\in\Cinfty(M)$, its Hamiltonian vector field
$Y_f\in\vf_{lh}(M)$ is an infinitesimal Noether symmetry.
\end{teor}
\begin{proof}
As $Y_f\in\vf_{lh}(M)$, then $\Lie(Y_f)\Omega=0$. Furthermore, as $f$
is a conserved quantity, then $\Lie(X_h)f=0$ and therefore
 $$
\Lie(Y_f) h=\inn(Y_f)\d h=\inn(Y_f)\inn(X_h)\Omega=
-\inn(X_h)\inn(Y_f)\Omega=-\inn(X_h)\d f=-\Lie(X_h)f=0 \ .
 $$
 \qed  \end{proof}

In general, for symmetries which are not of Noether-type,
there is not a so direct way of obtaining constants of motion, 
except in some particular cases like in the following:

\begin{teor}
\label{teo:biham}
If $Y\in\vf (M)$ is an infinitesimal dynamical symmetry such that
$\Lie(Y)h\not=0$, then the function
$f=\Lie(Y)h$ is a conserved quantity.
\end{teor}
\begin{proof}
We have that
$$
\Lie(X_h)f=\Lie(X_h)\Lie(Y)h=\Lie([X_h,Y])h+\Lie(Y)\Lie(X_h)h=0 \ .
$$
\qed  \end{proof}

For other kinds of symmetries, the way to obtain conserved quantities is, in general, 
more complicated (see, for instance, \cite{AA-78,CMR-2002,Cr-83,LMR-99}, and \cite{NRR-2018}
and the references quoted therein,
where a complete classification of the symmetries of Hamiltonian systems is done).

\subsection{Actions of Lie groups on symplectic manifolds}
\label{cami1}

In this section, we analyze the transformations
(i.e., diffeomorphisms) of a dynamical system
which are generated by Lie groups.
In fact, a family of diffeomorphisms in a differentiable manifold
endowed with the operation of composition
has the structure of group and when the
variations of these transformations are considered, the set of
these variations has a smooth differentiable structure.
In this case, the corresponding group of transformations is a Lie group. 
(For more information about actions of Lie groups, see,
for instance, \cite{AM-78,Ar-89,CDW-87,CP-adg,Gi-hss,GGK-2002,LM-sgam,MSSV-85,Ok-87,So-ssd,Wa-fsmlg}
\footnote{We would also like to highlight the work of {\it J.A. L\'azaro-Cam\'i}, 
who compiled the main results presented in Sections 2.3.5, 2.3.6, and 2.3.7 in an unpublished note. }.
See also the appendix \ref{Liega} for a review about Lie groups).

\begin{definition}
Let $G$ be a Lie group and $M$ a differentiable manifold.
A \textbf{left-action} of $G$ on $M$ is a map
$\phi \colon G \times M \to M$
verifying the following properties:
\begin{description}
\item[{\rm (i)}] \ 
$\phi(g_1g_2,{\rm p}) = \phi (g_1,\phi (g_2,{\rm p}))$;
for $g_1,g_2\in G$ and ${\rm p}\in M$.
\item[{\rm (ii)}]
If $e\in G$ denotes the neutral element, then
$\phi (e,{\rm p}) = {\rm p}$.
\end{description}
A \textbf{right action} is defined in an analogous way,
changing the operation law in $G$ by the opposite.
In any case, $M$ is said to be a
\textbf{left $G$-manifold} (resp., a \textbf{right $G$-manifold}).

Since, for every $g \in G$ and $p\in M$, the map \ 
$\Phi_g \colon {\rm p}\mapsto \phi (g,{\rm p})$ \ 
is a diffeomorphism in $M$, we have that a left-action of $G$ in $M$
is also a homomorphism
$$
\begin{array}{ccccc}
\Phi&\colon&G&\to&\Diff (M)
\\
& &g&\mapsto&\Phi_g
\end{array} \ .
$$
\end{definition}

From now on, we only consider left-actions.

\begin{definition}
let $G$ be a Lie group, $M$ a differentiable manifold,
$\Phi$ an action of $G$ on $M$, and ${\rm p}\in M$.
\begin{enumerate}
\item
The \textbf{isotropy group} of $p$  (with respect to $\Phi$) is the subgroup of $G$
$$
G_{\rm p}:= \{ g \in G \ | \ \Phi_g({\rm p})={\rm p} \} \ .
$$
\item
The \textbf{orbit} of ${\rm p}$ (with respect to $\Phi$) is the set
$$
{\cal O}_{\rm p}:= \{{\rm p}' \in M \ | \ {\rm p}' = \Phi_g({\rm p}) \ , \ \mbox{\rm for every $g \in G$} \}
$$
\item
The action is \textbf{effective} or \textbf{faithful}
if $\dst\bigcap_{{\rm p} \in M}G_{\rm p}=\{ e \}$
or, what is equivalent, if
$\Phi_g = {\rm Id}_M \ \Longleftrightarrow \ g=e$
(that is, $\Phi$ is injective).
\item
The action is \textbf{free} if the following map, defined for a fixed ${\rm p}\in M$,
$$
\begin{array}{ccc}
G&\to&M
\\
g&\mapsto&\Phi_g({\rm p})
\end{array}
$$
is injective, for every ${\rm p}\in M$
(this means that there are no invariant points in $M$ under the action).
\item
The action is  \textbf{transitive} if, for every ${\rm p}_1,{\rm p}_2\in M$, 
there exists $g \in G$ such that \ $\Phi_g({\rm p}_1)={\rm p}_2$ \ 
or, what is the equivalent, if \
${\cal O}_{\rm p}=M$, for every ${\rm p}\in M$
(that is, $\Phi$ has only one orbit).
In this case, $M$ is said to be a \textbf{homogeneous $G$-space}.
\item
The action is \textbf{proper} if the anti-image of every  compact set by the map
$(g,{\rm p}) \mapsto (\Phi_g({\rm p}),{\rm p})$, for every $g\in G$ and ${\rm p}\in M$,
is also a compact set.
\end{enumerate}
\end{definition}

\begin{definition}
Let $M_1,M_2$ be differentiable manifolds,
$G$ a Lie group and
$\Phi_1 \colon G\times M_1 \to M_1$, $\Phi_2 \colon G\times M_2 \to M_2$ 
actions of $G$ on $M_1$ and $M_2$.
A map $F \colon M_1 \to M_2$ is \textbf{equivariant}
with respect to these actions if, for every $g\in G$, the following diagram commutes
$$
\begin{array}{ccccc}
& M_1 & \mapping{\Phi_{1_g}} & M_1 &
\\
F & \Big\downarrow & & \Big\downarrow & F
\\
& M_2 & \mapping{\Phi_{2_g}} & M_2 &
\end{array}
$$
\end{definition}

Remember that the {\sl Lie algebra} of  a Lie group $G$ is
${\bf g}:=\Tan_eG$ or, equivalently, the set of
{\sl left-invariant vector fields} of $\vf(G)$
(that is, the set of vector fields which are invariant
by the action induced on $\vf(G)$ by the left-action $\Phi$ of $G$ on $G$). 
Then, every action of $G$ on $M$ induces a Lie algebra-homomorphism 
$$
\begin{array}{ccccc}
\xi&\colon&{\bf g}&\to&\vf(M)
\\
& &X&\mapsto&\xi_X
\end{array}
$$
defined as follows: the uniparametric subgroup of $G$
generated by $X$ determines another 
uniparametric subgroup of transformations in $M$
$$
\begin{array}{ccccc}
\sigma_t&\colon&M&\to&M
\\
& & {\rm p}&\mapsto&\Phi_{\alpha(-t)}({\rm p}) := \Phi_{{\rm exp}(-tX_e)}({\rm p})
\end{array} \ ,
$$
where ${\rm exp}\colon\Tan_eG\to G$ denotes the 
{\sl exponential map} of the group. Then, the map
$$
\begin{array}{ccccc}
\sigma&\colon&\Real \times M&\to&M
\\
& &(t,{\rm p})&\mapsto&\sigma_t({\rm p}) 
\end{array}
$$
is the flux of some vector field $\xi_X \in\vf(M)$;
that is, for ${\rm p}\in M$ and $t \in \Real$,
$\sigma (t,{\rm p})$ is the integral curve of $\xi_X$
passing through ${\rm p}$. This vector field is given by
$$
\xi_X({\rm p}) = \frac{\d }{\d t}\sigma (t,{\rm p}) = 
\frac{\d }{\d t}\sigma_t({\rm p}) = 
\frac{\d }{\d t}\Phi_{{\rm exp}(-tX_e)}({\rm p})\Big\vert_{t=0} \ .
$$
Hence, for every $f \in\Cinfty(M)$, we have that
$$
(\xi_X(f))({\rm p}) = \frac{\d }{\d t}f(\Phi_{{\rm exp}(-tX_e)}({\rm p}))\Big\vert_{t=0}\ .
$$

\begin{definition}
For $X \in {\bf g}$,  the vector field $\xi_X\in\vf(M)$ is the
\textbf{infinitesimal generator} or \textbf{fundamental vector field}
of the action $\Phi$ associated to $X$.
\end{definition}

\begin{definition}
Let $\Phi$ be an action of $G$ on $M$.
\begin{enumerate}
\item
A $k$-form $\vartheta \in\df^k(M)$ is \textbf{$G$-invariant} by $\Phi$ 
if $\Phi_g^*\vartheta =\vartheta$, for every $g \in G$;
or what is equivalent, if
$\Lie(\xi_Z)\vartheta = 0$, for every $Z \in {\bf g}$.
\item
A vector field $Z\in\vf(M)$ is \textbf{$G$-invariant} by $\Phi$ 
if $\Phi_{g_*}Z =Z$, for every $g \in G$;
or what is equivalent, if
$\Lie(\xi_X)Z= 0$, for every $X\in{\bf g}$.
\end{enumerate}
\end{definition}

\begin{prop}
Let $\psi \colon G_1 \to G_2$ be a Lie group homomorphism.
Then the following diagram commutes:
$$
\begin{array}{ccccc}
& G_1 & \mapping{\psi} & G_2 &
\\
{\rm exp} & \Big\uparrow & & \Big\uparrow & {\rm exp}
\\
& \Tan_{e_1}G_1 \simeq {\bf g}_1 & \mapping{\psi_*} 
& \Tan_{e_2}G_2 \simeq {\bf g}_2 &
\end{array}
$$
that is,
$\psi({\rm exp}(tX_{e_1})) = {\rm exp}(t(\psi_*X)_{e_2})$.
\label{conmutexp}
\end{prop}
\begin{proof}
Immediate observing that  \
$t \mapsto {\rm exp}(t(\psi_*X)_{e_2})$ \ 
is the only uniparametric subgroup of $G_2$
whose tangent vector field at $t=0$ is $(\psi_*X)_{e_2}$.
\\ \qed  \end{proof}

A very relevant type of group actions are the following:

\begin{definition}
Let $G$ be a Lie group, $(M,\Omega)$ a symplectic manifold, and
$\Phi \colon G\times M \to M$ an action of $G$ on $M$.
We say that $\Phi$ is a \textbf{symplectic action} of $G$ on $M$
(also that $G$  \textbf{acts symplectically} on $M$ by $\Phi$)
if $\Phi_g$ is a symplectomorphism, for every $g \in G$;
that is, $\Phi_g^*\Omega = \Omega$.
Then $M$ is said to be a \textbf{symplectic $G$-space}.
\end{definition}

According to this definition, if $G$ acts symplectically on $M$ by $\Phi$,
then the fundamental vector field $\xi_X$ associated to $X$ by $\Phi$ 
is a locally Hamiltonian vector field and, conversely,
if  $\xi_X\in\vf_{lh}(M)$, for every $X \in{\bf g}$, then 
$\Phi$ is a symplectic action of $G$ on $M$. Therefore:

\begin{prop}
Let $\xi\colon{\bf g}\to\vf(M)$ be the map such that 
$\xi(X):=\xi_X$, for $X\in {\bf g}$.
Then, $\Phi$ is a symplectic action of $G$ on $M$ if, and only if,
${\rm Im}\,\xi\subseteq\vf_{lh}(M)$.

This means that
$\Lie(\xi_X)\Omega = 0$, for every $X \in {\bf g}$,
(that is, $\inn(\xi_X)\Omega$ is a closed form).
\end{prop}

\begin{definition}
Let $G$ be a Lie group, $(M,\Omega )$ a symplectic manifold and
$\Phi \colon G \to M$ a symplectic action of $G$ on $M$.
We say that $\Phi$ is a \textbf{strongly symplectic action} of $G$ on $M$ 
if $\xi_X\in\vf_{h}(M)$, for every $X \in{\bf g}$;
or, what is equivalent, $\inn(\xi_X)\Omega$ is an exact form.

In this case $M$ is said to be a
\textbf{strongly symplectic $G$-space} and
$\Phi$ is a \textbf{Hamiltonian action} of $G$ on $M$.
Otherwise, it is a \textbf{locally Hamiltonian action} of $G$ on $M$.
\end{definition}

In order to discuss the obstruction for a symplectic action to be strongly symplectic, 
let ${\bf g}_h:=\{ X\in{\bf g}\ \mid\ \xi(X)=\xi_X\in\vf_H(M)\}$
and consider the following sequences of Lie algebras:
\beq
\begin{array}{ccccccccc}
0&\to&{\bf g}_h&\to&
{\bf g}&\to&{\bf g}/{\bf g}_h&\to&0
\\
& &\Big\downarrow \ \xi& &\Big\downarrow \ \xi& &\Big\downarrow \ \widetilde\xi& &
\\
0&\to&{\cal X}_h(M)&\to&
{\cal X}_{lh}(M)&\to&H^1(M)&\to&0
\end{array}\ ,
\label{diag}
\eeq 
where $H^1(M)$ denotes the {\sl first de Rham's cohomology group} of $M$,
and $\widetilde \xi$ is a Lie algebra homomorphism which makes the diagram commutative.
Then, the image of {\bf g} by $\xi$ is in ${\cal X}_h(M)$
(that is, the action is strongly symplectic) if, and only if,
$\widetilde \xi = 0$.

There are two specially relevant cases
for which every symplectic action is strongly symplectic:
when $M$ is simply connected
(then every closed form is exact and,
therefore, ${\cal X}_{lh}(M) ={\cal X}_h(M)$) and
if $G$ is semisimple (then ${\bf g} = [{\bf g},{\bf g}]$ 
and  $[{\cal X}_{lh}(M),{\cal X}_{lh}(M)] \subset {\cal X}_h(M)$).

As a particular case, we have:

\begin{definition}
Let $G$ be a Lie group, $(M,\Omega)$ an exact symplectic manifold 
(that is, $\Omega=\d\Theta$ for some $\Theta\in\df^1(M)$)
and $\Phi \colon G \to M$ an action of $G$ to $M$.
$\Phi$ is said to be an
\textbf{exact action} of $G$ on $M$
if, $\Phi_g^*\Theta = \Theta$, for every $g \in G$.
\end{definition}

Of course,  every exact action is strongly symplectic.

\subsection{Comomentum and momentum maps}
\label{cami2}

Next, we introduce new geometric elements
which are very relevant in the theory of symplectic group actions.
First, following \cite{So-ssd} we define:

\begin{definition}
Let $G$ a Lie group which acts symplectically on a symplectic manifold $(M,\Omega )$.
A \textbf{comomentum map} associated to this action
is every Lie algebra linear application (if it exists)
$$
\begin{array}{ccccc}
{\rm j}^*&\colon&{\bf g}&\to&\Cinfty (M)
\\
& &X&\mapsto&f_X:={\rm j}^*(X)
\end{array} \ ,
$$
such that the following diagram commutes
$$
\xymatrix{
& & \ar[dl]_<(0.45){{\rm j}^*} {\bf g} \ar[d]^<(0.45){\xi}&
\\
0 \to \Real 
\ar[r]&\Cinfty (M)\ar[r]_<(0.25){\sharp_\Omega\circ\d}&\vf_{lh}(M)\ar[r]&H^1(M) \to 0
} \ ,
$$
or, what is equivalent,
$$
\inn(\xi_X)\Omega = \d f_X \ .
$$
\label{com}
\end{definition}

The obstruction to the existence of comomentum maps
is given in the following:

\begin{prop}
Let $G$ be a Lie group which acts symplectically on $(M,\Omega)$.
A comomentum map associated to this action exists
if, and only if, the map \ 
$\widetilde\xi \colon {\bf g}/{\bf g}_h \to H^1(M)$ \ 
in \eqref{diag} reduces to be $\widetilde \xi = 0$.
\end{prop}
\begin{proof}
By definition, if a comomentum map exists, then \ 
$\sharp_\Omega~\circ~\d~\circ~{\rm j}^* = \xi$.
But, as ${\rm Im}\ \xi \subset {\rm Im}\ \sharp_\Omega$,
this implies that $\widetilde\xi = 0$.

Conversely, if $\widetilde \xi = 0$, then \
${\rm Im}\ \xi := \xi ({\bf g}) \subset {\cal X}_h(M)$ \
and, for all $X \in {\bf g}$, there exists 
$f_X \in \Cinfty (M)$ such that $\inn(\xi_X)\Omega = \d f_X$.
\\ \qed  \end{proof}

Observe that this is also the obstruction
for a symplectic action to be strongly symplectic.
Hence:

\begin{prop}
Let $G$ be a Lie group which acts symplectically on $(M,\Omega)$.
A comomentum map associated to this action exists
if, and only if, the action is strongly symplectic.
\end{prop}

\begin{remark}{\rm 
If a comomentum map exists, it is not unique.
In fact, if $f_C \colon {\bf g} \to \Real$
is a continuous linear function, then $f_C \in {\bf g}^*$
and then $f_C(X) = ctn.$, for every $X \in {\bf g}$.
Therefore, if ${\rm j}^*$ is a comomentum map,
so is ${\rm j}^{'*}={\rm j}^*+f_C$.
}\end{remark}

Furthermore, if a comomentum map exists, it is not a
Lie algebra homomorphism necessarily since,
for $X,Y \in {\bf g}$, we have that
$$
\sharp_\Omega \circ \d \{ f_X,f_Y \} =
\sharp_\Omega\inn([\xi_X,\xi_Y])\Omega
= -\sharp_\Omega\inn(\xi_{[X,Y]})\Omega 
= -\sharp_\Omega \d f_{[X,Y]}
$$
therefrom
$\d \{ f_X,f_Y \} = -\d f_{[X,Y]}$ and then
$\{ f_X,f_Y \} = -f_{[X,Y]} + \sigma (X,Y)$
where $\sigma \colon {\bf g}\times{\bf g} \to \Real$
is a skew symmetric bilinear function
whose existence measures the obstruction for the
comomentum map to be a Lie algebra homomorphism.
This leads to state:

\begin{definition}
Let $G$ be a Lie group which acts symplectically on $(M,\Omega)$.
The action is said to be a \textbf{Poissonian action} 
or also a \textbf{strongly Hamiltonian action} if
\begin{description}
\item[{\rm (i)}] \ 
There exists a comomentum map for this action.
\item[{\rm (ii)}]
The comomentum map is a Lie algebra (anti)homomorphism.
\end{description}
Then the triple $(M,\Omega,{\rm j}^*)$ is called a
\textbf{Hamiltonian $G$-space}.
\end{definition}

As a particular case, we have that:

\begin{prop}
If $(M,\Omega )$ is an exact symplectic manifold
and the action of $G$ on $M$ is exact, then it is Poissonian.
\label{exact}
\end{prop}
\begin{proof}
Let $\Omega =\d \Theta$.
Defining $f_X := -\Theta (\xi_X)$, then
${\rm j}^* \colon X \mapsto f_X$ is a comomentum map and
$$
f_{[X,Y]}=-\Theta (\xi_{[X,Y]}) = 
\Theta ([\xi_X,\xi_Y]) = \inn([\xi_X,\xi_Y])\Theta
=\Lie(\xi_X)\inn(\xi_Y)\Theta = 
-\Lie(\xi_X)f_Y = -\{ f_X,f_Y \} \ .
$$
\qed  \end{proof}

\begin{definition}
\cite{So-ssd}.
Let $G$ be a Lie group which acts 
in a strongly symplectic way on a symplectic manifold $(M,\Omega)$.
A \textbf{momentum map} associated to this action
is the dual map of a comomentum map; that is, a map
$$
\begin{array}{ccccc}
{\rm J}&\colon&M&\to&{\bf g}^*
\\
& &{\rm p}&\mapsto&\mu
\end{array}
$$
such that, for every $X \in {\bf g}$ and for ${\rm p}\in M$,
$$
\langle X,{\rm J}({\rm p})\rangle:={\rm j}^*(X)({\rm p})=f_X({\rm p}) \ .
$$
(Observe that the map
$\mu:= {\rm J}({\rm p})\colon{\bf g}\to\Real$ such that 
$X\mapsto f_X({\rm p})$, for $X\in{\bf g}$,
is, in fact, an element of ${\bf g}^*$).
\end{definition}

Obviously, the obstruction to the existence of momentum maps is
the same that for the comomentum maps.
In particular, as an immediate corollary of Proposition
\ref{exact} we have:

\begin{prop}
If $(M,\Omega )$ is an exact symplectic manifold
and the action of $G$ on $M$ is exact,
then a momentum map exists. It is given by
$\langle X,{\rm J}({\rm p})\rangle:= -\Theta (\xi_X)({\rm p})$,
for every $X \in {\bf g}$, for ${\rm p}\in M$.
\end{prop}

\begin{remark}{\rm 
As a more general result,
if $G$ is a connected Lie group and $\Phi$ a strongly symplectic action
of $G$ on $(M,\Omega)$, then the action is Poissonian
if, and only if,
the momentum maps associated to this action 
are $Ad^*$-equivariant, that is,
for every $g \in G$, we have that $Ad^*\circ {\rm J}=\Phi_g\circ{\rm J}$;
where $Ad^*\colon{\bf g}^*\to{\bf g}^*$ is the {\sl coadjoint action} of $G$.
(See \cite{So-ssd} for the proof).
}\end{remark}

\subsection{Reduction by symmetries}
\label{cami3}

Finally, we introduce the concept of 
{\sl group of symmetries} of a Hamiltonian system,
its relation with the momentum map,
and give the main insights about the theory of reduction by symmetries.

\begin{definition}
Let $G$ be a Lie group, $(M,\Omega,h)$ a Hamiltonian system, 
and $\Phi \colon G \to M$ an action of $G$ on $M$.
We say that $G$ is a \textbf{symmetry group} of the Hamiltonian system if
\begin{description}
\item[{\rm (i)}] \ 
$\Phi$ is a symplectic action.
\item[{\rm (ii)}]
For every $g \in G$, we have $\Phi_g^*h = h$.
\end{description}
The diffeomorphism $\Phi_g$, for every $g \in G$, is called a
\textbf{symmetry} of the Hamiltonian system.
\end{definition}

From this definition, it is obvious that 
$G$ is a symmetry group of a Hamiltonian system $(M,\Omega,h)$
if, and only if,  $\Lie(\xi_X)h = 0$, for every $X \in {\bf g}$. 

In this context, the Noether Theorem is stated as follows::

\begin{teor}
({\sl \textbf{Noether's theorem}}).
Let $G$ be a symmetry group which acts in a strongly symplectic way
on the Hamiltonian system $(M,\Omega,h)$.
Then, for every $X \in {\bf g}$, the functions $f_X$
are constants of motion; that is, $\Lie(X_h)f_X = 0$.
\end{teor}
\begin{proof}
We have that $\inn(X_h)\Omega=\d h$, then
$$
\Lie(X_h)f_X=\inn(X_h)\d f_X 
= \inn(X_h)\inn(\xi_X)\Omega
=-\inn(\xi_X)\inn(X_h)\Omega = -\inn(\xi_X)\d h
=-\Lie(\xi_X)h = 0 \ .
$$
%(A similar result can also be stated for locally Hamiltonian systems).
\\ \qed  \end{proof}

One of the most important features on the study of
Hamiltonian systems with symmetry is the so-called
{\sl reduction theory}.
Now we are going to state the main results on this topic
(see, for instance, \cite{CM-2009,MMOPR-2007,MR-99,MW-74,MW,Or-2002,OR-2002,OR-2004,We-77} for deeper explanations and details).

First, we have to introduce the following concept:

\begin{definition}
Let $G$ be a Lie group and a strongly symplectic action
of $G$ on $(M,\Omega)$.
Let ${\rm J}$ be a momentum map associated to this action. Then, 
$\mu \in {\bf g}^*$ is a 
\textbf{weakly regular value} of ${\rm J}$ if:
\begin{description}
\item[{\rm (i)}] \
${\rm J}^{-1}(\mu )$ is a submanifold of $M$.
\item[{\rm (ii)}]
For every ${\rm p}\in {\rm J}^{-1}(\mu )$, 
$\Tan_{\rm p}({\rm J}^{-1}(\mu)) = \ker \Tan_{\rm p}{\rm J}$.
\end{description}
If $\Tan_{\rm p}{\rm J}$ is surjective, then $\mu$ is said to be a \textbf{regular value}. 
(Of course,  every regular value is weakly regular).
\end{definition}

Then we have:

\begin{teor}
Let $G$ be a Lie group acting on a symplectic manifold 
$(M,\Omega )$ such that the action is 
Poissonian, free and proper.
Let $\mu$ be a weakly regular value of the
momentum map associated to this action,
and let $G_{\mu}$ be the isotropy group of
${\rm J}^{-1}(\mu )$ for the action of $G$ on $M$.
%(or, what is the same thing, of $\mu$  with respect of the coadjoint action of $G$ on ${\bf g}^*$).
Then:
\begin{enumerate}
\item
The submanifold 
${\rm J}^{-1}(\mu )$ is stable under the action of $G_\mu$ and
so the quotient $J^{-1}(\mu )/G_\mu$ is well-defined
(it is called the
\textbf{orbit space} of ${\rm J}^{-1}(\mu )$).
\item
{\rm (Marsden--Weinstein's theorem)}.
Denote by $\pi \colon{\rm J}^{-1}(\mu ) \to{\rm J}^{-1}(\mu )/G_\mu$
the canonical projection and by
$\iota \colon{\rm J}^{-1}(\mu ) \hookrightarrow M$ the embedding.
Then, ${\rm J}^{-1}(\mu )/G_\mu$ is a differentiable manifold
which is endowed with a (unique) symplectic structure
$\widehat\Omega$ such that
$\iota^*\Omega = \pi^*\widehat\Omega$.
\item
If $G$ is a symmetry group of the
Hamiltonian system $(M,\Omega ,h)$,
then $({\rm J}^{-1}(\mu )/G_\mu ,\widehat\Omega ,\widehat h)$
is a Hamiltonian system, where
$\iota^*h = \pi^*\widehat h$.
\end{enumerate}
\end{teor}
\begin{proof}
\begin{enumerate}
\item
In fact, if $g \in G_{\mu}$, for all ${\rm p}\in{\rm J}^{-1}(\mu )$ we have
$$
\Phi_g({\rm p}) = \Phi_g({\rm J}^{-1}(\mu)) =
{\rm J}^{-1}(g\mu )) = {\rm J}^{-1}(\mu ) \ ,
$$
then ${\rm J}^{-1}(\mu )$ is invariant under the action of $G_{\mu}$
and the quotient is well-defined.
\item
We omit he proof of this statement because it is long and heavy.
It can be found in the above-mentioned references.
\item
It is evident, since $h$ is $G$-invariant
(and hence it is also $G_\mu$-invariant),
therefore it projects onto 
${\rm J}^{-1}(\mu )/G_\mu$.
\end{enumerate}
\qed  \end{proof}

%%%%%%%%%%%%%%%%%%%%%%%%%%%%%%%%%%%%%%%%%%%%%%%%%%%%%%%%%%%%%%%%%%%%%%%%%%%

\chapter{Symplectic mechanics (II): Autonomous Lagrangian dynamical systems}
\protect\label{sdl}

The aim of this chapter is to state a geometrical formulation 
for a particular kind of dynamical systems: those which are described by
Lagrangian functions and that, consequently, are {\sl variational}.
These systems admit a description as Hamiltonian dynamical systems,
but they present relevant particular characteristics.
This formulation gives rise to two formalisms: the so-called {\sl Lagrangian formalism} 
and its associated {\sl canonical Hamiltonian formalism}.
These kinds of systems include, as particular cases, a wide range of mechanical systems
which are very relevant in physics: the {\sl conservative Newtonian systems}, 
which are those whose configuration space is endowed with a metric
(they are studied in Chapter \ref{sdn}).

There are several alternative ways to state the geometric description for these kinds of systems. 
The geometric framework where the theory is developed is the tangent bundle of a manifold, $\Tan Q$,
which represents the phase space of the system.
A selected function on it, the {\sl Lagrangian} $\Lag$, 
is a depository of the physical information of the system.
Then, a first formulation consists in using $\Lag$ to define the so-called
{\sl Legendre transformation} which connects $\Tan Q$ and $\Tan^* Q$
and, if this map is a (local) diffeomorphism, use it to translate the Hamiltonian description previously stated in $\Tan^*Q$ to $\Tan Q$ (see, for instance, \cite{AM-78}).

The second and more general approach consists in using the canonical structures of the tangent bundle
to construct, starting from $\Lag$, the geometrical objects needed to establish the dynamics of the system.
This is the so-called {\sl Klein} or {\sl Cartan formulation} of the Lagrangian formalism, which was stated in the classical references
\cite{Ga-52,Go-69,Grif-72a,Grif-72b,klein} and developed later by many other authors 
\cite{CMSVZ-76,Cr-81,Cr-83b,CP-adg,dLe89,SCC-84}.
It has the advantage that is more direct and
does not need any previously established Hamiltonian formalism
(in fact, the existence of this Hamiltonian formalism is not 
assured  for some singular Lagrangian systems).
This is the procedure that we follow in this exposition.

Another elegant alternative was developed in \cite{Tu-76a,Tu-76b},
using the concept of {\sl special symplectic manifold} and of
{\sl Lagrangian submanifold}, and obtaining both the Lagrangian and Hamiltonian formalisms
for Lagrangian systems and their equivalence.
Finally, another unified Lagrangian-Hamiltonian formalism for these systems
is given in \cite{SR-83}. 
These two approaches have subsequently been generalized 
and applied in many physical contexts.

The structure of the chapter is as follows: 
first, we introduce the mathematical framework which is needed 
to state the Lagrangian formalism and its associated canonical Hamiltonian formalism; 
that is, the canonical geometric structures of the tangent and the cotangent bundles of a manifold.
Next, we develop the Lagrangian formalism for {\sl dynamical Lagrangian systems}
and its associated {\sl canonical Hamiltonian formalism}, 
their equivalence, and the {\sl Hamilton--Jacobi theory}
for the Hamiltonian formalism.
Furthermore, we describe the {\sl Skinner-Rusk formalism,} which is
a nice formulation that unifies both the Lagrangian and the Hamiltonian formalism.
We also discuss the symmetries and conserved quantities in both formalisms,
paying special attention to those which are canonical lifts of diffeomorphisms and vector fields (which are called {\sl natural\/}),
and studying the equivalence of Lagrangians in this context.
Finally, we study the variational formulation introducing the {\sl Hamilton} and the {\sl Hamilton--Jacobi variational principles}
for the Lagrangian and the Hamiltonian formalisms, respectively.
Some relevant physical systems, the harmonic oscillator and the Kepler problem, are analyzed using this geometric treatment.

\section{Geometric structures of the tangent and cotangent bundles}
\protect\label{egft}

(See \cite{Cr-81,Cr-83b,CP-adg,dLe89,Ga-52,Go-69,Grif-72a,Grif-72b,klein,SCC-84}).

The tangent bundle of a manifold $Q$, denoted as   $\Tan Q$ (whose construction and main characteristics are reviewed in the appendix \ref{sec:tangentb})
is endowed with three canonical geometric structures:
the {\sl vertical subbundle}, the {\sl vertical endomorphism} and the {\sl Liouville vector field}.
Furthermore, there are some characteristic vector fields in $\Tan Q$,
those whose integral curves are obtained as solutions to 
{\sl second-order differential equations} in $Q$.
Similarly, the {\sl cotangent bundle} of $Q$, denoted as $\Tan^*Q$
is endowed with some canonical differential forms.
Next, we present and discuss all these topics in detail.

Along this section, $Q$ is a differentiable manifold with $\dim\,Q=n$.

\subsection{The vertical subbundle. Vertical lift}

Considering the canonical projection $\tau_Q \colon \Tan Q \to Q$,
and its tangent map $\Tan\tau_Q\colon\Tan\Tan Q\to\Tan Q$.

\begin{definition}
Let $(q,u)\in\Tan Q$ and ${\rm V}_{(q,u)}(\tau_Q):=\ker \Tan_{(q,u)}\tau_Q$.
The \textbf{vertical subbundle} of $\Tan\Tan Q$ is the vector bundle (of rank $n$)
${\rm V}(\tau_Q)\to\Tan Q$, where
$$
{\rm V}(\tau_Q) := \bigcup_{(q,u)\in\Tan Q}{\rm V}_{(q,u)}(\tau_Q) \ .
$$

The sections of this bundle ${\rm V}(\tau_Q)\to\Tan Q$ are called \textbf{vertical vector fields},
and the set of all these vector fields is denoted by
$\vf^{V(\tau_Q)}(\Tan Q)$.
\end{definition}

It is immediate to prove that, in natural coordinates of $\Tan Q$,
the expression of these vector fields is
\(\dst f^i\derpar{}{v^i}\); that is, $\vf^{V(\tau_Q)}(\Tan Q)$
is locally generated by the set \(\dst\left\{\derpar{}{v^i}\right\}\).

Another interpretation of the fibers of ${\rm V}(\tau_Q)$ is the following:
for every $q\in Q$, consider the $n$-dimensional vector space $\Tan_qQ$,
as a differentiable manifold and the natural immersion
$$
\begin{array}{cccc}
j_q \colon & \Tan_qQ & \to & \Tan Q
\\
& u & \mapsto & (q,u)
\end{array} \ .
$$
Observe that $\tau_Q\circ j_x$ is the constant map equal to $x$. 
If $u \in \Tan_qQ$, we have that
$$
\Tan_uj_q \colon  \Tan_u\Tan_qQ \to \Tan_{(q,u)}\Tan Q \ ;
$$
but, as $\tau_Q\circ j_q$ is a constant map,
then $\Tan_u(\tau_Q\circ j_q)=\Tan_{(q,u)}\tau_Q\circ\Tan_uj_q=0$.
This is equivalent to say that
${\rm Im}\,\Tan_uj_q\subseteq\ker \Tan_{(q,u)}\tau_Q={\rm V}_{(q,u)}(\tau_Q)$
and, as $j_q$ is an immersion,
$$
\dim\,{\rm Im}\,\Tan_uj_q=\dim\,(\Tan_u\Tan_qQ=n=\dim\,{\rm V}_{(q,u)}(\tau_Q) \ ,
$$
hence ${\rm Im}\,\Tan_uj_q={\rm V}_{(q,u)}(\tau_Q)$,
(since both of them have the same dimension and the first one is a subset of the second one).

As a consequence of the above discussion, ${\rm V}_{(q,u)}(\tau_Q)$
is identified naturally with $\Tan_u\Tan_qQ$,
through the isomorphism induced by $\Tan_uj_q$ onto its image.
Furthermore, as $\Tan_qQ$ is a vector space,
if $u \in \Tan_qQ$, we have that $\Tan_qQ$
is canonically identified with $\Tan_u(\Tan_qQ)$ by means of the directional derivative
\footnote{
Remember that, if $F$ is an $n$-dimensional real vector space and
$u \in F$, the natural identification between $F$ and $\Tan_uF$
is given as follows:
\begin{eqnarray*}
F & \to & \Tan_uF
\\
v & \mapsto & {\rm D}_v(u) \ ,
\end{eqnarray*}
where ${\rm D}_v(u)$ denotes the directional derivative with respect to the vector $v$ at the point $u$;
that is, if $f \colon \Real^n\simeq F\to \Real$ is any differentiable function, then
$$
({\rm D}_v(u))f\equiv{\rm D}_vf(u):= \lim_{t \mapsto 0}\frac{f(u+tv)-f(u)}{t} \ .
$$
If $x^1,\ldots ,x^n$ are coordinates in $F$ and
$v=(\lambda^1,\ldots ,\lambda^n)$, then
\(\dst ({\rm D}_v(u))f=\lambda^i\derpar{f}{x^i}\Big\vert_u\);
hence
\(\dst {\rm D}_v(u)=\lambda^i\derpar{}{x^i}\Big\vert_u\),
and the identification is immediate.
}.
Thus, ${\rm V}_{(q,u)}(\tau_Q)$ can also be identified with
 $\Tan_qQ$.
In other words,
${\rm V}(\tau_Q)$ is the pull-back to $\Tan Q$ of the bundle
$\Tan Q$ over $Q$ by means of the map $\tau_Q$;
that is:
$$
\begin{array}{ccccc}
{\rm V}(\tau_Q) & \simeq &\tau_Q^*(\Tan Q)& \to & \Tan Q
\\
& & \tau_{\Tan Q} \downarrow &  & \downarrow \tau_Q
\\
& & \Tan Q& \mapping{\tau_Q} & Q
\end{array} \ .
$$
In this way, we have constructed the  isomorphism of vector spaces
\begin{eqnarray*}
\Tan_qQ & \to & \Tan_u\Tan_q Q\simeq{\rm V}_{(q,u)}(\tau_Q)\subset\Tan_{(q,u)}\Tan Q
\\
v & \mapsto & {\rm D}_v(q,u) \ ,
\end{eqnarray*}
where, if $f \in\Cinfty (\Tan Q)$, then
$$
({\rm D}_v(q,u))f\equiv{\rm D}_vf(q,u)=\lim_{t \mapsto 0}\frac{f(q,u+tv)-f(q,u)}{t} \ .
$$

\begin{definition}
The vector ${\rm D}_v(q,u)$ is called the
\textbf{vertical lift} of $v$ to the point $(q,u)$,
and the map
$$
\begin{array}{ccccc}
\lambda_q^{(q,u)} & \colon &\Tan_qQ & \to & \Tan_{(q,u)}\Tan Q
\\
& & v & \mapsto & {\rm D}_v(q,u)
\end{array}
$$
which implements the above isomorphism is called the
\textbf{vertical lift}.
\end{definition}

In natural coordinates, if \(\dst v=\lambda^i\derpar{}{q^i}\Big\vert_{(q,u)}\),
then \(\dst\lambda_q^{(q,u)}(v)=\lambda^i\derpar{}{v^i}\Big\vert_{(q,u)}\).

The vertical lift is extended in a natural way to the vector fields of $\vf(Q)$ and the mapping $\lambda^V:\vf(Q) \to \vf^{V(\tau_Q)}(\Tan Q)$,
is $\Cinfty(Q)$-lineal.

In natural coordinates, if \(\dst X=f^i\derpar{}{q^i}\),
then its vertical lift is
\(\dst\lambda^V(X)\equiv  X^V=\tau^*_Qf^i\derpar{}{v^i}\).

\subsection{The canonical (or vertical) endomorphism}

\begin{definition}
Let $(q,u)\in\Tan Q$. The map
$$
\begin{array}{cccc}
J_{(q,u)} \colon & \Tan_{(q,u)}\Tan Q & \to & \Tan_{(q,u)}\Tan Q
\\
& Y & \mapsto & {\rm D}_{(\Tan_{(q,u)}\tau_Q)Y}(q,u)=
\lambda_q^{(q,u)}(\Tan_{(q,u)}\tau_Q(Y))
\end{array}
$$
is called the \textbf{canonical} or \textbf{vertical endomorphism} (at $(q,u)$).
\end{definition}

Observe that $J_{(q,u)}Y$ is the vertical lift of $(\Tan_{(q,u)}\tau_Q)Y$
to the point $(q,u)$ (that is, $J_{(q,u)}$ consists in projecting
to $\Tan_qQ$ and to lift vertically).

It is clear that the image of $J_{(q,u)}$ is in
${\rm V}_{(q,u)}(\tau_Q)$ and, as $\Tan_{(q,u)}\tau_Q$ is a surjective map,
it coincides with ${\rm V}_{(q,u)}(\tau_Q)=\ker\, J_{(q,u)}$.
Therefore:

\begin{prop}
The canonical endomorphism has the following properties:
for $(q,u)\in\Tan_qQ$,
\ben
\item
${\rm Im}\, J_{(q,u)}={\rm V}_{(q,u)}(\tau_Q)=\ker\, J_{(q,u)}$.
\item
$(J_{(q,u)})^2=0$.
\een
\end{prop}

The action of $J$ is extended in a natural way to vector fields:
$$
J\colon \vf(\Tan Q) \to\vf^{V(\tau_Q)}(\Tan Q)\subset\vf(\Tan Q) \ ,
$$
and to differential forms:
$$\begin{array}{ccccc}
J^*&\colon&\df^k(\Tan Q)&\longrightarrow&\df^k(\Tan Q) \\
 & & \alpha & \mapsto & \inn(J)\alpha
\end{array} \ ,
$$
where $J^*\alpha=\inn(J)\alpha$ is defined as
$$
[\inn(J)\alpha](\moment{X}{1}{k}):=\alpha(J(X_1),\ldots ,
J(X_i),\ldots ,J(X_k)) \ .
$$
In particular, if $k=1$, then we have that 
$\inn(J)\alpha=\alpha\circ J$, which is a
$1$-form verifying that
$\inn(X)(\inn(J)\alpha)=0$, for every $X\in\vf^{V(\tau_Q)}(\Tan Q)$.
These kinds of differential forms are called {\sl $\tau_Q$-semibasic forms}.

\noindent{\bf Local expressions}:
If $(q^i,v^i)$ are natural coordinates in $\Tan Q$, we have
$$
J_{(q,u)}\derpar{}{q^i}\Big\vert_{(q,u)}=
{\rm D}_{\derpar{}{q^i}\Big\vert_q}(q,u)=\derpar{}{v^i}\Big\vert_{(q,u)}
\quad , \quad
J_{(q,u)}\derpar{}{v^i}\Big\vert_{(q,u)}={\rm D}_0(q,u)=0 \ ;
$$
then, the local expression of $J_{(q,u)}$ is
$$
J_{(q,u)}=\d q^i\Big\vert_{(q,u)}\otimes\derpar{}{v^i}\Big\vert_{(q,u)} \ ,
$$
and, by extension,
$$
J=\d q^i\otimes\derpar{}{v^i} \ .
$$
Observe that $J$ is a tensor field of type $(1,1)$ in $\Tan Q$.

\subsection{The Liouville vector field}

\begin{definition}
Let $q\in Q$ and $(q,v)\in\Tan Q$.
Consider the vertical lift of  the vector $v\in \Tan_qQ$
to the point $(q,v)$; that is, ${\rm D}_v(q,v)$.
This operation allows us to construct a vertical vector field
$\Delta \in \vf(\Tan Q)$, which is called the \textbf{Liouville vector field},
as follows:
$$
\begin{array}{cccc}
\Delta \colon & \Tan Q & \to & \Tan\Tan Q
\\
& (q,v) & \mapsto & ((q,v),{\rm D}_v(q,v))
\end{array}
$$
\end{definition}

\noindent{\bf Local expression}:
In local coordinates $(q^i,v^i)$ en $\Tan Q$;  
let $f\colon\Tan Q \to \Real$ be a function and $(q,v)\in\Tan Q$; 
then we have that
$$
\Delta_{(q,v)}f={\rm D}_v(q,v)f=
\lim_{t \mapsto 0}\frac{f(q,v+tv)-f(q,v)}{t} \ ;
$$
therefore
\begin{eqnarray*}
\Delta_{(q,v)}q^i &=&
\lim_{t \mapsto 0}\frac{q^i(q,v+tv)-q^i(q,v)}{t}
=\lim_{t \mapsto 0}\frac{q^i(q)-q^i(q)}{t}=0 \ ,
\\
\Delta_{(q,v)}v^i&=&
\lim_{t \mapsto 0}\frac{v^i(q,v+tv)-v^i(q,v)}{t}
=\lim_{t \mapsto 0}\frac{(v+tv)(q^i)-v(q^i)}{t}=v(q^i)=v^i(q,v) \ ,
\end{eqnarray*}
and hence, the local expression of $\Delta$ is
$$
\Delta=v^i\derpar{}{v^i} \ .
$$

Bearing in mind that the fibers $\Tan_qQ$ of the
tangent bundle are vector spaces, and that
$\Delta_{(q,v)}=v\in\Tan_qQ$, for every $(q,v)\in\Tan Q$, 
it is usual to say that $\Delta$ is the vector field that generates
the dilatations along the fibers of $\Tan Q$.
Another way to see this interpretation is considering the local
expression of $\Delta$ and its associated system of differential equations
$$
\frac{d q^i}{d t}=0 \quad , \quad
\frac{d v^i}{d t}=v^i \ ,
$$
whose general solution is
$$
q^i(t)=A^i \quad ,\quad
v^i(t)=B^ie^t \quad ;\quad (A^i,B^i \ ctns.) \ .
$$
Thus, the flux of $\Delta$ is
$$
\begin{array}{ccccc}
F^\Delta&\colon&\Real\times\Tan Q&\longrightarrow&\Tan Q \\
 & & (t,q^i,v^i) & \mapsto & (q^i,v^ie^t) 
\end{array} \ .
$$
Hence, the elements of the uniparametric local group of diffeomorphisms generated by $\Delta$ are
$$
\begin{array}{ccccc}
F^\Delta_t&\colon&\Tan Q&\longrightarrow&\Tan Q \\
 & & (q^i,v^i) & \mapsto & (q^i,v^ie^t) \ ;
\end{array}
$$
that is, they are homothetics on the fibers with positive reason.

\subsection{Holonomic curves. Second Order Differential Equations}

A curve in $\Tan Q$ is not necessarily the canonical lift
of a curve in the base manifold $Q$
\footnote{
See the Appendix \ref{canliftq}}. 
Then we define:

\begin{definition}
A curve $\sigma\colon (a,b)\subseteq\Real\to\Tan Q$ is
\textbf{holonomic} if there exists $\gamma\colon (a,b)\subseteq\Real\to Q$
such that $\sigma=\widetilde\gamma$.
\end{definition}

In natural coordinates, $\sigma(t)=(q^i(t),v^i(t))$,
is a canonical lift if, and only if,
$v^i(t)=\dot q^i(t)$.
 
\begin{definition}
Let $\sigma\colon (a,b)\subseteq\Real\to\Tan Q$ a curve
in $\Tan Q$. Then $\tau_Q\circ\sigma\colon (a,b)\subseteq\Real\to Q$
is a curve in $Q$, which is called the
\textbf{curve in the base manifold} associated with $\sigma$.
\end{definition}

Using this, we can characterize holonomic curves as follows: 

\begin{prop}
A curve $\sigma\colon (a,b)\subseteq\Real\to\Tan Q$ is 
holonomic if, and only if,
$\widetilde{(\tau_Q\circ\sigma)}=\sigma$.
\end{prop}
\begin{proof}
($\Longrightarrow$) \quad
If $\sigma$ is holonomic then there exists $\gamma\colon (a,b)\subseteq\Real\to Q$
such that $\sigma=\widetilde  \gamma$; therefore
$$
\tau_Q\circ\sigma=\tau_Q\circ\widetilde  \gamma=\gamma
\quad \Longrightarrow \quad
\widetilde{(\tau_Q\circ\sigma)}=\widetilde  \gamma=\sigma \ .
$$
($\Longleftarrow$) \quad
The converse is immediate.
\\ \qed  \end{proof}

Finally, we have the following fundamental result:

\begin{prop}
\label{lema2}
\ben
\item
Let $\varphi\colon Q\to Q$ be a diffeomorphism and
$\Tan\varphi\colon \Tan Q \to \Tan Q$ its canonical lift to $\Tan Q$. Then
$$
(\Tan\varphi)^*J = J\quad , \quad
(\Tan\varphi)_*\Delta=\Delta \ .
$$
\item
Let $Z\in\vf (Q)$ and $Z^C\in\vf(\Tan Q)$  its canonical lift to $\Tan Q$.
Then the canonical endomorphism $J$ and the Liouville vector field $\Delta$
are invariant by the uniparametric group of local diffeomorphisms generated by $Z^C$.
\label{constq}
\een
\end{prop}
\begin{proof}
\bit
\item
The result for $J$ is  a straightforward consequence of the local expressions
of $J$ and $\Tan\varphi$.

The result for $\Delta$ is  a straightforward consequence of the property
$\Tan\varphi\circ F_t^\Delta=F_t^\Delta\circ\Tan\varphi$,
where $F_t^\Delta$ is an element uniparametric group of local diffeomorphisms generated by $\Delta$.
\item
It is immediate from the above result, taking the
uniparametric group of local diffeomorphisms generated by $Z$ and $Z^C$.
\eit
\qed  \end{proof}

This means that the canonical lifts of diffeomorphisms and vector fields
to the tangent bundle preserve the canonical structures of this bundle.

\begin{definition}
A vector field $X\in\vf(\Tan Q)$ is a
\textbf{Second Order Differential Equation ({\sc sode})}
or also a \textbf{holonomic vector field}
\footnote{It is also said that $X$ satisfies the 
\textbf{second-order condition}.}
if its integral curves are holonomic.
\end{definition}

The local interpretation of this definition is as follows:

\begin{prop}
The necessary and sufficient condition for a vector field
$X\in\vf(\Tan Q)$ to be a {\sc sode} is that its local expression
in any natural system of coordinates in $\Tan Q$ is
\beq
X=v^i\derpar{}{q^i}+g^i\derpar{}{v^i} \ ;
\label{cvedso}
\eeq
that is,
$X_{(q,v)}=(v^i,g^i(q,v))$, for every $(q,v)\in\Tan Q$.
\end{prop}
\begin{proof}
($\Longrightarrow$)\quad
In natural coordinates, the general expression of $X\in\vf(\Tan Q)$ is
\(\dst X=f^i\derpar{}{q^i}+g^i\derpar{}{v^i}\).
Let $\sigma\colon (a,b)\subset\Real\to \Tan Q$
be an integral curve of $X$, with $\sigma(t)=(q^i(t),v^i(t))$; then
\beq
\frac{d q^i}{d t}=(f^i\circ\sigma)(t)=f^i(q^j(t),v^j(t))
\quad , \quad
\frac{d v^i}{d t}=(g^i\circ\sigma)(t)=g^i(q^j(t),v^j(t)) \ .
\label{edo1}
\eeq
Furthermore, if $X$ is a {\sc sode}, by definition there exists
$\gamma\colon (a,b)\subset\Real\to Q$, with $\gamma(t)=(q^i(t))$,
such that $\sigma=\widetilde  \gamma$; that is,
$\sigma(t)=(q^i(t),\dot q^i(t))$.
Therefore, we have that $\dot q^i(t)=v^i(t)$, for $i=1,\ldots ,n$;
and hence, going to the first group of equations (\ref{edo1}), we obtain that
$$
v^i(t)=\frac{d q^i}{d t}=f^i(q^j(t),v^j(t)) \quad , \quad t\in (a,b) \ ,
$$
and the result follows.

\quad($\Longleftarrow$)\quad
If \(\dst X=v^i\derpar{}{q^i}+g^i\derpar{}{v^i}\), then its integral curves
$\sigma\colon (a,b)\subset\Real\to \Tan Q$, with $\sigma(t)=(q^i(t),v^i(t))$,
are determined by the system
$$
\frac{d q^i}{d t}=(v^i\circ\sigma)(t)=v^i(t)
\quad , \quad
\frac{d v^i}{d t}=(g^i\circ\sigma)(t)=g^i(q^j(t),v^j(t)) \ ;
$$
that is, $\sigma(t)=(q^i(t),\dot q^i(t))$.
Therefore, they are holonomic curves and $X$ is a {\sc sode}.
\\ \qed  \end{proof}

\begin{remark}{\rm 
Observe that, with the above conditions,
the second group of equations (\ref{edo1}) is written as
$$
\frac{d^2 q^i}{d t^2}=f^i(q^j,\dot q^j) \ ,
$$
which is a system of  second order ordinary differential equations,
whose solution completely determines the integral curves of $X$.
This fact justifies the name {\sc sode} for this kind of vector fields.
}\end{remark}

From this result, we obtain the following intrinsic characterization: 

\begin{prop}
The necessary and sufficient condition for $X\in\vf(\Tan Q)$ to be a {\sc sode} is that 
$$
J(X)=\Delta \ .
$$
\end{prop}
\begin{proof}
In natural coordinates, the general expression of $X\in\vf(\Tan Q)$ is
\(\dst X=f^i\derpar{}{q^i}+g^i\derpar{}{v^i}\); then
$$
J(X)=f^i\derpar{}{v^i}=v^i\derpar{}{v^i}\equiv\Delta
\quad \Longleftrightarrow \quad f^i=v^i \ ,
$$
and the result follows from the above proposition.
\\ \qed  \end{proof}

There is another intrinsic characterization as follows:
the bundle $\Tan \Tan Q$ has two natural projections
$$
\begin{array}{ccc}
\Tan\Tan Q&\mapping{\Tan\tau_Q}&\Tan Q
\\
\tau_{\Tan Q}\downarrow& &\downarrow\tau_Q
\\
\Tan Q&\mapping{\tau_Q}&Q 
\end{array} \ .
$$
Taking natural coordinates in the bundles $\Tan Q$ and
$\Tan\Tan Q$, and reminding the expression (\ref{matrix1}),
for every $((q,v),Y_{(q,v)})\equiv (q^i,v^i;u^i,w^i)\in\Tan\Tan Q$,
we have that
$$
\begin{array}{ccccccc}
\tau_{\Tan Q}((q,v),Y_{(q,v)})&=&\tau_{\Tan Q}(q^i,v^i;u^i,w^i)&=&(q^i,v^i)& & \\
\Tan\tau_Q((q,v),Y_{(q,v)})&=&\tau_{\Tan Q}(q^i,v^i;u^i,w^i)&=&
(\tau_Q(q^i,v^i),\Tan_{(q^i,v^i)}\tau_Q(u^i,w^i))&=&(q^i,u^i) \ .
\end{array}
$$
By definition, $X\in\vf(\Tan Q)$ is a section of the projection
$\tau_{\Tan Q}$; that is, a map $X\colon\Tan Q\to\Tan\Tan Q$ 
such that $\tau_{\Tan Q}\circ X={\rm Id}_{\Tan Q}$. Then:

\begin{prop}
The necessary and sufficient condition for a vector field
$X\in\vf(\Tan Q)$ to be a {\sc sode} is that $X$ is a section
of the projection $\Tan\tau_Q$; that is,  a map
$X\colon\Tan Q\to\Tan\Tan Q$ such that
$$
\Tan\tau_Q\circ X ={\rm Id}_{\Tan Q}
$$
\end{prop}
\begin{proof}
In natural coordinates, the general expression of
$X\in\vf(\Tan Q)$ is
\(\dst X=f^i\derpar{}{q^i}+g^i\derpar{}{v^i}\);
that is, for every
$(q,v)\equiv (q^i,v^i)\in\Tan Q$, it is a map
$$
X(q,v)=X(q^i,v^i)=(q^i,v^i;f^i(q,v),g^i(q,v)) \ .
$$
Therefore
$$
(\Tan\tau_Q\circ X)(q,v)=(\Tan\tau_Q\circ X)(q^i,v^i)=
\Tan\tau_Q(q^i,v^i;f^i(q,v),g^i(q,v))=(q^i,f^i(q,v)) \ ;
$$
and, as ${\rm Id}_{\Tan Q}(q,v)=(q^i,v^i)$, we have that
$$
(\Tan\tau_Q\circ X)(q,v)={\rm Id}_{\Tan Q}(q,v)\quad
\Longleftrightarrow \quad f^i=v^i \ ,
$$
and the result follows.
\\ \qed  \end{proof}

\begin{remark}{\rm 
Summarizing, we have proved that the following assertions are equivalent:
\begin{enumerate}
\item
A vector field $X\in\vf(\Tan Q)$ is a {\sc sode}.
\item
The integral curves of $X$ are canonical lifts of curves in $Q$.
\item
The expression of $X$ in a system of natural coordinates in $\Tan Q$,
is (\ref{cvedso}).
\item
$J(X)=\Delta$.
\item
$\Tan\tau_Q\circ X ={\rm Id}_{\Tan Q}$.
\end{enumerate}
}\end{remark}

\subsection{Canonical forms in the cotangent bundle}

An essential characteristic of the
cotangent bundle is the following:

\begin{teor}
The cotangent bundle $\Tan^*Q$ is endowed with a canonical $2$-form
which is closed and nondegenerate and, hence, it
is a {\sl symplectic manifold}.
\label{ficotvs}
\end{teor}
\begin{proof}
As we know,
$\Tan^*Q := \{ {\rm p}= (q,\xi) \ \mid \ q \in Q \ ,  \ \xi \in \Tan^*_qQ \}$.
As $\Tan^*Q$ is a differentiable manifold,
we can consider its tangent bundle $\Tan\Tan^*Q$
and the tangent structure induced by
$\pi_Q \colon \Tan^*Q \to Q$,
which must be understood as a fiber map
$\Tan_{(q,\xi)}\pi_Q \colon \Tan_{(q,\xi)}\Tan^*Q \to \Tan_q Q$.
Now, we prove the theorem constructively using the fiber projection.
\begin{enumerate}
\item
We construct the  differential $1$-form
$\Theta \in {\mit\Omega}^1(\Tan^*Q)$
as follows: for every ${\rm p}=(q,\xi) \in \Tan^*Q$ and
$X_{\rm p} \in \Tan_{\rm p}\Tan^*Q$,
$$
\Theta_{\rm p}(X_{\rm p}) := \xi(\Tan_{\rm p}\pi_Q(X_{\rm p})) \ .
$$
Its expression in a natural chart of coordinates of $\Tan^*Q$
is obtained in the following way:
as we have seen in proposition \ref{prop:cotbun} of the Appendix \ref{sec:cotbun},
$\xi = p_i({\rm p})\,\d q^i\mid_q$;
furthermore, every $X_{\rm p} \in \Tan_{\rm p}(\Tan^*Q)$ has the local expression
\(\dst X_{\rm p}= A^j\derpar{}{q^j}\Big\vert_{\rm p}+B_j\derpar{}{p_j}\Big\vert_{\rm p}\),
and the general expression of $\Theta_{\rm p}$ is
$\Theta_{\rm p}= a_i\,\d q^i\mid_{\rm p} + b^i\,\d p_i\mid_{\rm p}$, then
$$
\Theta_{\rm p}(X_{\rm p}) = a_iA^i + b^iB_i \ ;
$$
but, from the definition,
\begin{eqnarray*}
\Theta_{\rm p}(X_{\rm p}) &:=& \xi(\Tan_{\rm p}\pi_Q(X_{\rm p}))
=(p_i(m)\d q^i\mid_q)
\left[\Tan_{\rm p}\pi_Q \Big(A^j\derpar{}{q^j}\Big\vert_{\rm p}+
B_j\derpar{}{p_j}\Big\vert_{\rm p}\Big)\right]
\\ &=&
(p_i({\rm p})\d q^i\mid_x)\Big( A^j\derpar{}{q^j}\Big\vert_q\Big)
= p_i({\rm p})A^j\delta^i_j = p_i({\rm p})A^i \ ,
\end{eqnarray*}
and as this holds for every $X_{\rm p}$ (that is, for every $A^i,B_i$),
from the above expressions we obtain that
$a_i=p_i(m)$ and $b^i=0$; that is, the final local expression of $\Theta$ is
$$
\Theta = p_i\,\d q^i \ .
$$
\item
We define the differential $2$-form
$$
\Omega := -\d\Theta \ ,
$$
whose local expression in a natural chart is
$$
\Omega = \d q^i \wedge \d p_i \ .
$$
\end{enumerate}
Obviously $\Omega$ is closed (because it is exact)
and  nondegenerate (as we can see from its local expression in coordinates).
\\ \qed  \end{proof}

Then, cotangent bundles of manifolds are the canonical models of symplectic manifolds, since they carry a canonical symplectic form.
Furthermore,  their natural coordinates are also the Darboux coordinates
for this canonical symplectic form.

\begin{definition}
The forms $\Theta \in {\mit\Omega}^1(\Tan^*Q)$
and $\Omega \in {\mit\Omega}^2(\Tan^*Q)$ are the
\textbf{$1$ and $2$ canonical forms} of the cotangent bundle
\footnote{ The $1$-form $\Theta$ is also called the {\sl tautological form} of $\Tan^*Q$.}.
\end{definition}

From Proposition \ref{Liouville} we have that:

\begin{prop}
The cotangent bundle of a differentiable manifold is an
{\sl oriented} manifold.
\end{prop}

The canonical lift of a vector field in $Q$ to the cotangent bundle $\Tan^*Q$
\footnote{
See the Appendix \ref{sec:cotbun}.}
can be characterized using the canonical forms. First we define:

\begin{definition}
Every vector field $Z\in\vf (Q)$
induces a function $F_Z\in\Cinfty (\Tan^*Q)$ defined as
$$
\begin{array}{ccccc}
F_Z&\colon&\Tan^*Q&\longrightarrow&\Real
\\
& &{\rm p}\equiv (q,\xi ) & \mapsto & \xi (Z_q)
\end{array}
$$
\end{definition}

Taking into account the definition of the canonical  1-form
$\Theta$ of $\Tan^*Q$ and Proposition \ref{prolevcan}
we have that, for every ${\rm p}\equiv (q,\xi)\in\Tan^*Q$,
$$
F_Z({\rm p})\equiv F_Z(q,\xi)=\xi (Z_q)=\xi (\Tan \pi_Q (Z^*_{\rm p}))=\Theta_{\rm p}(Z^*_{\rm p}) \ ,
$$
and thus we have proved that:

\begin{prop}
If $Z\in\vf (Q)$,  then $F_Z=\Theta (Z^*)$.
\end{prop}

Bearing this in mind, we obtain that:

\begin{prop}
If $Z\in\vf (Q)$, then the canonical lift of $Z$ to $\Tan^*Q$
is the only vector field $Z^*\in\vf (\Tan^*Q)$ such that
$\inn (Z^*)\Omega =\d F_Z$
\footnote{
That is, with the terminology of Section \ref{icch},
$Z^*$ is the  (global) Hamiltonian vector field associated with the function $F_Z$.}.
\label{hamlev}
\end{prop}
\begin{proof}
In fact, using the Cartan formula and taking into account Proposition \ref{levdif},
$$
\inn (Z^*)\Omega =-\inn (Z^*)\d\Theta =
\d\inn (Z^*)\Theta -\Lie (Z^*)\Theta =\d (\Theta (Z^*))=\d F_Z
$$
since, according to Proposition \ref{levdif},
as the local uniparametric groups of diffeomorphisms of $Z^*$
are canonical lifts, 
they let  the canonical forms of $\Tan^*Q$ invariant, and then
$\Lie (Z^*)\Theta =0$.
\\ \qed \end{proof}

\noindent {\bf Local expression}:
Bearing in mind this characterization, it is easy to obtain the local expression
of $Z^*$ in a chart of canonical coordinates 
$(U;q^i,p_i)$ of $\Tan^*Q$. So, if
\ $\displaystyle Z\vert_{\pi_Q (U)}=f^i(q)\derpar{}{q^i}$,\
we have that $F_Z(q^i,p_i)=p_if^i(q^j)$, and then
\beq
Z^*\vert_U=f^i\derpar{}{q^i}-p_j\derpar{f^j}{q^i}\derpar{}{p_i} \ .
\label{canlifcot}
\eeq

Another relevant property is:

\begin{prop}
If $\varphi\colon Q\to Q$ is a diffeomorphism,
then $(\Tan^*\varphi)^*\Theta =\Theta$ and,
as a consequence,
$(\Tan^*\varphi)^*\Omega =~\Omega$~\footnote{
With the terminology introduced in Section \ref{tcs},
we say that $\Tan^*\varphi$ is a {\sl simplectomorphism}}.
\label{levdif}
\end{prop}
\begin{proof}
For every $(q,\xi)\in\Tan^*Q$ and
$V\in\Tan_{(q,\xi)}(\Tan^*Q)$,
using the definition of $\Theta$ and the first property of Proposition \ref{exerc2}, we obtain
\beann
\big((\Tan^*\varphi)^*\Theta\big)_{(\varphi(q),\xi)}(V) &=&
\Theta_{(q,\Tan^*\varphi(\xi))}\big(\Tan_{(\varphi(q),\xi)}(\Tan^*\varphi) (V)\big)
\\ &=&
(\varphi_q^*\xi)\Big(\Tan_{(q,\Tan^*\varphi(\xi))}\pi_Q \big(\Tan_{(\varphi(q),\xi)}(\Tan^*\varphi)(V)\big)\Big)
\\ &=&
(\varphi_q^*\xi)\big(\Tan_{(\varphi(q),\xi)}(\pi_Q\circ\Tan^*\varphi)(V)\big)=
\xi\Big(\Tan_{\varphi(q)}\varphi\big(\Tan_{(\varphi(q),\xi)}(\pi_Q\circ\Tan^*\varphi)(V)\big)\Big)
\\ &=&
\xi\Big(\Tan_{(\varphi(q),\xi)} \big(\varphi\circ\pi_Q\circ\Tan^*\varphi)(V)\big)\Big)=
\xi\big(\Tan_{(\varphi(q),\xi)}\pi_Q(V)\big)=
\Theta_{(\varphi(q),\xi)}(V) \, .
\eeann
The result for $\Omega$ is immediate from here.
\\ \qed  \end{proof}

Finally, from Proposition \ref{levdif} we obtain:

\begin{prop}
If $Z\in\vf (Q)$, then
$$
\Lie(Z^*)\Theta=0 \quad , \quad
\Lie(Z^*)\Omega=0 \ .
$$
\label{inf}
\end{prop}
\begin{proof}
To prove this, it suffices to take the local uniparametric groups of diffeomorphisms 
of $Z$ and their canonical lifts.
\\ \qed \end{proof}

\subsection{Fiber derivative of a function}

Let $F\in\Cinfty(\Tan Q)$.
Given $q \in Q$, we consider the function $F_q \colon \Tan_qQ \to \Real$, 
which is the restriction of $F$ to the fiber $\Tan_qQ$.
If $(q,u)\in \Tan Q$, the differential of $F_q$ at this point
is an element of $\Tan_q^*Q$; so ${\rm D}_{(q,u)}F_q \in \Tan^*_qQ$. Then:

\begin{definition}
The \textbf{fiber derivative} of $F$ is the map
$$
\begin{array}{ccccc}
{\cal F}F & \colon & \Tan Q & \to & \Tan^*Q
\\
 & & (q,u) & \mapsto & (q,{\rm D}F_q(u))
\end{array} \ .
$$
\end{definition}

Observe that $\pi_Q \circ{\cal F}F = \tau_Q$;
that is, ${\cal F}F$ preserve the fibers.

\noindent {\bf Local expressions}:
Consider a natural chart $(U;q^i,v^i)$ of $\Tan Q$
and the corresponding canonical chart $({\cal F}F(U);q^i,p_i)$ in $\Tan^* Q$.
Observe that a basis of $\Tan_qQ$ is
\(\dst \left\{ \derpar{}{q^i}\Big\vert_q\right\} \),
and the coordinates in $\Tan_qQ$ are $(v^1,\ldots ,v^n)$.
Then, the Jacobian matrix of ${\rm D}_{(q,u )}F_q$ is
$$
\left( \derpar{F}{v^1}\Big\vert_{(q,u)} \ldots
\derpar{F}{v^n}\Big\vert_{(q,u)} \right) \ ,
$$
and if  \(\dst v=\lambda^i\derpar{}{q^i}\Big\vert_q \in \Tan_qQ\), then
$$
({\rm D}_{(q,u)}F_q)(v) =
\left( \derpar{F}{v^1}\Big\vert_{(q,u)} \ldots
\derpar{F}{v^n}\Big\vert_{(q,u)} \right)
\left(\begin{matrix}\lambda^1 \\ \vdots \\ \lambda^n \end{matrix}\right)\ ,
$$
and therefore
$$
 {\cal F}F(x,u) = \left( q,\derpar{F}{v^i}\Big\vert_{(q,u)}\d q^i
\Big\vert_q\right) \ ;
$$
that is, the expression in coordinates of ${\cal F}F$ is
$$
q^i \circ{\cal F}F = q^i \quad , \quad p_i \circ{\cal F}F =\derpar{F}{v^i}
\ .
$$

\section{Lagrangian formalism for Lagrangian dynamical systems}

Next we introduce the so-called {\sl Lagrangian formalism} of the
Lagrangian dynamical systems.
(See
\cite{CMSVZ-76,Cr-81,Cr-83b,CP-adg,dLe89,Ga-52,Go-69,Grif-72a,Grif-72b,klein,SCC-84}
as general references).

First, we state what are these kinds of systems.

\subsection{Lagrangian dynamical systems}
\label{poslagdyn}

The physical and geometrical foundations of this formulation are the following:

\begin{pos}
{\rm (First Postulate of Lagrangian mechanics\/)}:
In a physical Lagrangian system,  the $n$ possible degrees of freedom of the system
are described by the domain of variation of a set of $n$ 
{\sl generalized coordinates}
\footnote{
If we have a mechanical system, the generalized coordinates 
correspond to ``kinematic'' degrees of freedom (position coordinates). But if we have a more generic physical system, 
such as an electric or a thermodynamical system,
then the generalized coordinates correspond to other kinds of physical magnitudes
(for instance, electric charge, temperature, etc.).
},
which determine locally the {\sl configuration space} of the system.
The {\rm states} of the system are locally described by means of
generalized coordinates and their corresponding
generalized velocities.

Geometrically, this means that
the configuration space $Q$ of a dynamical system
with $n$ degrees of freedom is a $n$-dimensional differentiable manifold,
and that
the phase space is the tangent bundle $\Tan Q$ of the configuration manifold $Q$,
which is called the
\textbf{state} or \textbf{phase space} of \textbf{positions-velocities} of the system.
\end{pos}

\begin{pos}
{\rm (Second Postulate of Lagrangian mechanics\/)}:
The observables or physical magnitudes of a 
dynamical system are functions of $\Cinfty (\Tan Q)$.

The result of a measure of an observable is the value that the function 
which represents the observable takes at a point of the phase space $\Tan Q$
(that is, in some state, following the first postulate).
\end{pos}

\begin{pos}
{\rm (Third Postulate of Lagrangian mechanics\/)}:
There is a function $\Lag \in \Cinfty (\Tan Q)$, called \textbf{Lagrangian function},
which carries the dynamical information of the system.
\end{pos}

As we will see,
from this function and using  the geometrical structures of the
tangent bundle, we can construct a differential form
$\Omega_\Lag\in\df^2(\Tan Q)$ and the so-called
{\sl Lagrangian Energy function} $E_\Lag\in\Cinfty(\Tan Q)$,
and we can set the dynamical equations of the system
(Postulate \ref{4poslag}).

Taking all of this into account, we define:

\begin{definition}
A \textbf{Lagrangian dynamical system} is a pair $(\Tan Q,\Lag )$,
where $Q$ is a manifold representing the configuration space of
a physical system and $\Lag \in \Cinfty (\Tan Q)$
is the Lagrangian function of the system.
\end{definition}

\begin{remark}{\rm 
It is important to point out that the dynamical systems which are
described in this way are {\sl autonomous}; that is, {\sl independent of time}
(a geometrical description of {\sl nonautonomous dynamical systems} 
is presented in Chapter \ref{chap:cosym}).
They are also {\sl first-order} systems,
which are those  whose Lagrangian function 
depends locally on the coordinates of position and velocity,
in contrast with those for which the dependence is also on the
{\sl generalized accelerations} or, in general, 
higher-order time-derivatives of the generalized positions.
(For the geometrical description of {\sl higher-order} systems, 
see, for instance, \cite{BGG-2015,proc:Cantrijn_Crampin_Sarlet86,
art:Carinena_Lopez92,LR-os,GM-86,GPR-higher,
art:Gracia_Pons_Roman92,art:Krupkova00,
art:Prieto_Roman11,art:Prieto_Roman12,art:Prieto_Roman15}).
}\end{remark}

In the next section we explain how, using the canonical structures of the
tangent bundle introduced in Section  \ref{egft} and starting from the Lagrangian function,
we can construct the dynamical-geometric structures which allow us to
set the dynamical equations in an intrinsic way.
In this way, we will complete the set of Postulates of the Lagrangian formalism.

\subsection{Geometric structures induced by the dynamics}

Given a Lagrangian function $\Lag \in \Cinfty (\Tan Q)$ and
using the canonical structures of the tangent bundle
(the vertical endomorphism 
$J \colon \vf (\Tan Q) \to \vf^{V(\tau_Q)}(\Tan Q)$
and the Liouville vector field $\Delta \in \vf^{V(\tau_Q)}(\Tan Q)$)
we can define the following elements:

\begin{definition}
The {\sl \textbf{Cartan}} or {\sl \textbf{Lagrangian $1$}} and {\sl \textbf{$2$-forms}} associated with $\Lag$ are
\beann
\Theta_\Lag&:=& \inn(J)\d\Lag=\d \Lag \circ J \in \df^1(\Tan Q) \ ,
\\
\Omega_\Lag &:=& -\d \Theta_\Lag \in \df^2(\Tan Q) \ .
\eeann
The {\sl \textbf{Lagrangian energy}} associated with $\Lag$ is the function
$$
E_\Lag := \Delta (\Lag )-\Lag \in \Cinfty (\Tan Q) \ .
$$
The function $A_\Lag := \Delta (\Lag ) \in \Cinfty (\Tan Q)$ is sometimes referred to as the {\sl \textbf{Lagrangian action function}} associated with $\Lag$.
\end{definition}

\begin{remark}{\rm 
\begin{itemize}
\item
The physical observable represented by the Lagrangian energy 
is the total energy of the system, and this justifies this terminology. 
\item
There are Lagrangian functions which, being different, give the same
Lagrangian form $\Omega_{\Lag}$ and the same Lagrangian energy $ E_{\Lag}$.
They are called {\sl gauge-equivalent Lagrangians} (and they are studied in Section \ref{flsltn}).
\item
It is important to point out that, for an arbitrary function
$\Lag \in \Cinfty (\Tan Q)$, the form $\Omega_\Lag$
has not constant rank in $\Tan Q$ necessarily.
The Lagrangian $\Lag$ is said to be {\sl \textbf{geometrically admissible}}
when this rank is constant.
The theory we are developing concerns Lagrangian systems of this kind
\footnote{
When this is not the case,
there exists a more general framework which consists in using {\sl Poisson manifolds} as phase states for the formalism (see, for instance, \cite{LM-sgam}).
}.
\end{itemize}
}\end{remark}

\noindent {\bf Local expressions}:
Consider a natural  chart $(U;q^i,v^i)$ in $\Tan Q$.
Remember that the local expression of the vertical endomorphism and 
the Liouville vector field on this chart are
$\displaystyle J = \d q^i\otimes \derpar{}{v^i}$ and
$\displaystyle \Delta = v^i\derpar{}{v^i}$. Then,
given a Lagrangian function $\Lag=\Lag (q^i,v^i)$,
the local expressions of the action and the Lagrangian energy associated with the Lagrangian $\Lag$ are 
$$
A_\Lag = v^i\derpar{\Lag}{v^i}
\quad , \quad
E_\Lag = v^i\derpar{\Lag}{v^i} - \Lag \ .
$$
The more general expression for a $1$-form in $\Tan Q$ is
$$
\Theta_\Lag = a_i(q^j,v^j) \d q^i + b_i(q^j,v^j) \d v^i \ ,
$$
then, for an arbitrary $Y\in\vf (\Tan Q)$,
$$
Y = A^i(q^j,v^j) \derpar{}{q^i} + B^i(q^j,v^j) \derpar{}{v^i} \ ,
$$
we have that
$$
\Theta_\Lag (Y) = a_iA^i +b_iB^i
$$
and furthermore, by definition,
$$
\Theta_\Lag (Y) := (\d \Lag \circ J)(Y) =
\left(\derpar{\Lag}{q^i}\d q^i + \derpar{\Lag}{v^i}\d v^i\right)
\left(A^i\derpar{}{v^i}\right)
=A^i \derpar{\Lag}{v^i} \ ;
$$
therefore, from both expressions we obtain that
\(\dst a_i = \derpar{\Lag}{v^i}\) and $b_i = 0$; that is
$$
\Theta_\Lag  = \derpar{\Lag}{v^i}\,\d q^i \ ,
$$
and, as a consequence,
$$
\Omega_\Lag  =
\frac{\partial^2\Lag}{\partial q^j \partial v^i}\, \d q^i \wedge \d q^j +
\frac{\partial^2\Lag}{\partial v^j \partial v^i}\, \d q^i \wedge \d v^j \ .
$$

It is important to remark that, for an arbitrary function
$\Lag \in \Cinfty (\Tan Q)$, although the Lagrangian $2$-form $\Omega_\Lag$
is always closed, it is not necessarily non-degenerated
(since, as it is seen from the above local expression, its rank
is determined by the second derivatives with respect to the velocities of the function $\Lag$). Then:

\begin{definition}
$\Lag \in \Cinfty (\Tan Q)$ is a  \textbf{regular Lagrangian function}
(and  $(\Tan Q,\Lag )$ is a \textbf{regular Lagrangian system})
if the Lagrangian $2$-form $\Omega_\Lag$ is nondegenerated.
Otherwise, $\Lag$ is a  \textbf{singular Lagrangian function}
(and  $(\Tan Q,\Lag )$ is said to be a \textbf{singular Lagrangian system})
\footnote{
In this exposition, we are mainly interested in
 regular Lagrangian systems.}.
\end{definition}

The non degeneracy of the Lagrangian $2$-form
can be characterized locally as follows:

\begin{prop}
Let $(\Tan Q,\Lag )$ be a Lagrangian system. Then
$\Lag$ is a regular Lagrangian if, and only if,
in a natural chart $(U;q^i,v^i)$ of $\Tan Q$, the Hessian matrix
\(\dst W= \left(\frac{\partial^2\Lag}{\partial v^j \partial v^i}\right)\)
is regular at every point of $U$.
\end{prop}
\begin{proof}
We have to prove that $\Omega_\Lag$ is nondegenerated if, and only if,
this condition holds. Then, it suffices to take into account that
$\Omega_\Lag \in \df^2(\Tan Q)$
is non degenerated if, and only if, $(\Omega_\Lag)^{\ n}$
is a volume form in $\Tan Q$. Therefore,
$$
\Omega_\Lag^{\ n} =
n!\ {\rm det}\left(\frac{\partial^2\Lag}{\partial v^j \partial v^i}\right)^n
\d q^1 \wedge \ldots \d q^n \wedge\d v^1 \wedge \ldots \d v^n \ ,
$$
which  is a non-vanishing form at every point when
\(\dst{\rm det}\left(\frac{\partial^2\Lag}{\partial v^j\partial v^i}\right)
\not= 0\)
everywhere.
\\ \qed \end{proof}

\subsection{Lagrangian dynamical equations.
 Euler--Lagrange equations}

The dynamical equation in the Lagrangian formalism is stated in the following:

\begin{pos}
\label{4poslag}
{\rm (Fourth Postulate of Lagrangian mechanics\/)}:
The dynamical trajectories of a Lagrangian system
$(\Tan Q,\Lag)$ are the integral curves
of a vector field $X_\Lag \in \vf (\Tan Q)$ such that:
\ben
\item
$X_\Lag$ is a solution to the equation
\beq
\inn (X_\Lag )\Omega_\Lag = \d E_\Lag \ .
\label{elm}
\eeq
\item
$X_\Lag$ is a {\sc sode}; that is, 
\beq
J (X_\Lag ) = \Delta \ .
\label{edso}
\eeq
\een
Equation (\ref{elm}) is called the
 \textbf{Lagrangian equation for vector fields,} and
a vector field $X_\Lag$ solution to \eqref{elm} (if it exists) is a
\textbf{Lagrangian dynamical vector field}.
If, in addition, the condition (\ref{edso}) holds, then
$X_\Lag$ it is called an \textbf{Euler--Lagrange vector field} of the system, 
\end{pos}

Bearing in mind the definition of integral curve and equation  \eqref{Hvecf},
it is immediate to prove that:

\begin{teor}
\label{teor:Lcurv}
The vector field $X_\Lag\in\vf(\Tan Q)$ is a Lagrangian vector field for a
Lagrangian system $(\Tan Q,\Lag)$ if, and only if,
the integral curves $c\colon I\subset\Real\to M$ of $X_\Lag$ are the solutions to the equations
 \begin{equation}
\label{Lcurv}
\inn( \widetilde c)(\Omega_\Lag\circ c)=\d E_\Lag\circ c \ ,
\eeq
Equation \eqref{Lcurv} is the \textbf{Lagrangian equation} for the
integral curves of $X_\Lag$.
If, in addition, $X_\Lag$ is a {\sc sode}, then $c$ is a holonomic curve
and \eqref{Lcurv} is the
\textbf{Euler--Lagrange equation} for the integral curves of $X_\Lag$.
\end{teor}

\begin{remark}{\rm 
It is usual to require that the Lagrangian dynamical equations are
obtained from a variational principle (as we will see
in Section \ref{lagvariational}) and
the necessary condition for it is that the integral curves
of the dynamical vector field $X_\Lag$ must be canonical lifts of
curves in the base $Q$ of the bundle $\Tan Q$ which represents the
Lagrangian phase space.
This is why $X_\Lag$ is asked to be a {\sc sode}.
}\end{remark}

\begin{definition}
Given a Lagrangian dynamical system $(\Tan Q,\Lag )$, 
the \textbf{Lagrangian problem} posed by  the system
consists in finding a vector field $X_\Lag \in \vf (\Tan Q)$
verifying  the conditions (\ref{elm}) and (\ref{edso}).
\end{definition}

\noindent {\bf Local expressions}:
Consider a natural chart $(U;q^i,v^i)$ of $\Tan Q$.
Bearing in mind  the local expressions of the several geometric elements
appearing in the dynamical equations we obtain that, if
\(\dst X_\Lag = A^i(q^j,v^j) \derpar{}{q^i} +
B^i(q^j,v^j) \derpar{}{v^i}\),
equation (\ref{elm}), written in coordinates, is:
\bea
\frac{\partial^2\Lag}{\partial v^j \partial v^i}B^i -
\left(\frac{\partial^2\Lag}{\partial q^j \partial v^i} +
\frac{\partial^2\Lag}{\partial v^j \partial q^i}\right)A^i +
\frac{\partial^2\Lag}{\partial v^i \partial q^j}v^i -
\derpar{\Lag}{q^j} &=& 0 \ ,
\label{eq11a}
\\
\frac{\partial^2\Lag}{\partial v^j \partial v^i}(A^i-v^i) &=& 0 \ .
\label{eq11}
\eea
As we know, condition (\ref{edso}) is locally equivalent to demand that
$A^i=v^i$. Furthermore, the integral curves
$\sigma\colon I\subseteq\Real\to\Tan Q$ of $X_\Lag$
are canonical lifts of curves
$\gamma\colon I\subseteq\Real\to Q$; therefore, if
$\gamma(t)=(q^i(t))$, then $\sigma(t)=(q^i(t),\dot q^i(t))$ and
$$
A^i = v^i = \frac{d q^i}{d t}
\quad , \quad
B^i = \frac{d^2q^i}{d t^2} \ ,
$$
and the combination of these expressions with equations (\ref{eq11a}) and (\ref{eq11}) leads to the
equation of  the integral curves which is
$$
\left(\frac{\partial^2\Lag}{\partial v^j \partial v^i}\circ\sigma\right)
\frac{d^2q^i}{d t^2} =
-\left(\frac{\partial^2\Lag}{\partial v^j \partial q^i}\circ\sigma\right)
\frac{d q^i}{d t} +
\derpar{\Lag}{q^j}\circ\sigma \ ,
$$
and can be written as
$$
W_{ji}(q^k(t),\dot q^k(t))\frac{d^2q^i}{d t^2} =F_j( q^k(t),\dot q^k(t)) \ ,
$$
or also in an equivalent form as
\beq
\frac{d}{d t}\left(\derpar{\Lag}{v^j}\circ\sigma\right)-
\derpar{\Lag}{q^j}\circ\sigma=0 \ .
\label{ELequats}
\eeq
This is the classical coordinate expression of the Euler--Lagrange equation \eqref{Lcurv}.

Therefore, for the case of regular Lagrangians we have the following result:

\begin{teor}
Let $(\Tan Q,\Lag)$ be a regular Lagrangian system.
Then, there exists a unique vector field $X_\Lag\in\vf (\Tan Q)$
which is the solution to the Lagrangian equation
(\ref{elm}) and it is a {\sc sode}.
\end{teor}
\begin{proof}
The existence and uniqueness are a straightforward consequence 
of the fact that $\Omega_\Lag$ is non-degenerated.

Furthermore, if the Lagrangian is regular, the Hessian matrix
$H_\Lag$ is regular and equation (\ref{eq11}) leads to
$A^i = v^i$, and thus the vector field $X_\Lag$ is a {\sc sode}.
(Observe that, in this case, all the coefficients $B^i$ are determined
by equations (\ref{eq11a})).
\\ \qed \end{proof}

Assuming that the Lagrangian $\Lag$ is regular, from \eqref{eq11a} and \eqref{eq11}, we can obtain the local expression of the Lagrangian dynamical vector field:
\begin{equation}\label{lagdynvf}
 X_\Lag= v^i\derpar{}{q^i}+
 W^{ik}\left(\derpar{\Lag}{q^i}
 -v^j\frac{\partial^2\Lag}{\partial q^j \partial v^i}\right)\derpar{}{q^k}\, ,
\end{equation}
where $ W^{ik}$ is the inverse matrix of the partial Hessian matrix of $\Lag$; that is,
$ W^{ik}W_{ji}=\delta^k_j$.

\begin{remark}{\rm 
\bit
\item
For Lagrangian systems, the triple
$(\Tan Q,\Omega_{\Lag},\d E_{\Lag})$ is a Hamiltonian system
which is regular or singular depending on the regularity of the Lagrangian $\Lag$.
Thus, the Lagrangian formalism of these kinds of systems is a Hamiltonian formalism
with certain additional characteristics.
\item
If the Lagrangian system is singular, then $(\Tan Q,\Omega_{\Lag},\d E_{\Lag})$
is a presymplectic Hamiltonian system and then equation (\ref{elm})
is not necessarily compatible everywhere on $\Tan Q$ 
and, even in the case that it has solution,
it is not unique and it is not necessarily a {\sc sode}.
In fact, if $X_\Lag^0$ is a solution, then $X_\Lag^0+Z$, with 
$Z\in\ker\,\Omega_\Lag$, is also a solution.
Then, in order to obtain the Euler--Lagrange equations \eqref{Lcurv},
the condition $J(X_\Lag)=\Delta$ must be added to the above Lagrangian equations.
In general, solutions $X_\Lag$ could exist only in some submanifold $S_f\hookrightarrow\Tan Q$,
and a suitable {\sl constraint algorithm} must be implemented in order to find 
this {\sl final constraint submanifold} $S_f$ (if it exists) where there are 
{\sc sode} vector fields $X_{\cal L}\in\vf(\Tan Q)$,
tangent to $S_f$, which are (not necessarily unique) solutions to the Lagrangian equations on $S_f$.
All these problems have been also widely studied
(see, for instance, \cite{BGPR-86,BGPR-87,CLR-87,CLR-88,GN-79,GN-80,Ka-82,MR-92}).
\eit
}\end{remark}

\section{Canonical Hamiltonian formalism for Lagrangian systems}

The {\sl canonical Hamiltonian formalism} of Mechanics was
initially introduced by {\it Hamilton, Lagrange, Poisson,
Ostrogadsky} and {\it Donkin}, among others, and constitutes the dual
formulation of the Lagrangian formalism.

\subsection{Legendre map}

The construction of the canonical Hamiltonian formalism is based on the introduction of a map:
the {\sl fiber derivative} of the Lagrangian function:

\begin{definition}
Let $(\Tan Q,\Lag )$ be a Lagrangian dynamical system.
The \textbf{Legendre map} associated with this system
is the fiber derivative of $\Lag$; that is the map
$$
\begin{array}{ccccc}
\Leg & \colon & \Tan Q & \to & \Tan^*Q
\\
 & & (q,u) & \mapsto & (q,{\rm D}\Lag_q(u))
\end{array} \ .
$$
\end{definition}

\noindent {\bf Local expression}:
Given a natural chart $(U;q^i,v^i)$ in $\Tan Q$
and the corresponding canonical coordinates
$(\Leg(U);q^i,p_i)$ in $\Tan^* Q$,
the local expression of $\Leg$ is
$$
q^i \circ \Leg = q^i \quad , \quad p_i \circ \Leg =\derpar{\Lag}{v^i} \ ,
$$
the coordinates $p_i$ are called the
{\sl \textbf{generalized momenta}} associated with the generalized coordinates $q^i$.

A relevant characteristic of the Legendre map is given by the following:

\begin{teor}
If $\Theta$ and $\Omega$ are the canonical forms in $\Tan^*Q$,
then $\Leg^*\Theta =\Theta_\Lag$ and $\Leg^*\Omega =\Omega_\Lag$.
\end{teor}
\begin{proof}
The local expression of $\Theta$ is $\Theta =p_i\d q^i$; therefore
$$
\Leg^*\Theta = (\Leg^*p_i)\d (q^i\circ\Leg ) =
\derpar{\Lag}{v^i}\d q^i = \Theta_\Lag \ ,
$$
and, bearing in mind that $\Leg^*$ commutes with the exterior differential,
we obtain the result for $\Omega$.
\\ \qed \end{proof}

Working also with local coordinates, it is immediate to prove that:

\begin{prop}
$\Lag$ is a regular Lagrangian if, and only if, $\Leg$ is a local diffeomorphism.
\end{prop}
\begin{proof}
As $\Lag$ is $\Cinfty$, the necessary and sufficient condition for
$\Leg$ to be a local diffeomorphism is that, for all ${\rm p}\in\Tan Q$,
the differential of $\Lag$ at ${\rm p}$,  ${\rm D}_{\rm p}\Leg$, is an isomorphism. 
Then it suffices to analyze the local expression 
of the Jacobian matrix of $\Leg$ at ${\rm p}$, which is
$$
{\cal H}_{{\cal F}\Lag}({\rm p})=
\left(\begin{matrix}
({\rm Id})_{n\times n}& (0)_{n\times n} \\
\Big(\displaystyle\frac{\partial^2\Lag}{\partial q^ij\partial v^i}({\rm p}) \Big)& \Big(\displaystyle\frac{\partial^2\Lag}{\partial v^j \partial v^i}({\rm p})\Big)
\end{matrix}\right) \ ;
$$
therefore, ${\rm det}\,{\cal H}_{{\cal F}\Lag}({\rm p})\not=0$
if, and only if, ${\rm det}\,\Big(\displaystyle\frac{\partial^2\Lag}{\partial v^i \partial v^j}({\rm p})\Big)\not=0$.
\\ \qed \end{proof}

This leads to:

\begin{definition}
$\Lag$ is a \textbf{hyperregular Lagrangian} if
$\Leg$ is a (global) diffeomorphism.
\end{definition}

Regarding the singular Lagrangians,  the most interesting ones are:

\begin{definition}
$\Lag$ is an \textbf{almost-regular Lagrangian} if:
\ben
\item
$\Leg(\Tan Q)\equiv P$ is a closed submanifold of $\Tan^*Q$.
\item
$\Leg$ is a submersion onto its image.
\item
For every ${\rm p}\in\Tan Q$, the fibers $\Leg^{-1}(\Leg({\rm p}))$ are connected submanifolds of $\Tan Q$.
\een
\end{definition}

\subsection{Canonical  Hamiltonian formalism. Equivalence with the Lagrangian formalism}
\label{canforma}

We study essentially the case of hiperregular systems.
Nevertheless, all the results hold also for the case de
regular systems in general,
changing $\Tan^*Q$ by $\Leg (\Tan Q)\subset\Tan^*Q$
(or locally, at least).
First, as $\Leg$ is a diffeomorphism we have:

\begin{prop}
Let $(\Tan Q,\Lag)$ be a hiperregular Lagrangian system.
Then there exists a unique function ${\rm h}\in\Cinfty (\Tan^*Q)$ such that
$\Leg^*{\rm h}=E_\Lag$,
which is called the \textbf{Hamiltonian function} associated with the system
$(\Tan Q,\Lag)$,
and the triple $(\Tan^*Q,\Omega,{\rm h})$ is the 
\textbf{canonical Hamiltonian system} associated with $(\Tan Q,\Lag)$.
\end{prop}

As $E_\Lag$ represents the energy
of the system in the Lagrangian formalism, the function ${\rm h}$
represents the same observable in the canonical Hamiltonian formalism.

Therefore, we have the regular Hamiltonian system
$(\Tan^*Q,\Omega,\d {\rm h})$, where $\Omega$ is the canonical $2$-form in
$\Tan^*Q$
fulfilling the Postulates and results established in Section \ref{shreh}.
In particular, the Hamiltonian equations for vector fields and their
integral curves read
\bea
\inn(X_{\rm h})\Omega=\d {\rm h} \quad &,& \quad X_{\rm h}\in\vf(\Tan^* Q) \ .
\label{elmh} \\
\inn(\widetilde c)(\Omega\circ c))=\d {\rm h}\circ c \quad &,& c\colon I\subseteq\Real\to\Tan^* Q \ ;
\label{elmh2}
\eea

Their local expressions in a natural chart of coordinates
$(U;q^i,p_i)$ in $\Tan^*Q$ are the following:
if $\displaystyle X_{\rm h}\vert_U=f^i\derpar{}{q^i}+g_i\derpar{}{p_i}\in\vf (\Tan^*Q)$,
then,
$$
f^i= \derpar{h}{p_i} \quad ; \quad g_i=-\derpar{h}{q^i}
$$
and its integral curves $c(t)=(q^i(t),p_i(t))$ solution to  \eqref{elmh2} are
the solutions to the system of first-order differential equations
\beq
\frac{d q^i}{d t} = \derpar{h}{p_i}(c(t))
\quad , \quad
\frac{d p_i}{d t} = -\derpar{h}{q^i}(c(t)) \ .
\label{hameq3}
\eeq

The relation between the Lagrangian and the canonical Hamiltonian formalisms of a (hyper)regular Lagrangian system is stated in the following:

\begin{teor} {\rm (Equivalence Theorem)}
\label{eqteorema}
Let $(\Tan Q,\Lag)$ be a (hyper)regular Lagrangian system.
\ben
\item
If $X_\Lag$ is the Lagrangian vector field 
solution to equations (\ref{elm}) and \eqref{edso},
then there exists a unique vector field
 $\Leg_*X_\Lag\equiv X_{\rm h}\in\vf (\Tan^*Q)$
which is the solution to equations (\ref{elmh}).

Conversely, if $X_{\rm h}$ is the Hamiltonian vector field
solution to equation (\ref{elmh}),
then there exists a unique vector field
 $\Leg_*^{-1}X_{\rm h}\equiv X_\Lag\in\vf (\Tan Q)$
which is the solution to equations (\ref{elm}) and \eqref{edso}.
\item
Equivalently, if $\gamma \colon I\subset\Real\to Q$ 
is a curve and its canonical lift
$\widetilde \gamma \colon I\subset\Real\to \Tan Q$ is a
solution to equation \eqref{Lcurv},
%(and hence it is an integral curve of the Euler--Lagrange vector field $X_\Lag$),
then $\zeta=\Leg\circ\widetilde \gamma$ is a curve
solution to equation \eqref{elmh2}.
%(and hence it is an integral curve of the Hamiltonian vector field $X_{\rm h}$).

Conversely, if $\zeta \colon I\subset\Real\to \Tan^*Q$ 
is a curve solution to equation \eqref{elmh2},
then $\widetilde\gamma=\widetilde{\pi_Q\circ\zeta} \colon I\subset\Real\to \Tan Q$ 
is a curve solution to equation \eqref{Lcurv}.
\een
\end{teor}
\begin{proof}
\ben
\item
For the first item we have that
$$
0=\inn (X_\Lag)\Omega_\Lag -\d E_\Lag=
\inn (X_\Lag)(\Leg^*\Omega) -\d (\Leg^*{\rm h})=\Leg^*[\inn (X_{\rm h})\Omega -\d {\rm h}]
$$
and as $\Leg$ is a diffeomorphism, this is equivalent to equation (\ref{elmh}).
The proof that $X_\Lag$ is a {\sc sode} is a consequence of the regularity.
\item
For the second item, if $\widetilde \gamma$
is a solution to equations \eqref{Lcurv} 
then, by Theorem \ref{teor:Lcurv}, it is an integral curve of the Euler--Lagrange vector field $X_\Lag$,
and this means that
$\dot{\widetilde\gamma}=X_\Lag\circ\widetilde\gamma$.
Then, as $\zeta=\Leg\circ\widetilde \gamma$ and $\Leg_*X_\Lag=X_{\rm h}$,  we have that
$$
\dot\zeta=\Tan\Leg\circ\dot{\widetilde\gamma}=
\Tan\Leg\circ X_\Lag\circ\widetilde\gamma=
X_{\rm h}\circ\Leg\circ\widetilde\gamma=X_{\rm h}\circ\zeta\ ;
$$
then $\zeta$ is an integral curve of the Hamiltonian vector field $X_{\rm h}$
and, by Theorem \ref{teo:Hcurv}, it is a solution to equation \eqref{elmh2}

Conversely, from a solution $\zeta$ to the equation \eqref{elmh2},
we construct the curve $\gamma=\pi_Q\circ\zeta$ and its
canonical lift $\widetilde\gamma$. Then, reasoning as above, we conclude that $\widetilde\gamma$ is a solution to \eqref{Lcurv}.

The following (commutative) diagram summarizes the situation
$$
\xymatrix{
\Tan\Tan Q 
\ar[rrrr]^<(0.45){\Tan\mathcal{FL}} 
 & \ & \ & \ &
\Tan\Tan^*Q   \\
\ & \ & \Real  \ar[dd]^<(0.65){\gamma}
 \ar[urr]^<(0.45){\widetilde\zeta} \ar[dll]_<(0.45){\widetilde\gamma}
 \ar[drr]^<(0.45){\zeta}  \ar[ull]_<(0.45){\widetilde{\widetilde\gamma}} & \ & \ \\
\Tan Q \ar[uu]^<(0.45){X_\Lag} \ar[drr]^<(0.45){\tau_Q}
\ar[rrrr]^<(0.35){\mathcal{FL}}
 & \ & \ & \ & 
\Tan^*Q  \ar[uu]_<(0.45){X_{\rm h}}  \ar[dll]_<(0.45){\pi_Q}  \\
& \ & Q & \ &
}
$$
\een
\qed \end{proof}

\begin{remark}{\rm 
If $(\Tan Q,\Lag)$ is an almost-regular Lagrangian system
and $j_0\colon P\hookrightarrow\Tan^*Q$ is the 
natural embedding of $P\equiv\Leg(\Tan Q)$ in $\Tan^*Q$,
then it can be proved that there exists ${\rm h}_0\in\Cinfty (P)$
such that $\Leg_0^*{\rm h}_0=E_\Lag$, where
$\Leg_0\colon\Tan Q\to P$ is defined by $\Leg=j_0\circ\Leg_0$.
This function is the {\sl canonical Hamiltonian function} of the system,
and has the same physical interpretation than ${\rm h}$.
Now, taking $\Omega_0=j_0^*\Omega$, the triple
$(P,\Omega_0,\d{\rm h}_0)$ is in this case the {\sl canonical Hamiltonian system}
associated with $(\Tan Q,\Lag)$, which is the equivalent to $(\Tan^*Q,\Omega,\d{\rm h})$ in the regular case, and is a presymplectic Hamiltonian system, in general.
In particular, the equation equivalent to (\ref{elmh}) is
$$
\inn (X_{{\rm h}_0})\Omega_0= \d {\rm h}_0 \quad ;\quad X_{{\rm h}_0}\in\vf(P) \ ,
$$
which, if $\Omega_0$ is a presymplectic form,
is incompatible in general and,  in the most interesting cases, $X_{{\rm h}_0}$ exists only in
some submanifold $P_f\hookrightarrow P$, and is tangent to it.
Moreover, the solution is not unique, since if $X_{{\rm h}_0}$ is a solution, then
$X_{{\rm h}_0}+Z$, for every $Z\in\ker\,\Omega_0$, is also a solution.

Details on the construction of the canonical Hamiltonian formalism for
almost-regular Lagrangians and a deeper study on the equivalence of
both, the Lagrangian and Hamiltonian formalisms, for this case
can be found, for instance, in
\cite{BGPR-86,CL-87,CLR-87,CLR-88,GN-79,GP-89,Ka-82,SR-83}).
}\end{remark}

\subsection{Discussion and comparison}

As a summary of the results in this section,
we state the essential characteristics of  the Lagrangian and canonical Hamiltonian
 formalisms of  the Lagrangian dynamical systems.

In the Lagrangian formalism we have that:
\ben
\item
Local description:
\ben
\item
To describe locally  the {\sl states} of the system we use
the {\sl generalized coordinates} $(q^i)$ ($i=1,\ldots ,n$), 
representing the degrees of freedom of the system, and
the {\sl generalized velocities} $(v^i)$ corresponding to each generalized coordinate.
\item
The dynamical information of the system is given by the
{\sl Lagrangian function} of the system, $\Lag (q^i,v^i)$.
\item
The dynamical  evolution of the system is described by the {\sl  Euler--Lagrange equations} \eqref{ELequats},
which are a second-order system of $n$ differential equations for the functions $q^i(t)$.
\item
The dynamical system  is {\sl regular} if
\(\dst{\rm det}\left(\derpar{^2\Lag}{v^i\partial v^j}(q,v)\right)\not= 0\)
in all  the points $(q,v)$.
In this case, given a set of initial conditions
(that is, a state of the system), the local solution to  the equations is unique.
\een
\item
In the geometrical description, the above characteristics are translated as:
\ben
\item
The phase space of the system is the tangent bundle
$\Tan Q$ of the manifold $Q$ which constitutes the configuration space of the system.
\item
The Lagrangian function $\Lag \in \Cinfty (\Tan Q)$
contains the dynamical information of the system.
\item
Starting from this function and using the geometric structures
of the tangent bundle (the {\sl  canonical endomorphism} and the
{\sl Liouville vector field\/}), one can construct
the {\sl  Lagrange $2$-form} $\Omega_\Lag$
and the {\sl Lagrangian energy} of the system, $ E_{\Lag}$,
and write the Lagrangian dynamical equations which are
$\inn (X_{\Lag})\Omega_{\Lag}=\d E_{\Lag}$,
together with the second-order condition $J(X_\Lag)=\Delta$.
The integral curves of $X_\Lag\in\vf(\Tan Q)$ are the
dynamical trajectories of the system.
\item
The dynamical system is {\sl regular} if
$\Omega_{\Lag}$ is a non-degenerated form.
In this case, the vector field $X_\Lag$ is unique and necessarily a {\sc sode}.
\een
\een

Concerning to the canonical Hamiltonian formalism, we have that:
\ben
\item
Local description:
\ben
\item
The states of the system are now described
by generalized coordinates of position $(q^i)$
and their corresponding generalized momenta $(p_i)$.
\item
The dynamical information of the system is given by the
{\sl Hamiltonian function} of the system, ${\rm h}(q^i,p_i)$.
\item
The dynamical evolution of the system is obtained from the
{\sl Hamilton equations} \eqref{hameq3},
which are a first-order system of $2n$ differential equations.
\item
The dynamical system  is {\sl regular} if $\Leg$ is a (local) diffeomorphism.

In this case, given a set of initial conditions
(that is, a state of the system), the solution to  the equations is unique.
\een
\item
In the geometrical description, the above characteristics are translated as::
\ben
\item
The phase space of the system is $\Leg(\Tan Q)$,
and is, whether the cotangent bundle
$\Tan^* Q$ of the manifold $Q$ which constitutes the
configuration space of the system, or an open set $U\subset\Tan^*Q$,
or a submanifold $P\hookrightarrow\Tan^*Q$
(depending on the regularity of $\Lag$).
\item
The dynamical information is given by ${\rm h}\in \Cinfty (\Tan^*Q)$,
(resp. ${\rm h}_0\in \Cinfty (P)$): the {\sl Hamiltonian function}.
\item
The Hamiltonian dynamical equations are obtained
using this function and the can\'onical  $2$-form
of the cotangent bundle $\Omega$, and they are
$\inn (X_{\rm h})\Omega=\d{\rm h}$ (resp. $\inn (X_{{\rm h}_0})\Omega_0=\d{\rm h}_0$).
The integral curves of the vector field solution are the
dynamical trajectories of the system.
\item
The dynamical systems  is {\sl regular} if
$\Leg$ is a (local) diffeomorphism.
\een
\een

The canonical Hamiltonian formalism is especially interesting by the following facts:
\bit
\item
Locally, it is manifest the asymmetry between the sets of variables $q^i$
and $v^i$ in  the  Lagrangian dynamical equations.
This does not happen in  the dynamical equations of the
Hamiltonian formalism, with the coordinates $q^i$ and $p_i$.
\item
Geometrically,
in the Lagrangian formalism, the Lagrange 2-form and the
Lagrangian energy, which appear in  the dynamical equations,
are obtained from the Lagrangian function and the canonical structures of the tangent bundle.
This means that the dynamical information is present in both the 
Lagrangian  2-form and the Lagrangian energy.

On the contrary, in the canonical Hamiltonian formalism,
the corresponding  geometric elements appearing
in  the dynamical equations are the canonical $2$-form
of the cotangent bundle and the Hamiltonian function;
the first of them contains only the geometric information
and the second one only the dynamical  information.
\item
Finally,  the characteristics of the canonical Hamiltonian formalism
are suitable in order to make the quantization of the physical system and,
in particular, to implement the so-called {\sl geometric quantization procedure}
(see, for instance, \cite{Bl-73,mfgq-98,Ga-83,Gos-2006,GS-77,Ki-gq,SW-76,Sn-80,So-ssd,Tu-85,Wo-80}).
\eit

In conclusion, a {\sl Lagrangian dynamical system} 
can be thought as a triple $(\Tan Q,\Omega_{\Lag},E_{\Lag})$
in the Lagrangian formalism, or
as a triple $(\Tan^* Q,\Omega,{\rm h})$ (or $(P,\Omega_0,{\rm h}_0)$),
in the canonical Hamiltonian formalism.
In both cases,  the phase spaces are differential manifolds
endowed with closed differential 2-forms (symplectic or presymplectic,
depending on the regularity of the system).

 \subsection{Geometric Hamilton--Jacobi theory}
\label{firstsec}

One of the most relevant features of the study of Hamiltonian systems is the
{\sl Hamilton--Jacobi theory}, which gives a way to integrate 
 Hamilton equations (and systems of first-order ordinary differential equations in general).
The classical theory is based on using canonical transformations
\cite{Ar-89,JS-98,Ga-70,LL-76,SC-71} and is also extended to singular systems \cite{Gomis1} 
or higher-order dynamics \cite{art:Constantelos84}.
The geometric description of the standard theory and other  geometric descriptions (for autonomous Hamiltonian systems)
is well-established in several texts, as \cite{AM-78,Ar-89,JS-98,BT-80,LM-sgam,MMMcq}.
Following the ideas stated in  \cite{AM-78,KV-1993},
the Hamilton--Jacobi theory has been formulated, recently, in a new more general geometric perspective
both for the Lagrangian and the Hamiltonian formalisms, 
for autonomous and nonautonomous mechanical systems \cite{CGMMR-06}. 
This framework has been extended to many other situations; namely,
to singular Lagrangian and singular Hamiltonian systems \cite{LMV-12,LMV-13,LOS-12},
holonomic and nonholonomic mechanics
\cite{BFS-14,HJTeam2,leones2,leones1,LOS-12,blo,OFB-11},
higher-order dynamical systems \cite{CLRP-13,CLRP-14},
control theory \cite{BLMMM-12,Wang2},
systems described using Poisson manifolds \cite{LMV-13,GP-2016} 
and Lie algebroids \cite{BMMP-10,LS-12},
dissipative systems described using contact manifolds \cite{LLLR-2023,dLLM-2021,
GP-2020},
discretization of dynamical systems \cite{BDM-12,OBL-11},
first-order classical field theories \cite{CLMV-14,LMM-09,LMMSV-12,LPRV-2017,DeLeon_Vilarino}
and higher-order field theories \cite{Vi-10,art:Vitagliano12},
partial differential equations in general \cite{Vi-11,Vi-14},
and other geometric applications and generalizations \cite{BLM-12,HJTeam3,GMP-2021,Wang3}.
(A review on this approach and some of its applications is given in \cite{RR-2021}).

Next, we review the foundations of this geometric framework of the Hamilton--Jacobi theory
for the Hamiltonian formalisms of autonomous dynamical systems.

Let $(\Tan^*Q,\omega,{\rm h})$ be a Hamiltonian system and $X_{\rm h}\in\vf(\Tan^*Q)$
the Hamiltonian vector field.

\begin{definition}
\label{genHJh}
The  \textbf{generalized Hamiltonian Hamilton--Jacobi problem}
consists in finding a vector field $X\in\vf(Q)$ and
a $1$-form $\alpha\in\df^1(Q)$
such that, if $\gamma\colon\Real\to Q$ is an integral curve
of $X$, then $\alpha\circ\gamma\colon\Real\to\Tan^*Q$
is an integral curve of $X_{\rm h}$; that is,
if $X\circ\gamma=\dot\gamma$, then
$\dot{\overline{\alpha\circ\gamma}}=X_{\rm h}\circ(\alpha\circ\gamma)$.
The pair $(X,\alpha)$ is said to be a  \textbf{solution to the generalized Hamiltonian Hamilton--Jacobi problem}.
\end{definition}

We have the following diagram:
$$
\xymatrix{
\ & \Tan Q \ar@/_1.5pc/[rr]_{\Tan\alpha}  & \ & \ar[ll]_{\Tan\pi_Q} \Tan(\Tan^*Q) \\
\ & \ & \ & \ \\
\Real  \ar[r]_{\gamma} \ & Q \ar[uu]^{X} \ar@/_1.5pc/[rr]_{\alpha} & \ & \ar[ll]_{\pi_Q} \Tan^*Q \ar[uu]_{X_{\rm h}}
}
$$

Solutions to the generalized Hamiltonian Hamilton--Jacobi problem are characterized as follows:

\begin{teor}
\label{HJhcns}
The following statements are equivalent:
\begin{enumerate}
\item
The pair $(X,\alpha)$ is a solution to the generalized Hamiltonian Hamilton--Jacobi problem.
\item 
The vector fields $X$ and $X_{\rm h}$ are $\alpha$-related; that is,
$X_{\rm h}\circ\alpha=\Tan\alpha\circ X$.

Therefore
$X=\Tan\pi_Q\circ X_{\rm h}\circ\alpha$ and, as a consequence,
the integral curves of $X_{\rm h}$ with initial conditions
in $\mathrm{Im}\,\alpha$ project onto the integral curves of 
$X$.

$X$ is called the \textbf{vector field associated with the form $\alpha$}.
\item
The submanifold ${\rm Im}\,\alpha$ of $\Tan^*Q$
is invariant by $X_{\rm h}$
(that is, $X_{\rm h}$ is tangent to ${\rm Im}\,\alpha$).
\item
$\inn(X)\d\alpha=-\d(\alpha^*H)$.
\end{enumerate}
\end{teor}
\begin{proof}
The equivalence between 1 and 2 is as follows:
If $(X,\alpha)$ satisfies the condition in Definition \ref{genHJh} then, 
for every integral curve of $X$,
$\gamma\colon\Real\to Q$, by definition $X\circ\gamma=\dot\gamma$,
then
$$
X_{\rm h}\circ\alpha\circ\gamma=\dot{\overline{\alpha\circ\gamma}}=
T\alpha\circ\dot\gamma=\Tan\alpha\circ X\circ\gamma \ ;
$$
and hence $X_{\rm h}\circ\alpha=\Tan\alpha\circ X$.
The proof of the converse is immediate.

From here, composing both members of this equality with $\Tan\pi_Q$
and taking into account that $\pi_Q\circ\alpha={\rm Id}_Q$,
we obtain that $X=\Tan\pi_Q\circ X_{\rm h}\circ\alpha$.
As a consequence,
we conclude that the $\pi_Q$-projection of the integral curves of $X_{\rm h}$
on ${\rm Im}\,\alpha$ are the integral curves of $X$.

For the equivalence between 2 and 3 we have that,
if $X_{\rm h}\circ\alpha=\Tan\alpha\circ X$, then 
$X_{\rm h}(\alpha({\rm q}))=\Tan\alpha(X({\rm q}))$, for every ${\rm q}\in Q$
and therefore $X_{\rm h}$ is tangent to ${\rm Im}\,\alpha$.
Conversely, if ${\rm Im}\,\alpha$ is invariant by $X_{\rm h}$ then
$X_{\rm h}(\alpha({\rm q}))\in\Tan_{\alpha({\rm q})}{\rm Im}\,\alpha$,
which implies that there exists ${\rm u}\in\Tan_{\rm q}Q$ such that
$X_{\rm h}(\alpha({\rm q}))=\Tan_{\rm q}\alpha({\rm u})$. Defining
$X$ by $\Tan_{\rm q}\alpha(X_{\rm q})=X_{\rm h}(\alpha({\rm q}))$,
we have that $X$ is differentiable, since $X=\Tan\pi_Q\circ X_{\rm h}\circ\alpha$,
and thus $X$ is a vector field in $Q$ satisfying that
$X_{\rm h}\circ\alpha=\Tan\alpha\circ X$, and then $\alpha$ is a solution to the
generalized Hamilton--Jacobi problem.

If $(X,\alpha)$ is a solution to the generalized Hamilton--Jacobi problem, by item 2 we have that $X=\Tan\pi_Q\circ X_{\rm h}\circ\alpha$
and 

Finally, the proof of the equivalence with 4 is the following:
from the Hamiltonian equation \eqref{elmh} we obtain that
$$
\alpha^*\inn(X_{\rm h})\omega=\alpha^*\d{\rm h}=\d(\alpha^*{\rm h}) \ .
$$
As $\Theta$ is the canonical form of $\Tan^*Q$, then
$\alpha^*\Theta=\alpha$, and
\beq
\alpha^*\Omega = -\alpha^*\d\Theta=-\d(\alpha^*\Theta)=-\d\alpha \ ,
\label{00}
\eeq
and, as $X$ and $X_{\rm h}$ are $\alpha$-related, we have
$$
\alpha^*\inn(X_{\rm h})\omega=\inn(X)\alpha^*\Omega=-\inn(X)\d\alpha
\ ,
$$
which yields condition 4. 
For the converse, first define
$Z_{\rm h}=X_{\rm h}\circ\alpha-\Tan\alpha\circ X\colon Q\to \Tan(\Tan^*Q)$,
which is a vector field along $\alpha$;
then we have to prove that $Z_{\rm h}=0$.
First, we have that $Z_{\rm h}$ is $\pi_Q$-vertical since,
as $\pi_Q\circ\alpha={\rm Id}_Q$,
\beann
\Tan\pi_Q\circ Z_{\rm h} &=& \Tan\pi_Q\circ(X_{\rm h}\circ\alpha-\Tan\alpha\circ X)=
\Tan\pi_Q\circ(X_{\rm h}\circ\alpha-\Tan\alpha\circ\Tan\pi_Q\circ X_{\rm h}\circ\alpha)
\\ &=&
\Tan\pi_Q\circ X_{\rm h}\circ\alpha-\Tan\pi_Q\circ X_{\rm h}\circ\alpha=0 \ .
\eeann
From equation \eqref{elmh}
and the hypothesis, as $\alpha^*\Omega=-\d\alpha$, we obtain
\beann
\alpha^*\inn(X_{\rm h})\Omega = \alpha^*\d{\rm h} &=& d(\alpha^*{\rm h})
\ ,
\\
\inn(X)\alpha^*\Omega=-\inn(X)\d\alpha &=& \d(\alpha^*{\rm h}) \ ,
\eeann
hence $\alpha^*\inn(X_{\rm h})\Omega-\inn(X)\alpha^*\Omega=0$.
Therefore, for every ${\rm q}\in Q$ and $Y_{\rm q}\in\Tan_{\rm q}Q$, we have
\beann
0 &=& (\alpha^*\inn(X_{\rm h})\Omega-\inn(X)\alpha^*\Omega)_{\rm q}(Y_{\rm q})=
\Omega_{\alpha({\rm q})}(X_{\rm h}(\alpha({\rm q})),\Tan_{\rm q}\alpha(Y_{\rm q}))-
\Omega_{\alpha({\rm q})}(\Tan_{\rm q}\alpha(X_{\rm q}),\Tan_{\rm q}\alpha(Y_{\rm q}))
\\
&=& \Omega_{\alpha({\rm q})}(Z_{\rm h}({\rm q}),\Tan_{\rm q}\alpha(Y_{\rm q})) \ .
\eeann
Furthermore, as ${\rm V}(\pi_Q)$ is a Lagrangian distribution in 
$(\Tan^*Q,\Omega)$, for every $\pi_Q$-vertical vector field $V\in\vf(\Tan^*Q)$ 
we have that
$\Omega_{\alpha({\rm q})}(Z_{\rm h}({\rm q}),V(\alpha({\rm q})))=0$;
but $\Tan_{\alpha({\rm q})}\Tan^*Q=
\Tan_{\alpha({\rm q})}({\rm Im}\,\alpha)\oplus{\rm V}_{\alpha({\rm q})}(\pi_Q)$,
and thus we have proved that
$$
\Omega_{\alpha({\rm q})}(Z_{\rm h}({\rm q}),Z(\alpha({\rm q})))=0 \ ;
$$
for every $Z\in\vf(\Tan^*Q)$. Then we conclude that $Z_{\rm h}=0$,
since $\Omega$ is non-degenerate;
or what is equivalent,
$X$ and $X_{\rm h}$ are $\alpha$-related, and thus
$(X,\alpha)$ is a solution to the generalized Hamilton--Jacobi problem
\\ \qed \end{proof}

To solve the generalized Hamilton--Jacobi problem is, in general, a very difficult task. 
Then, it is usual to state a less general version of the problem
which leads to the standard version of the Hamilton--Jacobi problem:

\begin{definition}
\label{HJh}
The \textbf{Hamiltonian Hamilton--Jacobi problem}
consists in finding a closed $1$-form $\alpha\in\df^1(\Tan^*Q)$
which is a solution to the generalized Hamiltonian Hamilton--Jacobi problem.
This form  $\alpha$ is said to be a \textbf{solution to the Hamiltonian Hamilton--Jacobi problem}.
\end{definition}

As $\d\alpha=0$, for every point in $Q$,
there exists a function $S$ in  a neighbourhood $U\subset Q$
such that $\alpha=\d S$.
Then, we say that $S$ is a {\sl \textbf{local generating function}} of the solution $\alpha$.

\begin{teor}
\label{cnsh}
The following statements are equivalent:
\begin{enumerate}
\item
The form  $\alpha\in\df^1(Q)$ is a solution to the Hamiltonian Hamilton--Jacobi problem.
\item
$\mathrm{Im}\,\alpha$ is a Lagrangian submanifold of
$\Tan^*Q$, which is invariant by $X_{\rm h}$,
and $S$ is a local generating function of this Lagrangian submanifold.
\item
The condition  $\d(\alpha^*{\rm h})=0$ holds
or, what is equivalent, the function
${\rm h}\circ\d S\colon Q\to\Real$ is locally constant.
\end{enumerate}
\end{teor}
\begin{proof}
These statements are consequences of Theorem \ref{HJhcns} and Definition \ref{HJh}.
In fact, if $\alpha$ is a solution to the Hamilton--Jacobi problem, 
as a consequence of \eqref{00},
$d\alpha=0$ is equivalent to $\alpha^*\Omega=0$. Then
${\rm Im}\,\alpha$ is a Lagrangian submanifold of $(\Tan^*Q,\Omega)$, 
which is contained in a level set of
${\rm h}$, because the condition $\inn(X)\d\alpha=-\d(\alpha^*{\rm h})$ 
implies that $d(\alpha^*{\rm h})=0$.
Notice that $\dim\,{\rm Im}\,\alpha=n$ and,
if $j_\alpha\colon{\rm Im}\,\alpha\hookrightarrow \Tan^*Q$
is the natural embedding, we have that $j_\alpha^*\Omega=0$.
\\ \qed \end{proof}

In natural coordinates of $\Tan^*Q$, condition (3) is the classical form of the
Hamiltonian Hamilton--Jacobi equation:
\beq
H\left(q^i,\derpar{S}{q^i}\right)= E\ (ctn.) \ .
\label{tres}
\eeq

Until now, we have only considered particular solutions $\alpha$
to the (generalized) Hamilton--Jacobi problem, which are given by particular solutions
to the partial differential equation \eqref{tres}.
Nevertheless, we are also interested in the general solution. Then:

\begin{definition}
\label{completeHJ}
Let $\Lambda\subseteq\Real^n$.
A family of  solutions $\{ \alpha_\lambda; \lambda\in\Lambda\}$
(which depends on $n$ parameters
$\lambda\equiv (\lambda_1,\ldots,\lambda_n)\in\Lambda$)
is a \textbf{complete solution} to the Hamilton--Jacobi problem 
if the map 
$$
\begin{array}{ccccc}
\Phi & \colon &Q\times\Lambda & \longrightarrow & \Tan^*Q \\
 & & (q,\lambda) & \mapsto & \alpha_\lambda(q)
\end{array}
$$
is a local diffeomorphism.
\end{definition}

\begin{remark}{\rm 
\bit
\item
Given a complete solution $\{ \alpha_\lambda; \lambda\in\Lambda\}$,
since $\d\alpha_\lambda=0$, for every $\lambda\in\Lambda$, 
there is a family of functions $\{ S_\lambda\}$ defined in 
open sets $U_\lambda\subset Q$ such that $\alpha_\lambda=\d S_\lambda$.
Therefore, we have a function
$$
\begin{array}{ccccc}
{\cal S} & \colon &\bigcap U_\lambda\times\Lambda\subset Q\times\Lambda & \longrightarrow & \Real \\
 & & (q,\lambda) & \mapsto & S_\lambda(q)
\end{array}
$$
which is locally defined, and is called a {\sl \textbf{local generating function}} of the 
complete solution $\{ \alpha_\lambda; \lambda\in\Lambda\}$.
\item
Every complete solution defines a Lagrangian foliation in $\Tan^*Q$
which is  transverse to the fibers,
and such that $X_{\rm h}$ is tangent to the leaves.
This foliation is locally defined by a family of functions which are the components of a map
$F\colon \Tan^*Q\stackrel{\Phi^{-1}}{\longrightarrow} Q\times\Lambda 
{\longrightarrow} \Lambda\subset\Real^n$.
Furthermore, these functions are a set of constants of motion of $X_{\rm h}$.

Conversely, from a set  $n$ first integrals $f_1,\ldots,f_n$ 
of $X_{\rm h}$ in involution, such that $\d f_1\wedge\ldots\wedge \d f_n\not=0$;
we can define a $\pi$-transversal Lagrangian foliation of $\Tan^*Q$
taking $f_i=\lambda_i$, with $\lambda_i\in\Real$,
and in this way we obtain a local complete solution $\{\alpha_\lambda,\lambda\in\Lambda\}$.
Then, from equations $f_i=\lambda_i$, we can  locally isolate
$p_i=p_i(q,\lambda)$, replace them in the expression of $X_{\rm h}$
and finally project to the basis, then obtaining the family of vector fields 
$\{X_\lambda\}$ associated with the local complete solution.
From a complete solution $\{ \alpha_\lambda;\lambda\in\Lambda\}$,
all the integral curves of  $X_{\rm h}$ are obtained starting from
the integral curves of the vector fields $\{ X_\lambda\}$
associated to this complete solution.
\item
This geometric framework for the Hamilton--Jacobi theory 
can be also stated in a very natural way for the Lagrangian formalism
(see \cite{CGMMR-06}).
\eit
}\end{remark}

The ``classical'' Hamilton--Jacobi problem for a (regular) Hamiltonian system
$(\Tan^*Q,\Omega,{\rm h})$ consists in 
obtaining a canonical transformation which leads 
the system to equilibrium \cite{Ar-89,JS-98,LL-76,SC-71}.
This transformation is given by a generating function, which is  
the solution to the {\sl Hamilton--Jacobi equation}.
From a geometric point of view, this canonical transformation
is associated with a foliation in the phase space of the system, $\Tan^*Q$,
which has the following properties:
it is invariant by the dynamics, transverse to the canonical projection of the cotangent bundle,
and is Lagrangian with respect to the
canonical symplectic structure of $\Tan^*Q$.
Then, the restriction of the Hamiltonian vector field $X_{\rm h}\in\vf(\Tan^*Q)$ to
each leaf  $S_\lambda$ of this foliation  projects onto a vector field $X_\lambda\in\vf(Q)$,
and the integral curves of $X_{\rm h}$ and  $X_\lambda$ are one-to-one related.
In this way,  all the dynamical trajectories are recovered from the integral curves
of all these vector fields $\{ X_\lambda\}$.
Thus, the geometric Hamilton--Jacobi problem consists in finding
this foliation and  the vector fields $\{ X_\lambda\}$.

Bearing this in mind, the  relation between the classical and the geometric Hamilton--Jacobi theories
is established through the equivalence of complete solutions 
and canonical transformations (see \cite{RR-2021,Vi-14}).

\begin{teor}
A complete solution $\{ \alpha_\lambda; \lambda\in\Lambda\}$
to the Hamilton--Jacobi problem provides a canonical transformation 
$\Phi\colon \Tan^*Q\to \Tan^*Q$ leading the system to equilibrium,
and conversely.
\end{teor}
\begin{proof}
Let $\{ \alpha_\lambda; \lambda\in\Lambda\}$ be a complete solution, and
let ${\cal S}$ be a generating function of it
in a neighbourhood of every point of $\Tan^*Q$.
As ${\cal S}={\cal S}(q^i,\lambda^i)$, the set $\lambda^i$ can be identified
with a subset  of coordinates $\lambda^i\equiv\tilde q^i$ in $\Tan^*Q\times\Tan^*Q$, 
and therefore ${\cal S}={\cal S}(q^i,\tilde q^i)$ can be thought as
the local expression of a generating function
of a  local canonical transformation $\Phi$, and hence
of an open set $W$ of the Lagrangian submanifold 
$graph\,\Phi\hookrightarrow \Tan^*Q\times\Tan^*Q$.
When this construction is done in every chart, we obtain the transformation $\Phi$
and the submanifold $graph\,\Phi$.
Finally, since (\ref{tres}) holds for every particular solution $S_\lambda$, we have that
$$
E={\rm h}\left(q^i(\tilde q,\tilde p),\derpar{{\cal S}}{q^i}(q(\tilde q,\tilde p),\tilde q)\right)= 
{\rm h}(\tilde q^i,\tilde p_i) \ .
$$

Conversely, starting from a canonical transformation $\Phi$ and a
generating function ${\cal S}={\cal S}(q^i,\tilde q^i)$;
if we take $\tilde q\equiv (\tilde q^i)=(\lambda^i)\equiv\lambda$, we obtain
a family of functions $\{ S_\lambda\}$ and then we get a
local complete solution $\{\alpha_\lambda=\d S_\lambda; \lambda\in\Lambda\}$
to the Hamiltonian Hamilton--Jacobi problem.
Doing this construction in every chart, we have the complete solution.
This means that, on each local chart of $\Tan^*Q$,
fixing the coordinates $\tilde q^i=\lambda^i$ of a point,
we obtain a local submanifold of $\Tan^*Q$ whose image by $\Phi^{-1}$ is
the image of a local section $\alpha_\lambda\colon Q\to \Tan^*Q$
which constitutes a particular solution to the Hamiltonian Hamilton--Jacobi problem.
\\ \qed \end{proof}

%%%%%%%%%%%%%%%%%%%%%%%%%%%%%%

\section{Skinner-Rusk unified Lagrangian-Hamiltonian formalism}
\label{SRuf}

In their seminal articles of 1983 \cite{SR-83,SR-83b}, 
{\it R. Skinner} and {\it R. Rusk} proposed a new geometric framework 
in order to unify the Lagrangian and the Hamiltonian formalisms
of first-order autonomous mechanical systems
into a single one formulation.
This is a simpler and elegant framework
which is particularly suitable for the treatment of singular systems.
Later, this nice formalism was generalized to many other types of physical systems;
such as nonautonomous dynamical systems \cite{BEMMR-2008,CMC-2002,GM-05},
vakonomic and nonholonomic mechanics \cite{CLMM-2002}
control systems \cite{BEMMR-2007,CMZ-10}, 
higher-order mechanics \cite{art:Prieto_Roman11,
art:Prieto_Roman12},  
dissipative systems (first and higher-order contact mechanics)  \cite{LGLMR-2021,LGMMR-2020},
and first-order and higher-order classical field theories
\cite{CLMV-09,CV-2007,LMM-2003,
ELMMR-04,PR-2015,RRS-2005,RRSV-2011,Vitagliano10}.

In this section, we describe the main features of this so-called {\sl Skinner-Rusk} 
or {\sl unified Lagrangian-Hamiltonian formalism} for autonomous dynamical systems.

\subsection{Unified bundle. Unified formalism}

This formalism is developed in the following bundle:

\begin{definition}
The  \textbf{unified bundle} or  \textbf{Pontryagin bundle}
 is ${\cal W}=\Tan Q\times_Q\Tan^*Q$
and has natural projections
$$
 \varrho_1\colon{\cal W}\to\Tan Q \ ,\
 \varrho_2\colon{\cal W}\to\Tan^*Q \ ,\
\varrho_0\colon{\cal W}\to Q \ ,
$$
\end{definition}

Natural coordinates in ${\cal W}$ are $(q^i,v^i,p_i)$.

\begin{definition}
A curve $c\colon\Real\rightarrow{\cal W}$
is \textbf{holonomic} in ${\cal W}$ if
$\varrho_1\circ c\colon\Real\to\Tan Q$ is holonomic.

A vector field $\Gamma \in\vf({\cal W})$  is a \textbf{holonomic vector field} in ${\cal W}$
if its integral curves  are holonomic in ${\cal W}$. 
\end{definition}

The coordinate expressions of holonomic curves and vector fields in  ${\cal W}$ are the following
\beann
c(t)&=&\Big(q^i(t),\frac{dq^i}{d t}(t),p_i(t),\Big) \ , \\
\Gamma&=& v^i\derpar{}{q^i}+F^i\derpar{}{v^i}+G_i\derpar{}{p_i}  \ .
\eeann

\begin{definition}
The unified bundle ${\cal W}$ is endowed with the following canonical structures:
\ben
\item
The {\sl \textbf{coupling function}} is the
map ${\cal C}\colon{\cal W}\to\Real$  defined by
$$
\begin{array}{rclcc} 
{\cal C} &\colon& \Tan Q \times_Q\Tan^*Q & \longrightarrow & \Real \\ 
& &(q,v_q,\xi_q)=(q^i,v^i,p_i)& \longmapsto & \langle v_q \mid \xi_q\ \rangle=v^ip_i
 \end{array} \ .
$$
\item
If $\Theta\in\df^1(\Tan^*Q)$ and $\Omega=-\d\Theta\in\df^2(\Tan^*Q)$
are the canonical forms in $\Tan^*Q$,
the \textbf{canonical forms} in ${\cal W}$ are  
$$
\Theta_{\cal W}:=\varrho_2^*\,\Theta\in\df^1({\cal W}) \quad , \quad
\Omega_{\cal W}:=-\d\Theta_{\cal W}=\varrho_2^*\,\Omega\in\df^2({\cal W}) \ .
$$
\een
\end{definition}

And using the coupling function, we introduce:

\begin{definition}
Given a Lagrangian function $\Lag\in\Cinfty(\Tan Q)$, if
${\mathfrak L}=\varrho_1^*\Lag\in\Cinfty({\cal W})$, the \textbf{Hamiltonian function} is defined as
$$
{\cal H}:={\cal C}-{\mathfrak L}\in\Cinfty({\cal W}) \ .
$$
\end{definition}

The coordinate expressions of these elements are
$$
\Theta_{\cal W}=p_i\d q^i \quad , \quad
\Omega_{\cal W}=\d q^i\wedge\d p_i\quad , \quad
{\cal H}=v^ip_i-{\mathfrak L}(q^i,v^i) \ .
$$

The triple $({\cal W},\Omega_{\cal W},{\cal H})$
is a presymplectic Hamiltonian system since
$\dst\ker\Omega_{\cal W}=\left\langle\derpar{}{v^i}\right\rangle$.
Then, the {\sl\textbf{dynamical problem}} for this system consists in finding $X_{\cal H}\in\vf({\cal W})$
which is a solution to the Hamiltonian equations
\begin{equation}
\label{Whamilton-contact-eqs0}
\inn(X_{\cal H})\Omega_{\cal W}=\d{\cal H} \ ,
\end{equation}
and then the integral curves $c\colon\Real\to{\cal W}$
 of $X_{\cal H}$ are solutions to the equations
\begin{equation}
\label{Whamilton-contact-eqs0b}
\inn(\widetilde c)(\Omega_{\cal W}\circ c)=\d{\cal H}\circ c \ .
\end{equation}

As $({\cal W},\Omega_{\cal W},{\cal H})$ is a presymplectic Hamiltonian system,
these equations are not compatible in ${\cal W}$.
In fact, for an arbitrary vector field in $\vf({\cal W})$,
$$
X_{\cal H} =  f^i\derpar{}{q^i}+F^i\derpar{}{v^i}+G_i\derpar{}{p_i} \ ,
$$ 
equations \eqref{Whamilton-contact-eqs0} give
\beq
 f^i= v^i \quad , \quad
 G_i= \displaystyle \derpar{\Lag}{q^i} \quad ,  \quad
p_i= \derpar{\Lag}{v^i}  \ . 
\label{firstuni}
\eeq
\begin{itemize}
\item
The first equations assure that $X_{\cal H}$ is a holonomic vector field in ${\cal W}$
(regardless of the regularity of the Lagrangian function).
\item
The second equations allow us to determine the component functions $G_i$.
\item
The third equations are compatibility conditions;
that is, {\sl compatibility constraints} defining a submanifold 
${\cal W}_0\hookrightarrow{\cal W}$
where vector fields $X_{\cal H}$ solution to \eqref{Whamilton-contact-eqs0} are defined.
Observe that these constraints give the Legendre map and hence ${\cal W}_0={\rm graph}({\Leg})$.
\end{itemize}
Thus, we have that
$$
X_{\cal H} \vert_{{\cal W}_0}= v^i\derpar{}{q^i}+F^i\derpar{}{v^i}+
\derpar{\Lag}{q^i}\derpar{}{p_i}\ ,
$$
where the functions $F^i$ are still undetermined.
Nevertheless, the constraint algorithm for presymplectic systems continues by demanding that
$X_{\cal H}$ is tangent to ${\cal W}_0$; that is, we have
$\dst{X_{\cal H}\Big(p_i- \derpar{\Lag}{v^i}\Big)}\Big\vert_{_{{\cal W}_0}} = 0,$
which gives the equations for the remaining coefficients $F^i$,
\beq
 \frac{\partial^2L}{\partial v^i\partial v^j}F^j+ \ \frac{\partial^2L}{\partial q^j\partial v^i}v^j- \derpar{\Lag}{q^i}=0 \quad (\text{on } \mathcal{W}_0)  \ .
\label{eluni}
\eeq

If $\Lag$ is regular, these equations are compatible and define a unique vector field $X_{\cal H}$ solution to \eqref{Whamilton-contact-eqs0} on ${\cal W}_0$,
and the last system of equations give the dynamical trajectories.
If $\Lag$ is singular, equations \eqref{eluni} can be compatible or not.
and, eventually, new compatibility constraints can appear
that define a new submanifold ${\cal W}_1\hookrightarrow{\cal W}_0$.
In that case, the constraint algorithm continues by demanding the tangency of solutions to the new constraint submanifold ${\cal W}_1$ and so on. 
In the most favorable cases,
there is a submanifold ${\cal W}_f \hookrightarrow {\cal W}_0$ (it could be ${\cal W}_f = {\cal W}_0$)
such that there exist holonomic vector fields $X_{\cal H}\in\vf({\cal W})$ defined on ${\cal W}_0$ and tangent to ${\cal W}_f$,
which are solutions to equations \eqref{Whamilton-contact-eqs0}
at support on ${\cal W}_f$.

\subsection{Recovering the Lagrangian and Hamiltonian formalisms}

Next, we study the equivalence of the unified formalism 
with the Lagrangian and the Hamiltonian formalisms.
We restrict our analysis to the hyperregular case
(the regular case is the same, at least locally).

Denoting by  $\jmath_0\colon{\cal W}_0\hookrightarrow{\cal W}$
the natural embedding, we have that
$$
(\varrho_1\circ\jmath_0)({\cal W}_0)=\Tan Q
\quad , \quad
(\varrho_2\circ\jmath_0)({\cal W}_0)=\Tan^*Q \ ,
$$
and being ${\cal W}_0$ the graph of the Legendre map,
the restricted projection $\varrho_1\circ\jmath_0$
is a diffeomorphism between
${\cal W}_0$ and $\Tan Q$.
The following diagram summarizes the situation:
$$
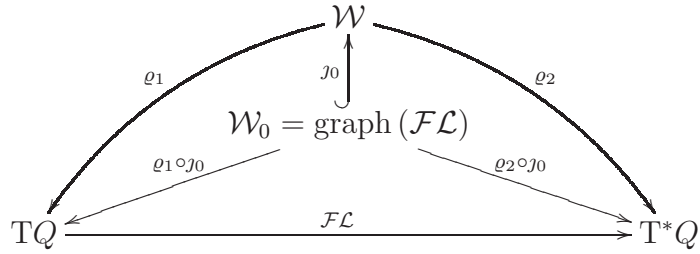
\xymatrix{
\ & \ & {\cal W} \ar@/_1.3pc/[ddll]_{\varrho_1} \ar@/^1.3pc/[ddrr]^{\varrho_2} & \ & \ \\
\ & \ & {\cal W}_0={\rm graph}\,(\Leg) \ar[dll]_{\varrho_1\circ\jmath_0} \ar[drr]^{\varrho_2\circ\jmath_0} \ar@{^{(}->}[u]^{\jmath_0} & \ & \ \\
\Tan Q\ar[rrrr]^<(0.45){\Leg}
& \ & \ & \ & \Tan^*Q  \\
}
$$
Therefore, functions, differential forms, and vector fields in ${\cal W}$ tangent to ${\cal W}_0$
can be restricted to ${\cal W}_0$, and then
they can be translated to the Lagrangian side 
by using this diffeomorphism, and to the Hamiltonian side
using the Legendre map and the projection $\varrho_2$.

In particular, if $c(t)=(q^i(t),v^i(t),p_i(t))$ is a solution to equation \eqref{Whamilton-contact-eqs0}
(or, what is equivalent, $c(t)$ is an integral curve of the vector field
$\,X_{\cal H}$ solution to the dynamical equations \eqref{Whamilton-contact-eqs0})
then equation \eqref{eluni} leads to
\beq
 \frac{d}{dt}\Big(\derpar{\Lag}{v^i}\circ c\Big)= \derpar{\Lag}{q^i}\circ c \ ,
\label{recuEL}
\eeq
and from equations \eqref{firstuni} we obtain
\beq
 \frac{dq^i}{dt}= v^i \quad , \quad
 \frac{dp_i}{dt}= \derpar{\Lag}{q^i}\circ c= -\derpar{{\cal H}}{q^i}\circ c \quad ,  \quad
p_i= \derpar{\Lag}{v^i}\circ c  \ . 
\label{recuH}
\eeq
From the first group of equations \eqref{recuH}, together with \eqref{recuEL},
we recover the Euler--Lagrange equations for the curves $c_\Lag(t)=(q^i(t),v^i(t))$.
Furthermore, bearing in mind the local expression of ${\cal H}$,
we have that $\dst \derpar{\Lag}{q^i}=-\derpar{{\cal H}}{q^i}$ and hence
the second group of equations \eqref{recuH} reads
$$
 \frac{dp_i}{dt}= -\derpar{{\cal H}}{q^i}\circ c
$$
and, using again the local expression of ${\cal H}$ and the first group of equations \eqref{recuH}, we get
$$
\derpar{{\cal H}}{q^i}\circ c=v^i=\frac{dq^i}{dt} \ ;
$$
finally, using the third group of equations \eqref{firstuni} (that is, the Legendre map)
we have that ${\cal H}=\Leg^*{\rm h}$, and these last equations become the {\sl Hamilton equations}
for the curves $c_{\rm h}(t)=(q^i(t),p_i(t))$.

In this way, for the dynamical trajectories, we can state.

\begin{teor}
\label{eqW}
Every curve $c\colon \Real\to{\cal W}$,
taking values in ${\cal W}_0$ can be viewed as
$c=(c_\Lag,c_{\rm h})$, where
$c_\Lag=\varrho_1\circ c\colon \Real \to\Tan Q$
and $c_{\rm h}=
\Leg\circ c_\Lag\colon\Real\to\Tan^*Q$.

If $c\colon\Real\to{\cal W}$,
with  ${\rm Im}\,c\subset{\cal W}_0$,
is a curve fulfilling equation \eqref{Whamilton-contact-eqs0b}, then
$c_\Lag$ is the lift to
$\Tan Q$ of the projected curve
$c_o=\varrho_0\circ c\colon\Real\to Q$ 
(that is, $c_\Lag$ is a holonomic curve),
and it is a solution to equation \eqref{Lcurv},
where $E_\Lag\in\Cinfty(\Tan Q)$ is such that ${\cal H}=\varrho_1^*E_\Lag$.
Moreover, the curve  
$c_{\rm h}=\varrho_2\circ c=\Leg\circ c_\Lag$
is a solution to equation \eqref{elmh2},
where ${\rm h}\in\Cinfty(\Tan^*Q)$ is such that 
${\cal H}=\varrho_1^*{\rm h}$.

Conversely, if $c_o\colon\Real\to Q$ is a curve such that
$\widetilde c_o\equiv c_\Lag$ is a solution to equation \eqref{Lcurv}, then the curve
$c=(c_\Lag,\Leg\circ c_\Lag)$
is a solution to equation \eqref{Whamilton-contact-eqs0b}
and $\Leg\circ c_\Lag$
is a solution to equation \eqref{elmh2}
 \label{mainteor1}.
\end{teor}

Now, the curves $c\colon\Real\to{\cal W}$ which are
solution to equation \eqref{Whamilton-contact-eqs0b}
are the integral curves of a holonomic vector field $X_{\cal H}\in\vf({\cal W})$ 
which is the solution to \eqref{Whamilton-contact-eqs0},
the curves $c_\Lag\colon\Real\to\Tan Q$
are the integral curves of the holonomic vector field $X_\Lag\in\vf(\Tan Q)$ which is the
solution to  \eqref{elm}, and  the curves $c_{\rm h}\colon\Real\to\Tan^*Q$
are the integral curves of the vector field $X_{\rm h}\in\vf(\Tan^*Q)$ 
which is the solution to  \eqref{elmh}.
Then, as a corollary of the above theorem, 
for the dynamical vector fields we have:

\begin{teor}
\label{eqL1}
Let $X_{\cal H}\in\vf({\cal W})$ be the solution to equations \eqref{Whamilton-contact-eqs0} (on ${\cal W}_0$),
which is tangent to  ${\cal W}_0$.  Then:

The vector field $X_\Lag\in\vf(\Tan Q)$, defined by
$X_\Lag\circ\varrho_1=\Tan\varrho_1\circ X_{\cal H}$,
is the solution to equations \eqref{elm} and \eqref{edso},
where $E_\Lag\in\Cinfty(\Tan Q)$ is such that ${\cal H}=\varrho_1^*E_\Lag$.

The vector field $X_{\rm h}\in\vf(\Tan^*Q)$, defined by
$X_{\rm h}\circ\varrho_2=\Tan\varrho_2\circ X_{\cal H}$,
is the solution to equations \eqref{elmh},
where ${\rm h}\in\Cinfty(\Tan^*Q)$ is such that 
${\cal H}=\varrho_1^*{\rm h}$.
Furthermore $\Leg_*X_\Lag=X_{\rm h}$.
\end{teor}

Summarizing, the main features of the unified formalism are:
\bit
\item
It assures holonomy (even in the non-regular case).
\item
It provides the Legendre map.
\item
It gives the Euler--Lagrange and the Hamilton equations.
\eit

\section{Symmetries of regular Lagrangian systems}

For Lagrangian dynamical systems,
the phase space of the system is either
$M=\Tan Q$ in the Lagrangian formalism, or $M=\Tan^* Q$,
in the canonical Hamiltonian formalism.
In these cases, there exist distinguished symplectic potentials:
the Lagrangian $1$-form $\Theta_\Lag\in\df^1(\Tan Q)$ and the
canonical $1$-form $\Theta\in\df^1(\Tan^*Q)$.
Furthermore, the symmetries of the
dynamical systems use to be canonical lifts
of diffeomorphisms or vector fields on the base manifold $Q$.
All of this lead to introduce new kinds of symmetries in that cases,
whose properties are studied next.

\subsection{Symmetries in the canonical Hamiltonian formalism}

Consider a (regular) canonical Hamiltonian system $(\Tan^*Q,\Omega,{\rm h})$,
and let $X_{\rm h}\in\vf_{lh}(\Tan^*Q)$ be the Hamiltonian vector field of the system.

Obviously, all we have stated for Hamiltonian dynamical systems
in Section \ref{secsym} holds for this particular case.
Nevertheless, new kinds of symmetry can be introduced for this situation.
Thus, in addition to the previous symmetries already defined,
we can consider the following particular cases:

\begin{definition}
A dynamical symmetry $\Phi\in {\rm Diff}\, (\Tan^*Q)$ of the canonical Hamiltonian system
is a \textbf{natural dynamical symmetry} if there exists $\varphi\in{\rm Diff}\, (Q)$ 
such that $\Phi=\Tan^*\varphi$ (that is,  $\Phi$ is the canonical lift
of a diffeomorphism in $Q$).
 \end{definition}

\begin{definition}
An infinitesimal dynamical symmetry $Y\in\vf (Q)$ of the canonical Hamiltonian system
is a \textbf{natural infinitesimal dynamical symmetry} if there exists $Z\in\vf (Q)$ 
such that $Y=Z^*$ (that is, $Y$ is the canonical lift of a vector field in $Q$
\footnote{
The terminology ``natural symmetries'' refers both to the diffeomorphism $\varphi$
or the vector field $Z$.}).
\label{gsdef2}
 \end{definition}

Remember that $(\Tan^*Q,\Omega)$ is an exact
symplectic manifold and a symplectic  potential of $\Omega$
is the canonical $1$-form $\Theta\in\df^1(\Tan^*Q)$.
Furthermore, for every $\varphi\in{\rm Diff}\,(Q)$,
we have that $(\Tan^*\varphi)^*\Theta=\Theta$, and hence
$(\Tan^*\varphi)^*\Omega=\Omega$.
In the same way, for every $Z\in\vf(Q)$,
we have that  $\Lie(Z^*)\Theta=0$, and then  $\Lie(Z^*)\Omega=0$;
therefore, if $Z^*$ is an infinitesimal natural dynamical symmetry,
then $Z^*\in\vf_H(\Tan^*Q)$ and, as we saw in Proposition \ref{hamlev}, 
the global Hamiltonian function of $Z^*$ is $f_Z=\inn (Z^*)\Theta$ (up to constants).
This leads to introduce the following particular type of Noether symmetries
for the canonical Hamiltonian system $(\Tan^*Q,\Omega,{\rm h})$:

\begin{definition}
A diffeomorphism $\Phi\in {\rm Diff}\, (\Tan^*Q)$ is a
\textbf{natural Noether symmetry} if:
\ben
\item
There exists a diffeomorphism $\varphi\in{\rm Diff}\, (Q)$ such that
$\Phi=\Tan^*\varphi$.
\item
$\Phi^*{\rm h}=(\Tan^*\varphi)^*{\rm h}={\rm h}$.
\een
\label{nCs}
\end{definition}

\begin{definition}
A vector field $Y\in\vf(\Tan^*Q)$ is an
 \textbf{infinitesimal natural Noether symmetry} if:
\ben
\item
There exists a vector field $Z\in\vf (Q)$ such that $Y=Z^*$.
\item
$\Lie(Y){\rm h}=\Lie(Z^*){\rm h}=0$.
\een
\label{nCsinf}
\end{definition}

Other particular cases of Noether symmetries in this formalism are:

\begin{definition}
A  Noether symmetry is \textbf{exact} if $\Phi^*\Theta=\Theta$.
\label{eCs}
\end{definition}

\begin{definition}
An infinitesimal Noether symmetry is \textbf{exact} if $\Lie(Y)\Theta=0$.
\label{eCsinf}
\end{definition}

For  infinitesimal exact Noether symmetries,
we have that their local Hamiltonian functions can be expressed as
$f_Y=\inn(Y)\Theta$ (see Proposition \ref{structure}) .

Obviously, every (infinitesimal) natural  Noether symmetry  is a
(infinitesimal) natural  dynamical  symmetry.
Moreover, as every canonical lift preserves the canonical forms
$\Theta\in\df^1(\Tan^*Q)$ and $\Omega\in\df^2(\Tan^*Q)$
(Propositions \ref{levdif} and \ref{inf}), we have that:

\begin{prop}
Every (infinitesimal) natural Noether symmetry is an
(infinitesimal) exact  Noether symmetry.
\end{prop}

At this point, it is possible to state the Noether Theorem as in Theorem \ref{Nth}.
In particular, for infinitesimal exact Noether symmetries,
the associated conserved quantities are 
$f_Y=\inn(Y)\Theta$ (up to constants).

Summarizing, the following table recovers the relations among
the several types
of symmetries of the canonical Hamiltonian systems:
$$
\begin{array}{ccc}
\left\{ \begin{array}{c}
\mbox{\rm Natural Noether symmetries}
\end{array}\right\}
& \subset &
\left\{ \begin{array}{c}
\mbox{\rm Natural dynamical symmetries}
\end{array}\right\}
\\ \cap & &  \\
\left\{ \begin{array}{c}
\mbox{\rm Exact Noether symmetries}
\end{array}\right\}
&  & \cap
\\ \cap &  &  \\
\{ \mbox{\rm Noether symmetries}\}
& \subset &
\{ \mbox{\rm Dynamical symmetries} \}
\end{array}
$$

\subsection{Lagrangian formalism: Lagrangian symmetries and Noether's Theorem}
\label{lfLsNt}

Let $(\Tan Q,\Omega_\Lag,E_\Lag)$ be a regular Lagrangian systems
and $X_\Lag\in\vf(\Tan Q)$ the Euler--Lagrange vector field solution of the system.

Also in this situation all the concepts and results about symmetries
established in Section \ref{secsym} are true and,
in this way, we can also introduce the concepts of
{\sl  (infinitesimal) Lagrangian dynamical symmetry},
{\sl  (infinitesimal) Lagrangian Noether symmetry}
and  {\sl (infinitesimal) exact  Lagrangian Noether symmetry},
and their properties and relations, including Noether's Theorem.
But, in addition, the study of symmetries en the Lagrangian formalism
presents some nuances that should be highlighted.

First, we define:

\begin{definition}
A diffeomorphism $\Phi\colon\Tan Q\to \Tan Q$ is a
\textbf{natural Lagrangian dynamical symmetry} if:
\ben
\item
There exists a diffeomorphism $\varphi\colon Q\to Q$ such that $\Phi=\Tan\varphi$.
\item
$\Phi_*X_\Lag=X_\Lag$.
\een
\label{dls}
\end{definition}

\begin{definition}
A vector field $Y\in\vf(\Tan Q)$ is an
\textbf{ infinitesimal natural Lagrangian dynamical symmetry} if
\ben
\item
There exists $Z\in\vf (Q)$ such that $Y=Z^{C}$.
 \item
$[Y,X_\Lag]=[Z^{C},X_\Lag]=0$
\footnote{
Or, more generically,
$[Y,X_\Lag]=[Z^{C},X_\Lag]=gX_\Lag$, for some function $g\in\Cinfty(\Tan Q)$).}.
\een
\label{dlsinf}
\end{definition}

As in the canonical Hamiltonian formalism,
the canonical lifts of diffeomorphisms and vector fields
preserve  the canonical structures of $\Tan Q$ (Proposition \ref{lema2}).
A consequence of this is the following:

\begin{prop}
\label{lema3}
\ben
\item
 Let $\varphi\colon Q\to Q$ be a diffeomorphism and
$\Phi=\Tan\varphi$ its canonical lift to $\Tan Q$. Then
$$
\Phi^*\Theta_\Lag=\Theta_{\Phi^*\Lag} \quad ,\quad
\Phi^*\Omega_\Lag=\Omega_{\Phi^*\Lag} \quad ,\quad
\Phi^*E_\Lag=E_{\Phi^*\Lag} \ .
$$
\item
Let $Z\in\vf (Q)$ and its canonical lift $Z^C$ to $\Tan Q$. Then
$$
\Lie(Z^C)\Theta_\Lag=0 \quad , \quad
\Lie(Z^C)\Omega_\Lag=0 \quad , \quad
\Lie(Z^C)E_\Lag=0 \ .
$$
\een
\end{prop}
\begin{proof}
It is a straightforward consequence of Proposition \ref{lema2} and the
definitions of  $\Theta_\Lag$, $\Omega_\Lag$, and $E_\Lag$.
In fact:
\ben
\item
For $\Phi=\Tan\varphi$, we obtain
 \beann
 \Phi^*\Theta_\Lag&=& \Phi^*(\d \Lag\circ J)=\d (\Phi^*\Lag)\circ J)=\Theta_{\Phi^*\Lag} \ ,
\\
\Phi^*\Omega_\Lag&=&\Phi^*(-\d\Theta_\Lag)=-\d\Phi^*\Theta_\Lag=\Omega_{\Phi^*\Lag} \ ,
 \\
\Phi^*E_\Lag&=&\Phi^*(\Delta(\Lag)-\Lag)=\Delta(\Phi^*\Lag)-\Phi^*\Lag=E_{\Phi^*\Lag} \ .
\eeann
\item
They are proved using the uniparametric groups of diffeomorphisms generated by
 the fluxes of $Z$ and $Z^C$, and the above item.
 \een
 \qed \end{proof}

Nevertheless, the Lagrangian forms $\Theta_\Lag$ and $\Omega_\Lag$
are not canonical structures of $\Tan Q$,
since they depend on the choice of a Lagrangian function $\Lag$
and hence they are not  invariant by these lifts, necessarily.
Thus,  for Lagrangian Noether symmetries, we can state the following definitions:

\begin{definition}
A diffeomorphism $\Phi\colon\Tan Q\to \Tan Q$ is a
\textbf{natural Lagrangian Noether symmetry}
if there exist a diffeomorphism $\varphi\colon Q\to Q$ such that
$\Phi=\Tan\varphi$ and it satisfies:
\ben
\item
$\Phi^*\Omega_\Lag=(\Tan\varphi)^*\Omega_\Lag=\Omega_\Lag$.
\item
$\Phi^*E_\Lag=(\Tan\varphi)^*E_\Lag=E_\Lag+c$ ($c\in\Real$)
\footnote{
It is usual to write simply that $\Phi^*E_\Lag =E_\Lag$.
}.
\een
\label{nls}
\end{definition}

\begin{definition}
A vector field $Y\in\vf(\Tan Q)$ is an
\textbf{infinitesimal natural Lagrangian Noether symmetry}
if there exists $Z\in\vf (Q)$ such that $Y=Z^{C}$ and it satisfies:
\ben
\item
 $\Lie(Y)\Omega_\Lag=\Lie(Z^{C})\Omega_\Lag=0$.
\item
$\Lie(Y)E_\Lag=\Lie(Z^{C})E_\Lag=0$.
\een
\label{nlsinf}
\end{definition}

Obviously, every (infinitesimal) natural Lagrangian Noether symmetry
is a (infinitesimal) natural Lagrangian dynamical symmetry.

Finally, as a particular case, we have:

\begin{definition}
A  Lagrangian Noether symmetry is \textbf{exact} if $\Phi^*\Theta_\Lag=\Theta_\Lag$.
\label{eCls}
\end{definition}

\begin{definition}
An infinitesimal  Lagrangian Noether symmetry is \textbf{exact}
if $\Lie(Y)\Theta_\Lag=0$.
\label{eClsinf}
\end{definition}

In these circumstances, it is possible to state the Lagrangian geometric version 
of the Noether Theorem.
First, observe that, if $Y\in\vf(\Tan Q)$ is a
infinitesimal natural Lagrangian Noether symmetry, then
Proposition \ref{structure} holds for these kinds of symmetries.
Thus, for every ${\rm p}\in \Tan Q$, there exists an open set $U_{\rm p}\ni{\rm p}$ and
$f_Y\in\Cinfty(U_{\rm p})$, which is unique up to the sum constant functions, such that
\beq
\inn(Z^C)\Omega_\Lag=\d f_Y \qquad \mbox{\rm (in $U_{\rm p}$)} \ .
\label{fA}
\eeq
Furthermore, there exists $\zeta_Y\in\Cinfty(U_{\rm p})$, defined as
 $\Lie(Z^C)\Theta_\Lag=\d\zeta_Y$, in $U_{\rm p}$, and such that
\bea
f_Y&=&\zeta_Y-\inn(Z^C)\Theta_\Lag=\zeta_Y-\Theta_\Lag(Z^C)=\zeta_Y-\d\Lag\circ J(Z^C)
\nonumber \\ &=&
\zeta_Y-\d\Lag(Z^V)=\zeta_Y-\inn(Z^V)\d\Lag=\zeta_Y-Z^V(\Lag) \ ,
\label{CV}
\eea
(up to the sum of constant functions in $U_{\rm p}$). Then:

\begin{teor}
 {\rm (Lagrangian Noether):}
If $Y=Z^C\in\vf (\Tan Q)$ (with $Z\in\vf(Q)$)
is an infinitesimal natural Lagrangian Noether symmetry,
then $f_Y=\zeta_Y-Z^V(\Lag)$
is a conserved quantity; that is, $\Lie (X_\Lag)f_Y=0$.
 \label{NthL}
\end{teor}
\begin{proof}
The proof is the same as in Theorem \ref{Nth}, taking into account (\ref{fA}) and (\ref{CV}).
\\ \qed \end{proof}

If the infinitesimal  Noether symmetry is exact, then $Y=Z^C$ and
we can take $f_Y=\inn(Y)\Theta_\Lag=Z^V(\Lag)$.

\subsection{Equivalent Lagrangians and Noether's Theorem}
\protect\label{flsltn}

It is evident that, if $\Phi\in{\rm Diff}(\Tan Q)$ (resp. $Y\in\vf (\Tan Q)$)
is a canonical lift of a diffeomorphism (resp. of a
vector field) of $Q$ to $\Tan Q$ which, in addition,
lets the Lagrangian function of the system invariant,
then the symplectic form $\Omega_{\Lag}$, the
Lagrangian energy $E_{\Lag}$ and hence the Euler--Lagrange vector field
$X_\Lag$ (that is,  the Euler--Lagrange equations) are also invariant by $\Phi$.
All of this assures that the conditions of Definitions \ref{nls} and \ref{nlsinf} hold.
Nevertheless, this requirement is too strong
because there are Lagrangian functions that, being different,
give the same form $\Omega_{\Lag}$ and the same
Euler--Lagrange equations. This leads to the following:

\begin{definition}
Two Lagrangian functions $\Lag_1,\Lag_2\in\Cinfty (\Tan Q)$
are \textbf{equivalent} if
$$
\Omega_{\Lag_1}=\Omega_{\Lag_2}
\quad {\rm and} \quad
 X_{\Lag_1}= X_{\Lag_2} \ .
$$
\end{definition}
For regular Lagrangians this definition is equivalent to the following:

\begin{prop}
Two regular Lagrangians $\Lag_1,\Lag_2\in\Cinfty (\Tan Q)$
are equivalent if
$$
\Omega_{\Lag_1}=\Omega_{\Lag_2}
\quad {\rm and} \quad
 E_{\Lag_1}= E_{\Lag_2}+c \ \mbox{\rm( $c\in\Real$)} \ .
$$
\label{gaugecarac}
\end{prop}
\begin{proof}
We must prove that, if $\Omega_{\Lag_1}=\Omega_{\Lag_2}$,
then $X_{\Lag_1}= X_{\Lag_2}$ is equivalent to
$E_{\Lag_1}= E_{\Lag_2}+c$.

If $X_{\Lag_1}= X_{\Lag_2}$, then %$J(X_{\Lag_1})=J(X_{\Lag_2})=\Delta$, and
$$
0=\inn(X_{\Lag_1})\Omega_{\Lag_1}-\d E_{\Lag_1}=\inn(X_{\Lag_2})\Omega_{\Lag_2}-\d E_{\Lag_1} \ ,
$$
which implies that $\d E_{\Lag_1}= \d E_{\Lag_2}$ and, hence, $E_{\Lag_1}= E_{\Lag_2}+c$.

Conversely, if $E_{\Lag_1}= E_{\Lag_2}+c$, then
$$
\inn(X_{\Lag_1})\Omega_{\Lag_1}=\d E_{\Lag_1}=\d E_{\Lag_2}=\inn(X_{\Lag_2})\Omega_{\Lag_2} \ ,
$$
and, as $\Omega_{\Lag_1}=\Omega_{\Lag_2}$,
this implies that $X_{\Lag_1}= X_{\Lag_2}$, since
${\Lag_1}$ and ${\Lag_2}$ are regular Lagrangians and the solution is unique.
\\ \qed \end{proof}

Next, we specify how equivalent Lagrangians are (see also \cite{AM-78}).
First we find the form of Lagrangian functions leading to vanishing Cartan forms:

\begin{prop}
\label{omega=0}
A Lagrangian $\Lag\in\Cinfty(\Tan Q)$ satisfies that
$\Omega_\Lag=0$ if, and only if, there exists a closed $1$-form
$\alpha\in\df^1(Q)$ in $Q$ and a  function
$f\in \mathcal{C}^\infty(Q)$, such that 
$\Lag=\widehat{\alpha}+\tau_Q^*f$ (up to a constant), 
where $\widehat\alpha\in\Cinfty (\Tan Q)$ is the function defined by
$$
\begin{array}{ccccc}
\widehat\alpha&\colon&\Tan Q&\longrightarrow&\Real
\\ & & (q,v) & \mapsto &\alpha_q(v)
\end{array}\ .
$$
\end{prop}
\begin{proof}
Let $\Omega_\Lag=-\d\Theta_\Lag$; then $\Theta_\Lag=\d\Lag\circ J$
is a closed and semibasic form in $\Tan Q$ and, as a consequence,
it is a basic form. Then, there exists $\alpha\in\df^1(Q)$ such that
\begin{equation}
\label{stau}
\d\Lag\circ J=\tau_Q^*\alpha \ .
\end{equation}
 Moreover, since $0=\d\Theta_\Lag=\d(\tau_Q^*\alpha)=\tau_Q^*(\d\alpha)$,
then $\d\alpha=0$; that is, $\alpha$ is a closed $1$-form in $Q$.
Furthermore, a simple computation in local coordinates shows that
 $\d\widehat\alpha\circ J=\tau_Q^*\alpha$, and from (\ref{stau}) we have that
 $$
\d\widehat\alpha\circ J=\tau_Q^*\alpha = \d\Lag\circ J \ .
$$
Then $\d(\Lag-\widehat\alpha)\circ J=0$, and therefore, the $1$-form 
$\d(\Lag-\widehat\alpha)$ is closed and semi-basic. 
As a consequence, $\d(\Lag-\widehat\alpha)$ is a basic $1$-form; 
that is, there exists $f\in\Cinfty(Q)$ such that
 $$
\d(\Lag-\widehat\alpha)=\tau_Q^*\d f=\d(\tau_Q^*f) \ ,
$$
and thus
 $\Lag=\widehat\alpha+\tau_Q^*f$ (up to a constant).

Conversely, let us suppose that $\Lag=\widehat\alpha+\tau_Q^*f$ 
(up to a constant). We have
$$
 \Theta_\Lag=\d\Lag\circ J=\d(\widehat\alpha+\tau_Q^*f)\circ J=
\d\widehat\alpha\circ J=\tau_Q^*\alpha  \ ,
$$
since $\d\tau_Q^*f$ vanishes on the vertical vector fields.
As $\alpha$ is closed, $\d\alpha=0$ and we obtain
$$
\Omega_\Lag=-\d\Theta_\Lag=-\d(\tau_Q^*\alpha)=-\tau_Q^*(\d\alpha)=0 \ .
$$
\qed \end{proof}

In a local chart of natural coordinates of $\Tan Q$,
the local expression of a closed $1$-form $\alpha$ is
\begin{displaymath}
\alpha =\alpha_i\d q^i=\derpar{g}{q^i}\d q^i \ ,
\end{displaymath}
for some local function $g$;
then the local expression of the function $\widehat\alpha$ is 
\begin{displaymath}
\widehat\alpha=\derpar{g}{q^i}v^i \ .
\end{displaymath} 

Now, from this result we obtain the explicit characterization of the equivalent Lagrangians:

\begin{prop}
\label{gaugeL}
Two regular Lagrangians $\Lag_1,\Lag_2\in\Cinfty (\Tan Q)$ are
equivalent if, and only if, $\Lag_1=\Lag_2+\widehat\alpha$  (up to a constant).
\end{prop}
\begin{proof}
Suppose that $\Lag_1,\Lag_2$ are equivalent.
As $\Omega_{\Lag_1}=\Omega_{\Lag_1}$, then
$\Omega_{\Lag_1-\Lag_2}=0$. Therefore, by Proposition \ref{omega=0}, 
there exist $\alpha\in Z^1(Q)$ and $f\in \mathcal{C}^\infty(Q)$ such that
$\Lag_1-\Lag_2=\widehat\alpha+\tau_Q^*f$ (up to a constant).
From Proposition \ref{gaugecarac} we know that 
$E_{\Lag_1}= E_{\Lag_2}$,  (up to a constant), or equivalently, 
$E_{\Lag_1}- E_{\Lag_2}=0$ (up to a constant). Then
$$
\begin{array}{lcl}
0&=&E_{\Lag_1}- E_{\Lag_2}=\Delta(\Lag_1)-\Lag_1
-\Delta(\Lag_2)+\Lag_2=\Delta(\Lag_1-\Lag_2)-(\Lag_1-\Lag_2)  \\ &=&
\Delta(\widehat\alpha+\tau_Q^*f)-(\Lag_1-\Lag_2) =
\widehat\alpha-(\Lag_1-\Lag_2)\quad \makebox{(up to a constant)}  \ .
 \end{array}
$$

Conversely, suppose that $\Lag_1=\Lag_2+\widehat\alpha$  (up to a  constant).
First, a simple computation gives
 $$
 \Omega_{\Lag_2}-\Omega_{\Lag_1}=
\d(\Theta_{\Lag_1}-\Theta_{\Lag_2})=\d(\d(\Lag_1-\Lag_2)\circ J)=
\d(\d\widehat\alpha\circ J)=\d(\tau_Q^*\alpha)=\tau_Q^*(\d\alpha)=0 \ .
$$
Thus $\Omega_{\Lag_1}=\Omega_{\Lag_2}$. Furthermore, 
as $\Delta(\widehat\alpha)=\widehat\alpha$, we have
 $$
 E_{\Lag_1}=\Delta(\Lag_1)-\Lag_1=
 \Delta(\Lag_2+\widehat\alpha)-(\Lag_2+\widehat\alpha)=
  E_{\Lag_2} + \widehat\alpha - \widehat\alpha= E_{\Lag_2}\quad
\makebox{\rm (up to a constant)} \ .
 $$
As $\Omega_{\Lag_1}=\Omega_{\Lag_2}$ and 
$ E_{\Lag_1}=E_{\Lag_2}$ (up to a constant), 
then $\Lag_1$ and $\Lag_2$ are equivalent (Proposition \ref{gaugecarac}).
\\ \qed \end{proof}

Taking this into account, one can define:

\begin{definition}
A \textbf{symmetry of the Lagrangian} is
a diffeomorphism $\Phi\colon\Tan Q\to\Tan Q$ such that
$\Lag$ and $\Phi^*\Lag$ are equivalent Lagrangians; that is,
$\Phi^*\Lag=\Lag+\widehat\alpha$ (up to constants),
where $\widehat\alpha\in\Cinfty (\Tan Q)$ is the function defined in Proposition \ref{omega=0}.
\end{definition}

\begin{definition}
An \textbf{ infinitesimal symmetry of the Lagrangian} is
a vector field $Y\in\vf (\Tan Q)$ such that  the uniparametric groups of diffeomorphisms 
generated by its flux are symmetries of the Lagrangian; that is,
$\Lie(Y)\Lag=\widehat\alpha$.
\end{definition}

A special case of these kinds of symmetries are:

\begin{definition}
A \textbf{strict symmetry of the Lagrangian} is a diffeomorphism 
$\Phi\colon\Tan Q\to\Tan Q$ such that $\Phi^*\Lag=\Lag$.
% ($c\in\Real$).
\end{definition}

\begin{definition}
An \textbf{ infinitesimal strict symmetry of the Lagrangian} is
a vector field $Y\in\vf (\Tan Q)$ such that  the uniparametric groups of diffeomorphisms 
generated by its flux are strict symmetries of the Lagrangian; that is,
$\Lie(Y)\Lag=0$.
\end{definition}

And, in particular, we define:

\begin{definition}
A (strict) symmetry of the Lagrangian  $\Phi\colon\Tan Q\to\Tan Q$ is 
said to be \textbf{natural} if there exists a diffeomorphism $\varphi\colon Q\to Q$ such that 
$\Phi=\Tan\varphi$.
\end{definition}

\begin{definition}
An infinitesimal (strict) symmetry of the Lagrangian $Y\in\vf (\Tan Q)$ is 
said to be \textbf{natural} if there
exists a vector field $Z\in\vf (Q)$ such that $Y=Z^C$.
\end{definition}

\begin{remark}{\rm 
A symmetry of the Lagrangian $\Phi\colon\Tan Q\to\Tan Q$
is not necessarily a Lagrangian Noether  symmetry since, in general, 
$\Phi^*\Omega_\Lag\not=\Omega_{\Phi^*\Lag}$ and
$\Phi^* E_\Lag\not= E_{\Phi^*\Lag}$,
as a simple calculation in coordinates shows.
In addition, it is not a Lagrangian dynamical symmetry.
Nevertheless, the following relation holds:
}\end{remark}

\begin{prop}
A diffeomorphism $\Phi\colon T^1_kQ\to T^1_kQ$ is a natural Lagrangian Noether symmetry
if, and only if, it is a natural symmetry of the Lagrangian.
\end{prop}
\begin{proof}
If $\Phi=\Tan\varphi$, for some diffeomorphism $\varphi\colon Q\to Q$,
according to Lemma (\ref{lema3}) we have that
$$
\Phi^*\Omega_\Lag=\Omega_{\Phi^*\Lag} \quad , \quad
\Phi^*E_\Lag=E_{\Phi^*\Lag} \ ,
$$
and then
$$
\left. \begin{array}{cccc}
\Phi^*\Omega_\Lag&=&\Omega_\Lag &  \\
\Phi^*E_\Lag&=&E_\Lag &  \mbox{\rm (up to constants)}
\end{array}\right\}
\ \Longleftrightarrow \
\left\{ \begin{array}{cccc}
\Omega_{\Phi^*\Lag}&=&\Omega_\Lag & \\
E_{\Phi^*\Lag}&=&E_\Lag &  \mbox{\rm (up to constants)}
\end{array}\right. \ ;
$$
that is, $\Phi$ is a natural Lagrangian Noether symmetry
if, and only if, $\Lag$ and $\Phi^*\Lag$ are equivalent Lagrangians and hence
$\Phi$ is a natural symmetry of the Lagrangian.
 \\ \qed \end{proof}

This result holds also for  infinitesimal symmetries,
as can be proved taking the flux of  the vector fields that generate them.
Thus, we have the following immediate corollary:

\begin{prop}
A vector field $Y\in\vf (\Tan Q)$ is an infinitesimal natural Lagrangian Noether symmetry
if, and only if, it is an infinitesimal natural symmetry of the Lagrangian.
\end{prop}

Finally, a version of Noether's Theorem
for infinitesimal natural strict symmetries of the Lagrangian can be established as follows:

\begin{teor}
 {\rm (Classical Noether for Lagrangian systems}).
Let  $Y=Z^C\in\vf (\Tan Q)$, with $Z\in\vf(Q)$,
be an infinitesimal natural strict symmetry of the Lagrangian.
Then $\widetilde  f=Z^V(\Lag)$ is a conserved quantity; that is, $\Lie (X_\Lag)\widetilde  f=0$.
 \label{NthLs}
\end{teor}
\begin{proof}
As every infinitesimal natural strict symmetries of the Lagrangian
is a natural symmetry of the Lagrangian and then,
it is a natural Noether Lagrangian symmetry, according to the above proposition;
then the result is a straightforward consequence of Theorem \ref{NthL},
since $\zeta_Y=\Lie(Y)\Theta_\Lag=\Lie(Z^C)\Theta_\Lag=0$.
\\ \qed \end{proof}

The following table summarizes the relation among  the different kinds
of symmetries in the Lagrangian formalism of  Lagrangian systems:
$$
\begin{array}{ccccccc}
\left\{ \begin{array}{c}
\mbox{\rm  Strict} \\ \mbox{\rm Lagrangian} \\ \mbox{\rm symmetries} 
\end{array}\right\}
& \supset &
\left\{ \begin{array}{c}
\mbox{\rm Strict} \\ \mbox{\rm  natural} \\ \mbox {\rm Lagrangian} \\ \mbox{\rm symmetries}  
\end{array}\right\}
& & & &
\\ \cap & & \cap & &  & & \\
\left\{ \begin{array}{c}
 \mbox {\rm symmetries} \\ \mbox{\rm  of the} \\ \mbox{\rm  Lagrangian}
\end{array}\right\}
& \supset &
\left\{ \begin{array}{c}
\mbox{\rm  Natural} \\ \mbox{\rm symmetries} \\ \mbox{\rm  of the} \\ \mbox{\rm Lagrangian} 
\end{array}\right\}
& = &
\left\{ \begin{array}{c}
 \mbox{\rm Natural} \\  \mbox{\rm Lagrangian} \\ \mbox{\rm Noether} \\ \mbox{\rm symmetries} \\ 

\end{array}\right\}
& \subset &
\left\{ \begin{array}{c}
\mbox{\rm  Natural} \\ \mbox{\rm Lagrangian}\\ \mbox{\rm  dynamical} \\ \mbox{\rm   symmetries} 
\end{array}\right\}
\\  & &  & & \cap & & \\
& & & &
\left\{ \begin{array}{c}
\mbox{\rm  Exact} \\ \mbox{\rm Lagrangian} \\ \mbox{\rm Noether} \\ \mbox{\rm   symmetries}
\end{array}\right\}
&  & \cap
\\  & &  & &  \cap & & 
\\  & &  & &
\left\{ \begin{array}{c}
\mbox{\rm  Lagrangian} \\ \mbox{\rm Noether} \\ \mbox{\rm  symmetries}
\end{array}\right\}
& \subset &
\left\{ \begin{array}{c}
 \mbox{\rm Lagrangian} \\ \mbox{\rm  dynamical} \\ \mbox{\rm  symmetries}
\end{array}\right\}
\end{array}
$$

\section{Variational formulation for Lagrangian systems}
\label{lagvariational}

One of the main characteristics of the Lagrangian systems is that 
they are variational; that is, their dynamical equations can be obtained 
from a variational principle, and the same happens for Hamiltonian systems
\cite{AA-78a,Bl-81,Bli,Cl-2018,NT-2006,Els,GF-63,He-68,La-70,lanczos,Le-2014,MFVMR-90,OR-83}.
Next, we present the variational formulation of Lagrangian systems.
In order to study this topic, we need to enlarge the phase space to include
the time coordinate, and so we have to consider the trivial bundle
$\Real\times\Tan Q$ (or $\Real\times\Tan^*Q$).

\subsection{Functional associated to a Lagrangian function}
\label{falf}

Let $(\Tan Q,\Lag )$ be a Lagrangian dynamical system.
We can associate to the Lagrangian function a functional defined on a suitable space of curves 
which are taken to be possible trajectories of the system. 

To get it,  consider the manifolds
$\Real\times\Tan Q$, $\Real\times Q$, and $\Real$
with the canonical projections
$$
\rho\colon\Real\times Q\to \Real  \ , \
\tau_{(1,0)}=id\times\tau_Q\colon\Real\times\Tan Q \to\Real\times Q \ , \
\tau_1=\rho\circ\tau_{(1,0)}\colon\Real\times\Tan Q \to \Real \ .
$$
The canonical projections $\rho$ and $\tau_1$ define a global coordinate in $\Real$
which is denoted $t$, and is physically identified with ``time''.
Then we can take the natural volume element $\d t$ in $\Real$ and its natural lift to
$\Real\times\Tan Q$, which is also denoted by $\d t$.

Now, let $\gamma \colon [a,b]\subset\Real \to Q$ be a curve, 
and let $\widetilde\gamma \colon [a,b]\subset\Real \to \Tan Q$ be
its canonical lift of $\gamma$ from $Q$ to the tangent bundle $\Tan Q$.
We denote $\mbox{\boldmath$\gamma$}\colon[a,b]\subset\Real \to\Real\times Q$ 
the curve given by $\mbox{\boldmath$\gamma$}(t)=(t,\gamma(t))$,
and by $\widehat{\mbox{\boldmath$\gamma$}} \colon [a,b]\subset\Real  \to\Real\times \Tan Q$ 
the curve $\widehat{\mbox{\boldmath$\gamma$}}(t)=(t,\widetilde\gamma(t))$.

Now, let $\alpha$ be a $1$-form in $\Real\times\Tan Q$.
As $\widehat{\mbox{\boldmath$\gamma$}}^*\alpha$ is a $1$-form in $\Real$,
we can define
$$
\int_{\widehat{\mbox{\boldmath$\gamma$}}} \alpha = \int_a^b\widehat{\mbox{\boldmath$\gamma$}}^*\alpha \ .
$$

Finally, given a (time-dependent) Lagrangian function 
$\Lag\colon\Real\times\Tan Q\to\Real$, consider the $1$-form
$\Lag\,\d t \in {\mit\Omega}^1(\Real\times\Tan Q)$.
Then we define

\begin{definition}
For every curve $\gamma \colon \Real \to Q$, the
\textbf{action} of $\Lag$ along $\gamma$ is the functional
$$
{\mathbf L}(\gamma ) := \int_{\widehat{\mbox{\boldmath$\gamma$}}} \Lag\d t \ .
$$
\end{definition}

\begin{remark}{\rm 
Usually we take curves 
$\gamma \colon \Real \to Q$ defined on a closed  interval
$I=[a,b]\subset\Real$, but we can extend the domain to the whole real numbers defining $\gamma$ out of the interval  as a constant function taking permanently the values on the two extreme points of the interval. Observe that, in this case, the extension may not be differentiable at the extreme points, but this is not a problem for the existence of the integral
\footnote{
From now on it is only necessary that the curve is of class $C^2$.}.

If the Lagrangian function is time-independent, 
it is enough to extend it from $\Tan Q$ to $\Real\times\Tan Q$.
}\end{remark}

\subsection{Hamilton variational principle}

The {\sl variational problem} consists in optimizing the above functional;
that is, to choose the curves for which the functional takes the extremal value. 
As this is a very complicated problem we restrict ourselves to look for the local extremal points of the functional; 
so we look for curves $\gamma$ such that the value of 
$L(\gamma)$ does not change for small variations of the curve,
but in first approximation. Next we are going to precise these ideas.

\begin{definition}
Let $\gamma \colon [a,b]\subset\Real \to Q$ be a curve with
$\gamma (a)=q_0$, $\gamma (b)=q_1$.
A \textbf{variation} of $\gamma$ is a map
$$
\mu \colon (-\epsilon ,\epsilon )\times [a,b]\subset\Real^2 \to Q
$$
such that
\begin{enumerate}
\item
$\mu (0,t)=\gamma (t)$, $t\in [a,b]$.
\item
$\mu (s,a)=q_0$, $\mu (s,b)=q_1$, for $s\in (-\epsilon ,\epsilon )$.
\end{enumerate}
We denote $\mu_s$ the curve defined by
$\mu_s \colon [a,b]\subset\Real \to Q$, with $\mu_s(t)=\mu (s,t)$.
\end{definition}

Observe that, given one of these curves  $\mu_s$, we can lift it to $\Tan Q$ and $\Real\times\Tan Q$ in the usual way.
We denote these lifts as $\widetilde \mu_s$ and $\widehat{\mbox{\boldmath$\mu$}}_s$, which are variations of $\widetilde\gamma$ and $\widehat{\mbox{\boldmath$\gamma$}}$ respectively.

There is a natural way to construct different variations of that type.
In fact, given $X \in \vf(Q)$ with
$X(q_0)=0$, $X(q_1)=0$; if $F_s$ is a local flux for $X$,
a variation of the curve $\gamma$ is obtained by
\beq
\label{variation}
\mu (s,t)=(F_s\circ\gamma)(t)=F_s(\gamma (t)) \ .
\eeq
We know that $F_s \colon Q \to Q$
is a diffeomorphism, and the corresponding lift
$\Tan F_s \colon \Tan Q \to \Tan Q$ is the local flux corresponding to the canonical lift $X^C\in \vf(\Tan Q)$ of $X$ to $\Tan Q$.
In the same way that \eqref{variation} is a variation of $\gamma$, 
we have that 
$$
\widetilde{\mu} (s,t)=(\Tan F_s\circ\widetilde\gamma)(t)=\Tan F_s(\widetilde\gamma (t))
$$
is a variation of $\widetilde\gamma$. 
Observe that $X^C$ is naturally a vector field in $\Real\times\Tan Q$ acting as the identity on the component $\Real$.
It is vertical with respect to the projection $\tau_1\colon\Real\times\Tan Q\to\Real$.

Given $\gamma$ and $F_s$ as above, the following properties hold: 
\begin{enumerate}
\item $\mu_s(t)=(F_s\circ\gamma)(t)=F_s(\gamma(t))$
\item $\widetilde{\mu}_s (t)=(\Tan F_s\circ\widetilde{\gamma})(t)=\Tan F_s(\widetilde{\gamma}(t))$
\item ${\mbox{\boldmath$\mu$}}_s(t)=(F_s\circ\mbox{\boldmath$\gamma$})(t)=F_s(\mbox{\boldmath$\gamma$})(t))=(t,F_s(\gamma(t))$
\item $\widehat{\mbox{\boldmath$\mu$}}_s(t)=(\Tan F_s\circ\widehat{\mbox{\boldmath$\gamma$}})(t)=\Tan F_s(\widehat{\mbox{\boldmath$\gamma$}}(t))=(t,\Tan F_s(\widetilde{\gamma}(t))$
\end{enumerate}

These kinds of variations obtained from specific vector fields on the manifold $Q$ are sufficient to obtain the local extreme of the functional ${\mathbf L}(\gamma)$ as we are going to see.

\begin{definition}
A curve $\gamma$ is a \textbf{local extreme} for the functional ${\mathbf L}$
if for every variation $\mu_s= F_s\circ\gamma$, with $F_s$ corresponding to a vector field $X$ as above, the following condition holds:
$$
\frac{d}{d s}\Big\vert_{s=0}\int_{\widehat{\mbox{\boldmath$\mu$}}_s}\Lag\d t =0 \ .
$$
This is the so-called \textbf{Hamilton variational principle}.
\end{definition}

To obtain a more useful way to manage with this definition, first we need the following results:

\begin{prop}
Let  $X\in\vf(Q)$ and $\mu_s$ be the corresponding variation of a curve $\gamma$. Then
\begin{enumerate}
\item
\(\dst \frac{d}{d s}\Big\vert_{s=0}\int_{\widehat{\mbox{\boldmath$\mu$}}_s}\Lag \d t=
\int_{\widehat{\mbox{\boldmath$\gamma$}}}(\Lie (X^C)\Lag)\,\d t\).
%=\int_{\widetilde  \gamma}\widetilde  X(\Lag )\d t\)~.
\item
If $\alpha \in {\mit\Omega}^1(\Tan Q)$, then \
$\displaystyle\frac{d}{d s}\Big\vert_{s=0}\int_{\widetilde{\mu}_s}\alpha =
\int_{\widetilde \gamma}\Lie (X^C)\alpha$.
\end{enumerate}
\end{prop}
\begin{proof}
Taking into account that
$\widetilde  \mu_s = \widetilde  F_s\circ\widetilde  \gamma$ and $\widehat{\mbox{\boldmath$\mu$}}_s=(\Tan F_s\circ\widehat{\mbox{\boldmath$\gamma$}})$,
we have:
\begin{enumerate}
\item
First,
\begin{eqnarray*}
\frac{d}{d s}\Big\vert_{s=0}\int_{\widehat{\mbox{\boldmath$\mu$}}_s}\Lag \d t &=&
\lim_{s\to 0}\frac{1}{s}\left(\int_{\widehat{\mbox{\boldmath$\mu$}}_s}\Lag \d t-
\int_{\widehat{\mbox{\boldmath$\mu$}}_0}\Lag \d t\right)
\\ &=&
\lim_{s\to 0}\frac{1}{s}\left(\int_a^b\widetilde  {\widehat\mu}_s^*(\Lag\,\d t)-
\int_a^b\widetilde  {\widehat\mu}_0^*(\Lag\,\d t)\right)
\\ &=&
\int_a^b\lim_{s\to 0}\frac{\widetilde  {\widehat\mu}_s^*\Lag -\widetilde  {\widehat\mu}_0^*\Lag}{s}\d t
=\int_a^b\widehat{\mbox{\boldmath$\gamma$}}^{*}\lim_{s\to 0}\frac{\Tan  F_s^*\Lag -\Tan  F_0^*\Lag}{s}\d t
\\&=&
\int_a^b\widehat{\mbox{\boldmath$\gamma$}}^{*}(\Lie ( X^C)\Lag) \d t =
%\int_a^b \widetilde  X(\Lag )\d t =
\int_{\widehat{\mbox{\boldmath$\gamma$}}} (\Lie(X^C)\Lag )\d t \ ,
\end{eqnarray*}
because 
$\widetilde  \mu_s^{*}\alpha = (\widetilde  \gamma^{*}\circ\Tan  F_s^{*})\alpha=\widetilde  \gamma^{*}(\Tan  F_s^{*}\alpha)$.
\item
Second,
\begin{eqnarray*}
\frac{d}{d s}\Big\vert_{s=0}\int_{\widetilde  \mu_s}\alpha &=&
\lim_{s\to 0}\frac{1}{s}\left(\int_{\widetilde  \mu_s}\alpha -
\int_{\widetilde  \mu_0}\alpha\right)
=\lim_{s\to 0}\frac{1}{s}\left(\int_a^b\widetilde  \mu_s^*\alpha -
\int_a^b\widetilde  \mu_0\alpha\right)
\\ &=&
\int_a^b\lim_{s\to 0}\frac{\widetilde  \mu_s^*\alpha -\widetilde  \mu_0^*\alpha}{s}
=
\int_a^b\widetilde\gamma^*\lim_{s\to 0}\frac{\Tan\widetilde  F_s^*\alpha -\Tan\widetilde  F_0^*\alpha}{s}
\\&=&
\int_a^b\widetilde\gamma^*\Lie (X^C)\alpha=
%\int_a^b\Lie (X^C)\alpha  =
\int_{\widetilde  \gamma} \Lie(X^C)\alpha \ .
\end{eqnarray*}
\end{enumerate}
\qed \end{proof}

Hence, the problem is to seek for the equation fulfilled by the curves 
$\gamma \colon [a,b]\subset\Real \to Q$
satisfying that $\gamma (a)=q_0$, $\gamma (b)=q_1$,
and such that
$$
\frac{d}{d s}\Big\vert_{s=0}\int_{\widetilde  \mu_s}\Lag \d t=
\int_{\widetilde  \gamma}\Lie (X^C)\d t = 0 \ ,
$$
for every variation $\mu_s=F_s\circ\gamma$ of $\gamma$.
Then, we need the following:

\begin{lem}
For every curve $\gamma \colon [a,b]\subset\Real \to Q$ we have
$$
\widetilde{\mbox{\boldmath$\gamma$}}^*(A_\Lag\,\d t)= 
\widetilde{\mbox{\boldmath$\gamma$}}^*(\Lag\,\d t) + \widetilde{\mbox{\boldmath$\gamma$}}^*(E_\Lag\,\d t)=
\widetilde{\gamma}^*\Theta_\Lag
\ .
$$
\end{lem}
\begin{proof}
If $t_0\in [a,b]$, let $(q^i,v^i)$ be
a canonical local coordinate system in a neighbourhood of
$\widetilde{\gamma} (t_0)$. 
On the one hand,
\begin{eqnarray*}
 \big(\widetilde{\gamma}^*(A_\Lag\,\d t)\big)_{t_0}\left(\frac{d}{d t}\Big\vert_{t_0}\right)
 &=&(\widetilde{\gamma}^*A_\Lag)(t_0)
=A_\Lag(\widetilde  \gamma (t_0))
\\ &=&
\left(v^i\derpar{\Lag}{v^i}\right)(\widetilde  \gamma (t_0)) =
\left(v^i\derpar{\Lag}{v^i}\right)(\gamma (t_0),\dot\gamma (t_0))=
\dot\gamma (t_0)\derpar{\Lag}{v^i}\Big\vert_{\widetilde  \gamma (t_0)} \ ,
\end{eqnarray*}
and, on the other hand,
\begin{eqnarray*}
(\widetilde{\gamma}^*\Theta_\Lag)_{t_0}\left(\frac{d}{d t}\Big\vert_{t_0}\right)
&=&
(\Theta_\Lag)_{t_0}\left(\Tan_{t_0}\widetilde\gamma\left(\frac{d}{d t}\Big\vert_{t_0}\right)\right)=
(\d\Lag\circ J)_{\widetilde\gamma (t_0)}
\left(\Tan_{t_0}\widetilde{\gamma}\left(\frac{d}{d t}\Big\vert_{t_0}\right)\right)
\\ &=&
(\d\Lag\circ J)_{\widetilde\gamma (t_0)}
\left(\dot\gamma^i(t_0)\derpar{}{q^i}\Big\vert_{\widetilde  \gamma (t_0)}
+\ddot\gamma^i(t_0)\derpar{}{v^i}\Big\vert_{\widetilde  \gamma (t_0)}\right)
\\ &=&
(\d\Lag)_{\widetilde\gamma (t_0)}
\left(\dot\gamma^i(t_0)\derpar{}{v^i}\Big\vert_{\widetilde\gamma (t_0)}\right)
=\dot\gamma^i(t_0)\derpar{}{v^i}\Big\vert_{\widetilde\gamma (t_0)} \ .
\end{eqnarray*}
and the result follows.
\\ \qed \end{proof}

\subsection{Euler--Lagrange equations}

Now we can obtain the equation for the curve $\gamma$ solution to the variational problem. 

\begin{teor}
The curves $\gamma$ which are local extremes for the functional ${\mathbf L}$
are the solutions to the Euler--Lagrange equations.
\label{hamprinc}
\end{teor}
\begin{proof}
We have that
\begin{eqnarray*}
0 &=&
\frac{d}{d s}\Big\vert_{s=0}\int_{\widehat{\mbox{\boldmath$\mu$}}_s}(\Lag \d t )
=
\frac{d}{d s}\Big\vert_{s=0}\int_a^b\widetilde  \mu_s^*\Lag \d t
=
\frac{d}{d s}\Big\vert_{s=0}
\int_a^b\widetilde{\mu}_s^*\Theta_\Lag -\widehat{\mbox{\boldmath$\mu$}}_s^*E_\Lag\d t
\\ 
&=&
\frac{d}{d s}\Big\vert_{s=0}\int_{\widehat{\mbox{\boldmath$\mu$}}_s}\Theta_\Lag -E_\Lag\d t =
\int_{\widetilde{\gamma}}\Lie (X^C)\Theta_\Lag -\Lie(X^C)(E_\Lag )\,\d t
\\
&=&
\int_{\widehat{\mbox{\boldmath$\gamma$}}}\d\inn (X^C)\Theta_\Lag +
\inn (X^C)\d\Theta_\Lag -\Lie(X^C)(E_\Lag )\d t \ .
\end{eqnarray*}
But, by Stokes Theorem, being
$X(q_0)=0=X(q_1)$ and $\Theta_\Lag$ a semibasic form, we have that 
\(\dst \int_{\widetilde  \gamma}\d\inn (\widetilde  X)\d\Theta_\Lag =0 \); hence
\begin{eqnarray*}
0 &=&
\int_{\widetilde{\gamma}}\inn (X^C)\d\Theta_\Lag -\Lie(X^C)(E_\Lag )\d t
=\int_a^b\widetilde  \gamma^*(\inn (X^C)\d\Theta_\Lag ) -
\widetilde  \gamma^*(\Lie(X^C)(E_\Lag ))\d t
\\ &=&
\int_a^b\left[\d\Theta_\Lag (X^C,(\Tan_t\widetilde  \gamma)\frac{d}{d t})-
(\Lie(X^C)(E_\Lag ))(\widetilde\gamma (t)\right]\d t \ ,
\end{eqnarray*}
because
$$
\left(\widetilde  \gamma^*(\inn (X^C)\d\Theta_\Lag \right)\frac{d}{d t}\Big\vert_t
=(\inn (X^C)\d\Theta_\Lag )\left(\Tan_t\widetilde  \gamma\frac{d}{d t}\right)
=(d\Theta_\Lag )_{\widetilde\gamma (t)}
\left(X^C,\Tan_t\gamma(\frac{d}{d t})\right) \ ,
$$
and we have
$$
0=\int_a^b\inn (X^C_{\widetilde\gamma (t)}
\left[ -\inn \left(\Tan_t\widetilde  \gamma\frac{d}{d t}\right)\d\Theta_\Lag -
\d E_\Lag\right]\d t \ .
$$
But $X^C$ is an arbitrary vector field; then, observing that
\(\dst
\Tan_t\widetilde\gamma\left(\frac{d}{d t}\right)=(\widetilde  \gamma (t),\dot{\widetilde  \gamma} (t))\),
we have
$$
-\inn (\widetilde  \gamma (t),\dot{\widetilde  \gamma} (t))\d\Theta_\Lag = \d E_\Lag \ ,
$$
and, as $-\d\Theta_\Lag =\Omega_\Lag$, we obtain that
$$
\inn (\widetilde  \gamma (t),\dot{\widetilde  \gamma} (t))\Omega_\Lag = \d E_\Lag \ .
$$
(Recall that 
$(\widetilde  \gamma (t),\dot{\widetilde  \gamma} (t))$
is the tangent vector to the curve $\widetilde  \gamma$ at the point $\gamma(t)$).

If we consider that the curves  $\widetilde  \gamma (t)$
which are the solutions to the problem are the integral curves of a vector field
$X_\Lag\in\vf(\Tan Q)$, observing that
the curve $\widetilde\gamma$ is a natural lift
to $\Tan Q$ of the curve $\gamma$ in $Q$, then this vector field must verify the following conditions:
\begin{enumerate}
\item
$\inn (X_\Lag )\Omega_\Lag =\d E_\Lag$.
\item
$X_\Lag$ is a {\sc sode} ($J(X_\Lag )=\Delta$).
\end{enumerate}
and, as we know, these are the geometrical expression of the Euler--Lagrange equations.
\\ \qed \end{proof}

As a final remark, observe that we have said that the vector field $X^C$ is arbitrary, and this is not exact. The real arbitrary vector field is $X$ instead of $X^C$. But the 1-form
\(\dst\inn\left(\Tan\widetilde\gamma\Big(\frac{d}{d t}\Big)\right)\d\Theta_\Lag +
\d E_\Lag\)
is semibasic and, hence, when we apply it to $X^C$ the obtained value depends only on $X$ and not on its canonical lift
 $X^C$ from $Q$ to $\Tan Q$.

\subsection{Hamilton--Jacobi variational principle. Relation with the canonical Hamiltonian formalism}
\label{hamvariational}

When we studied the Hamiltonian formalism in the previous Sections,
we commented about the possibility of obtaining the dynamical equations from a variational principle. 
Following the discussion on the equivalence between the Lagrangian and the Hamiltonian formalisms,
we can establish the relation between the corresponding variational formalisms.

We have studied above the variational approach to the Lagrangian formalism, then using the corresponding geometrical and dynamical elements we can relate it with the Hamiltonian one.
Thus, let $\Leg \colon \Tan Q \to \Tan^*Q$ be the Legendre map associated to the Lagrangian $\Lag$.
If $\gamma \colon [a,b]\subset\Real \to Q$ is a curve we have
\begin{eqnarray*}
{\mathbf L}(\gamma ) &=&
\int_{\widehat{\mbox{\boldmath$\gamma$}}}\Lag\d t=
\int_{\widehat{\mbox{\boldmath$\gamma$}}}\Theta_\Lag -E_\Lag\,\d t =
\int_{\widehat{\mbox{\boldmath$\gamma$}}}\Leg^*\Theta -\Leg^*{\rm h}\,\d t 
\\ &=&
\int_{\widehat{\mbox{\boldmath$\gamma$}}}\Leg^*(\Theta -{\rm h}\,\d t) =
\int_{\Leg\circ\widehat{\mbox{\boldmath$\gamma$}}}\Theta - {\rm h}\,\d t \ ;
\end{eqnarray*}
where we have extended the Legendre transformation $\Leg \colon \Tan Q \to \Tan^*Q$ (with the same notation) to
$\Leg \colon\Real\times \Tan Q \to \Real\times\Tan^*Q$,
 as the identity on $\Real$
(see Definition \ref{legmap}).
Note that, if $\gamma$ is a solution to the Hamilton variational problem, that is of the Lagrangian formalism, then 
$\Leg\circ\widehat{\mbox{\boldmath$\gamma$}}$ is the solution to the following variational problem: 

\begin{definition}
The \textbf{Hamilton-Jacobi variational principle} consists in
finding the curves $\zeta \colon [a,b]\subset\Real \to \Tan^*Q$,
with fixed extremes, such that they are extremal for the functional
$$
{\mathbf H}(\zeta ):=\int_\zeta \Theta -{\rm h}\d t \ .
$$
\end{definition}

By a similar computation as above, 
although something simpler because there are no canonical lifts,
we obtain that 

\begin{teor}
The curves $\zeta$ which are local extremes for the functional ${\mathbf H}$
are the solutions to the Hamilton equations.
\end{teor}
\begin{proof}
In fact, let $\zeta\colon [a,b]\subset\Real\to \Tan^*Q$ a curve with 
$\zeta(a)=A,\zeta(b)=B$, and $Z\in\vf(\Tan^{*} Q)$ with $Z(A)=0,Z(B)=0$. 
Let $F_{s}$ the flux of $Z$ and consider the variation $\eta_{s}=F_{s}\circ\zeta$.
We have
$$
{\mathbf H}(F_{s}\circ\zeta)=
\int_a^b(F_{s}\circ\zeta)^*\Theta-({\rm h}\circ F_{s}\circ\zeta)\d t=
\int_a^b\zeta^*(F_{s}^*\Theta)-\zeta^*(F_{s}^*{\rm h})\d t\ .
$$
Hence:
$$
\frac{d}{d s}\Big\vert_{s=0}{\mathbf H}(F_{s}\circ\zeta)=
\int_a^b\zeta^*(\Lie(Z)\Theta)-\zeta^*(\Lie(Z){\rm h})\d t=
\int_\zeta \Lie(Z)\Theta-(\Lie(Z){\rm h})\d t\ .
$$
And following the same lines as in Theorem \ref{hamprinc},
we conclude that the tangent vector to the curve $\zeta$ satisfies the Hamilton equations
$$
\inn(\widetilde\zeta)\Omega=\d {\rm h}(\zeta(t))\, ,
$$
with $\Omega = -\d\Theta$, the symplectic form of the cotangent bundle.
If we suppose that the curve is an integral curve of a vector field $X_{\rm h}$ then
$$
\inn(X_{\rm h})\Omega=\d {\rm h}\, .
$$
\qed\end{proof}

\begin{remark}{\rm 
The Hamilton--Jacobi variational principle can be stated for every
Hamiltonian system $(M,\Omega,\alpha)$ in general, taking a symplectic potential
$\theta\in\df^1(M)$, such that $\Omega=-\d\theta$,
and a (local) Hamiltonian function $h\in\Cinfty(M)$,
such that $\alpha=\d h$.
}\end{remark}

\vspace{0.2cm}

Summarizing, 
in the texts of classical mechanics, all the above results are collected
stating the so-called:

\noindent{\bf  Minimal action principles} \
{\it
Given a Lagrangian system $(\Tan Q,\Lag )$ and the corresponding associated canonical 
Hamiltonian system $(\Tan^*Q,\Omega,{\rm h})$, we have:

\noindent\textbf{Hamilton's Principle of minimal action}:
The dynamics of the Lagrangian system
$(\Tan Q,\Lag)$ is given by  the curves
$\gamma \colon [a,b]\subset\Real \to Q$,
with fixed extremes, such that they minimize the functional
$$
{\mathbf L}(\gamma ):=\int_{\widehat{\mbox{\boldmath$\gamma$}}} \Lag\d t \ .
$$
\textbf{Hamilton--Jacobi principle of minimal action}:
The dynamics of the Hamiltonian system
$(\Tan^*Q,\Omega,\d{\rm h})$ is given by the curves
$\zeta \colon [a,b]\subset\Real \to \Tan^*Q$,
with fixed extremes, such that they minimize the functional
$$
{\mathbf H}(\zeta ):=\int_\zeta \Theta -{\rm h}\,\d t \ .
$$
}

\section{Examples}

In this last section, we study two of the most typical and relevant mechanical systems
using the symplectic Lagrangian, the Hamiltonian, and the unified formalisms.

\subsection{Harmonic oscillator}
\label{sho}

The classical harmonic oscillator is a mechanical system 
made of a point-particle with mass $m$
moving in $\Real$, submitted to a recuperative force ({\sl Hook's law\/}).
The configuration bundle is $Q=\Real$, with coordinate $(q)$.

\subsubsection{Lagrangian formalism}

The Lagrangian formalism takes place in $\Tan Q\simeq\Real^2$, with coordinates 
$(q,v)$, and the Lagrangian function is 
$$
\Lag=\frac{1}{2}(mv^2-kq^2) \quad , \quad k\in\Real^+ \ .
$$
The Lagrangian elements are
$$
E_\Lag=\frac{1}{2}(mv^2+kq^2) \quad , \quad
\Theta_\Lag=mv\d q \quad , \quad
\Omega_\Lag = m\,\d q\wedge\d v \ ,
$$
and the Lagrangian is regular.
For $\displaystyle X_\Lag=f\derpar{}{q}+g\derpar{}{v}$, equation \eqref{elm} gives
$$
\inn(X_\Lag)\Omega_\Lag=
m\,(f\, \d v-g\, \d q)=
mv\, \d v+kq\, \d q=\d E_\Lag\ ,
$$
which leads to
$$
f=v \quad , \quad mg=-kq \ .
$$
So the Euler--Lagrange vector field is
$$
X_\Lag=v\derpar{}{q}-\frac{k}{m}q\derpar{}{v} \ ,
$$
and its integral curves $(q(t),v(t))$ are the solutions to
$$
\frac{dq}{dt} =v \quad , \quad m\frac{dv}{dt}=-kq
\quad \Longrightarrow \quad m\frac{d^2q}{dt^2}=-kq  \ ,
$$
which is the Euler--Lagrange equation for the system.

\subsubsection{Hamiltonian formalism}

For the Hamiltonian formalism, $\Tan^*Q\simeq\Real^2$,
with coordinates $(q,p)$. First, the Legendre transformation is
$$
\Leg^*q=q  \quad , \quad \Leg^*p=mv  \ ,
$$
which is a diffeomorphism (the Lagrangian is hyperregular).
The canonical Hamiltonian function is
$$
{\rm h}=\frac{p^2}{2m}+kq^2 \ .
$$
As $\Omega=\d q\wedge\d p$,
for $\displaystyle X_{\rm h}=F\derpar{}{q}+G\derpar{}{p}$,
equation \eqref{elmh} gives
$$
\inn(X_{\rm h})\Omega= F\, \d p-G\, \d q-=
\frac{p}{m}\, \d p+kq\, \d q=\d{\rm h}\ ,
$$
which leads to
$$
F=\frac{p}{m} \quad , \quad G=-kq  \ .
$$
So the Hamiltonian vector field is
$$
X_{\rm h}=\frac{p}{m}\derpar{}{q}-kq\derpar{}{p} \ ,
$$
and its integral curves $(q(t),p(t))$ are the solutions to
$$
m\frac{dq}{dt} =p \quad , \quad
\frac{dp}{dt}=-kq \ ,
$$
which are the Hamilton equations for the system.

Observe that, using the Legendre map, the Hamilton and
the Euler--Lagrange equations of the system are, in fact, equivalent.
Obviously, we have that $\Leg_*X_\Lag=X_{\rm h}$.

%{\color{red} (HACER LA ECUACION DE HAMILTON--JACOBI PARA ESTE EJEMPLO)}

\subsubsection{Unified Lagrangian--Hamiltonian formalism}

Consider the unified bundle ${\cal W}=\Tan Q\times_Q\Tan^*Q\simeq\Real^3$
with coordinates $(q,v,p)$. On it, we have the canonical presymplectic form
$$
\Omega_{\cal W}=\d q\wedge\d p 
$$
and the Hamiltonian function
$$
{\cal H}=pv-\frac{1}{2}(mv^2-kq^2) \ .
$$
For $\displaystyle X_{\cal H}=f\derpar{}{q}+g\derpar{}{v}+G\derpar{}{p}$,
equation \eqref{Whamilton-contact-eqs0} gives
$$
\inn(X_{\cal H})\Omega_{\cal W}= f\, \d p-G\, \d q=
kq\, \d q+(p-mv)\,\d v+v\, \d p=\d{\cal H}\ ,
$$
which leads to
$$
f=v \quad , \quad G=-kq \quad , \quad p=mv  \ .
$$
The last equation is a constraint which defines the submanifold
${\cal W}_0\hookrightarrow{\cal W}$ and gives the Legendre map.
Therefore, the Hamiltonian vector field is
$$
X_{\cal H}\vert_{{\cal W}_0}=v\derpar{}{q}+g\derpar{}{v}-kq\derpar{}{p} \ .
$$
Then, the tangency condition leads to
$$
X_{\cal H}(p-mv)=-kq-gm=0 \ \Longleftrightarrow\
g=-\frac{kq}{m} \quad \mbox{\rm (on ${\cal W}_0$)} \ 
$$
and then
$$
X_{\cal H}\vert_{{\cal W}_0}=v\derpar{}{q}-\frac{kq}{m}\derpar{}{v}-kq\derpar{}{p} \ .
$$
Its integral curves $(q(t),v(t),p(t))$ are the solutions to
$$
\frac{dq}{dt}=v \quad , \quad
m\frac{dv}{dt}=-kq \quad , \quad
\frac{dp}{dt}=-kq \ .
$$
Te first two equations are equivalent to
$$
m\frac{d^2q}{dt^2}=-kq \ ,
$$
which is the Euler--Lagrange equation of the system. 
Furthermore, using the constraint $p=mv$ (the Legendre map),
the first and third equations are
$$
\frac{dq}{dt}=\frac{p}{m} \quad , \quad
\frac{dp}{dt}=-kq \ ;
$$
which are the Hamilton equations for the system.

\subsection{Central forces: the Kepler problem}
\label{Kp}

The {\sl Kepler problem} consists in studying 
the motion of a particle of mass $m$
under the action of {\sl Newtonian central forces}.
It is well-known that the motion of such a particle is on a plane and hence
$Q=\Real^2$. We take polar coordinates $(r,\phi)$ in the plane
(with origin at the center of the force).

\subsubsection{Lagrangian formalism}

The Lagrangian formalism takes place in $\Tan Q\simeq\Real^4$, with local coordinates 
$(r,\phi,v_r,v_\phi)$. The Lagrangian function is 
$$
\Lag=\frac{1}{2}m(v_r^2+r^2v_\phi^2)-\frac{K}{r} \quad , \quad K\not=0 \ ;
$$
therefore
\beann
E_\Lag&=&\frac{1}{2}m(v_r^2+r^2v_\phi^2)+\frac{K}{r} \ , \\
\Theta_\Lag&=&m(v_r\,\d r+r^2v_\phi\,\d\phi) \ , \\
\Omega_\Lag&=&m(\d r\wedge\d v_r+r^2\d\phi\wedge\d v_\phi-
2rv_\phi\,\d r\wedge\d \phi) \ ,
\eeann
and the Lagrangian is regular.
For $\displaystyle \X_\Lag=f_r\derpar{}{r}+f_\phi\derpar{}{\phi}+
g_r\derpar{}{v_r}+g_\phi\derpar{}{v_\phi}$, 
equation \eqref{elm} gives
\beann
\inn(X_\Lag)\Omega_\Lag&=&
m\,[f_r\, \d v_r+f_\phi r^2\, \d v_\phi-(g_r-2rv_\phi f_\phi)\, \d r-
(g_\phi r^2+2rv_\phi f_r)\, \d\phi] \\
&=& mv_r\, \d v_r+mr^2v_\phi\, \d v_\phi+\Big(mrv_\phi^2-\frac{K}{r^2}\Big)\d r=\d E_\Lag \ ,
\eeann
which leads to
$$
f_r=v_r \quad , \quad f_\phi=v_\phi \quad , \quad 
mg_r=2mrv_\phi f_\phi-mrv_\phi^2+\frac{K}{r^2} \quad , \quad 
g_\phi=-\frac{2v_\phi f_r}{r}   \ ,
$$
and then the Euler--Lagrange vector field is
$$
X_\Lag=v_r\derpar{}{r}+v_\phi\derpar{}{\phi}+\Big(rv_\phi^2+\dst\frac{K}{mr^2}\Big)\derpar{}{v_r}
-\frac{2v_\phi v_r}{r}\derpar{}{v_\phi} \ .
$$
Then, its integral curves $(r(t),\phi(t),v_r(t),v_\phi(t))$ are the solutions to
\beann
\frac{dr}{dt} =v_r \quad , \quad \frac{d\phi}{dt} =v_\phi \quad , \quad
m\frac{dv_r}{dt}=mrv_\phi^2+\frac{K}{r^2}  \quad , \quad
\frac{dv_\phi}{dt}=-\frac{2v_\phi v_r}{r} &\Longrightarrow& \\
\Longrightarrow \quad m\frac{d^2r}{dt^2}=mr\Big(\frac{d\phi}{dt}\Big)^2+\frac{K}{r^2}  
\quad , \quad  \frac{d^2\phi}{dt^2}=
-\frac{2}{r}\,\frac{d\phi}{dt}\,\frac{dr}{dt}
&\Longrightarrow& \\
\Longrightarrow \quad \frac{d}{dt}\Big(m\frac{dr}{dt}\Big)=mr\Big(\frac{d\phi}{dt}\Big)^2+\frac{K}{r^2}  
\quad , \quad \frac{d}{dt}\Big(mr^2\frac{d\phi}{dt}\Big)=0 &\ , &
\eeann
which are the Euler--Lagrange equations for this system.

There is a Lagrangian exact Noether symmetry 
which is generated by the vector field $\dst Y=\derpar{}{\phi}$, 
since
\bea
\Lie(Y)\Theta_\Lag&=& 
\Lie\left(\derpar{}{\phi}\right)\left(m(v_r\,\d r+r^2v_\phi\,\d\phi)\right) 
\nonumber \\ &=&
\d\inn\left(\derpar{}{\phi}\right)\left(m(v_r\,\d r+r^2v_\phi\,\d\phi)\right)+
\inn\left(\derpar{}{\phi}\right)\,\d\left(m(v_r\,\d r+r^2v_\phi\,\d\phi)\right) 
\nonumber\\ &=&
\d(mr^2v_\phi)+\inn\left(\derpar{}{\phi}\right)
\left(m(\d v_r\wedge\d r+r^2\d v_\phi\wedge\d\phi-
2rv_\phi\,\d \phi\wedge\d r)\right) 
\nonumber \\ &=&
m(2rv_\phi\d r+r^2\d v_\phi-r^2\d v_\phi-2rv_\phi\d r)=0 \ , 
\label{NsymKep1} \\
\Lie(Y)E_\Lag&=& 
\Lie\left(\derpar{}{\phi}\right)\left(\frac{1}{2}m(v_r^2+r^2v_\phi^2)+\frac{K}{r}\right)=0 \ ,
\label{NsymKep2}
\eea
and hence its associated conserved quantity is
$$
f_Y=\inn\left(\derpar{}{\phi}\right)\Theta_\Lag=mr^2v_\phi \ ;
$$
that is, the angular momentum, as the last Euler--Lagrange equation shows.

\subsubsection{Hamiltonian formalism}

For the Hamiltonian formalism, $\Tan^*Q\simeq\Real^4$,
with local coordinates $(r,\phi,p_r,p_\phi)$. First, the Legendre transformation is,
$$
\Leg^*r=r  \quad , \quad \Leg^*\phi=\phi  \quad , \quad
\Leg^*p_r=mv_r  \quad , \quad \Leg^*p_\phi=mr^2v_\phi\ ,
$$
which is a diffeomorphism (the Lagrangian is hyperregular).
The canonical Hamiltonian function is
$$
{\rm h}=\frac{p_r^2}{2m}+\frac{p_\phi^2}{2mr^2}+\frac{K}{r} \ .
$$
As $\Omega=\d r\wedge\d p_r+\d \phi\wedge\d p_\phi$,
for $\displaystyle X_{\rm h}=F_r\derpar{}{r}+F_\phi\derpar{}{\phi}+
G_r\derpar{}{p_r}+G_\phi\derpar{}{p_\phi}$,
equation \eqref{elmh} gives
$$
\inn(X_{\rm h})\Omega= 
F_r\, \d p_r+F_\phi\, \d p_\phi-G_r\, \d r-G_\phi\, \d \phi=
\frac{p_r}{m}\, \d p_r+\frac{p_\phi}{mr^2}\, \d p_\phi-
\left(\frac{p_\phi^2}{mr^3}+\frac{K}{r^2}\right)\, \d r=\d{\rm h}\ ,
$$
which leads to
$$
F_r=\frac{p_r}{m} \quad , \quad F_\phi=\frac{p_\phi}{mr^2} \quad , \quad
G_r=\frac{p_\phi^2}{mr^3}+\frac{K}{r^2} \quad , \quad G_\phi=0  \ .
$$
Then, the Hamiltonian vector field is
$$
X_{\rm h}=\frac{p_r}{m}\derpar{}{r}+\frac{p_\phi}{mr^2}\derpar{}{\phi}+\left(\frac{p_\phi^2}{mr^3}+\frac{K}{r^2}\right)\derpar{}{p_r} \ ,
$$
and its integral curves $(r(t),\phi(t),p_r(t),p_\phi(t))$ are the solutions to
$$
m\frac{dr}{dt} =p_r \quad , \quad mr^2\frac{d\phi}{dt} =p_\phi \quad , \quad
\frac{dp_r}{dt}=\frac{p_\phi^2}{mr^3}+\frac{K}{r^2} \quad , \quad
\frac{dp_\phi}{dt}=0 \ ,
$$
which are the Hamilton equations for this system.

As in the above example, using the Legendre map
one can check that the Hamilton and
the Euler--Lagrange equations of the system are, in fact, equivalent.
Obviously $\Leg_*X_\Lag=X_{\rm h}$.

The Hamiltonian exact Noether symmetry is again the
vector field $\dst Y=\derpar{}{\phi}$, since
\beann
\Lie(Y)\Theta&=& 
\Lie\left(\derpar{}{\phi}\right)\left(p_r\,\d r+p_\phi\,\d\phi\right) =\\ &=&
\d\inn\left(\derpar{}{\phi}\right)\left(p_r\,\d r+p_\phi\,\d\phi\right) +
\inn\left(\derpar{}{\phi}\right)\,\d\left(p_r\,\d r+p_\phi\,\d\phi\right)  \\ &=&
\d p_\phi-\d p_\phi=0 \ , 
\\
\Lie(Y){\rm h}&=& 
\Lie\left(\derpar{}{\phi}\right)\left(\frac{p_r^2}{2m}+\frac{p_\phi^2}{2mr^2}+\frac{K}{r}\right)=0 \ ,
\eeann
and, as the last Hamilton equation shows,
 its associated conserved quantity is again the angular momentum
$$
f_Y=\inn\left(\derpar{}{\phi}\right)\Theta=p_\phi \ .
$$

This Hamiltonian system is also a good and simple example to show how the geometric method of reduction by symmetries proceeds.
The symmetry group $G$ is the group of rotations on the orbit plane.
The Lie algebra ${\bf g}$ is spanned by the vector field
$\displaystyle\xi\equiv\derpar{}{\phi}$ and hence ${\bf g}^*=\{\d\phi\}$.
Thus, the set of fundamental vector fields $\tilde{\bf g}$ is generated by
the vector field $\displaystyle\tilde\xi\equiv Y=\derpar{}{\phi}\in\vf(\Tan^*Q)$.
The action $\Phi\colon G\times\Tan^*Q\to\Tan^*Q$ is effective, free, and proper,
and is a strongly symplectic action on the symplectic manifold $(\Tan^*Q,\Omega)$, since it is exact, as we have seen.
In this way, the momentum map is given by
$$
({\rm J}(r,p_r,\phi,p_\phi))\left(\derpar{}{\phi}\right):=p_\phi
\quad \mbox{\rm (for every $(r,p_r,\phi,p_\phi)\in\Tan^*Q$)}
$$
and, for every weakly regular value
$\mu =\mu_0\d\phi\in{\bf g}^*$, the level sets of this map are
$$
{\rm J}^{-1}(\mu):=\{ (r,p_r,\phi,p_\phi)\in\Tan^*Q\ \vert\ p_\phi=\mu_0\} \ ;
$$
They are defined by the constraints $\zeta :=p_\phi-\mu_0$;
that is, the hypersurfaces of constant angular momentum in $\Tan^*Q$,
and hence $\dst\derpar{}{\phi}$ is tangent to all of them.
On each one, we have the presymplectic Hamiltonian system
$({\cal J}^{-1}(\mu),\Omega_\mu,{\rm h}_\mu)$, where
$$
\Omega_\mu :=j_\mu^*\Omega=\d p_r\wedge\d r
\quad , \quad
{\rm h}_\mu=\frac{p_r^2}{2m}+\frac{K}{r} \ ;
$$
hence $\displaystyle\ker\Omega_\mu=\left\langle\derpar{}{\phi}\right\rangle$.
In this case, $G_\mu=G$ and, applying the Marsden--Weinstein reduction theorem,
this presymplectic system reduces to another symplectic one,
$({\cal J}^{-1}(\mu )/G,\widehat\Omega,\widehat{\rm h})$,
where the local coordinate are $(r,p_r)$, and
$$
\widehat\Omega =\d p_r\wedge\d r
\quad , \quad
\widehat{\rm h}=\frac{p_r^2}{2m}+\frac{K}{r} \ .
$$
The Hamiltonian equation $\inn(\widehat X)\widehat\Omega =\d\widehat{\rm h}$,
with $\widehat X\in\vf({\cal J}^{-1}(\mu)/G)$,
gives the Hamiltonian vector field
$$
\widehat X=\frac{p_r}{m}\derpar{}{r}+\left(\frac{p_\phi^2}{mr^3}+\frac{K}{r^2}\right)\derpar{}{p_r} \ ,
$$
whose integral curves $(r(t),p_r(t))$ are the solutions to the Hamilton equations
$$
\frac{dr}{dt}=\frac{p_r}{m} \quad , \quad 
\frac{dp_r}{dt}=\frac{p_\phi^2}{mr^3}+\frac{K}{r^2} \ ,
$$
In order to obtain the complete set of Hamiltonian equations of the system,
first, remember that the Hamiltonian vector field $X_{\rm h}\in\vf(\Tan^*Q)$ is tangent to the level sets $p_\phi=ctn.$, 
and second that, by the Legendre map, $p_\phi=mr^2v_\phi$.
Therefore, for the integral curves of $X_{\rm h}\in\vf(\Tan^*Q)$, we have that
$$
\frac{d\phi}{dt}=\frac{p_\phi}{mr^2} \quad , \quad
\frac{dp_\phi}{dt}=0 \ .
$$

\subsubsection{Unified Lagrangian--Hamiltonian formalism}

Consider the unified bundle ${\cal W}=\Tan Q\times_Q\Tan^*Q\simeq\Real^6$
with coordinates $(r,\phi,v_r,v_\phi,p_r,p_\phi)$. 
The canonical presymplectic form is
$$
\Omega_{\cal W}=\d r \wedge\d p_r+\d\phi \wedge\d p_\phi \ ,
$$
and the Hamiltonian function is
$$
{\cal H}=p_rv_r+p_\phi v_\phi-\frac{1}{2}m(v_r^2+r^2v_\phi^2)+\frac{K}{r} \ .
$$
For $\displaystyle X_{\cal H}=f_r\derpar{}{r}+
f_\phi\derpar{}{\phi}+g_r\derpar{}{v_r}+g_\phi\derpar{}{v_\phi}+G_r\derpar{}{p_r}+G_\phi\derpar{}{p_\phi}$,
equation \eqref{Whamilton-contact-eqs0} gives
\beann
\inn(X_{\cal H})\Omega_{\cal W}&=&
 f_r\, \d p_r+f_\phi\, \d p_\phi-G_r\, \d r-G_\phi\, \d\phi
\\ &=&
-\Big(\frac{K}{r^2}+mrv_\phi^2\Big)\d r+(p_r-mv_r)\,\d v_r+(p_\phi-mr^2v_\phi)\,\d v_\phi+v_r\, \d p_r+v_\phi\, \d p_\phi
\\ &=& \d{\cal H}\ ,
\eeann
which leads to
$$
f_r=v_r \quad , \quad f_\phi=v_\phi \quad , \quad
G_r=\frac{K}{r^2}+mrv_\phi^2 \quad , \quad G_\phi=0 \quad , \quad
p_r=mv_r \quad , \quad p_\phi=mr^2v_\phi \ .
$$
The last two equations are constraints defining the submanifold
${\cal W}_0\hookrightarrow{\cal W}$ which give the Legendre map.
The Hamiltonian vector field is
$$
X_{\cal H}\vert_{{\cal W}_0}=
v_r\derpar{}{r}+v_\phi\derpar{}{\phi}+g_r\derpar{}{v_r}+g_\phi\derpar{}{v_\phi}
+\Big(\frac{K}{r^2}+mrv_\phi^2\Big)\derpar{}{p_r} \ ,
$$
and the tangency condition leads to
\beann
X_{\cal H}(p_r-mv_r)=\frac{K}{r^2}+mrv_\phi^2-g_rm=0 &\Longleftrightarrow&
g_r=\frac{K}{mr^2}+rv_\phi^2 \quad \mbox{\rm (on ${\cal W}_0$)} \ , \\
X_{\cal H}(p_\phi-mr^2v_\phi)=-m(g_\phi r^2+2f_rrv_\phi)=0 &\Longleftrightarrow&
g_\phi=-\frac{2v_rv_\phi}{r} \quad \mbox{\rm (on ${\cal W}_0$)} \ ;
\eeann
therefore
$$
X_{\cal H}\vert_{{\cal W}_0}=
v_r\derpar{}{r}+v_\phi\derpar{}{\phi}+
\Big(rv_\phi^2+\dst\frac{K}{mr^2}\Big)\derpar{}{v_r}-\frac{2v_rv_\phi}{r}\derpar{}{v_\phi}
+\left(\frac{K}{r^2}+mrv_\phi^2\right)\derpar{}{p_r} \ ,
$$
and its integral curves $(r(t),\phi(t),v_r(t),v_\phi(t),p_r(t),p_\phi(t))$ are the solutions to
$$
\frac{dr}{dt}=v_r \ , \
\frac{d\phi}{dt}=v_\phi \ , \
\frac{dv_r}{dt}=\frac{K}{mr^2}+rv_\phi^2 \ , \
\frac{dv_\phi}{dt}=-\frac{2v_rv_\phi}{r} \ , \
\frac{dp_r}{dt}=\frac{K}{r^2}+mrv_\phi^2 \ , \
\frac{dp_\phi}{dt}=0 \ .
$$
Te first four equations are equivalent to
$$
m\frac{d^2r}{dt^2}=mr\Big(\frac{d\phi}{dt}\Big)^2+\frac{K}{r^2}  \quad , \quad
 \frac{d^2\phi}{dt^2}=
-\frac{2}{r}\,\frac{d\phi}{dt}\,\frac{dr}{dt}\ ,
$$
which are the Euler--Lagrange equation of the system. 
Furthermore, using the constraints $p_r=mv_r$ 
and $p_\phi=mr^2v_\phi$ (that is, the Legendre map),
the first, second, fifth, and sixth equations are
$$
\frac{dr}{dt}=\frac{p_r}{m} \quad , \quad
\frac{d\phi}{dt}=\frac{p_\phi}{mr^2} \quad , \quad
\frac{dp_r}{dt}=\frac{p_\phi^2}{mr^3}+\frac{K}{r^2}  \quad , \quad
\frac{dp_\phi}{dt}=0 \ ;
$$
which are the Hamilton equations for the system.

%%%%%%%%%%%%%%%%%%%%%%%%%%%%%%%%%%%%%%%%%%%%%%%%%%%%%%%%%%%%%%%%%%%%%%%%%%%%%%%%%%%%%%%%%%%%%%%%%%%%%%%%%%%%%%%%%%

\chapter{Cosymplectic mechanics: Nonautonomous dynamical systems}
\label{chap:cosym}

In the previous chapters, we have studied autonomous Hamiltonian and Lagrangian systems;
that is, dynamical systems described by Hamiltonian or Lagrangian functions
which are independent of time.
Now we are going to analyze the case of {\sl nonautonomous
dynamical systems}, which are described by 
time-dependent Hamiltonian or Lagrangian functions.

The geometrical description of nonautonomous Hamiltonian and Lagrangian systems
can be made using different approaches.
For instance, one can use the so-called {\sl contact formalism}
\cite{AM-78,CPT-hctd,EMR-gstds,EMR-sdtc,SC-81}, or a generalization
of it, the {\sl jet bundle formalism}, using {\sl jet} and {\sl  fiber bundles}
\cite{Cr-95,EMR-gstds,MS-98,MV-03,Saunders89}.
However, these kinds of systems can also be described as symplectic
Hamiltonian systems by means of the
{\sl extended formalisms} \cite{ACI-gct,EMR-gstds,Ku-td,Ra2,Ra1,St-2005}, as singular (presymplectic) dynamical systems
 \cite{CGIR-87,EMR-gstds}, or using the Lagrangian-Hamiltonian {\sl unified formalism}  \cite{BEMMR-2008,CMC-02}.
Nevertheless, one of the most elegant and simpler geometric description
of time-dependent systems is the {\sl cosymplectic formulation}
\cite{CLL-92,CLM-91}, and this is where we focus our attention in this chapter.

As in the above chapters, first we state the geometrical foundations on which this formulation is based, which are the {\sl cosymplectic manifolds} and their properties.
Next we introduce the concept of {\sl cosymplectic Hamiltonian system}
and we describe, in particular, the Lagrangian and Hamiltonian
formalisms of nonautonomous Lagrangian dynamical systems using this formulation.
Symmetries, conserved quantities and the theorem of Noether are also discussed in this context.
The chapter is completed with a brief presentation of two 
other very common formulations of time-dependent mechanics: 
the {\sl contact} and the {\sl extended symplectic formulations}, 
and showing their equivalence with the cosymplectic picture.
Finally, some examples of the previous chapter are analyzed 
for the case in which the Lagrangians are time-dependent.

\section{Notions on cosymplectic geometry}

First, we establish the basic foundations of
cosymplectic geometry 
(see, for instance, \cite{CLL-92,dN-2013,CLM-91} for details).

\subsection{Cosymplectic manifolds}

\begin{definition}
\label{deest}
Let $M$ be a differentiable manifold of dimension $2n+1$.
A \textbf{cosymplectic structure}  on $M$ is a couple
$(\eta,\omega)$, where $\eta\in\df^1(M)$ and $\omega\in\df^2(M)$
are closed forms such that
$\eta\wedge\omega^n$ is a volume form.
Then, $(M,\eta,\omega)$ is called a  \textbf{cosymplectic manifold}.

If $\eta\wedge\omega^n$ is not a volume form
and ${\rm dim}\,M$ is arbitrary, then we say that $(\eta,\omega)$
is a \textbf{precosymplectic structure}  on $M$ and
$(M,\eta,\omega)$ is a  \textbf{precosymplectic manifold}.

The (pre)cosymplectic structure is said to be \textbf{exact} if 
$\omega =\d \theta$, for some $\theta\in\df^1(M)$.
\end{definition}

\begin{prop}
If $(\eta ,\omega)$ is a cosymplectic structure on $M$,
then there exists a unique vector field
$R\in\vf(M)$, called the \textbf{Reeb vector field},
which is characterized by the conditions
\beq
\inn(R)\eta=1 \quad ,\quad \inn(R)\omega=0 \ .
\label{rvf}
\eeq
\end{prop}
\begin{proof}
Observe that, by the second condition, $R\in\ker\,\omega$.
From the condition that $\eta\wedge\omega^n$ is a volume form,
we have that ${\rm rank}\,\omega=2n$ and hence
$\ker\,\omega$ is a $1$-dimensional $\Cinfty(M)$-module.
Therefore, the first condition allows us to select one generator of this module.
\\ \qed \end{proof}

The local structure of cosymplectic manifolds is given by the following extension of Darboux Theorem  \cite{dLe89,dLGRR-2023}:

\begin{teor} {\rm (Darboux)} \
Let $(M,\eta,\omega)$ be a cosymplectic manifold.
 Then, for every ${\rm p}\in M$, there exists an open neighborhood $U \subset M$, ${\rm p}\in U$, 
which is the domain of a local chart of coordinates $(t,x^i,y_i)$, 
$1\leq i \leq n$, such that
$$
\eta\vert_U=\d t \quad ,\quad 
\omega\vert_U=\d x^i\wedge\d y_i
\quad , \quad R\vert_U=\derpar{}{t} \ .
$$
These are called \textbf{Darboux} or \textbf{canonical coordinates}
of the cosymplectic manifold.
\label{darbo}
\end{teor}
\begin{proof}
The idea of the proof is the following:
in a cosymplectic manifold there is a symplectic foliation
which is made of  the leaves of the distribution generated by $\ker\,\eta$
(which are $2n$-dimensional submanifolds and, on each one of them, the restriction of $\omega$ is a symplectic form
since it has maximal rank $2n$, by the condition of the volume form).
Then we take coordinates adapted to the foliation and,
on each leaf, apply the symplectic Darboux Theorem.
In this way we have local coordinates $(x^i,y_i,\tilde z)$
such that $\omega=\d x^i\wedge\d y_i$.
Finally, we write $\eta$ as a combination of all of these coordinates,
$\eta=f(q,p,\tilde z)\d\tilde z$,
with $f$ a nonvanishing function, and then we can redefine the coordinate $z$.
In these coordinates, the Reeb vector field has the expression given by the theorem.
\\ \qed \end{proof}

\begin{remark}{\rm 
For precosymplectic manifolds, there is a similar result   \cite{dLe89}.
In fact, if $(M,\eta,\omega)$ is a precosymplectic manifold with ${\rm rank}\,\omega=2r<{\rm dim}\,M-1\equiv m-1$; then,
for every point on $M$, there exists a local chart $(U; t, x^i, y_i, z^j)$, where $1\leq i\leq r$, $1\leq j\leq m-2r-1$, such that
\begin{equation*}
\eta\vert_U=\d t\quad,\quad \omega\vert_U=\d x^i\wedge\d y_i \ .
\end{equation*}
These local charts are the so-called {\sl precosymplectic charts},
and their coordinates are the {\sl canonical coordinates} or 
{\sl  Darboux coordinates} of the precosymplectic manifold in this chart \cite{dLGRR-2023}.
In addition, for precosymplectic manifolds the solution to equations \eqref{rvf} is not unique
and Reeb vector fields are not uniquely defined.
}\end{remark}

\noindent {\bf Canonical model}:
The canonical model for cosymplectic manifolds is the following:
consider the manifold $\Real\times\Tan^*Q$ with canonical projections
$$
\pi_1\colon\Real \times\Tan^*Q \rightarrow \Real \ , \
\pi_2\colon\Real \times \Tan^*Q \rightarrow\Tan^*Q \ , \
\pi_0\colon\Real \times \Tan^*Q \rightarrow Q \ , \
\pi_{1,0}\colon\Real \times \Tan^*Q \rightarrow
\Real\times Q \ .
$$
If $(q^i)$ are local coordinates on $U \subseteq Q$, the
induced local coordinates  $(t,q^i ,p_i)$
on $\pi_0^{-1}(U)=\Real \times \Tan^*U$ are given by
$$
t({\rm t},\alpha_q) = {\rm t} \quad , \quad
q^i({\rm t},\alpha_q) = x^i(q) \quad , \quad
 p_i({\rm t},\alpha_q) =
\alpha_q\left(\displaystyle\frac{\partial}{\partial q^i}\Big\vert_q \right) \ ,
$$
for ${\rm t}\in\Real$, $q\in Q$ and $\alpha_q\in\Tan_q^*Q$.
We define the differential forms on $\Real\times\Tan^*Q$,
$$
\eta=\pi_1^*\d t\, , \quad \theta= \pi_2^*\Theta\, ,
\quad
\omega= \pi_2^*\Omega\, ,
$$
where $\Theta$ and $\Omega$ are the canonical forms on $\Tan^*Q$. In local coordinates we have
$$
\eta=\d t \quad , \quad \theta =  p_i \d q^i \quad   , \quad
\omega = \d q^i\wedge\d p_i\, .
$$
Hence $(\Real\times \Tan^*Q,\eta,\omega)$ is a
cosymplectic manifold, and the natural coordinates of $\Real\times\Tan^*Q$ 
are Darboux coordinates for this canonical cosymplectic structure. 
Furthermore, $\displaystyle\derpar{}{t}$ is its Reeb vector field.

\noindent {\bf Almost-canonical cosymplectic manifolds}:
There is another kind of cosymplectic manifolds which are specially relevant:
those which are of the form $M=\Real\times N$, where
$(N,\Omega)$ is a symplectic manifold.
Then,  denoting by
$$
\pi_\Real\colon\Real \times N \rightarrow \Real \quad , \quad
\pi_N\colon\Real \times N \rightarrow N
$$
 the canonical projections, we have the differential forms
$$
\eta=\pi_\Real^*\d t\quad , \quad \omega=\pi_M^*\Omega\ .
$$
The conditions given in Definition \ref{deest}
are verified and hence $\Real\times N$ is a cosymplectic manifold.
From the Darboux Theorem \ref{darbo} we have local coordinates $(t,x^i,y_i)$ on $\Real\times N$.
These kinds of $k$-cosymplectic manifolds are sometimes
called {\sl almost-canonical $k$-cosymplectic manifolds}.
Observe that the standard model is a particular class of these kinds of $k$-cosymplectic manifolds,
where $N=\Tan^*Q$.

Every cosymplectic manifold $(M,\eta,\omega)$ is endowed with
the natural vector bundle isomorphism
$$
\begin{array}{rcl}
\flat_{(\eta,\omega)}\colon \Tan M & \to & \Tan^*M\\ \noalign{\medskip}
({\rm p},X_{\rm p}) & \mapsto & 
\big({\rm p},\inn(X_{\rm p})\omega_{\rm p} + ((\inn(X_{\rm p})\eta_{\rm p})\eta_{\rm p}\big)\,.
\end{array}
$$
and its inverse $\sharp_{(\eta,\omega)}=\flat_{(\eta,\omega)}^{-1}\colon \Tan^*M \to \Tan M$.
Their natural extensions are the $\Cinfty(M)$-module isomorphisms
which are denoted with the same notation,
$$
\begin{array}{rccl}
   \flat_{(\eta,\omega)}\colon & \vf(M) & \longrightarrow & \df^1(M) \\
   & X & \longmapsto & \inn(X)\omega+(\inn(X)\bmeta)\bmeta
\end{array} \ ,
$$
and its inverse $\sharp_{(\eta,\omega)}=\flat_{(\eta,\omega)}^{-1}\colon \df^1(M)\to \vf(M)$.
In particular, for the Reeb vector field, we have that
$\flat_{(\eta,\omega)}(R)=\eta$.

\subsection{Hamiltonian, gradient, and evolution vector fields}

Using the natural $\Cinfty(M)$-module isomorphism $\flat_{(\eta,\omega)}$
introduced in the above section,
one can associate to every function $f\in\mathcal{C}^\infty(M)$
some particular vector fields:

\begin{definition}
\label{GradHamEvol}
Let $(M,\eta,\omega)$ be a cosymplectic manifold and $f\in\mathcal{C}^\infty(M)$.

The \textbf{Hamiltonian vector field} associated with $f$
is the vector field \ $X_f\in\vf(M)$ defined by
\ $\flat_{(\eta,\omega)}(X_f):=\d f-R(f)\eta$.

 The \textbf{gradient vector field} associated with $f$
is the vector field \ ${\rm grad}\,f\in\vf(M)$ defined by
\ $\flat_{(\eta,\omega)}({\rm grad}\, f):=\d f$.

The \textbf{evolution vector field}  associated with $f$
is the vector field ${\cal E}_f\in\vf(M)$ defined by
\ ${\cal E}_f:=R+ X_f$;\ or, equivalently
\ $\flat_{(\eta,\omega)}({\cal E}_f)=\d f-(R(f)-1)\eta$.
\end{definition}

These vector fields can be equivalently characterized as follows:

\begin{prop}
The \textbf{Hamiltonian vector field} associated with $f$
is determined by the equations:
\beq
\inn(X_f)\eta=0 \quad ,\quad \inn(X_f)\omega=\d f-R(f)\eta\ .
\label{Havf}
\eeq

The \textbf{gradient vector field} associated with $f$ is determined by the equations:
\beq
 \inn({{\rm grad}\, f})\eta = R(f) \quad ,\quad 
\inn({{\rm grad}\, f})\omega=\d f-R(f)\eta\ .
\label{gravf}
\eeq

The \textbf{evolution vector field}  associated with $f$
is determined by the equations:
\beq
\inn({\cal E}_f)\eta=1 \quad,\quad 
\inn({\cal E}_f)\omega = \d f- R(f)\eta\ .
\label{evovf}
\eeq
\end{prop}
\begin{proof}
For every $f\in\Cinfty(M)$,
if  $X_f$ is the Hamiltonian vector field associated with $f$,
using the definitions of the isomorphism $\flat_{(\eta,\omega)}$ and of $X_f$, first we have that,
\beq
\label{aux01}
\inn(X_f)\d\eta+(\inn(X_f)\eta)\eta=
\flat_{(\eta,\omega)}(X_f)=\d f-R(f)\,\eta \ ,
\eeq
and contracting both members with $R$ and using \eqref{rvf}, we get
$$
(\inn(X_f)\eta)\inn(R)\eta=
\inn(R)\d f-R(f)\,\inn(R)\eta=
R(f)-R(f)=0
\ \Longleftrightarrow\ \inn(\X_f)\eta=0 \ ,
$$
since this holds for every $f$.
Now, going to \eqref{aux01}, we obtain that
$$
\inn(X_f)\d\eta=\d f-R(f)\,\eta \ .
$$

Repeating the same procedure for the gradient and the evolution vector fields associated with $f$,
we obtain the corresponding equations for ${\rm grad}\,f$ and ${\cal E}_f$.
\\ \qed \end{proof}

As in the above chapters, from these results it is immediate to obtain that:

\begin{prop}
The equations for the integral curves $c\colon I\subset\Real\to M$
of the Hamiltonian, gradient, and evolution vector fields associated with
$f\in\mathcal{C}^\infty(M)$ are, respectively:
\bea
\inn(\widetilde{c})(\eta\circ c)=0 \ & ,& \ 
\inn(\widetilde{c})(\omega\circ c)=(\d f-R(f)\eta)\circ c\ , \nonumber \\
 \inn(\widetilde{c})(\eta\circ c)) = R(f)\circ c \ & , &\
\inn(\widetilde{c})(\omega\circ c))=(\d f-R(f)\eta)\circ c\ , \nonumber \\
\inn(\widetilde{c})(\eta\circ c))=1 \ & ,& \
\inn(\widetilde{c})(\omega\circ c)) = (\d f- R(f)\eta)\circ c\ .
\label{evoic}
\eea
\end{prop}

\noindent{\bf Local expressions}:
In Darboux coordinates on $M$, we have that $\dst R(f)=\derpar{f}{t}$; therefore
$\dst \d f- R(f)\eta=\derpar{f}{x^i}\d x^i+\derpar{f}{y_i}\d y_i$, and from \eqref{Havf}, \eqref{gravf}, and \eqref{evovf},
we obtain
\footnote{
 Observe that the local expression of the Hamiltonian vector field in the cosymplectic formulation is the same as the Hamiltonian vector field in the symplectic formulation;
which justifies the terminology.},
\bea
    X_f &=& \frac{\partial f}{\partial y_ i}\frac{\partial}{\partial x^i}- \frac{\partial f}{\partial x^i}\frac{\partial}{\partial y_i}\ , \nonumber \\
\noalign{\medskip}
    {\rm grad}\,f &=& \frac{\partial f}{\partial t}\frac{\partial}{\partial t} + \frac{\partial f}{\partial y_ i}\frac{\partial}{\partial x^i}- \frac{\partial f}{\partial x^i}\frac{\partial}{\partial y_i}\ ,\nonumber \\
\noalign{\medskip}
    {\cal E}_f &=& \frac{\partial}{\partial t} + \frac{\partial f}{\partial y_ i}\frac{\partial}{\partial x^i}- \frac{\partial f}{\partial x^i}\frac{\partial}{\partial y_i}\ .
\label{evocoor1}
\eea
Therefore, if $c(s)=(t(s),x^i(s), y_i(s))$ is an integral curve of some of these vector fields, these expressions imply that $c(s)$ should satisfy the following systems of differential equations,
\bea
\frac{dt}{ds}=0 \ &,& \
\frac{dx^i}{ds}=\frac{\partial f}{\partial y_i} \quad ,\quad 
\frac{dy_i}{ds}= -\frac{\partial f}{\partial x^i}\ , \nonumber \\
\frac{dt}{ds}=\frac{\partial f}{\partial t} \ &,& \
\frac{dx^i}{ds}=\frac{\partial f}{\partial y_i} \quad ,\quad 
\frac{dy_i}{ds}= -\frac{\partial f}{\partial x^i}\ , \nonumber \\
\frac{dt}{ds}=1 \ &,& \
\frac{dx^i}{ds}=\frac{\partial f}{\partial y_i} \quad ,\quad 
\frac{dy_i}{ds}= -\frac{\partial f}{\partial x^i}\ ;
\label{evocoor2}
\eea

As in the symplectic case, a Poisson bracket can be defined now, using gradient vector fields. Indeed, it is immediate to prove that:

\begin{prop}
\label{4.6}
Every cosymplectic manifold $(M,\eta,\omega)$ is
a {\sl Poisson manifold}, with the {\sl Poisson bracket} defined by
\[
    \{f,q\}:=\omega({\rm grad}\, f,{\rm grad}\, g) \quad ; \quad 
f,g\in\mathcal{C}^\infty(M) \ .
\]
\end{prop}

The expression of this Poisson bracket in Darboux coordinates is the usual one:
\beq
\label{PBcoor}
    \{f,g\}= \frac{\partial f}{\partial y_i}\frac{\partial g}{\partial x^i}- \frac{\partial f}{\partial x^i}\frac{\partial g}{\partial y_i}\,.
\eeq

As a consequence of these definitions and properties,
many of the results stated in Section \ref{sympgeom}
about Poisson brackets and canonical transformations for 
symplectic manifolds can be extended also to this case.

\section{Nonautonomous Hamiltonian dynamical systems}

\subsection{Nonautonomous Hamiltonian systems}
\label{naspos}

Bearing in mind the above considerations and the set
of postulates for autonomous dynamical systems,
 we can state the 
following postulates for the geometric study of nonautonomous Hamiltonian dynamical systems.

\begin{pos}
{\rm (First Postulate of nonautonomous Hamiltonian mechanics\/)}:
The phase space of a regular (resp. singular)
nonautonomous dynamical system 
is a differentiable manifold $M$ endowed with a cosymplectic 
(resp. precosymplectic) structure $(\eta,\omega)$:
\label{axi1}
\end{pos}

\begin{pos}
{\rm (Second Postulate of nonautonomous  Hamiltonian mechanics\/)}:
The observables or physical magnitudes of a nonautonomous dynamical  system
are functions of $\Cinfty (M)$.
\label{axi2}
\end{pos}

\begin{pos}
{\rm (Third Postulate of nonautonomous  Hamiltonian mechanics\/)}:
The dynamics of a nonautonomous dynamical  system
is given by a function $h\in\Cinfty(M)$
(or, in general, a closed $1$-form $\alpha \in Z^1(M)$, such that $\alpha=\d h$, locally)
which is called the \textbf{Hamiltonian function}
(or  the \textbf{Hamiltonian $1$-form}) of the system.
This function represents the energy of the system.
\end{pos}

\begin{pos}
{\rm (Fourth Postulate of nonautonomous  Hamiltonian mechanics\/)}:
The dynamical trajectories of a nonautonomous dynamical system are the integral curves
of the evolution vector field ${\cal E}_h\in \vf(M)$ associated with $h$;
that is, of the vector field solution to equations \eqref{evovf}.
Thus, they are the solutions to equations \eqref{evoic}.
\end{pos}

Then we define:

\begin{definition}
A \textbf{regular nonautonomous} or \textbf{cosymplectic Hamiltonian dynamical system}
is a set $(M,\eta,\omega;h)$,
where $(M,\eta,\omega)$ is a cosymplectic manifold 
and $h\in\Cinfty(M)$ is the Hamiltonian function of the system.
If $(M,\eta,\omega)$ is a precosymplectic manifold, then $(M,\eta,\omega;h)$
is said to be a \textbf{singular nonautonomous} or \textbf{precosymplectic Hamiltonian dynamical system}.
\label{stdhr}
\end{definition}

\begin{definition}
Given a nonautonomous Hamiltonian dynamical system $(M,\eta,\omega;h)$, 
the \textbf{Hamiltonian problem} posed by  the system
consists in finding the evolution vector field ${\cal E}_h \in \vf (M)$
associated with $h$ (if it exists).
\end{definition}

In addition, we have:

\begin{prop}
\label{teo-evoeqs}
If $(M,\eta,\omega;h)$ is a regular nonautonomous Hamiltonian system, 
then there exists a unique evolution vector field ${\cal E}_h\in\vf(M)$;
that is, a unique vector field which is the solution to equations \eqref{evovf}.
\end{prop}
\begin{proof}
It is immediate because,
if the system is regular, the existence of the isomorphism $\sharp_{(\eta,\omega)}$ is assured.
\\ \qed \end{proof}

\begin{remark}{\rm 
If $(M,\eta,\omega,h)$ is a precosymplectic Hamiltonian system
equations \eqref{evovf} are not necessarily compatible everywhere on $M$ 
and a suitable {\sl constraint algorithm} must be implemented in order to find 
a {\sl final constraint submanifold} $P_f\hookrightarrow M$
(if it exists) where there are evolution vector fields ${\cal E}_h\in\vf(M)$,
tangent to $P_f$ which are solutions to equations \eqref{evovf} on $P_f$
(they are not necessarily unique).
Singular nonautonomous systems and the
corresponding constraint algorithms are studied in \cite{CF-93,CLM-94,LMM-96,LMM-96b,LMMMR-02,Vig-00}.
}\end{remark}

\section{Nonautonomous Lagrangian dynamical systems}

\subsection{Geometric elements}
\label{geomRxTQ}

In order to develop the cosymplectic Lagrangian formalism, first
we have to extend some canonical structures of
$\Tan Q$ to $\Tan(\Real\times\Tan Q)$.
Notice that, as $\Real\times\Tan Q$ is a product manifold, we can write
$$
\Tan(\Real\times\Tan Q) =\Tan\Real\oplus_{\Real\times\Tan Q}
\Tan(\Tan Q)
$$
and this splitting extends in a natural way to vector fields.
Thus, any operation on tangent vectors to $\Tan Q$
acts on tangent vectors to $\Real\times\Tan Q$.
In particular, the canonical endomorphism $J$
of $\Tan(\Tan Q)$ yields a \textsl{canonical endomorphism}
${\cal J} \colon \Tan (\Real\times\Tan Q) \to \Tan (\Real\times\Tan Q)$ and,
similarly, the Liouville vector field on $\Tan Q$
yields a \textsl{Liouville vector field}
$\Delta \in \vf(\Real\times\Tan Q)$ which is the Liouville vector field 
of the vector bundle structure $\pi_{1,0}\colon\Real\times\Tan Q\to\Real\times Q$.
In natural coordinates, their local expressions are
$$
{\cal J} =
\frac{\partial}{\partial v^i} \otimes \d q^i
\quad ,\quad
\Delta = 
v^i\, \frac{\partial}{\partial v^i}
\,.
$$

\begin{definition}
\label{de652}
Let  ${\bf c} \colon\Real \rightarrow \Real\times Q$ be a curve.
We can write ${\bf c} = (c_1,c_2)$, where
$c_1 \colon\Real \rightarrow \Real$ and $c _2\colon\Real \rightarrow Q$.
The \textbf{lift} of ${\bf c}$ 
to $\Real\times\Tan Q$ is the curve
$$
\widehat{\bf c}=(c_1,\widetilde{c_2})\colon\Real \longrightarrow \Real\times\Tan Q  \,,
$$
where $\widetilde{c_2}$ is the canonical lift of~$c_2$ to $\Tan Q$.
The curve $\widehat{\bf c}$ is said to be \textbf{holonomic} on $\Real\times\Tan Q$.
\\
A vector field  $\Gamma \in \vf(\Real\times\Tan Q)$ 
is said to be a  {\sc sode} on $\Real\times\Tan Q$ if its integral curves are holonomic. 
\end{definition}

As in the autonomous case, the last definition can be 
equivalently expressed in terms of the canonical structures of
$\Real\times\Tan Q$, and so it is immediate to prove that:

\begin{prop}
A vector field $\Gamma \in\vf(\Real\times\Tan Q)$ 
is a \textsc{sode} if, and only if,
${\cal J}(\Gamma) = \Delta$.
\end{prop}

In coordinates, if 
${\bf c}(t)=(s(t),q^i(t))$, 
then
$$
\widehat{\bf c}(t) =
\left(s(t),q^i(t),\frac{d q^i}{d t}(t)\right) \ ,
$$
and the local expression of a {\sc sode} is
$$
\Gamma= g\,\frac{\partial}{\partial s} +
v^i \frac{\partial}{\partial q^i} +
f^i \frac{\partial}{\partial v^i} \ ;
$$
so, in coordinates a {\sc sode} defines the following system of
differential equations
$$
\frac{d s}{d t}=g(q,\dot q,s)  \quad , \quad  
\frac{d^2 q^i}{d t^2}=f^i(q,\dot q,s) \ .
$$

\subsection{Lagrangian formalism. Nonautonomous Lagrangian systems}

The foundations of the Lagrangian formulation of
(first-order) nonautonomous dynamical system are analogous to those given
in Section \ref{poslagdyn} for autonomous Lagrangian systems; 
and they can be stated as follows:

\begin{pos}
{\rm (First Postulate of nonautonomous Lagrangian mechanics\/)}:
The configuration space of a system with $n$ degrees of freedom 
is $\Real\times Q$, where $Q$ is a $n$-dimensional differentiable manifold.
The phase space is the bundle $\Real\times\Tan Q$.
\end{pos}

\begin{pos}
{\rm (Second Postulate of nonautonomous Lagrangian mechanics\/)}:
The observables or physical magnitudes of the system are functions of $\Cinfty (\Real\times\Tan Q)$.
\end{pos}

\begin{pos}
{\rm (Third Postulate of nonautonomous Lagrangian mechanics\/)}:
There is a function $\Lag\in\Cinfty (\Real\times\Tan Q)$, 
called the \textbf{Lagrangian function},
which contains the dynamical information of the system.
\end{pos}

Let $\Lag\in\Cinfty(\Real\times\Tan Q)$ be a Lagrangian function.
As in the autonomous case, using the canonical structures in 
$\Real\times\Tan Q$ we can introduce the {\sl Lagrangian forms} $\theta_\Lag\in\df^1(\Real\times\Tan Q)$ and
 $\omega_\Lag\in\df^2(\Real\times\Tan Q)$ associated with $\Lag$,
 which are defined as follows:
 $\theta_\Lag={\cal J}(\d\Lag)$ and $\omega_\Lag=-\d\theta_\Lag$.
They have the  local expressions
$$
%\label{am2}
\theta_\Lag= \displaystyle\frac{\displaystyle\partial
\Lag}{\displaystyle\partial v^i}\, \d q^i\ \quad ,\quad 
\omega_\Lag=\d q^i\wedge\d \left(\frac{\displaystyle\partial
\Lag}{\displaystyle\partial v^i}\right) \  .
$$
In the same way, we define the {\sl energy Lagrangian function} associated with $\Lag$ as
$E_\Lag:=\Delta(\Lag)-\Lag$, whose local expression is
$$
E_\Lag=v^i\derpar{\Lag}{v^i}-\Lag\ .
$$

\begin{definition}
\label{legmap}
Given a Lagrangian 
$\Lag\in\Cinfty(\Real\times\Tan Q)$,  the \textbf{Legendre map}
associated with $\Lag$ is the fiber derivative of~$\Lag$,
considered as a function on the vector bundle
$\pi_0 \colon\Real \times\Tan Q\to Q \times \Real$;
that is, the map
${\rm F}\Lag \colon\Real\times\Tan Q \to\Real\times\Tan^*Q$ 
given by
$$
{\rm F}\Lag (t,q,v_q) = \left(t,{\rm F}\Lag_t (q,v_q)\right)
\,,
$$
where $\Lag_t\colon\Tan Q\to\Real$ denotes the restriction of $\Lag$ to each fiber of the bundle $\pi_1\colon\Real\times\Tan Q \to\Real$
(that is, the Lagrangian $\Lag$ with $t$ ``freezed'').
\end{definition}

In natural coordinates we have
$$
t \circ {\rm F}\Lag = t \quad , \quad
q^i \circ {\rm F}\Lag = q^i \quad , \quad p_i \circ{\rm F}\Lag =\derpar{\Lag}{v^i} \ .
$$

Observe that, considering the canonical cosymplectic manifold
$(\Real\times\Tan^*Q,\eta,\omega)$, where $\omega=-\d\theta$, 
then the Lagrangian forms can also be defined as 
\beq
\label{flom}
\theta_\Lag={\rm F}\Lag^{\;*}\theta\quad , \quad\omega_\Lag={\rm F}\Lag^{\;*}\omega\ .
\eeq

And, as in the autonomous case, we also have:

\begin{definition}\label{de811}
A Lagrangian function $\Lag\in\Cinfty(\Real\times\Tan Q)$ is said
to be \textbf{regular}  (resp. \textbf{hyperregular}) if the corresponding
Legendre map ${\rm F}\Lag$ is a local (resp. global) diffeomorphism.
Elsewhere, $\Lag$ is called a \textbf{singular} Lagrangian.

A singular Lagrangian function $\Lag\in\Cinfty(\Real\times\Tan Q)$
 is called \textbf{almost-regular} if $\mathcal{P}:= {\rm F}\Lag(\Real\times\Tan Q)$ is
a closed submanifold of \ $\Real\times\Tan^*Q$ (we will denote
the natural embedding by $\jmath_0\colon\mathcal{P}\hookrightarrow\Real\times\Tan^*Q$), 
${\rm F}\Lag$ is a submersion onto its image, and
the fibers ${\rm F}\Lag^{-1}(FL({\rm p}))$, for every ${\rm p}\in \Real\times\Tan Q$, are
connected submanifolds of $\Real\times\Tan Q$.
\end{definition}

Once again, as in the autonomous case, it is immediate to prove that $\Lag$ is regular if, and only
if, the matrix $\left(\frac{\displaystyle\partial^2\Lag}
{\displaystyle\partial v^i
\partial v^j}\right)$ is regular everywhere.
Therefore, the following equivalences are also immediate:

\begin{prop}  
\label{equinona}
The following conditions are equivalent:
\ben
\item
The Lagrangian $\Lag$ is regular.
\item
The Legendre map ${\rm F}\Lag$  is a local diffeomorphism. 
\item
The pair 
$(\d t,\omega_\Lag)$ is a cosymplectic structure on $\Real\times\Tan Q$.
\een
\end{prop}

If $\Lag$ is not regular, then $(\d t,\omega_\Lag)$ is a precosymplectic structure on $\Real\times\Tan Q$.
Thus, if $\Lag$ is a regular (resp. singular) Lagrangian,
we have that $(\Real\times\Tan Q,\d t,\omega_\Lag;E_\Lag)$ is a 
cosymplectic (resp. precosymplectic) Hamiltonian system.
Then we define:

\begin{definition}
A \textbf{nonautonomous} or {\rm pre)cosymplectic Lagrangian dynamical system}
is a pair $(\Real\times\Tan Q,\Lag )$,
where $Q$ is an $n$-dimensional manifold,
and $\Lag \in \Cinfty (\Real\times\Tan Q)$ is the Lagrangian function of the system.
\label{stdlr}
\end{definition}

And now we can state:

\begin{pos}
{\rm (Fourth Postulate of nonautonomous Lagrangian mechanics\/)}:
The dynamical trajectories of a nonautonomous Lagrangian system are the integral curves
of a vector field $\Gamma_\Lag\in \vf(\Real\times\Tan Q)$ satisfying the conditions:
\ben
\item
$\Gamma_\Lag$ is the evolution vector field associated with $E_\Lag$;
that is, it is a solution to the equations
\beq
\inn(\Gamma_\Lag)\d t=1 \quad,\quad
\inn (\Gamma_\Lag)\omega_\Lag = \d E_\Lag- R_\Lag(E_\Lag)\d t \ .
\label{nelm}
\eeq
\item
$\Gamma_\Lag$ is a {\sc sode}:
\ ${\cal J}(\Gamma_\Lag ) = \Delta$.
\een
Therefore, these trajectories are the holonomic curves 
${\bf c}\colon I\subset\Real\to\Real\times\Tan Q$
which are the solutions to equations
\beq
\inn(\widetilde{\bf c})(\d t\circ{\bf c})=1 \quad  , \quad
\inn(\widetilde{\bf c})(\omega_\Lag\circ{\bf c})= (\d E_\Lag -R_\Lag(E_\Lag)\d t)\circ {\bf c}\ .
\label{neuq}
\eeq
Equations \eqref{neuq} are the
 \textbf{nonautonomous Euler--Lagrange equations for curves}.
Equations \eqref{nelm} are called the
 \textbf{nonautonomous Lagrangian equations for vector fields} and
a vector field solution to them (if it exists) is a
\textbf{nonautonomous Lagrangian dynamical vector field}.
If, in addition, $\Gamma_\Lag$ is a {\sc sode}, then
it is called a \textbf{nonautonomous Euler--Lagrange vector field} of the system, 
\end{pos}

In this postulate, $R_\Lag$ is a Reeb vector field of the structure
$(\Real\times\Tan Q,\d t,\omega_\Lag)$, determined by the corresponding equations to \eqref{rvf} which are
$$
\inn(R_\Lag)\d t=1 \quad , \quad \inn(R_\Lag)\omega_\Lag=0 \ .
$$

\begin{definition}
Given a nonautonomous Lagrangian dynamical system $(\Real\times\Tan Q,\Lag)$, 
the \textbf{Lagrangian problem} posed by  the system
consists in finding a {\sc sode} vector field $\Gamma_\Lag \in \vf (\Real\times\Tan Q)$
solutions to (\ref{nelm}).
\end{definition}

\noindent {\bf Local expressions}:
Consider a natural chart $(U;t,q^i,v^i)$ of $\Real\times\Tan Q$.
Bearing in mind  the local expressions of the several geometric elements
appearing in the nonautonomous dynamical equations,
first we obtain that the Reeb vector field is given by
$$
\displaystyle R_\Lag=\derpar{}{t}+R^i\derpar{}{v^i} \ ,
$$ 
where the functions $R^i$ are obtained from 
$$
\frac{\partial^2\Lag}{\partial t\partial v^j}+\frac{\partial^2\Lag}{\partial v^j\partial v^i}R^i=0 \ ,
$$
and then we have that
$$
R_\Lag(E_\Lag)=-\derpar{\Lag}{t}  \ .
$$
Therefore, if \(\dst \Gamma_\Lag = f\derpar{}{t}+A^i\derpar{}{q^i} +
B^i\derpar{}{v^i}\),
equations (\ref{nelm}), written in coordinates, lead to $f=1$ and
\beann
0&=&\frac{\partial^2\Lag}{\partial v^j \partial v^i}B^i+
\left(\frac{\partial^2\Lag}{\partial q^j \partial v^i}+
\frac{\partial^2\Lag}{\partial v^j \partial q^i}\right)A^i+
\frac{\partial^2\Lag}{\partial v^i \partial q^j}v^i+
\frac{\partial^2\Lag}{\partial v^j \partial t}-
\derpar{\Lag}{q^j} \ ,
\\
0&=&\frac{\partial^2\Lag}{\partial v^j \partial v^i}(A^i-v^i) \ .
\eeann
To demand that $\Gamma_\Lag$ is a {\sc sode} is locally equivalent to demand that
$A^i=v^i$. Furthermore, the integral curves of $\Gamma_\Lag$
are holonomic; that is, they are of the form ${\bf c}(t)=(t,q^i(t),\dot q^i(t))$ and
$$
A^i = v^i = \frac{d q^i}{d t}
\quad , \quad
B^i = \frac{d^2q^i}{d t^2} \ ,
$$
and the combination of these expressions with the above equations leads to the
equation of  the integral curves of $\Gamma_\Lag$ which is
$$
\left(\frac{\partial^2\Lag}{\partial v^j \partial v^i}\circ{\bf c}\right)
\frac{d^2q^i}{d t^2}+
\left(\frac{\partial^2\Lag}{\partial v^j \partial q^i}\circ{\bf c}\right)
\frac{d q^i}{d t}+
\left(\frac{\partial^2\Lag}{\partial v^j \partial t}\circ{\bf c}\right)
-\derpar{\Lag}{q^j}\circ{\bf c}=0 \ ,
$$
or also in an equivalent form as in \eqref{ELequats},
which is the classical coordinate expression of the Euler--Lagrange equations.

\begin{remark}{\rm
If the Lagrangian $\Lag$ is singular (in particular almost-regular), 
then the existence of solutions to equations \eqref{nelm}
is not assured except, perhaps, in a
submanifold of $\Real\times\Tan Q$. Furthermore,
when these solutions exist, they are not {\sc sode}, in general.
Thus, in order to recover the Euler--Lagrange equations \eqref{neuq}
for the integral curves of $\Gamma_\Lag$,
the condition ${\cal J}(\Gamma_\Lag)=\Delta$ must be added to the
Lagrangian equations  \eqref{nelm}.
Then, as in the autonomous case, in general,
a constraint algorithm must be implemented in order to find
a submanifold $S\hookrightarrow\Real\times\Tan Q$ where the existence of 
{\sc sode} vector fields solutions to the Lagrange equations on $S$,
and tangent to $S$, are assured.
Furthermore, Reeb vector fields $R_\Lag$ are not uniquely defined by
equations \eqref{rvf} for this case; nevertheless
the dynamics and the constraint algorithm are independent 
of the selected Reeb vector field $R_\Lag$ \cite{CLM-94}. }
\end{remark}

\subsection{Canonical Hamiltonian formalism}

As in the autonomous case analyzed in Section \ref{canforma},
we study the case of hiperregular systems;
although all the results hold also for the case of
regular systems,
changing $\Real\times\Tan^*Q$ by ${\rm F}\Lag (\Real\times\Tan Q)\subset\Real\times\Tan^*Q$.
First, as ${\rm F}\Lag$ is a diffeomorphism we have:

\begin{prop}
Let $(\Real\times\Tan Q,\Lag)$ be a hiperregular nonautonomous Lagrangian system.
Then there exists a unique function ${\rm h}\in\Cinfty (\Real\times\Tan^*Q)$ such that
${\rm F}\Lag^*{\rm h}=E_\Lag$,
which is the \textbf{Hamiltonian function} associated with the system
$(\Real\times\Tan Q,\Lag)$.
The triple $(\Real\times\Tan^*Q,\d t,\omega,{\rm h})$ is the 
\textbf{canonical nonautonomous Hamiltonian system} associated with $(\Real\times\Tan Q,\Lag)$,
where $\omega\in\df^2(\Real\times\Tan^*Q)$
is given in \eqref{flom}.
\end{prop}

Therefore, we have the cosymplectic Hamiltonian system
$(\Real\times\Tan^*Q,\d t,\omega;{\rm h})$, fulfilling the postulates and results established in Section \ref{naspos}.
and hence the Hamiltonian equations for vector fields
${\cal E}_{\rm h}\in\vf(\Real\times\Tan^*Q)$ and their
integral curves ${\bf c}\colon\Real\to\Real\times\Tan^*Q$ read
\bea
\inn({\cal E}_{\rm h})\d t=1 \quad&,&\quad 
\inn({\cal E}_{\rm h})\omega = \d{\rm h}- R({\rm h})\d t\ , 
\label{celmh1} \\
\inn(\widetilde{\bf c})(\d t\circ{\bf c})=1 \ & ,& \
\inn(\widetilde{\bf c})(\omega\circ{\bf c}) = (\d{\rm h}- R({\rm h})\d t)\circ{\bf c}\ .
\label{celmh2} 
\eea

In canonical coordinates, the local expression 
of the dynamical vector field ${\cal E}_{\rm h}$ solution to equations 
\eqref{celmh1} is given by \eqref{evocoor1}
and the equations of its integral curves \eqref{celmh2} are
equations \eqref{evocoor2}, with $f={\rm h}$.
Then, since $\displaystyle\frac{dt}{ds}=1$ implies $t(s)=s+ const.$
($t$ is an affine transformation of $s$), we deduce that
$$
\frac{dx^i}{dt} = \frac{\partial f}{\partial y_i} \quad ,\quad 
\frac{dy_i}{dt}= -\frac{\partial f}{\partial x^i}\ ,
%\label{evocoor3}
$$
which are the {\sl cosymplectic Hamilton equations}.

The relation between the Lagrangian and the canonical Hamiltonian formalisms of a (hyper)regular Lagrangian system is stated as in the autonomous case as follows:

\begin{teor} {\rm (Equivalence Theorem)}
Let $(\Real\times\Tan Q,\Lag)$ be a (hyper)regular Lagrangian system.
\ben
\item
If $\Gamma_\Lag$ is the Lagrangian vector field 
solution to equations (\ref{nelm}),
then there exists a unique vector field
 ${\rm F}\Lag_*\Gamma_\Lag\equiv {\cal E}_{\rm h}\in\vf (\Real\times\Tan^*Q)$
which is the solution to equations (\ref{celmh1}).

Conversely, if ${\cal E}_{\rm h}$ is the evolution vector field
solution to equation (\ref{celmh1}),
then there exists a unique vector field
 ${\rm F}\Lag_*^{-1}{\cal E}_{\rm h}\equiv \Gamma_\Lag\in\vf (\Real\times\Tan Q)$
which is the solution to equations (\ref{nelm}) and \eqref{edso}.
\item
Equivalently, if $\mbox{\boldmath$\gamma$}\colon I\subset\Real\to\Real\times Q$ 
is a curve and its canonical lift
$\widehat{\mbox{\boldmath$\gamma$}} \colon I\subset\Real\to\Real\times\Tan Q$ is a
solution to equation \eqref{neuq},
then ${\mbox{\boldmath$\zeta$}}={\rm F}\Lag\circ\widehat{\mbox{\boldmath$\gamma$}}$ is a curve
solution to equation \eqref{celmh2}.

Conversely, if ${\mbox{\boldmath$\zeta$}}\colon I\subset\Real\to\Real\times\Tan^*Q$ 
is a curve solution to equation \eqref{celmh2},
then $\widehat{\mbox{\boldmath$\gamma$}}=\widehat{\pi_{(1,0)}\circ{\mbox{\boldmath$\zeta$}}} \colon I\subset\Real\to\Real\times\Tan Q$ 
is a curve solution to equation \eqref{neuq}.
\een
\label{4.26}
\end{teor}
\begin{proof}
The proof follows the same pattern as in Theorem \ref{eqteorema}.
\\ \qed\end{proof}

\begin{remark}{\rm 
If the Lagrangian is almost-regular, then there exists
${\rm h}_0\in\Cinfty(\mathcal{P})$ such that 
${\rm F}\Lag_0^*{\rm h}_0=E_\Lag$, where
${\rm F}\Lag_0\colon\Real\times\Tan Q\to \mathcal{P}$ is defined by
$\jmath_0\circ{\rm F}\Lag_0={\rm FL}$. 
Now, taking $\omega_0=\jmath_0^*\omega$, the set
$(\mathcal{P},\d t,\omega_0;{\rm h}_0)$ is a precosymplectic 
Hamiltonian system which is called the {\sl canonical Hamiltonian system}
associated with $(\Real\times\Tan Q,\Lag)$.
In particular, the equations equivalent to (\ref{celmh1}) are
$$
\inn(X_0)\d t=1\quad ,\quad
\inn(X_0)(\jmath_0^*\omega_0) =
\d {\rm h}_0-R_0({\rm h}_0)\d t \  .
$$
where $X_0\in\vf({\cal P})$. The existence of such a vector field
$X_0$ solution to the above equations is not
assured except, perhaps, on a submanifold 
$\mathcal{P}_f\hookrightarrow \mathcal{P}$
to which it is tangent. Moreover, the solution is not unique.
Furthermore, although Reeb vector fields $R$ are not uniquely defined by
equations \eqref{rvf} for this case,
the dynamics and the constraint algorithm are independent 
of the selected Reeb vector field $R_0$.
(Details on the construction of the canonical Hamiltonian formalism for
almost-regular nonautonomous Lagrangians and a deeper study on the constraint algorithms 
and the equivalence between both formalisms
can be found, for instance, in
\cite{CF-93,CLM-94,HL-dstd,LMM-96b,LMMMR-02,Vig-00}).
}\end{remark}

\section{Unified Lagrangian-Hamiltonian formalism for nonautonomous systems}

This section is devoted to explain the extension of the 
Skinner-Rusk unified formalism, which was presented in
Section \ref{SRuf}, to the case of nonautonomous dynamical systems, using the cosymplectic setting
(see \cite{BEMMR-2008,CMC-2002,GM-05}
for different but equivalent approaches).
This extended formalism is quite similar to the autonomous case.

\subsection{Extended unified bundle. Unified nonautonomous formalism}

\begin{definition}
Let $Q$ be a $n$-dimensional differentiable manifold.
The \textbf{extended unified bundle} or \textbf{extended Pontryagin bundle} is
 ${\mathcal M}:=
 \Real\times\Tan Q \times_Q\Tan^*Q$, and has natural projections
$$
\kappa_1\colon{\mathcal M} \to\Real \times\Tan Q \quad , \quad
\kappa_2\colon{\mathcal M} \to\Real \times\Tan^*Q \quad , \quad
\kappa_0\colon{\mathcal M} \to\Real \times Q 
$$
\end{definition}

Its natural coordinates are $(t,q^i,v^i,p_i)$. 

\begin{definition}
A curve ${\bf c}\colon\Real\rightarrow{\cal M}$
is \textbf{holonomic} on ${\cal M}$ if
$\kappa_1\circ{\bf c}\colon\Real\to\Real\times\Tan Q$ is holonomic.

A vector field $\Gamma \in\vf({\cal M})$  is a \textbf{holonomic vector field} on ${\cal M}$
if its integral curves  are holonomic on ${\cal M}$. 
\end{definition}

The coordinate expressions of holonomic curves and vector fields on  ${\cal M}$ are the following
\beann
{\bf c}(t)=\Big(s(t),q^i(t),\frac{dq^i}{d t}(t),p_i(t),\Big) \ , \\
\Gamma= g\derpar{}{s}+v^i\derpar{}{q^i}+F^i\derpar{}{v^i}+G_i\derpar{}{p_i}  \ .
\eeann

The extended unified bundle ${\cal M}$ is endowed with the following canonical structures:

\begin{definition}
\ben
\item
The \textbf{coupling function} on ${\mathcal M}$
is the map ${\rm C}\colon{\mathcal M}\to\Real$  defined by:
$$
\begin{array}{ccccl}
{\rm C} & : & {\mathcal M}=\Real
\times\Tan Q \times_Q\Tan^*Q& \longrightarrow & \Real \\
\noalign{\medskip} & & (t, q,v_q,\xi_q)=(t,q^i,v^i,p_i)
 & \mapsto &  \langle v_q,\xi_q\rangle=v^ip_i
\end{array} \ .
$$
\item
If $\theta\in\df^1(\Real\times\Tan^*Q)$ and $\omega\in\df^2(\Real\times\Tan^*Q)$
are the canonical forms on $\Real\times\Tan^*Q$,
then the canonical forms on ${\mathcal M}$ are
$$
\Theta_{\mathcal M}:=\kappa_2^*\,\theta\in\df^1({\mathcal M})
\quad , \quad
\Omega_{\mathcal M}:=\kappa_2^*\,\omega=-\d\Theta_{\mathcal M}\in\df^2({\mathcal M}) \ ,
$$
We denote also by $\d t$ the pull-back to ${\mathcal M}$ of the canonical $1$-form on $\Real$.
\een
\end{definition}

\begin{definition}
Given a Lagrangian function $\Lag\in\Cinfty \left(\Real \times \Tan Q \right)$,
if ${\mathfrak L}=\kappa_1^*\Lag$, the \textbf{Hamiltonian function} is defined as
$$
{\rm H}={\rm C}-{\mathfrak L}\in\Cinfty({\mathcal M}) \ .
$$
\end{definition}

In coordinates, we have
$$
\Theta_{\cal M}=p_i\d q^i \quad , \quad
\Omega_{\cal M}=\d q^i\wedge\d p_i\quad , \quad
{\rm H}=v^ip_i-{\mathfrak L}(t,q^i,v^i) \ .
$$

The triple $({\cal M},\d t,\Omega_{\cal M})$ is a precosymplectic manifold, since
$\dst\ker\Omega_{\cal M}=\left\langle\derpar{}{t},\derpar{}{v^i}\right\rangle$,
where the Reeb vector field can be taken to be
$\dst R=\derpar{}{t}$;
and hence
$({\cal M},\d t,\Omega_{\cal M},{\rm H})$ is a precosymplectic Hamiltonian system.
The {\sl\textbf{dynamical problem}} for this system consists in finding
an evolution vector field $X_{\rm H}\in\vf({\cal M})$,
which is a solution to the equations
\begin{equation}
\label{Mhamilton-contact-eqs0}
\inn(X_{\rm H})\d t=1 \quad , \quad
\inn(X_{\rm H})\Omega_{\cal M}=\d{\rm H}-R({\rm H})\d t \ ,
\end{equation}
and then the integral curves ${\bf c}\colon\Real\to{\cal W}$
 of $X_{\rm H}$ are solutions to the equations
\begin{equation}
\label{Mhamilton-contact-eqs0b}
\inn(\widetilde{\bf c})(\d t\circ{\bf c})=1 \quad , \quad
\inn(\widetilde{\bf c})(\Omega_{\cal W}\circ{\bf c}))=\left(\d{\rm H}-R({\rm H})\d t\right)\circ{\bf c} \ .
\end{equation}
But,  as $({\cal M},\Omega_{\cal M},{\rm H})$ is a precosymplectic Hamiltonian system,
these equations are not compatible on ${\cal M}$.
In fact, for an arbitrary vector field in $\vf({\cal M})$,
$$
X_{\rm H} =g\derpar{}{s}+ f^i\derpar{}{q^i}+F^i\derpar{}{v^i}+G_i\derpar{}{p_i} \ ,
$$ 
equations \eqref{Mhamilton-contact-eqs0} leads to
\beq
g= 1 \quad , \quad
f^i= v^i \quad , \quad
G_i= \displaystyle \derpar{\Lag}{q^i} \quad ,  \quad
p_i= \derpar{\Lag}{v^i}  \ . 
\label{Mfirstuni}
\eeq
which give different kind of information:
\begin{itemize}
\item
The first equation fixes the evolution parameter $s=t$.
\item
The second equations assure that $X_{\rm H}$ is a holonomic vector field on ${\cal M}$
(regardless the regularity of the Lagrangian function).
\item
The third equations allow us to determine the component functions $G_i$.
\item
The fourth equations are compatibility conditions;
that is, {\sl compatibility constraints} defining a submanifold 
${\cal M}_0\hookrightarrow{\cal M}$
where vector fields $X_{\rm H}$ solution to \eqref{Mhamilton-contact-eqs0} are defined.
As in the autonomous case, these constraints give the Legendre map and ${\cal M}_0={\rm graph}({{\rm F}\Lag})$.
\end{itemize}
In this way, we have obtained
$$
X_{\rm H} \vert_{{\cal M}_0}= 
\derpar{}{t}+v^i\derpar{}{q^i}+F^i\derpar{}{v^i}+
\derpar{\Lag}{q^i}\derpar{}{p_i}\ ,
$$
where the functions $F^i$ are still undetermined.
Now, the constraint algorithm continues by demanding the tangency of
$X_{\rm H}$ to the submanifold ${\cal M}_0$; that is, we have
to impose that
$\dst{X_{\rm H}\Big(p_i- \derpar{\Lag}{v^i}\Big)}\Big\vert_{_{{\cal M}_0}} = 0$,
and then we obtain the equations for the functions $F^i$:
\beq
 \frac{\partial^2L}{\partial v^i\partial v^j}F^j+ \frac{\partial^2L}{\partial q^j\partial v^i}v^j+\frac{\partial^2\Lag}{\partial t\partial v^i}- \derpar{\Lag}{q^i}=0 \quad (\text{on } \mathcal{M}_0)  \, .
\label{Meluni}
\eeq
If $\Lag$ is a regular Lagrangian, these equations are compatible, and they have a unique vector field $X_{\rm H}$ which is the solution to \eqref{Mhamilton-contact-eqs0} on ${\cal M}_0$,
and the last system of equations give the dynamical trajectories
(i.e., the solutions to equations \eqref{Mhamilton-contact-eqs0b}
on ${\cal M}_0$).
If $\Lag$ is singular, equations \eqref{Meluni} can be compatible or not
and, eventually, new compatibility constraints can appear,
defining a new submanifold ${\cal M}_1\hookrightarrow{\cal M}_0$.
Then, the constraint algorithm continues by demanding 
the tangency of solutions to ${\cal M}_1$, and so on. 
In the most favorable cases,
there is a submanifold ${\cal M}_f \hookrightarrow {\cal M}_0$ 
such that there exist holonomic vector fields $X_{\rm H}\in\vf({\cal M})$ defined on ${\cal M}_0$ and tangent to ${\cal M}_f$,
which are solutions to equations \eqref{Mhamilton-contact-eqs0}
at support on ${\cal M}_f$ (they are not unique necessarily).

\subsection{Recovering the nonautonomous Lagrangian and Hamiltonian formalisms}

The analysis of the equivalence among the unified, the Lagrangian, and Hamiltonian formalisms is made again
for the hyperregular case (the regular case is locally the same.
See \cite{BEMMR-2008} for details about the singular case).

As in the autonomous case, denoting by 
$\jmath_0\colon\mathcal{M}_0\hookrightarrow\mathcal{M}$
the natural embedding, we have the projections
$$
(\kappa_1\circ\jmath_0)\colon\mathcal{M}_0\to\Real\times\Tan Q
\quad , \quad
(\kappa_2\circ\jmath_0)\colon\mathcal{M}_0\to\Real\times\Tan^*Q \ ,
$$
where $\kappa_1\circ\jmath_0$ is a diffeomorphism.
The diagram that summarizes the situation is
$$
\xymatrix{
\ & \ & {\cal M} \ar@/_1.3pc/[ddll]_{\kappa_1} \ar@/^1.3pc/[ddrr]^{\kappa_2} & \ & \ \\
\ & \ & {\cal M}_0={\rm graph}\,({\rm F}\Lag) \ar[dll]_{\kappa_1\circ\jmath_0} \ar[drr]^{\kappa_2\circ\jmath_0} \ar@{^{(}->}[u]^{\jmath_0} & \ & \ \\
\Real\times\Tan Q\ar[rrrr]^<(0.45){{\rm F}\Lag}
& \ & \ & \ & \Real\times\Tan^*Q  \\
}
$$
Hence, functions, differential forms, and vector fields on ${\cal M}$ 
tangent to ${\cal M}_0$ can be restricted to ${\cal M}_0$, and then
they can be translated to the Lagrangian side 
by using the diffeomorphism $\kappa_1\circ\jmath_0$, 
and to the Hamiltonian side using the Legendre map and the projection 
$\kappa_2$.

Therefore, if ${\bf c}(t)=(t,q^i(t),v^i(t),p_i(t))$ is a solution to equation \eqref{Mhamilton-contact-eqs0},
and hence it is an integral curve of the vector field
$\,X_{\rm H}$ solution to equations \eqref{Mhamilton-contact-eqs0},
then \eqref{Meluni} leads to
\beq
 \frac{d}{dt}\Big(\derpar{\Lag}{v^i}\circ{\bf c}\Big)= \derpar{\Lag}{q^i}\circ{\bf c} \ ,
\label{MrecuEL}
\eeq
and from equations \eqref{Mfirstuni} we obtain
\beq
 \frac{dq^i}{dt}= v^i \quad , \quad
 \frac{dp_i}{dt}= \derpar{\Lag}{q^i}\circ{\bf c}= -\derpar{{\rm H}}{q^i}\circ{\bf c} \quad ,  \quad
p_i= \derpar{\Lag}{v^i}\circ{\bf c}  \ . 
\label{MrecuH}
\eeq
From the first group of these equations, together with \eqref{MrecuEL},
we recover the Euler--Lagrange equations for the curves 
${\bf c}_\Lag(t)=(t,q^i(t),v^i(t))$.
In addition, $\dst \derpar{\Lag}{q^i}=-\derpar{{\rm H}}{q^i}$, and hence
the second group of equations \eqref{MrecuH} is
$$
 \frac{dp_i}{dt}= -\derpar{{\rm H}}{q^i}\circ{\bf c}
$$
and, using the local expression of ${\rm H}$ and the first group of equations \eqref{MrecuH}, we get
$$
\derpar{{\rm H}}{q^i}\circ{\bf c}=v^i=\frac{dq^i}{dt} \ ;
$$
but, using  the Legendre map (that is, the third group of equations \eqref{Mfirstuni})
we have that ${\rm H}={\rm F}\Lag^*{\rm h}$, and these last equations become the {\sl Hamilton equations}
for the curves ${\bf c}_{\rm h}(t)=(t,q^i(t),p_i(t))$.

Every curve ${\bf c}\colon \Real\to{\cal M}$,
taking values on ${\cal M}_0$ can be viewed as
${\bf c}=({\bf c}_\Lag, {\bf c}_{\rm h})$, where
${\bf c}_\Lag=\kappa_1\circ{\bf c}\colon \Real \to\Real\times\Tan Q$
and ${\bf c}_{\rm h}={\rm F}\Lag\circ {\bf c}_\Lag\colon\Real\to\Real\times\Tan^*Q$.
Thus, we have proved that:

\begin{teor}
\label{eqM}
If ${\bf c}\colon\Real\to{\cal M}$,
with  ${\rm Im}\,{\bf c}\subset{\cal M}_0$,
is a curve fulfilling equations \eqref{Mhamilton-contact-eqs0b}, then
${\bf c}_\Lag$ is the lift to
$\Real\times\Tan Q$ of the projected curve
${\bf c}_o=\kappa_0\circ{\bf c}\colon\Real\to\Real\times Q$ 
(that is, ${\bf c}_\Lag$ is a holonomic curve),
and it is a solution to equation \eqref{neuq},
where $E_\Lag\in\Cinfty(\Real\times\Tan Q)$ is such that ${\rm H}=\kappa_1^*E_\Lag$.
Moreover, the curve  
${\bf c}_{\rm h}=\kappa_2\circ{\bf c}={\rm F}\Lag\circ{\bf c}_\Lag$
is a solution to equation \eqref{celmh2},
where ${\rm h}\in\Cinfty(\Real\times\Tan^*Q)$ is such that 
${\rm H}=\kappa_1^*{\rm h}$.

Conversely, if ${\bf c}_o\colon\Real\to Q$ is a curve such that
$\widetilde {\bf c}_o\equiv{\bf c}_\Lag$ is a solution to equation \eqref{neuq}, then the curve
${\bf c}=({\bf c}_\Lag,{\rm F}\Lag\circ{\bf c}_\Lag)$
is a solution to equation \eqref{Mhamilton-contact-eqs0b}
and ${\rm F}\Lag\circ{\bf c}_\Lag$
is a solution to equation \eqref{celmh2}.
\end{teor}

The curves ${\bf c}\colon\Real\to{\cal M}$ which are
solution to equation \eqref{Mhamilton-contact-eqs0b}
are the integral curves of the holonomic vector field $X_{\rm H}\in\vf({\cal M})$ 
which is the solution to equations \eqref{Mhamilton-contact-eqs0},
the curves ${\bf c}_\Lag\colon\Real\to\Real\times\Tan Q$
are the integral curves of the holonomic vector field 
$\Gamma_\Lag\in\vf(\Real\times\Tan Q)$ which is the
solution to equations \eqref{nelm}, and  the curves ${\bf c}_{\rm h}\colon\Real\to\Real\times\Tan^*Q$
are the integral curves of the vector field ${\cal E}_{\rm h}\in\vf(\Real\times\Tan^*Q)$ which is the
solution to equations  \eqref{celmh1}.
Therefore:

\begin{teor}
\label{eqL2}
Let $X_{\rm H}\in\vf({\cal M})$ be the solution to equations \eqref{Mhamilton-contact-eqs0} (on ${\cal M}_0$),
which is tangent to  ${\cal M}_0$.  Then:

The vector field $\Gamma_\Lag\in\vf(\Real\times\Tan Q)$, defined by
$\Gamma_\Lag\circ\kappa_1=\Tan\kappa_1\circ X_{\rm H}$,
is the {\sc sode} solution to equations \eqref{nelm},
where $E_\Lag\in\Cinfty(\Tan Q)$ is such that ${\rm H}=\kappa_1^*E_\Lag$.

The vector field ${\cal E}_{\rm h}\in\vf(\Real\times\Tan^*Q)$, defined by
${\cal E}_{\rm h}\circ\kappa_2=\Tan\kappa_2\circ X_{\rm H}$,
is the solution to equations \eqref{celmh1},
where ${\rm h}\in\Cinfty(\Real\times\Tan^*Q)$ is such that 
${\rm H}=\kappa_1^*{\rm h}$.
Furthermore ${\rm F}\Lag_*\Gamma_\Lag={\cal E}_{\rm h}$.
\end{teor}

\section{Symmetries of regular nonautonomous dynamical systems}

We restrict the study of symmetries of regular nonautonomous dynamical system
to the case in which the phase space is an
almost-canonical cosymplectic manifold; i.e., $\Real\times M$,
with $M$ a symplectic manifold, and, in particular, to the canonical case
$M=\Tan^*Q$.
Remember that, in this case, the natural projection
$\pi_\Real\colon\Real\times M\to\Real$ defines
a global canonical coordinate $t$.

This study is inspired in the ideas introduced mainly in \cite{albert} and follows
the scheme given in Section \ref{secsym}.
(See also \cite{SC-81} for a very interesting description of symmetries in the Lagrangian formalism).

\subsection{Symmetries of nonautonomous Hamiltonian systems. Noether Theorem}

Let $(\Real\times M,\eta,\omega;h)$ be a regular nonautonomous (almost-canonical) Hamiltonian system, and ${\cal E}_h\in\vf(\Real\times M)$
its evolution vector field.

\begin{definition}
A function $f \in \Cinfty (\Real\times M)$ is a \textbf{conserved quantity}
or a \textbf{constant of motion} if
$$
\Lie ({\cal E}_h) f = 0 \ ;
$$
\end{definition}

 \begin{definition}
A \textbf{dynamical symmetry} of the system $(\Real\times M,\eta,\omega;h)$
 is a diffeomorphism $\Phi\colon\Real\times M \to\Real\times M$ satisfying that
$$
\Phi_*{\cal E}_h={\cal E}_h \ .
$$
\end{definition}

\begin{definition}
An \textbf{infinitesimal dynamical symmetry} of the system $(\Real\times M,\eta,\omega;h)$
is a vector field $Y\in\vf(M)$ such that  
the local diffeomorphisms  generated by its flux
are dynamical symmetries of the system; that is,
$$
\Lie(Y){\cal E}_h=[Y,{\cal E}_h]=0 \ .
$$
 \end{definition}

It is immediate to check that the results obtained in Section \ref{dynsym}
for the dynamical symmetries of autonomous Hamiltonian systems also hold in this case.

\begin{definition}
\label{ksns}
 A \textbf{cosymplectic Noether symmetry} of the system $(\Real\times M,\eta,\omega;h)$ is a diffeomorphism
$\Phi\colon \Real\times M\to \Real\times M$
satisfying the following conditions:

(a)\quad $\Phi^*\omega=\omega$, \quad
(b)\quad $\Phi^*t=t$, \quad
(c)\quad $\Phi^*h=h$.

If the cosymplectic structure is exact, and $\omega=-\d\theta$,
then a cosymplectic Noether symmetry
is said to be \textbf{exact} if $\Phi^*\theta=\theta$.
\end{definition}

\begin{definition}
\label{cosymsym}
An \textbf{infinitesimal cosymplectic Noether symmetry}
of the system $(\Real\times M,\eta,\omega;h)$ is a vector field $Y\in\vf( \Real\times M)$ whose local flux are local cosymplectic Noether symmetries;
that is, it satisfies that:

(a)\quad $\Lie(Y)\omega=0$, \quad
(b)\quad $\Lie(Y)t=\inn(Y)\eta=0$, \quad
(c)\quad $\Lie(Y)h=0$.

If the cosymplectic structure is exact,
an infinitesimal cosymplectic Noether symmetry
is said to be \textbf{exact} if $\Lie(Y)\theta=0$.
\end{definition}

For the Reeb vector field, we have the following property:

\begin{prop}
\label{lematec}
\ben
\item
If $\Phi\in{\rm Diff}\, (\Real\times M)$ is a cosymplectic Noether symmetry, then $\Phi_*R=R$.
\item
If $Y\in\vf(\Real\times M)$
is an infinitesimal cosymplectic Noether symmetry,
then $[Y,R]=~0$.
\een
\end{prop}
\begin{proof}
\ben
\item
As $\Phi$ is a  cosymplectic Noether symmetry then
$\Phi^*\omega=\omega$ and $\Phi^*\eta=\eta$; then,
taking into account the conditions \eqref{rvf} that characterize the Reeb vector field, we have
\beann
\inn(\Phi_*R)\eta&=&\inn(R)(\Phi^{-1})^*\eta=\inn(R)\eta=1 \ , \\
\inn(\Phi_*R)\omega&=&\inn(R)(\Phi^{-1})^*\omega=\inn(R)\omega=0 \ ,
\eeann
and the result is a consequence of the unicity of the Reeb vector field,
and so $\Phi_*R=R$.
\item
We have that
\beann
\inn([Y,R])\omega=
\Lie(Y)\inn(R)\omega-\inn(R)\Lie(Y)\omega=0 & \Longleftrightarrow &
[Y,R]\in\ker\,\omega\ ,
\\
\inn([Y,R])\eta =
\Lie(Y)\inn(R)\eta-\inn(R)\Lie(Y)\eta=0
& \Longleftrightarrow & [Y,R]\in\ker\,\eta\ , \eeann 
and then
$[Y,R]\in\ker\,\omega\cap\ker\,\eta=\{ 0\}$ (because $\eta\wedge\omega^n$ is a volume form).
\een
\qed\end{proof}

\begin{remark}{\rm 
\begin{itemize}
 \item
The condition $\Phi^*t=t$ means that
cosymplectic Noether symmetries generate transformations along
the fibers of the projection
 $\pi_\Real\colon \Real\times M\longrightarrow\Real$; that
is, they leave the fibers of the projection
 $\pi_\Real\colon \Real\times M\longrightarrow\Real$ invariant or,
 what means the same thing,
$\pi_\Real\circ \Phi=\pi_\Real$.
\item
In the case of infinitesimal cosymplectic Noether symmetries, the
analogous condition is $\inn(Y)\d t=0$, which means that $Y$ has
the local expression
$\displaystyle Y=f^i\frac{\partial}{\partial x^i}+g_i\frac{\partial}{\partial y_i}$.
This means that $Y$ is tangent to the fibers of the projection
 $\pi_\Real$. Thus, these infinitesimal symmetries only generate 
transformations along these fibers, or, what means the same thing, 
the local flux of the generators $Y$ leave the fibers of the projection
 $\pi_\Real$ invariant.
 Furthermore, as a consequence of the above proposition,
 and taking into account that $\displaystyle R=\derpar{}{t}$,
in this local expression for $Y$
the component functions $f^i,g_i$ do not depend on the coordinate $t$.
\end{itemize}
}\end{remark}

\begin{prop}
\ben
\item
If $\Phi\in{\rm Diff}\, (\Real\times M)$ is a cosymplectic Noether symmetry, then it is a
dynamical symmetry.
\item
If $Y\in\vf(\Real\times M)$ is an infinitesimal cosymplectic Noether symmetry, then it is
an infinitesimal dynamical symmetry.
\een
\label{simdin2}
\end{prop}
\begin{proof}
\ben
\item
As $\Phi$ is a  cosymplectic Noether symmetry, then
$\Phi^*\omega=\omega$, $\Phi^*\eta=\eta$, $\Phi^*h=h$,
and $\Phi_*R=R$; therefore, from 
$0=\inn({\cal E}_h)\omega-\d h+(R(h))\eta$ we obtain that
\beann
0&=&\Phi^*[\inn({\cal E}_h)\omega-\d h+(R(h))\eta]=
\inn(\Phi_*^{-1}{\cal E}_h)(\Phi^*\omega)-\Phi^*\d h
+((\Phi_*^{-1}R)(\Phi^*h))(\Phi^*\eta)]
\\ &=&
\inn(\Phi_*^{-1}{\cal E}_h)\omega-\d h+(R(h))\eta \ ,
\eeann
and from $0=\inn({\cal E}_h)\eta-1$ we get
$$
0=\Phi^*(\inn({\cal E}_h)\eta-1)=\inn(\Phi_*^{-1}{\cal E}_h)(\Phi^*\eta)-1=
\inn(\Phi_*^{-1}{\cal E}_h)\eta-1 \ ;
$$
but as the system is regular, the evolution vector field solution is unique,
then $\Phi_*^{-1}{\cal E}_h={\cal E}_h$, and hence the result holds.
\item
The result for infinitesimal cosymplectic Noether symmetries follows
from the above item using the local uniparametric group of diffeomorphisms
associated to their fluxes
(or also, using the properties \eqref{rvf} that characterize the
Reeb vector field and reasoning as in the first item).
\een
\qed \end{proof}

In addition, it is immediate to prove that, if $Y_1,Y_2\in\vf(\Real\times M)$
are infinitesimal cosymplectic Noether symmetries, then so is $[Y_1,Y_2]$.

As in the above chapters,
in order to establish Noether's theorem for non-autonomous systems, 
we will stick to the case of infinitesimal Noether symmetries:

\begin{prop}
Let $Y\in\vf(\Real\times M)$ be an
infinitesimal cosymplectic Noether symmetry.
Then, for every $p\in\Real\times M$, there is an open
neighbourhood $U\ni p$, such that:
\begin{enumerate}
\item
There exist ${\cal F}\in\Cinfty(U)$,
which is unique up to a constant function, such that
\begin{equation}
 \inn(Y)\omega=\d {\cal F} \quad , \quad \mbox{\rm (on $U$)}\;.
 \label{funo}
 \end{equation}
\item
There exist $\zeta\in\Cinfty(U)$, verifying that
$\Lie(Y)\theta=\d\zeta$, on $U$; and then
$$
{\cal F}=\inn(Y)\theta-\zeta \quad , \quad \mbox{\rm (up to a
constant function on $U$)}\ 
$$
Then, for exact infinitesimal cosymplectic Noether
symmetries we have that
${\cal F}=\inn(Y)\theta$ (up to a constant function).
 \label{fdos2}
\end{enumerate}
 \label{structure2}
\end{prop}
\begin{proof}
Recalling the conditions set out in Definition \ref{cosymsym}:
\begin{enumerate}
\item
It is a consequence of the Poincar\'e Lemma and the condition
$$
0=\Lie(Y)\omega=\inn(Y)\d\omega+\d\inn(Y)\omega=\d\inn(Y)\omega\ .
$$
\item
We have that
$$
\d\Lie(Y)\theta=\Lie(Y)\d\theta=-\Lie(Y)\omega=0 
$$
and hence $\Lie(Y)\theta$ are closed forms. Therefore, by the
Poincar\'e Lemma, there exists $\zeta\in\Cinfty(U)$, verifying
that $\Lie(Y)\theta=\d\zeta$, on $U$. Furthermore, as
(\ref{funo}) holds on $U$, we obtain that
$$
\d\zeta=\Lie(Y)\theta= \d\inn(Y)\theta+\inn(Y)\d\theta=
\d\inn(Y)\theta-\inn(Y)\omega=\d \{\inn(Y)\theta-{\cal F}\}
$$
and the result holds.
\end{enumerate}
\qed \end{proof}

\begin{remark}{\rm 
Observe that, using Darboux coordinates on $\Real\times M$,
the item 2 of Proposition \ref{structure2} tells us that the conserved quantities
associated with  infinitesimal cosymplectic Noether symmetries
do not depend on the coordinate $t$
(since the generators of these symmetries,
the vector fields $Y$, neither depend on them).
}\end{remark}

Finally, the classical Noether's Theorem can be stated
for these kinds of symmetries as follows:

\begin{teor}
\label{Nthsec}
{\rm (Noether's Theorem):}
 If $Y\in\vf(\Real\times M)$ is an infinitesimal cosymplectic Noether symmetry then,
for every $p\in \Real\times M$, there is an open neighbourhood
$U\ni p$ such that the function ${\cal F}=\inn(Y)\theta-\zeta$
is a conserved quantity.
\end{teor}
\begin{proof}
From (\ref{funo}) one obtains
\beann
\Lie({\cal E}_h) {\cal F} &=&
\inn ({\cal E}_h) \d{\cal F} =
\inn ({\cal E}_h) \inn (Y)\omega =
 - \inn(Y) \inn ({\cal E}_h)\omega
\\ &=&   - \inn(Y)\d h +\inn(Y)((R(h))\eta)=
 - \Lie(Y)h +(R(h))\inn(Y)\eta=0 \ ,
\eeann
because $Y$ is an infinitesimal cosymplectic Noether symmetry.
\\ \qed \end{proof}

\bigskip
As a final particular situation, consider now the case of the canonical model $\Real\times\Tan^*Q$.
If $\varphi\colon Q\to Q$ is a diffeomorphism,
we construct the diffeomorphism 
$\Phi:=({\rm {\rm Id}_{\mathbb R}},\Tan^*\varphi) \colon
\Real\times\Tan^*Q\longrightarrow \Real\times\Tan^*Q$,
where $\Tan^*\varphi\colon\Tan^*Q\to\Tan^*Q$
is the canonical lift of $\varphi$ to $\Tan^*Q$.
Then $\Phi$ is said to be the {\sl \textbf{canonical lift}} of $\varphi$ to $\Real\times\Tan^*Q$.
Any transformation $\Phi$ of this kind is called a {\sl \textbf{natural transformation}} of $\Real\times\Tan^*Q$.

In the same way, given a vector field $Z\in \vf(Q)$
we can define its {\sl \textbf{complete lift}}
to $\Real\times\Tan^*Q$ as the vector field
$Y\in\vf(\Real\times\Tan^*Q)$
whose local flux is the canonical lift of 
the local flux of $Z$ to $\Real\times\Tan^*Q$; 
that is, $Y=Z^*$,
where $Z^*$ denotes the complete lift of $Z$ to $\Real\times\Tan^*Q$.
Any vector field $Y$ of this kind is called a {\sl \textbf{natural infinitesimal transformation}} of $\Real\times\Tan^*Q$.

Then we define:

\begin{definition}
Let $(\Real\times\Tan^*Q,\eta,\omega;h)$ be a regular nonautonomous Hamiltonian system.

A symmetry (resp. cosymplectic Noether symmetry) $\Phi\in\Cinfty(\Real\times\Tan^*Q)$
is said to be \textbf{natural} if $\Phi$ is a natural transformation of $\Real\times\Tan^*Q$.

An infinitesimal symmetry (resp. infinitesimal cosymplectic Noether symmetry) $Y\in\vf( \Real\times\Tan^*Q)$
is said to be \textbf{natural} if $Y$ is a natural infinitesimal transformation of $\Real\times\Tan^*Q$.
\end{definition}

\begin{remark}{\rm 
The study of symmetries and conserved quantities in the Lagrangian formalism of nonautonomous dynamical systems
follows the same patterns as in the canonical Hamiltonian formalism, considering the system
$(\Real\times\Tan Q,\eta_\Lag,\omega_\Lag,E_\Lag)$
and and bearing in mind the development made in Section \ref{lfLsNt}
(see also \cite{Cr-77,SC-81}).
}\end{remark}

\section{Other geometrical settings for nonautonomous dynamical systems}

As we have said in the beginning of this chapter, there are other ways to describe geometrically
time-dependent dynamical systems.
Next we briefly review two of the best known of them.

\subsection{Contact formulation}

(See \cite{AM-78,CLM-94,CPT-hctd,EMR-gstds,EMR-sdtc,SC-81,SCC-84} for more details).

The manifold where this formalism is developed is the same as in the cosymplectic case;
that is $\Real\times M$, where $(M,\Omega)$ is a symplectic manifold.
In particular, it is usual to consider the canonical case, where  $M=\Tan^*Q$ (Hamiltonian formalism),
or $M=\Tan Q$ (Lagrangian formalism),
and $\Real\times Q$ is the configuration space of the system.

Thus, in the Lagrangian description of this contact formulation,
the dynamics takes place on the manifold $\Real\times\Tan Q$.
Then, given an admissible time-dependent Lagrangian function
${\cal L} \in\Cinfty(\Real\times\Tan Q)$
(then  ${\rm rank}\,\omega_\Lag=2n$),
we define the so-called {\sl Poincar\' e-Cartan's forms} associated with $\Lag$ as
$$
\mbox{\boldmath $\Theta$}_{\cal L}:=\theta_{\cal L}-E_\Lag\,\d t
\quad , \quad
\mbox{\boldmath $\Omega$}_\Lag:=-\d\Theta_{\cal L} = \omega_{\cal L}+\d E_\Lag\wedge\d t \ ;
$$
The expression in coordinates of $\mbox{\boldmath $\Theta$}_{\cal L}$ is
$$
\mbox{\boldmath $\Theta$}_{\cal L}= \derpar{{\cal L}}{v^i}\d q^i -
\Big(v^i \derpar{{\cal L}}{v^i}-{\cal L}\Big)\,\d t \ .
$$

If $\Lag$ is a regular Lagrangian, then the pair  $(\Real\times\Tan Q,\mbox{\boldmath $\Omega$}_\Lag)$
is said to be a {\sl regular time-dependent Lagrangian system}.
The dynamical Lagrangian equations for a vector field 
$X_{\cal L}\in\vf(\Real\times\Tan Q)$ are
\beq
\inn(X_{\cal L})\mbox{\boldmath $\Omega$}_\Lag = 0
\quad , \quad
\inn(X_{\cal L})\d t = 1 \ ,
\label{contacteq}
\eeq
and the equations for the integral curves ${\bf c}\colon\Real\to\Real\times\Tan Q$ of $X_{\cal L}$ are
$$
\inn(\widetilde{\bf c})(\mbox{\boldmath $\Omega$}_\Lag\circ{\bf c}) = 0
\quad , \quad
\inn(\widetilde{\bf c})(\d t\circ{\bf c}) = 1 \ .
$$

These equations are compatible and determinate and,
in a local chart of coordinates on $\Real\times\Tan Q$, its unique solution is
$$
X_{\cal L} =\derpar{}{t}+ v^i \derpar{}{q^i} + 
\Bigl(\displaystyle\frac{\partial^2{\cal L}}{\partial{v^i}\partial{v^j}}\Bigr)^{-1}
\Big(\derpar{{\cal L}}{q^j} -
v^k\frac{\partial^2{\cal L}}{\partial v^k \partial v^j}\Big)\derpar{}{v^i} \ .
$$
Notice that $X_{\cal L}$ is a {\sc sode} and hence
its integral curves are solutions to the Euler--Lagrange equations.
Moreover, we obtain the additional information
$$
\frac{dE_{\cal L}}{dt}\circ{\bf c} = \derpar{{\cal L}}{t}\circ{\bf c} \ ,
$$
on every integral curve of $X_\Lag$, ${\bf c}\colon I\subseteq\Real\to\Real\times\Tan Q$.
This gives the non-conservation law of the energy.
%Thus, the Lagrangian contact and cosymplectic formulations are equivalent. 

In the Hamiltonian formalism, as in the cosymplectic formulation,
the dynamical information is given by the 
{\sl time-dependent Hamiltonian function}
${\rm h} \in\Cinfty(\Real\times\Tan^*Q)$,
which may be defined only locally.
Then, if $\omega=\pi_2^*\Omega$, where $\Omega$ is the canonical $2$-form on $\Tan^*Q$,
we define the form (perhaps only locally if ${\rm h}$
is a locally Hamiltonian function)
$$
\mbox{\boldmath $\Omega$}_ {\rm h }:= \omega + \d{\rm h} \wedge\d t \in Z^2(\Real\times\Tan^*Q) \ ,
$$
which is called the {\sl Hamilton--Cartan $2$-form} associated with ${\rm h}$.
Observe that  
$\mbox{\boldmath $\Omega$}_{\rm h}=-\d\mbox{\boldmath $\Theta$}_{\rm h}$, with
$\mbox{\boldmath $\Theta$}_{\rm h}=\theta-{\rm h}\,\d t$,
and $\theta=\pi_2^*\Theta$, where $\Theta$ is the Liouville's $1$-form on $\Tan^*Q$. Obviously $\omega=-\d\theta$.
Then, the pair $(\Real\times\Tan^*Q,\mbox{\boldmath $\Omega$}_{\rm h})$ is said to be a {\sl time-dependent regular (locally) Hamiltonian system}.
Its dynamical equations for a vector field 
$X_{\rm h}\in\vf(\Real\times\Tan^*Q)$ are
\beann
\inn(X_{\rm h)}\mbox{\boldmath $\Omega$}_{\rm h} = 0
\quad ; \quad
\inn(X_{\rm h})\d t = 1 \ ,
\eeann
and for the integral curves ${\bf c}\colon\Real\to\Real\times\Tan^*Q$ of $X_{\rm h}$, the equations are
$$
\inn(\widetilde{\bf c})(\mbox{\boldmath $\Omega$}_{\rm h}\circ{\bf c}) = 0
\quad , \quad
\inn(\widetilde{\bf c})(\d t\circ{\bf c})) = 1 \ .
$$

This system of equations is compatible and determinate.
In fact, in a chart of canonical coordinates  on $\Real\times\Tan^*Q$,
the unique solution is the vector field
$$
X_{\rm h} =\derpar{}{t}+ \derpar{{\rm h}}{p_i} \derpar{}{q^i} -
      \derpar{{\rm h}}{q^i} \derpar{}{p_i} \ .
$$
and, if ${\bf c}(s)=(t(s),q^i(s), p_i(s))$ is an integral curve of $X_{\rm h}$,
then the second equation implies that $t=s+ctn.$ and $c(t)$
is a solution to the Hamilton equations
\[
\frac{dq^i}{dt} = \frac{\partial {\rm h}}{\partial p_i}\circ{\bf c} \quad ,\quad 
\frac{dp_i}{dt}= -\frac{\partial {\rm h}}{\partial q^i}\circ{\bf c}\ .
\]
%Thus the Hamiltonian contact and cosymplectic formulations are also equivalent. 

The connection between the Lagrangian and the
Hamiltonian contact formalisms for nonautonomous systems is performed by the Legendre map
${\rm F}\Lag\colon\Real\times\Tan Q \to \Real\times\Tan^*Q$
defined in Definition \ref{legmap}.
Then, as ${\rm F}\Lag$  is a (local) diffeomorphism,
then all the geometrical objects in
$\Real\times\Tan Q$ are ${\rm F}\Lag$-projectable;
in particular we have that
$$
{\rm F}\Lag^*{\rm h}=E_\Lag \quad ; \quad
{\rm F}\Lag^*\mbox{\boldmath $\Omega$}_{\rm h} =\mbox{\boldmath $\Omega$}_{\cal L} \ ,
$$
and therefore, for the solutions
$$
{\rm F}\Lag_*X_{\cal L} = X_{\rm h} \ .
$$

Finally, the equivalence between the cosymplectic and the contact (Lagrangian) formulations is stated in the following:

\begin{prop}
Let $(\Real\times\Tan Q,\Lag)$ a time-dependent Lagrangian system.
Then equations \eqref{contacteq} and \eqref{nelm} are equivalents.
\end{prop}
\begin{proof}
From \eqref{contacteq} we have
$$
0=\inn(X_{\cal L})\mbox{\boldmath $\Omega$}_\Lag =
\inn(X_{\cal L})\omega_\Lag+ \inn(X_{\cal L})\d E_\Lag\,\d t-
\d E_\Lag\, \inn(X_{\cal L})\d t=
\inn(X_{\cal L})\omega_\Lag+ \inn(X_{\cal L})\d E_\Lag\,\d t-
\d E_\Lag \ .
$$
Now, taking the Reeb vector field $R_\Lag$, we obtain that
\beann
0&=&\inn(R_\Lag)\inn(X_{\cal L})\omega_\Lag+ \inn(X_{\cal L})\d E_\Lag\,\inn(R_\Lag)\d t-\inn(R_\Lag)\d E_\Lag=
 \inn(X_{\cal L})\d E_\Lag-\inn(R_\Lag)\d E_\Lag \\
& &  \Longleftrightarrow \ \inn(X_{\cal L})\d E_\Lag=\inn(R_\Lag)\d E_\Lag \ ,
\eeann
and coming to the first equation we obtain that
$$
\inn(X_{\cal L})\omega_\Lag=
\d E_\Lag- \inn(R_{\cal L})\d E_\Lag\,\d t \ ,
$$
which together with $\inn(X_{\cal L})\d t=1$ are equations \eqref{nelm}.
\\ \qed\end{proof}

Thus, the Lagrangian contact and cosymplectic formulations are equivalent. 
The equivalence between the corresponding Hamiltonian equations is proved in the same way
and lead to the same Hamilton equations.

\begin{remark}
({\sl Cartan's Theorem\/}): {\rm
It can be proved (see \cite{AM-78} ) that, if $(\Real\times\Tan^*Q,\mbox{\boldmath $\Omega$}_{\rm h})$
is a regular time-dependent Hamiltonian system, then
$\mbox{\boldmath $\Theta$}_{\rm h}\wedge\mbox{\boldmath $\Omega$}_{\rm h}^{\ n}$ is a volume form on $\Real\times\Tan^*Q$ 
if, and only if, the condition $X_{\rm h}({\rm h})-{\rm h}\not=0$ holds everywhere.
Then, $\mbox{\boldmath $\Theta$}_{\rm h}$ is said to be a {\sl contact form} and the pair $(\Real\times\Tan^*Q,\mbox{\boldmath $\Theta$}_{\rm h})$  is called a {\sl contact manifold}.
(See Definition \ref{definition-contact-manifold}).

In the same way, if $(\Real\times\Tan Q,\mbox{\boldmath $\Omega$}_\Lag)$ is a regular time-dependent Lagrangian system,
then $\mbox{\boldmath $\Theta$}_\Lag\wedge\mbox{\boldmath $\Omega$}_\Lag^{\ n}$ is a volume form on $\Real\times\Tan Q$
if the condition $X_\Lag(E_\Lag)-E_\Lag\equiv-\Lag\not=0$ everywhere.
Then $\mbox{\boldmath $\Theta$}_\Lag$ is a contact form, and $(\Real\times\Tan Q,\mbox{\boldmath $\Theta$}_\Lag)$ is a contact manifold.

All this justifies the name of {\sl contact formulation}.
\label{Cartan}
}\end{remark}

\subsection{Extended symplectic formulation}

(See \cite{ACI-gct,EMR-gstds,Ku-td,Ra2,Ra1} for details).

In this formulation,  the Lagrangian formalism is done
starting from the extended configuration space $\Real\times Q$
and taking its tangent bundle
$\Tan(\Real \times Q) \cong
 \Tan\Real\times\Tan Q \cong
 \Real \times \Real\times\Tan Q$
which is called the {\sl extended velocity phase space} and has the canonical projections
$$
\nu \colon\Tan(\Real\times Q)  \to \Real\times\Tan Q \quad ; \quad
\varpi \colon\Tan (\Real\times Q) \to \Real \ .
$$
Then, we have the following extensions of the
geometrical objects which are defined in  $\Real\times\Tan Q$,
\beann
E_{{\cal L}_{ext}}&:=&\nu^*E_\Lag + \varpi 
\\
{\mathbf \Theta}_{{\cal L}_{ext}}&:=&
\nu^*\Theta_{\cal L} +E_{{\cal L}_{ext}} \wedge\d t
\\
{\mathbf \Omega}_{{\cal L}_{ext}}&:=&
-\d{\mathbf \Theta}_{{\cal L}_{ext}} =\nu^*\Omega_\Lag -\d E_{{\cal L}_{ext}} \wedge\d t \ .
\eeann
Another way to obtain these geometric elements is constructing
the {\sl extended Lagrangian}
\beq
{\cal L}_{ext} := \nu^*{\cal L} +\frac{1}{2}\varpi^2\in\Cinfty(\Tan(\Real\times Q))
\label{lext}
\eeq
and using the canonical structures (the vertical endomorphism and the Liouville vector field)
of the tangent bundle $\Tan(\Real\times Q)$ to construct them.
It is obvious that the Lagrangian ${\cal L}$ is regular if, and only if,
$\Omega_{\Lag_{ext}}$ is a symplectic form and then
$(\Tan(\Real\times Q),\Omega_{{\cal L}_{ext}},E_{{\cal L}_{ext}})$
is a regular Hamiltonian system. This means that there exists a unique vector field
$X_{{\cal L}_{ext}} \in\vf(\Tan(\Real\times Q))$ 
which is the solution to the  Lagrangian equation of motion in this formalism, which is
$$
\inn(X_{{\cal L}_{ext}}){\mathbf \Omega}_{{\cal L}_{ext}} = \d E_{{\cal L}_{ext}} \ ,
$$
and whose integral curves ${\bf c}\colon\Real\to\Tan(\Real\times Q)$ are the solutions to the equation
$$
\inn(\widetilde{\bf c})({\mathbf \Omega}_{{\cal L}_{ext}}\circ{\bf c}) = \d E_{{\cal L}_{ext}}\circ{\bf c}\ .
$$

In coordinates, this solution is
$$
X_{{\cal L}_{ext}} = \derpar{}{t} - \derpar{E_{{\cal L}_{ext}}}{t}\derpar{}{\varpi}  +
v^i\derpar{}{q^i} +
\Bigl(\displaystyle\frac{\partial^2{\cal L}}{\partial{v^i}\partial{v^j}}\Bigr)^{-1}
\Big(\derpar{{\cal L}}{q^j} -
v^k\frac{\partial^2{\cal L}}{\partial v^k \partial v^j}\Big)\derpar{}{v^i} \ .
$$
Nevertheless, the real manifold of physical states is not
$\Tan(\Real\times Q)$ but $\Real \times \Tan Q$.
But as $\ker\,\rho_*=\Big\langle\displaystyle\derpar{}{\varpi} \Big\rangle$ and
$\displaystyle\Big[\derpar{}{\varpi},X_{{\cal L}_{ext}} \Big]=0$,
this vector field $X_{{\cal L}_{ext}}$ is $\rho$-projectable and, actually, we have that
$$
\rho_*X_{{\cal L}_{ext}} =X_{\cal L} \ ;
$$
Note that $X_{\cal L}$ is a {\sc sode} on $\Tan(\Real\times Q)$,
although $X_{{\cal L}_{ext}}$ is not so on $\Tan(\Real\times Q)$,
and the integral curves of $X_\Lag$ are solutions to the Euler--Lagrange equations.
Thus, the Lagrangian contact (or the cosymplectic) formulation is recovered in this way.

For the Hamiltonian formalism, we take the
{\sl extended momentum phase space}
$\Tan^*(\Real \times Q) \cong
 \Tan^*\Real\times\Tan^*Q \cong
 \Real \times \Real^*\times\Tan^*Q$
(or, more generally, $M \times \Real \times \Real^*$,
where $M$ is any symplectic manifold), and its  canonical projections
\beann
pr_1\colon\Tan^*(\Real \times Q)\to\Tan^*\Real \cong\Real \times \Real^*
&;&
\mu\colon \Tan^*(\Real \times Q) \to \Real\times\Tan^*Q
\\
pr_2\colon \Tan^*(\Real \times Q)\to \Tan^*Q
&;&
u\colon\Tan^*(\Real \times Q) \to \Real
\eeann
If $\Omega \in\df^2(\Tan^*Q)$ and $\Omega_{\Real} \in\df^2(\Tan^*\Real )$
are the canonical symplectic forms of $\Tan^*Q$ and $\Tan^*\Real $, respectively,
then $\Tan^*(\Real \times Q)$
can be endowed with the following symplectic structure
$$
{\mathbf \Omega}_{ext} := pr_2^* \Omega +pr_1^* \Omega_{\Real} \ .
$$
The dynamical information is still given by the time-dependent Hamiltonian function
${\rm h} \in \Cinfty(\Tan^*(\Real \times Q))$.
Then, we can introduce the {\sl extended time-dependent Hamiltonian function}
$$
{\rm h}_{ext} := \mu^* {\rm h} + u \in \Cinfty(\Tan^*(\Real \times Q))
$$
and it is immediate to prove that
$$
{\mathbf \Omega}_{ext} = \mu^* \omega - \d{\rm h}_{ext} \wedge\d t \ .
$$
Thus we have the symplectic (locally) Hamiltonian system
$(\Tan^*(\Real \times Q),\Omega_{ext},{\rm h}_{ext})$.
The Hamiltonian dynamical equation for $X_{ext}\in\vf(\Tan^*(\Real \times Q))$ is
$$
\inn(X_{ext}){\mathbf \Omega}_{ext} = \d{\rm h}_{ext}  \ ,
$$
and for its integral curves ${\bf c}\colon\Real\to\Tan^*(\Real\times Q)$ 
$$
\inn(\widetilde{\bf c})({\mathbf \Omega}_{ext}\circ{\bf c}) = \d{\rm h}_{ext}\circ{\bf c}  \ .
$$

The unique solution to the above equation, in a chart of natural coordinates
$(t,u,q^i,p_i)$ of $\Tan^*(\Real \times Q)$, is
\beann
X_{ext}&=&
\derpar{}{t} -
\derpar{{\rm h}_{ext}}{t}\derpar{}{u}+
\derpar{{\rm h}_{ext}}{p_i} \derpar{}{q^i} -
\derpar{{\rm h}_{ext}}{q^i} \derpar{}{p_i}
\\
&=&
\derpar{}{t} -
\derpar{(\mu^*{\rm h})}{t}\derpar{}{u}+
\derpar{(\mu^*{\rm h})}{p_i} \derpar{}{q^i} -
\derpar{(\mu^*{\rm h})}{q^i} \derpar{}{p_i} \ .
\eeann
Then the integral curves of this vector field are solutions to the  Hamilton equations,
$$
\frac{dq^i}{dt} = \derpar{\mu^*{\rm h}}{p_i} \quad ,\quad
\frac{dp_i}{dt} = - \derpar{\mu^*{\rm h}}{q^i} \quad , \quad
\frac{du}{dt} = - \derpar{\mu^*{\rm h}}{t} \ ,
$$
where the last equation shows that the energy is not conserved.

As the physical phase space is $\Real \times \Tan^*Q$ 
and not $\Tan^*(\Real \times Q))$,
the real dynamical vector field  cannot be $X_{ext}$ really, 
but another one on $\Real \times \Tan^*Q$ which must be related to it.
But, as in the Lagrangian formalism, the vector field $X_{ext}$ is $\mu$-projectable
and, in fact, we have that
$$
\mu_* X_{ext} = X_{\rm h} \ .
$$
In this way, the equivalence with the Hamiltonian contact 
(or the cosymplectic) formulation is proved.

In order to establish the equivalence between the Lagrangian and Hamiltonian formalisms
in this extended symplectic picture, we need to connect the
extended velocity and momentum phase spaces through a suitable ``Legendre map'';
that is, the fiber derivative of some Lagrangian function on $\Tan(\Real \times Q)$.
This function cannot be $\rho^*{\cal L}$ because its fiber derivative
${\rm F}(\rho^*{\cal L})$ is not a (local) diffeomorphism since
$\rho^*{\cal L}$ is not a regular Lagrangian on $\Tan(\Real\times Q)$.
The correct choice consists in taking the extended Lagrangian \eqref{lext}
and its fiber derivative
$$
{\rm F}{\cal L}_{ext} : \Tan (\Real\times Q ) \to \Tan^*(\Real \times Q) \ ,
$$
which in coordinates is
\beann
\Leg_{ext}^{\quad *}\,q^i = q^i \quad , \quad
\Leg_{ext}^{\quad *} \,p_i = \derpar{\rho^*{\cal L}}{v^i} \quad , \quad
\Leg_{ext}^{\quad *} \,t = t \quad , \quad 
\Leg_{ext}^{\quad *} \,u = \varpi  \ .
\eeann
Observe that ${\rm F}{\cal L}_{ext}$ is just an extension of $\Leg$
such that the following diagram commutes
\beann
\Tan ( \Real\times Q) & \mapping{{\cal F}{\cal L}_{ext}} &
\Tan^*(\Real \times Q)
\\
\nu \biggr \downarrow \  &          & \quad \biggr \downarrow \ \mu
\\
\Real\times\Tan Q        & \mapping{\Leg}      & \Real\times\Tan^*Q 
\eeann
and ${\cal F}{\cal L}_{ext}$ is a (local) diffeomorphism if, and only if, $\Leg$ is so
or, what is equivalent, if ${\cal L}$ is regular.
Therefore
$$
{\rm F}{\cal L}_{ext}^*{\mathbf \Omega} = {\mathbf \Omega}_{{\cal L}_{ext}}
\quad , \quad
{\rm F}{\cal L}_{ext}^*{\rm h}_{ext} = E_{{\cal L}_{ext}}
$$
and then
$$
{\rm F}{\cal L}_{{ext}\ *}X_{{\cal L}_{ext}} = X_{ext}  \ .
$$

\section{Examples}

In this last section,
let us consider the examples studied at the end of the previous chapter, 
in which the Lagrangian functions of the systems have been modified to include time dependence.

\subsection{The forced harmonic oscillator with periodic perturbation}

We recover the system analyzed in Section \ref{sho},
but now the oscillator is also submitted to an external periodic force.
In the cosymplectic formalism,
the configuration bundle is $\Real\times Q=\Real^2$, with coordinates $(t,q)$.

\subsubsection{Lagrangian formalism}

The cosymplectic Lagrangian formalism takes place on $\Real\times \Tan Q\simeq\Real^3$, with coordinates 
$(t,q,v)$, and the dynamics can be described by the Lagrangian function
$$
\Lag=\frac{1}{2}(mv^2-kq^2)+Aq\,\cos w t \quad ; \quad k,w\in\Real^+ \ .
$$
Now we have
$$
E_\Lag=\frac{1}{2}(mv^2+kq^2)-Aq\,\cos w t \quad , \quad
\theta_\Lag=mv\,\d q \quad , \quad
\omega_\Lag=m\,\d q\wedge\d v \ ,
$$
and $\Lag$ is a regular Lagrangian.
The Reeb vector field for this Lagrangian is just $\dst R_\Lag=\derpar{}{t}$.
Then, for $\displaystyle \Gamma_\Lag=\lambda\derpar{}{t}+f\derpar{}{q}+g\derpar{}{v}$, 
equations \eqref{nelm} give $\lambda=1$ and
$$
\inn(\Gamma_\Lag)\omega_\Lag=
m\,(f\, \d v-g\, \d q)=
mv\, \d v+(kq-A\,\cos w t)\, \d q=\d E_\Lag+\derpar{\Lag}{t}\,\d t\ ,
$$
which leads to
$$
f=v \quad , \quad mg=-kq+A\,\cos w t \ .
$$
So the Euler--Lagrange vector field is
$$
\Gamma_\Lag=\derpar{}{t}+v\derpar{}{q}+\frac{-kq+A\,\cos w t}{m}\derpar{}{v} \ ,
$$
and its integral curves $(q(t),v(t))$ are the solutions to
$$
\frac{dq}{dt} =v \quad , \quad m\frac{dv}{dt}=-kq+A\,\cos w t
\quad \Longrightarrow \quad m\frac{d^2q}{dt^2}=-kq+A\,\cos w t  \ ,
$$
which is the Euler--Lagrange equation for this time-dependent forced harmonic oscillator.

\subsubsection{Hamiltonian formalism}

For the Hamiltonian formalism, $\Real\times\Tan^*Q\simeq\Real^3$,
with coordinates $(t,q,p)$. First, the Legendre transformation is
$$
\Leg^*t=t  \quad , \quad \Leg^*q=q  \quad , \quad \Leg^*p=mv  \ ,
$$
which is a diffeomorphism (the Lagrangian is hyperregular).
The canonical Hamiltonian function is
$$
{\rm h}=\frac{p^2}{2m}+kq^2-Aq\,\cos w t  \ .
$$
As $\omega=\d q\wedge\d p$, the Reeb vector field is $\dst R=\derpar{}{t}$ and,
for $\displaystyle {\cal E}_{\rm h}=\Lambda\derpar{}{t}+F\derpar{}{q}+G\derpar{}{p}$,
equations \eqref{celmh1} give $\Lambda=1$ and
$$
\inn({\cal E}_{\rm h})\omega= F\, \d p-G\, \d q=
\frac{p}{m}\, \d p+(kq-A\,\cos w t )\, \d q=\d{\rm h}-R\derpar{{\rm h}}{t}\,\d t\ ,
$$
which leads to
$$
F=\frac{p}{m} \quad , \quad G=-kq+A\,\cos w t  \ .
$$
So the evolution vector field is
$$
{\cal E}_{\rm h}=\derpar{}{t}+\frac{p}{m}\derpar{}{q}-(kq-A\,\cos w t)\derpar{}{p}\ ,
$$
and its integral curves $(q(t),p(t))$ are the solutions to
$$
m\frac{dq}{dt} =p \quad , \quad
\frac{dp}{dt}=-kq+A\,\cos w t \ ,
$$
which are the Hamilton equations for this time-dependent forced harmonic oscillator.

Using the Legendre map, we check again that the Hamilton and
the Euler--Lagrange equations of the system are equivalent.
Obviously, we have that $\Leg_*\Gamma_\Lag={\cal E}_{\rm h}$.

\subsubsection{Unified Lagrangian--Hamiltonian formalism}

The extended unified bundle is ${\cal M}=\Real\times\Tan Q\times_Q\Tan^*Q\simeq\Real^4$
with coordinates $(t,q,v,p)$. The Reeb vector field is $\dst R=\derpar{}{t}$, 
the canonical forms are
$$
\Theta_{\cal M}=p\,\d q \quad , \quad
\Omega_{\cal M}=-\d\Theta_{\cal M}=\d q\wedge\d p
$$
and the Hamiltonian function is
$$
{\rm H}=pv-\frac{1}{2}(mv^2-kq^2)-Aq\,\cos w t  \ .
$$
For $\displaystyle X_{\rm H}=g\derpar{}{t}+f\derpar{}{q}+F\derpar{}{v}+G\derpar{}{p}$,
equations \eqref{Mhamilton-contact-eqs0} give $g=1 $ and
$$
\inn(X_{\rm H})\Omega_{\cal W}= f\, \d p-G\, \d q=
(kq-A\,\cos w t)\, \d q+(p-mv)\,\d v+v\, \d p=\d{\rm H}-R({\rm H})\d t\ ,
$$
which leads to
$$
f=v \quad , \quad G=-(kq-A\,\cos w t ) \quad , \quad p=mv  \ .
$$
The last equation is the constraint defining the submanifold
${\cal M}_0\hookrightarrow{\cal M}$ and gives the Legendre map.
The evolution vector field is
$$
X_{\rm H}\vert_{{\cal M}_0}=\derpar{}{t}+v\derpar{}{q}+G\derpar{}{v}-
(kq-A\,\cos w t )\derpar{}{p} \ ,
$$
and the tangency condition leads to
$$
X_{\rm H}(p-mv)=-(kq-A\,\cos w t )-Gm=0 \ \Longleftrightarrow\
G=-\frac{kq-A\,\cos w t}{m} \quad \mbox{\rm (on ${\cal M}_0$)} \ ;
$$
therefore
$$
X_{\rm H}\vert_{{\cal M}_0}=\derpar{}{t}+v\derpar{}{q}
-\frac{kq-A\,\cos w t}{m}\derpar{}{v}-(kq-A\,\cos w t)\derpar{}{p} \ .
$$
and its integral curves $(t,q(t),v(t),p(t))$ are the solutions to
$$
\frac{dq}{dt}=v \quad , \quad
m\frac{dv}{dt}=-(kq-A\,\cos w t) \quad , \quad
\frac{dp}{dt}=-(kq-A\,\cos w t) \ .
$$
Te first two equations are equivalent to
$$
m\frac{d^2q}{dt^2}=-kq+A\,\cos w t \ ,
$$
which is the Euler--Lagrange equation of the system. 
Furthermore, using the constraint $p=mv$ (the Legendre map),
the first and third equations are
$$
\frac{dq}{dt}=\frac{p}{m} \quad , \quad
\frac{dp}{dt}=-kq+A\,\cos w t \ ;
$$
which are the Hamilton equations for the system.

\subsection{A nonautonomous system with central forces}

Consider the {\sl Kepler problem} analyzed in Section \ref{Kp}
in the case where the mass of the particle subjected to the central field is not constant but a (strictly positive) function of time $m(t)$.
%\footnote{This is, for example, the case of the motion of comets when approaching the Sun, in which they lose part of their masses.}.
As in the standard case, the motion is on a plane;
hence the configuration bundle is $\Real\times Q=\Real^3$
and the coordinates are $(t,r,\phi)$, where $(r,\phi)$
are polar coordinates on the adequate open set $U\subset Q=\Real^2$.

\subsubsection{Lagrangian formalism}

The velocity phase space for the cosymplectic Lagrangian formalism is 
$\Real\times\Tan Q\simeq\Real^5$, with local coordinates 
$(t,r,\phi,v_r,v_\phi)$ and the Lagrangian function is 
$$
\Lag=\frac{1}{2}m(t)(v_r^2+r^2v_\phi^2)-\frac{K}{r} \quad , \quad K\not=0 \ ;
$$
therefore
\beann
E_\Lag&=&\frac{1}{2}m(t)(v_r^2+r^2v_\phi^2)+\frac{K}{r} \ , \\
\theta_\Lag&=&m(t)(v_r\,\d r+r^2v_\phi\,\d\phi) \ , \\
\omega_\Lag&=&m(t)(\d r\wedge\d v_r+r^2\d\phi\wedge\d v_\phi-
2rv_\phi\,\d r\wedge\d \phi)-
\dot m(t)(v_r\,\d t\wedge\d r+r^2v_\phi\,\d t\wedge\d\phi) \ ,
\eeann
where $\dst\dot m(t)=\frac{dm(t)}{dt}$.
The Lagrangian is regular and the Reeb vector field for this system is
$$
 R_\Lag=\derpar{}{t}+\frac{\dot m(t)}{m(t)}\Big(v_r\derpar{}{v_r}+v_\phi\derpar{}{v_\phi}\Big)\ .
$$
If $\displaystyle \Gamma_\Lag=\lambda\derpar{}{t}+f_r\derpar{}{r}+f_\phi\derpar{}{\phi}+
g_r\derpar{}{v_r}+g_\phi\derpar{}{v_\phi}$, 
equations  \eqref{nelm} give $\lambda=1$ and
\beann
\inn(\Gamma_\Lag)\omega_\Lag&=& 
\dot m(t)\,[-v_r\,\d r-r^2v_\phi\,\d\phi+
(v_rf_r+r^2v_\phi f_\phi)\,\d t]+
 \\
& & m(t)\,[f_r\, \d v_r+f_\phi r^2\, \d v_\phi-(g_r-2rv_\phi f_\phi)\, \d r-
(g_\phi r^2+2rv_\phi f_r)\, \d\phi] \\
&=& \dot m(t)\,(v_r^2+r^2v_\phi^2)\,\d t+
m(t)\,v_r\, \d v_r+m(t)\,r^2v_\phi\, \d v_\phi+\Big(m(t)\,rv_\phi^2-\frac{K}{r^2}\Big)\d r \\ 
&=& \d E_\Lag+\derpar{\Lag}{t}\,\d t\ ,
\eeann
which leads to
\beann
f_r=v_r \quad , \quad f_\phi=v_\phi \quad , \quad 
v_rf_r+r^2v_\phi f_\phi=v_r^2+r^2v_\phi^2 \ , \\
m(t)\,g_r=m(t)\,(2rv_\phi f_\phi-rv_\phi^2)-\dot m(t)\,v_r+\frac{K}{r^2} \quad , \quad 
g_\phi=-\frac{2v_\phi f_r}{r}-\frac{\dot m(t)}{m(t)}v_\phi   \ , 
\eeann
and then the Euler--Lagrange vector field is
$$
\Gamma_\Lag=\derpar{}{t}+v_r\derpar{}{r}+v_\phi\derpar{}{\phi}+\Big(rv_\phi^2
-\frac{\dot m(t)}{m(t)}v_r+\dst\frac{K}{m(t)\,r^2}\Big)\derpar{}{v_r}-\Big(\frac{\dot m(t)}{m(t)}v_\phi+\frac{2v_\phi v_r}{r} \Big)\derpar{}{v_\phi} \ .
$$
Then, its integral curves $(r(t),\phi(t),v_r(t),v_\phi(t))$ are the solutions to
\beann
\frac{dr}{dt} =v_r \ , \ \frac{d\phi}{dt} =v_\phi \ , \
m(t)\,\frac{dv_r}{dt}=m(t)\,rv_\phi^2-\dot m(t)\,v_r+\frac{K}{r^2}  \ , \
\frac{dv_\phi}{dt}=\frac{\dot m(t)}{m(t)}v_\phi-\frac{2}{r}v_\phi v_r &\Longrightarrow& \\
\Longrightarrow \quad m(t)\,\frac{d^2r}{dt^2}=m(t)\,r\Big(\frac{d\phi}{dt}\Big)^2-\dot m(t)\,\frac{dr}{dt}+\frac{K}{r^2}  \quad , \quad \frac{d^2\phi}{dt^2}=
-\frac{\dot m(t)}{m(t)}\,\frac{d\phi}{dt}-\frac{2}{r}\,\frac{d\phi}{dt}\,\frac{dr}{dt}
&\Longrightarrow& \\
\Longrightarrow \quad 
\frac{d}{dt}\Big(m(t)\,\frac{dr}{dt}\Big)=m(t)\,r\Big(\frac{d\phi}{dt}\Big)^2+\frac{K}{r^2}  \quad , \quad 
\frac{d}{dt}\Big(m(t)r^2\frac{d\phi}{dt}\Big)=0 \ ,  & &
\eeann
which are the Euler--Lagrange equations for this system.

There is an infinitesimal Lagrangian exact Noether symmetry 
which is the vector field $\dst Y=\derpar{}{\phi}$,  because
$\dst\Lie(Y)\d t=0$ trivially,
and a calculation similar to the one done in the example of Section \ref{Kp} leads to $\Lie(Y)\theta_\Lag=0$
and $\Lie(Y)E_\Lag=0$.
Therefore, the associated conserved quantity is
$$
f_Y=\inn\left(\derpar{}{\phi}\right)\theta_\Lag=m(t)r^2v_\phi \ ;
$$
that is, the angular momentum, as the last Euler--Lagrange equation shows.

\subsubsection{Hamiltonian formalism}

The cosymplectic Hamiltonian formalism
takes place on $\Real\times\Tan^*Q\simeq\Real^5$,
with local coordinates $(t,r,\phi,p_r,p_\phi)$. First, the Legendre transformation is
$$
\Leg^*t=t  \quad , \quad
\Leg^*r=r  \quad , \quad \Leg^*\phi=\phi  \quad , \quad
\Leg^*p_r=m(t)\,v_r  \quad , \quad \Leg^*p_\phi=m(t)\,r^2v_\phi\ ,
$$
which is a diffeomorphism because the Lagrangian is hyperregular.
The canonical time-dependent Hamiltonian function is
$$
{\rm h}=\frac{p_r^2}{2m(t)}+\frac{p_\phi^2}{2m(t)\,r^2}+\frac{K}{r} \ .
$$
As $\omega=\d r\wedge\d p_r+\d \phi\wedge\d p_\phi$,
the Reeb vector field is $\dst R=\derpar{}{t}$.
Then, the evolution vector field obtained from equations \eqref{celmh1} is
$$
{\cal E}_{\rm h}=\derpar{}{t}+\frac{p_r}{m(t)}\derpar{}{r}+\frac{p_\phi}{m(t)\,r^2}\derpar{}{\phi}+\Big(\frac{p_\phi^2}{m(t)\,r^3}+\frac{K}{r^2}\Big)\derpar{}{p_r} \ ,
$$
and its integral curves $(r(t),\phi(t),p_r(t),p_\phi(t))$ are the solutions to
$$
m(t)\,\frac{dr}{dt} =p_r \quad , \quad m(t)\,r^2\frac{d\phi}{dt} =p_\phi \quad , \quad
\frac{dp_r}{dt}=\frac{p_\phi^2}{m(t)\,r^3}+\frac{K}{r^2} \quad , \quad
\frac{dp_\phi}{dt}=0 \ ,
$$
which are the Hamilton equations for this system.

Although the Hamilton equations have the same aspect as in the
autonomous case (see Section \ref{Kp}), one can observe that,
as $m=m(t)$, combining these equations and
using the Legendre map, we obtain, in fact,
the Euler--Lagrange equations for this nonautonomous system,
and $\Leg_*\Gamma_\Lag={\cal E}_{\rm h}$.

Once again, a Hamiltonian exact Noether symmetry is given by 
the vector field $\dst Y=\derpar{}{\phi}$, since
$\dst\Lie(Y)\d t=0$, $\Lie(Y)\theta=0$,
and $\Lie(Y){\rm h}=0$.
Then, as the last Hamilton equation shows,
the associated conserved quantity is again the angular momentum
$\dst f_Y=\inn\left(\derpar{}{\phi}\right)\theta=p_\phi$.

\subsubsection{Unified Lagrangian--Hamiltonian formalism}

The extended unified bundle is ${\cal M}=\Real\times\Tan Q\times_Q\Tan^*Q\simeq\Real^7$,
with natural coordinates $(t,r,\phi,v_r,v_\phi,p_r,p_\phi)$.
The Reeb vector field $\dst R=\derpar{}{t}$,
the canonical forms are
$$
\Theta_{\cal M}=p_r\,\d r +p_\phi\,\d\phi  \quad , \quad
\Omega_{\cal M}=-\d\Theta_{\cal M}=\d r \wedge\d p_r+\d\phi \wedge\d p_\phi \ ,
$$
and the Hamiltonian function is
$$
{\rm H}=p_rv_r+p_\phi v_\phi-\frac{1}{2}m(t)(v_r^2+r^2v_\phi^2)+\frac{K}{r} \ .
$$
For $\displaystyle X_{\rm H}=f\derpar{}{t}+f_r\derpar{}{r}+
f_\phi\derpar{}{\phi}+F_r\derpar{}{v_r}+F_\phi\derpar{}{v_\phi}+G_r\derpar{}{p_r}+G_\phi\derpar{}{p_\phi}$,
equation \eqref{Mhamilton-contact-eqs0} gives $g=1$ and
\beann
\inn(X_{\rm H})\Omega_{\cal M}&=&
 f_r\, \d p_r+f_\phi\, \d p_\phi-G_r\, \d r-G_\phi\, \d\phi
\\ &=&
-\Big(\frac{K}{r^2}+m(t)\,rv_\phi^2\Big)\d r+(p_r-m(t)\,v_r)\,\d v_r
\\ & &
+(p_\phi-m(t)\,r^2v_\phi)\,\d v_\phi+v_r\, \d p_r+v_\phi\, \d p_\phi
\\ &=& \d{\rm H}-R({\rm H})\d t\ ,
\eeann
which leads to
$$
f_r=v_r \ , \ f_\phi=v_\phi \ , \
G_r=\frac{K}{r^2}+m(t)\,rv_\phi^2 \ , \ G_\phi=0 \ , \
p_r=m(t)\,v_r \ , \ p_\phi=m(t)\,r^2v_\phi \ .
$$
The last two equations are constraints defining the submanifold
${\cal M}_0\hookrightarrow{\cal M}$ which give the Legendre map.
The evolution vector field is
$$
X_{\rm H}\vert_{{\cal M}_0}=\derpar{}{t}+
v_r\derpar{}{r}+v_\phi\derpar{}{\phi}+F_r\derpar{}{v_r}+F_\phi\derpar{}{v_\phi}
+\Big(\frac{K}{r^2}+m(t)\,rv_\phi^2\Big)\derpar{}{p_r} \ ,
$$
and the tangency condition leads to
\beann
X_{\rm H}(p_r-m(t)v_r)=\frac{K}{r^2}+m(t)\,rv_\phi^2-F_rm(t)-\dot m(t)\,v_r=0
\quad \mbox{\rm (on ${\cal M}_0$)} \\ \Longleftrightarrow
F_r=\frac{K}{m(t)\,r^2}+rv_\phi^2-\frac{\dot m(t)}{m(t)}v_r \quad \mbox{\rm (on ${\cal M}_0$)} \ , \\
X_{\rm H}(p_\phi-m(t)\,r^2v_\phi)=-m(t)\,F_\phi r^2-m(t)\,2f_r rv_\phi-\dot m(t)\,r^2v_\phi=0 
\quad \mbox{\rm (on ${\cal M}_0$)} \\ \Longleftrightarrow
F_\phi=-\frac{2v_rv_\phi}{r}-\frac{\dot m(t)}{m(t)}r^2v_\phi \quad \mbox{\rm (on ${\cal M}_0$)} \ ;
\eeann
therefore
\beann
X_{\rm H}\vert_{{\cal M}_0}&=&\derpar{}{t}+
v_r\derpar{}{r}+v_\phi\derpar{}{\phi}+
\Big(rv_\phi^2+\dst\frac{K}{m(t)\,r^2}-\frac{\dot m(t)}{m(t)}v_r\Big)\derpar{}{v_r}
\\ & & -\Big(\frac{2v_rv_\phi}{r}+\frac{\dot m(t)}{m(t)}r^2v_\phi\Big)\derpar{}{v_\phi}
+\left(\frac{K}{r^2}+m(t)\,rv_\phi^2\right)\derpar{}{p_r} \ ,
\eeann
and its integral curves $(t,r(t),\phi(t),v_r(t),v_\phi(t),p_r(t),p_\phi(t))$ are the solutions to
\beann
\frac{dr}{dt}=v_r \ &,& \
\frac{d\phi}{dt}=v_\phi \ ,
\\
\frac{dv_r}{dt}=\frac{K}{m(t)\,r^2}+rv_\phi^2-\frac{\dot m(t)}{m(t)}v_r \ &,& \
\frac{dv_\phi}{dt}=-\frac{2v_rv_\phi}{r}-\frac{\dot m(t)}{m(t)}r^2v_\phi \ , 
\\
\frac{dp_r}{dt}=\frac{K}{r^2}+m(t)\,rv_\phi^2 \ &,& \
\frac{dp_\phi}{dt}=0 \ .
\eeann
Te first four equations are equivalent to
$$
m(t)\frac{d^2r}{dt^2}=
m(t)\,r\Big(\frac{d\phi}{dt}\Big)^2+\frac{K}{r^2}-\dot m(t)\frac{dr}{dt}  \quad , \quad
 \frac{d^2\phi}{dt^2}=
-\frac{2}{r}\,\frac{d\phi}{dt}\,\frac{dr}{dt}-\frac{\dot m(t)}{m(t)}r^2\frac{d\phi}{dt} \ ,
$$
which are the Euler--Lagrange equation of the system. 
Furthermore, using the constraints $p_r=m(t)\,v_r$ 
and $p_\phi=m(t)\,r^2v_\phi$ (tat is, the Legendre map),
the first, second, fifth, and sixth equations are
$$
\frac{dr}{dt}=\frac{p_r}{m(t)} \quad , \quad
\frac{d\phi}{dt}=\frac{p_\phi}{m(t)\,r^2} \quad , \quad
\frac{dp_r}{dt}=\frac{p_\phi^2}{m(t)\,r^3}+\frac{K}{r^2}  \quad , \quad
\frac{dp_\phi}{dt}=0 \ ;
$$
which are the Hamilton equations for the system.

%%%%%%%%%%%%%%%%%%%%%%%%%%%%%%%%%%%%%%%%%%%%%%%%%%%%%%%%%%%%%%%%%%%%%%%%%%%%%%%%%%%%%%%%%%%%%%%%%%%%%%%%%%%%%%%%%%

\chapter{Riemannian mechanics: Newtonian dynamical systems}
\protect\label{sdn}

In this chapter, we analyze the geometry of Newtonian dynamical systems, 
which are classical  mechanical systems on a (semi)Riemannian manifold:
they are also called {\sl purely mechanical systems}. 
This is the geometrization of the classical {\sl Newtonian Mechanics} or {\sl (semi)Riemannian mechanics}.
As we will see, a special type of these systems, the conservative ones,
are just a particular kind of the Lagrangian systems studied in the above chapter.
Some basic references are \cite{AM-78,Ar-89,CC-2005,CM-2023,CDW-87,GN-2014,Ol-02}.

As in the above chapters,
we start reviewing the basic notions and properties on connections in manifolds and Riemannian geometry,
that are needed for the development of this chapter.
Next we introduce the general features about Newtonian dynamical systems and their properties: 
we state geometrically the classical Newton equations of dynamics, we study conservation laws,
and describe some particular types of these kinds of systems.
This geometric framework is also used to study the case of systems with holonomic and nonholonomic constraints and their corresponding variational principles.
Finally, nonautonomous Newtonian systems are also briefly analyzed.

\section{Connections in manifolds. Riemannian manifolds}
\label{conexion}

(For a detailed account of these subjects and the proof of the results, see for example 
\cite{Con2001,GHV-72,Lang-95,Lee2013, Lee2018,GN-2014,Sp-72}).

\subsection{Connections and covariant derivatives}

\begin{definition}
A \textbf{connection} on a manifold $M$ is a map
$$
\nabla\colon\vf(M)\times\vf(M)\to\vf(M) \quad ,\quad (X,Y)\mapsto \nabla_X Y \ ,
$$
satisfying the following properties:
\begin{description}
\item[{\rm (i)}] $\nabla_X Y$ is $\Real$-linear in $Y$.
\item[{\rm (ii)}]  $\nabla_X Y$ is $\Cinfty(M)$-linear in $X$.
\item[{\rm (iii)}] If $f\in\Cinfty(M)$, then $\nabla_X (fY)=(\Lie(X)f)Y+f\nabla_X Y$.
\end{description}
Then $\nabla_X Y$ is said to be the \textbf{covariant derivative} of $Y$ with respect to $X$.
\end{definition}

Connections and covariant derivatives
satisfy the following properties:
\begin{itemize}
\item 
For every $p\in M$, the result of $(\nabla_X Y)(p)$ 
depends only on the value of $X$ at $p$ and 
on the value of $Y$  on a neighbourhood of $p$. 
Thus, we can calculate $(\nabla_X Y)(p)$ in a local coordinate system 
around $p$, and $\nabla_u Y$ is well-defined for $u\in\Tan_{p}M$. 
\item 
(Localization of a connection). 
 If $U\subset M$ is an open set, then there exists a connection $\nabla^{U}$ on $U$ defined as:  
$$
\nabla^{U}_X|_{U} Y|_{U}:=(\nabla_X Y)|_{U} \ .
$$
\item
Let $(U,x^i)$ be a local chart; then the vector fields 
$\displaystyle\derpar{}{x^i}$ define a basis of 
$\vf(U)$ and
$$
\nabla_{\frac{\partial}{\partial x^{i}}}\frac{\partial}{\partial x^{j}}=\Gamma_{ij}^{k}\frac{\partial}{\partial x^{k}}\ ,
$$
where $\Gamma_{ij}^{k}\in\Cinfty(U)$ are the so-called 
{\sl Christoffel symbols} of the connection $\nabla$ 
for the local basis $\left\{\displaystyle\derpar{}{x^i}\right\}$ or in the local chart $(U,x^i)$.
Then, if $\displaystyle X=X^{i}{\frac{\partial}{\partial x^{i}}}$ and 
$\displaystyle Y=Y^{j}{\frac{\partial}{\partial x^{j}}}$,
$$
(\nabla_X Y)|_{U}=(\Lie(X)Y^{j}){\frac{\partial}{\partial x^{j}}}+
X^{i}Y^{j}\Gamma_{ij}^{k}{\frac{\partial}{\partial x^{k}}}\,.
$$
Observe that this expression says that, at $p\in M$,
the covariant derivative depends only on the values of $Y$ along a curve 
$\gamma\colon(-\epsilon,\epsilon)\subset\Real\to M$ 
such that $\gamma'(0)=X_{p}$.

This also holds if we have a local basis 
$(X_i)$ for vector fields on an open set $U\subset M$, not necessarily local coordinate vector fields.
\item
If $M$ is a parallelizable manifold, and $X_{1},\ldots, X_{n}$ 
is a basis of $\vf(M)$, then any family of differentiable functions 
$\Gamma_{ij}^{k}\in\Cinfty(M)$ define a connection on $M$. 
That is, the Christoffel symbols determine the connection 
and they allow calculating the covariant derivative 
in the open set where they are defined.

In particular, if $M=\Real^n$, the standard connection is defined 
taking the Christoffel symbols equal to zero in the natural coordinate system.
\item
Let $(x^{1},\ldots, x^{n})$ and $(\tilde{x}^{1},\ldots, \tilde{x}^{n})$ be two coordinate systems on the same open set, and $\Gamma_{ij}^{k}$ and $\widetilde{\Gamma}_{\alpha\beta}^{\gamma}$ be the corresponding Christoffel symbols, then
$$
\Gamma_{ij}^{k}=\widetilde{\Gamma}_{\alpha\beta}^{\gamma}
\frac{\partial \tilde{x}^{\alpha}}{\partial x^{i}}
\frac{\partial \tilde{x}^{\beta}}{\partial x^{j}}
\frac{\partial x^{k} }{\partial \tilde{x}^{\gamma}}+
\frac{\partial^{2} \tilde{x}^{\delta} }{\partial x^{i}\partial x^{j}}\frac{\partial x^{k} }{\partial \tilde{x}^{\delta} }\ ,
$$
hence they are not the components of any tensor field.
Thus, it is a nonsense to say that a connection is null. 
The Christoffel symbols can vanish in a coordinate system and be different from zero in another one.
\item
The set of connections on a manifold is not a vector space but an affine space.
\item 
Let $\nabla^{\alpha}$ be a family  of connections on $M$ and $f_{\alpha}$ be a family of differentiable functions such that their supports are a locally finite family of subsets of $M$. Then $\nabla=f_{\alpha}\nabla^{\alpha}$ defined as 
$$
\nabla_{X}Y:=f_{\alpha}\nabla^{\alpha}_{X}Y
$$
is a connection if, and only if, $\displaystyle\sum_{\alpha}f_{\alpha}=1$.
\item
Using the above item, we can state that 
any paracompact manifold admits a connection and, conversely,
if there exists a connection on $M$, then $M$ is paracompact.
(Remember that to say that $M$ is paracompact is equivalent to say
 that every open covering has a subordinate partition of unity, 
or that $M$ is metrizable, or that every connected component has a numerable basis of open sets).
\end{itemize}

\subsection{Covariant derivative of tensor fields}

\begin{definition}
Let $(M,\nabla)$ be a differentiable manifold with a connection, 
and let $\mathcal{T}(M)$ be the algebra of tensor fields on $M$.
For $X\in\vf(M)$ the
\textbf{covariant derivative} by $X$ defined by the connection 
is the unique map $\nabla_{X}\colon\mathcal{T}(M)\to\mathcal{T}(M)$ 
verifying the following properties:
\begin{description}
\item[{\rm (i)}] It is the Lie derivative $\Lie(X)$ on $\Cinfty(M)$.
\item[{\rm (ii)}] It reduces to the well known $\nabla_{X}$ on $\vf(M)$.
\item[{\rm (iii)}] It is $\Real$-linear.
\item[{\rm (iv)}] It respects the type of tensor field; that is, it maps $\mathcal{T}_{h}^{k}(M)$ onto $\mathcal{T}_{h}^{k}(M)$.
\item[{\rm (v)}] 
{\rm (Leibniz rule)}: $\nabla_{X}(S\otimes T)=(\nabla_{X}S)\otimes T+S\otimes(\nabla_{X}T)$.
\item[{\rm (vi)}] It commutes with the internal contractions.
\end{description}
That is, $\nabla_{X}$ is a derivation of the $\Real$-algebra $\mathcal{T}(M)$, with zero degree, which commutes with the internal contractions. 
\end{definition}

Observe that, as a consequence of the definition, we have that,
for $f\in\Cinfty(M)$ and $S,T\in\mathcal{T}(M)$,
$$
\nabla_{X}(fS\otimes T)=
(\Lie(X)f)(S\otimes T)+f(\nabla_{X}(S\otimes T))=\nabla_{X}(S\otimes fT) \ .
$$

The covariant derivative has the following properties:
\begin{itemize}
\item 
For $p\in M$, the value of $(\nabla_{X} T)(p)$ only depends on $X(p)$ and the value of the tensor field $T$ on an open set containing $p$. This property has the same consequences as in the case of derivation of vector fields.
\item
The above properties defining $\nabla_{X}$ allow us to calculate $\nabla_{X}T$ for any tensor field $T$ providing we know how is the actuation of $\nabla_{X}$ on functions and vector fields. 
For example, let $X,Y\in\mathcal{T}^{1}(M)$ and $\alpha\in\Omega^{1} (M)=\mathcal{T}_{1}(M)$, then,
$$
\langle \nabla_{X}\alpha,Y\rangle= \nabla_{X}\langle\alpha,Y\rangle-\langle\alpha, \nabla_{X}Y\rangle \ .
$$     
In particular, we have that
$\nabla_{\frac{\partial}{\partial x^{i}}}\d x^{j}=-\Gamma_{ik}^{j}\d x^{k}$.
\item
In a local coordinate system $(U,x^i)$, if
$$
X=X^{j}{\frac{\partial}{\partial x^{i}}}, \,\,\, T=T_{j_{1},
\ldots,j_{s}}^{{i_{1},\ldots,i_{r}}}
\d x^{j_{1}}\otimes\ldots\otimes\d x^{j_{s}}\otimes{\frac{\partial}{\partial x^{i_{1}}}}\otimes\ldots{\frac{\partial}{\partial x^{i_{r}}}} \ ,
$$
then
$$
 \nabla_{X}T=S_{j_{1},
\ldots,j_{s}}^{{i_{1},\ldots,i_{r}}}
\d x^{j_{1}}\otimes\ldots\otimes\d x^{j_{s}}\otimes{\frac{\partial}{\partial x^{i_{1}}}}\otimes\ldots{\frac{\partial}{\partial x^{i_{r}}}}\ ;
$$
where
$$
S_{j_{1},
\ldots,j_{s}}^{{i_{1},\ldots,i_{r}}}=
X^{k}(\partial_{k}T_{j_{1},
\ldots,j_{s}}^{{i_{1},\ldots,i_{r}}}+
\sum_{n=1}^{r} T_{j_{1},
\ldots,j_{s}}^{{i_{1},\ldots,l,\dots,i_{r}}}\Gamma_{kl}^{i_{n}}-
\sum_{n=1}^{s} T_{j_{1},
\ldots,l,\dots,j_{s}}^{{i_{1},\ldots,i_{r}}}\Gamma_{k_{n}}^{l})\ .
$$
\end{itemize}

\begin{definition}
Let $T\in\mathcal{T}_{s}^{r}(M)$, $X_{i},Y\in\vf(M)$, and $\beta^{j}\in\Omega^{1}(M)$.
The tensor field $\nabla T\in\mathcal{T}_{s+1}^{r}(M)$ defined by
$$
(\nabla T)(X_{1},\ldots,X_{s},Y,\beta^{1},\ldots,\beta^{r}):=(\nabla_{Y} T)(X_{1},\ldots,X_{s},\beta^{1},\ldots,\beta^{r}) \ ,
$$
is called the {\sl \textbf{covariant differential}} of $T$.
\end{definition}

With the same notation for $T$ as in the above item, we have 
$$
 \nabla T=S_{j_{1},
\ldots,j_{s};k}^{{i_{1},\ldots,i_{r}}}
\d x^{j_{1}}\otimes\ldots\otimes\d x^{j_{s}}\otimes\d x^{k}\otimes{\frac{\partial}{\partial x^{i_{1}}}}\otimes\ldots{\frac{\partial}{\partial x^{i_{r}}}} \ ,
$$
where
$$
S_{j_{1},
\ldots,j_{s};k}^{{i_{1},\ldots,i_{r}}}=
\derpar{}{x^k}T_{j_{1},
\ldots,j_{s}}^{{i_{1},\ldots,i_{r}}}+
\sum_{n=1}^{r} T_{j_{1},
\ldots,j_{s}}^{{i_{1},\ldots,l,\dots,i_{r}}}\Gamma_{kl}^{i_{n}}-
\sum_{n=1}^{s} T_{j_{1},
\ldots,l,\dots,j_{s}}^{{i_{1},\ldots,i_{r}}}\Gamma_{k_{n}}^{l}\ .
$$
And for functions $f\in\Cinfty(M)$, we have $\nabla f=\d f$.

For a vector field $\displaystyle X=X^i{\frac{\partial}{\partial x^{i}}}$, we have that 
$$\nabla X=\Big(\derpar{X^i}{x^k}X^i+X^l\Gamma_{kl}^i\Big)\d x^k\otimes{\frac{\partial}{\partial x^{i}}}\in\mathcal{T}_{1}^{1}(U)\ .
$$
Observe that $(\nabla X)(Y)=\nabla_YX$.

\subsection{Covariant derivative along curves}

Remember that a {\sl vector field along a map}
$F\colon N\to M$, is another map $Z\colon N\to\Tan M$ 
such that $Z(q)\in \Tan_{F(q)}M$, for every $q\in N$; that is $\tau_Q\circ Z=F$. 
The set of vector fields along a map $F$ is a $\Cinfty(N)$-modulus, denoted by $\vf(F)$.

If $\gamma\colon I\subset\Real\to M$ is a curve, 
then we have the set $\vf(\gamma)$ as a $\Cinfty(I)$-modulus. 
One distinguished element of this set is the velocity vector $\gamma'\in\vf(\gamma)$. 
If $X\in\vf(M)$, then $X\circ\gamma\in\vf(\gamma)$;
but not every element of $\vf(\gamma)$ is of this kind.

If $\displaystyle Y=Y^{j}\derpar{}{x^i}$ is a vector field, 
then, by the above properties of the covariant derivative, 
$\nabla_{\gamma'(t)}Y$ is well-defined as
$$
\nabla_{\gamma'(t)}Y:=(D(Y^{j}\circ\gamma)(t){\frac{\partial}{\partial x^{i}}}|_{\gamma(t)}+\Gamma_{ij}^{k}(\gamma(t))D\gamma^{i}(t)(Y^{j}(\gamma(t)){\frac{\partial}{\partial x^{k}}}|_{\gamma(t)} \ ,
$$
where $\displaystyle D=\frac{d}{d t}$ on the functions $f:I\to\Real$, and $\gamma^{i}$ are the components of $\gamma$. Then:

\begin{definition}
Given a curve $\gamma\colon I\subset\Real\to M$,
the \textbf{covariant derivative} along $\gamma$
is the unique map $\nabla_{t}\colon\vf(\gamma)\to\vf(\gamma)$ 
such that
\begin{description}
\item[{\rm (i)}] 
$\nabla_{t}$ is $\Real$-linear.
\item[{\rm (ii)}]
 If $f\in\Cinfty(I)$ and $\mathbf{w}\in\vf(\gamma)$, then
$\nabla_{t}(f\mathbf{w})=(Df)\mathbf{w}+f\nabla_{t}\mathbf{w}$.
\item[{\rm (iii)}] 
If $X\in\vf(M)$, then $\nabla_{t}(X\circ\gamma) (t)=\nabla_{\gamma'(t)}X$.
\end{description}
The vector field along $\gamma$ given by $\nabla_{t}\mathbf{w}$ is called the
 \textbf{covariant derivative} of $\mathbf{w}$ along $\gamma$.
\end{definition}

Other notations for $\nabla_{t}\mathbf{w}$ are
$\nabla_{t}^{\gamma}\mathbf{w}=\nabla_{\gamma'(t)}\mathbf{w}=\nabla_{\frac{d}{d t}}\mathbf{w}$.

If $(U,x^i)$ is a local chart with $\gamma(I)\subset U$, 
then a basis for $\vf(\gamma)$ is given by  $\displaystyle\derpar{}{x^i}\circ\gamma$ and,
for $\mathbf{w}\in\vf(\gamma)$, then  $\mathbf{w}=w^{i}\Big(\displaystyle\derpar{}{x^i}\circ\gamma\Big)$,
with $w^{i}\in\Cinfty(I)$. Writing 
$\gamma'=\dot{q}^{i}\displaystyle\derpar{}{x^i}\circ\gamma$, 
we have that
$$
\nabla_{t}\mathbf{w}=\left(\dot{w}^{i}+(\Gamma_{ij}^{k}\circ\gamma)\dot{q}^{i}{w}^{j}\right){\frac{\partial}{\partial x^{k}}}\circ\gamma\ .
$$

Similarly as above, we can define the concept of tensor field along a map 
or along a curve, and the covariant derivative of tensor fields along a curve 
with the same rules as above. 
If we know the action of  $\nabla_{t}$ on functions and vector fields,
we can obtain the covariant derivative of any kind of tensor field 
along the curve $\gamma$.
In particular, if $ \boldsymbol{\beta}\in\Omega^{1}(\gamma)$ 
and $\mathbf{w}\in\vf(\gamma)$, then
$$
D\langle \boldsymbol{\beta},\mathbf{w}\rangle=\langle\nabla_{t} \boldsymbol{\beta},\mathbf{w}\rangle+\langle \boldsymbol{\beta},\nabla_{t}\mathbf{w}\rangle\ ,
$$
which, for $\boldsymbol{\beta}=b_{j}\d x^{j}\circ\gamma$, gives the following expression in a local chart
$$
\nabla_{t} \boldsymbol{\beta}=\left(\dot{b}_{j}-(\Gamma_{ij}^{k}\circ\gamma)\dot{q}^{i}{b}_{k}\right)\d x^{j}\circ\gamma\ .
$$

\subsection{Parallel transport along a curve. Geodesic curves}

\begin{definition}
Let $(M,\nabla)$ be a manifold with a connection
and $\gamma\colon I\subseteq\Real\to M$ a curve. 
We say that $\mathbf{w}\in\vf(\gamma)$ is \textbf{parallel along $\gamma$}
 if $\nabla_{t}\mathbf{w}=0$.
\end{definition}

The following properties hold:
\begin{itemize}
  \item
(Existence of parallel vector fields). Given $(M,\nabla)$ and 
$\gamma\colon I\to M$, if $t_o\in I$ and $\mathbf{w}_o\in\Tan _{\gamma(t_o)}M$, 
then there exists an unique $\mathbf{w}\in\vf(\gamma)$ parallel and such that $\mathbf{w}(t_o)=\mathbf{w}_o$.
  
This result comes from the linearity of the differential equation 
defining parallel vector fields along a curve $\gamma\colon I\to M$. 
We say that $\mathbf{w}$ is the {\sl parallel transport} of $\mathbf{w}_o$ 
along $\gamma$, with respect to the connection $\nabla$.
  \item 
Observe that the map $\Tan _{\gamma(t_o)}M\to\vf(\gamma)$ 
given by $\mathbf{w}_o\mapsto\mathbf{w}$ is linear and bijective.
\item
If $t\in I$, then the vector $\mathbf{w}(t)$ is called the {\sl parallel transport} of $\mathbf{w}_o$ from $\gamma(t_o)$ to $\gamma(t)$. 
Thus, for $t_o,t\in I$, we have a {\sl parallel transport operator} along $\gamma$, denoted $\tau_{t_o,t}\colon\Tan _{\gamma(t_o)}M\to \Tan _{\gamma(t)}M$, which is a linear isomorphism. 
  \item 
Let $(\mathbf{e_i})$ be a basis of $\Tan _{\gamma(t_o)}M$ 
and $\mathbf{w}_i$ the corresponding parallel transported elements 
of $\vf(\gamma)$, then:
 \begin{enumerate}
  \item 
$(\mathbf{w}_i(t))$ is a basis of $\Tan _{\gamma(t)}M$ for every $t\in I$.
  \item 
$(\mathbf{w}_i)$ is a basis of $\vf(\gamma)$ as $\Cinfty(I)$-module.
  \end{enumerate}
  \item 
For $t_o,t\in I$, consider $\tau_{t_o,t}\colon\Tan _{\gamma(t_o)}M\to \Tan _{\gamma(t)}M$ which is the parallel transport operator along $\gamma$. 
Then it is easy to show that,
  $$
  (\nabla_{t}\mathbf{w})(t_o)=\lim_{t\to t_o}\frac{\tau_{t_o,t}^{-1}(\mathbf{w}(t))-\mathbf{w}(t_o)}{t-t_o}\ ;
  $$
that is, the parallel transport determines the covariant derivative.
\end{itemize}

\begin{definition}
A curve $\gamma\colon I\to M$ is a \textbf{geodesic} of the connection $\nabla$
if $\gamma'$, as a vector field along $\gamma$, 
is parallel along $\gamma$; that is $\nabla_t{\gamma'}=0$. 
\end{definition}

\begin{itemize}
  \item
If $\gamma=(q^i(t))$ in local coordinates,
then the curve $\gamma$ is a geodesic if, and only if,
it satisfies the following second order differential equation
$$
\ddot{q}\,^k+(\Gamma_{ij}^k\circ\gamma)\dot{q}\,^i\dot{q}\,^j=0\ .
$$  
As a consequence, given $u_p\in\Tan_pM$ there is a unique geodesic with initial condition $u_p$.	
\item 
{\sl Reparametrization of geodesics}: 
Let $\gamma\colon I_1\subset\Real\to M$ be a geodesic and 
let $\varphi\colon I_2\subset\Real\to I_1$ be a diffeomorphism 
between open intervals on $\Real$. 
Then the reparametrized curve $\gamma\circ\varphi$ is a geodesic
if, and only if, $\varphi$ is an affine map. 
\item
A vector field $X\in\vf(M)$ is said to be {\sl parallel} 
if it is parallel along any curve on $M$. 
This is equivalent to each of the following expressions:
\begin{description}
  \item[i)] $\nabla_uX=0$, for every tangent vector $u$.
  \item[ii)] $\nabla_YX=0$, for every vector field $Y$. 
  \item[iii)] $\nabla X=0$.
\end{description}
\item 
The above notions can be extended to tensor fields,
so defining {\sl tensor fields parallel along a curve} 
or {\sl parallel tensor fields}.
\end{itemize}

\subsection{Torsion and curvature of a connection}

\begin{definition}
Let $(M,\nabla)$ be a manifold endowed with a connection.
The map ${\mathit T}\colon\vf(M)\times\vf(M)\to\vf(M)$ 
defined by
 $$
 {\mathit T}(X,Y):=\nabla_XY-\nabla_YX-[X,Y]
 $$
is $\Cinfty(M)$-bilinear, and then it defines a tensor field 
${\mathit T}\in{\mathcal T}_2^1(M)$ which is called 
the \textbf{torsion tensor} of the connection. 

A connection $\nabla$ is said to be \textbf{symmetric} if its torsion is zero.
\end{definition}

 The properties of the torsion tensor are:
\begin{itemize}
\item
The torsion tensor ${\mathit T}$ is antisymmetric.
 \item 
Let $(E_i)$ be a local reference frame of vector fields, $(E^i)$ its dual reference, 
$\Gamma_{ij}^k$ the corresponding Christoffel symbols,
and $[E_i,E_j]=c_{ij}^kE_k$; then
 $$
 {\mathit T}={\mathit T}_{ij}^k\, E^i\otimes E^j\otimes E_k \quad  , \quad 
{\mathit T}_{ij}^k=\Gamma_{ij}^k-\Gamma_{ji}^k-c_{ij}^k \ .
 $$ 
 Hence, $\nabla$ is symmetric if, in any system of coordinates,
the Christoffel symbols $\Gamma_{ij}^k$ are symmetric in $(i,j)$.
 \item 
It is easy to show that:
	\begin{description}
 \item[i)] 
Two connections on $M$ are the same if, and only if,
they have the same geodesics and the same torsion tensor.
  \item[ii)] 
Given a connection on $M$, there exist another, unique, 
with the same geodesics and null torsion. 
 \end{description} 
\end{itemize}

\begin{definition}
 The map ${\mathit R}\colon\vf(M)\otimes\vf(M)\otimes\vf(M)\to\vf(M)$ 
defined by
 $$
 {\mathit R}(X,Y)Z:=\nabla_X\nabla_Y Z-\nabla_Y\nabla_X Z-\nabla_{[X,Y]}Z 
 $$
 is $\Cinfty(M)$-trilinear, then it is a tensor field 
${\mathit R}\in{\mathcal T}_3^1(M)$ which is called 
the \textbf{curvature tensor} of the connection $\nabla$. 

A connection is said to be \textbf{flat} if its curvature tensor is zero.
\end{definition}

 The properties of the curvature tensor are:
\begin{itemize}
\item
The curvature tensor ${\mathit R}$ is antisymmetric in $(X,Y)$. 
  \item 
With the same notations as above,
in a local reference frame $(E_i)$ we have that 
 $$
 {\mathit R}={\mathit R}_{ijk}^l\; E^i\otimes  E^j\otimes E^k\otimes E_l \ ,
 $$
 with
 $$
 {\mathit R}_{ijk}^l=\Gamma_{jk,i}^l-\Gamma_{ik,j}^l+\Gamma_{jk}^m\Gamma_{im}^l-\Gamma_{ik}^m\Gamma_{jm}^l-c_{ij}^m\Gamma_{mk}^l \ .
 $$
  \item 
When the torsion of $\nabla$ is zero, then ${\mathit R}(X,Y)Z+{\mathit R}(Z,X)Y+{\mathit R}(Y,Z)X=0$. This is called the {\sl Bianchi identity}.
  \item 
It can be proved that, for a manifold with a connection,
the following conditions are equivalent:
  \begin{description}
  \item[{\rm (i)}] The connection is flat.
  \item[{\rm (ii)}] On a neighbourhood of any point, there exist a local chart with coordinate vector fields which are parallel.
  \item[{\rm (iii)}] The parallel transport along a curve between two points of the manifold  does not depend on the chosen curve.
  \end{description}
\end{itemize}

\subsection{Riemannian manifolds}
\label{Riemann}

\begin{definition}
Let $M$ be a differentiable manifold. 
A \textbf{Riemannian metric} on $M$ is a 2-covariant tensor field 
${\tt g}\in {\mathcal T}_2(M)$ symmetric and positive defined;
that is
\begin{description}
  \item[{\rm (i)}] The map ${\tt g}\colon\vf(M)\times \vf(M)\to\Cinfty(M)$ 
  is $\Cinfty(M)$-bilinear and symmetric 
  \item[{\rm (ii)}] For every $p\in M$, and $u_p\in\Tan_p M$, 
we have ${\tt g}(u_p,u_p)\geq 0$, and it is zero if, and only if, $u_p=0$.
\end{description}
A \textbf{Riemannian manifold} $(M,{\tt g})$ is a differentiable manifold endowed with a Riemannian metric.

Changing the positive defined condition by nondegeneracy of ${\tt g}$, 
we say that we have a \textbf{pseudo Riemannian  manifold}. 
In the case that the signature of a pseudo Riemannian metric ${\tt g}$ 
is $(1,m-1)$, or $(m-1,1)$, we say that $(M,{\tt g})$ is a 
\textbf{Lorentzian manifold} and ${\tt g}$ is a \textbf{Lorentz metric}.
\end{definition}

In the sequel we will suppose that $(M,{\tt g})$ is a Riemannian manifold unless indicated.

We have the following properties:
\begin{itemize}
\item
We introduce the following notations and terminology:
\begin{enumerate}
  \item 
  If $u_p,v_p\in\Tan_p M$, then $(u_p\mid v_p):={\tt g}(u_p,v_p)$ is the {\sl scalar product} of $u_p$ and $v_p$.
  \item 
  For  $u_p\in\Tan_p M$, the {\sl norm} of $u_p$ is $\parallel u_p\parallel:={\tt g}(u_p,u_p)^{1/2}$.
  \item 
  For $X,Y\in\vf(M)$, the scalar product of $X$ and $Y$ is $(X\mid Y):={\tt g}(X,Y):M\to \Real$.
  \item 
  For $u_p,v_p\in\Tan_p M$, if ${\tt g}(u_p,v_p)=0$ we say that they are {\sl orthogonal}.
  \item 
  Two vector fields $X,Y\in\vf(M)$ are orthogonal if ${\tt g}(X,Y)=0$.
  \item 
  If $u_p,v_p\in\Tan_p M$ and ${\tt g}$ is Riemannian, 
  then the {\sl angle} $\theta$ 
  between $u_p$ and $v_p$: $\theta$ is defined by 
  $\displaystyle\cos\theta:=
  \frac{{\tt g}(u_p,v_p)}{\parallel u_p\parallel\parallel v_p\parallel}$. 
\end{enumerate}
\item
If $(U,x^i)$ is a local chart,
the local expression of ${\tt g}$ is 
${\tt g}|_U=g_{ij}\,\d x^i\otimes\d x^j$, where 
$$
g_{ij}={\tt g}\left(\frac{\partial}{\partial x^i},\frac{\partial}{\partial x^j}\right)\,.
$$
The matrix $(g_{ij})$ is symmetric and positive defined.
\item
In general, if $(E_i)$ is a basis for vector fields on a manifold, 
then ${\tt g}=g_{ij}\,E^i\otimes E^j$, where $g_{ij}={\tt g}(E_i,E_j)$. 
Remember that $(E^j)$ is the dual basis corresponding to $(E_i)$.
\item 
{\sl The standard metric on $\Real^m$}: 
The standard Riemannian metric on $\Real^m$ is given by 
${\tt g}=\delta_{ij}\d x^i\otimes\d x^j$,
in the canonical coordinate system.
\item
{\sl The Minkowski space}: 
On the vector space 
$\Real^m$ we introduce the Minkowski metric given by 
${\tt g}=\eta_{ij}\,\d x^i\otimes\d x^j$,
 where the diagonal of $(\eta_{ij})$ is $(-1,1,\ldots,1)$ and all the other elements are zero. 
 This is the so-called {\sl Minkowski metric}
  and $(\Real^m,{\tt g})$ is a Lorentzian manifold.
\item
Let ${\mathit j}\colon N\hookrightarrow M$ 
an immersed submanifold of a Riemannian manifold $(M,{\tt g})$. 
Then $(N,{\mathit j}^*{\tt g})$ is a Riemannian manifold. 
Observe that this can not be true in the pseudo Riemannian case.
\item
Let $(M_1,{\tt g}_1)$ and $(M_2,{\tt g}_2)$ be Riemannian manifolds. 
An {\sl isometry} between $M_1$ and $M_2$ is a diffeomorphism 
$F\colon M_1\to M_2$ such that $F^*{\tt g}_2={\tt g}_1$. 
If $F$ is defined only on an open set then it is called a local isometry.
\item 
An {\sl isometry} of $(M,{\tt g})$ is a diffeomorphism on $M$, such that it leaves invariant the metric tensor ${\tt g}$.

An {\sl infinitesimal isometry} of $(M,{\tt g})$ is a vector field on $M$ 
such that the diffeomorphisms defined by its flux $F^t$ 
are isometries of $(M,{\tt g})$.
The infinitesimal isometries are called {\sl Killing vector fields}.
Clearly, $X\in\vf(M)$ is a Killing vector field if, and only if, $\Lie(X){\tt g}=0$.
\item 
In a paracompact differentiable manifold, there exists a Riemannian metric. This is a consequence of the existence of partitions of unity.
 
For the existence of a Lorentzian metric on a paracompact manifold $M$,
it is necessary that there exists a vector field without critical points. 
This is true for connected compact manifolds with null Euler--Poincar\'e characteristic.
\end{itemize}

\subsection{Some natural constructions on Riemannian manifolds}

Let $(M,{\tt g})$ be a Riemannian manifold.
\begin{itemize}
  \item 
For every point $p\in M$, the metric ${\tt g}$ defines an isomorphism
$\hat{\tt g}_p\colon\Tan_pM\to\Tan_p^*M$, such that
$\langle\hat{\tt g}_p(u_p),v_p\rangle:={\tt g}(u_p,v_p)$. 
This family of isomorphisms collects into a global one
$\hat{\tt g}\colon\Tan M\to\Tan^*M$, 
which is an isomorphism of vector bundles.

This isomorphism can be extended to the corresponding 
$\Cinfty(M)$-modules of sections; that is, we have the map
${\tt g}^\flat\colon\vf(M)\to\Omega^1(M)$ defined by
${\tt g}^\flat(X)=X^\flat:=\hat{\tt g}\circ X$,
which is an isomorphism with inverse 
${\tt g}^\sharp\colon\Omega^1(M)\to\vf(M)$ such that
$\alpha\mapsto\alpha^\sharp=\hat{\tt g}^{-1}\circ\alpha$. 
\item 
In a local chart $(U,x^i)$, the isomorphism $\hat{\tt g}$ 
and its inverse are given respectively by
$$ 
(x^i,v^i)\mapsto (x^i, v^ig_{ij}(x))\quad ,\quad
(x^i,\alpha_i)\mapsto(x^i, \alpha_i g^{ij}(x)) \ ,
$$ 
where $(g^{ij})$ is the inverse matrix of $(g_{ij})$.
\item
The above isomorphisms are usually called 
{\sl musical isomorphisms}, and they are written as 
$v^\flat$, for the linear form with coordinates $v_j=v^ig_{ij}$, 
and $\alpha^\sharp$, for the tangent vector with coordinates 
$\alpha^j= \alpha_i g^{ij}$. 

Obviously, these isomorphisms can be extended to any kind of tensor fields, and we can lift and low indexes using the metric tensor.
\item
Let $f\in\Cinfty(M)$; the {\sl gradient} of $f$ 
is the vector field ${\rm grad}\,f:=\hat{{\tt g}}^{-1}\circ\d f$. 
Then we have that $({\rm grad}\, f|X)=\langle\d f,X\rangle$.
In local coordinates, 
$\dst{\rm grad}\,f=g^{ij}\frac{\partial f}{\partial x^i}\frac{\partial}{\partial x^j}$.
\item 
In a Riemannian manifold, there exists a canonical measure $V_{\tt g}$ 
called {\sl Riemannian volume}. 
First, in the open set of a local chart $(U,x^i)$, 
the integral of a function $f\colon U\to \Real$ is given by
$$
\int_U f\,\d V_{\tt g}=\int_{\varphi(U)}\hat{f}(x^1,\ldots,x^m)\sqrt{{\mid\det\,(g_{ij})}\mid}\,\,\d x^1\ldots\d x^m\ ,
$$
where $\hat{f}(x^1,\ldots,x^m)$ is the local expression of $f$.
\item
IF $M$ is an oriented manifold, then this measure comes from 
a canonical volume form $\Omega_g$ usually called 
{\sl canonical Riemannian volume form}. 
In a local chart, this volume form is given by  
$$
\Omega_{\tt g}=\sqrt{{\mid\det\,(g_{ij})}\mid}\ \d x^1\ldots\d x^m \ .
$$
\end{itemize}

\subsection{The Levi-Civita connection}

\begin{definition}
Let $(M,{\tt g})$ be a Riemannian manifold.
A connection $\nabla$ on $M$ is a \textbf{Riemannian connection} 
if the metric tensor field ${\tt g}$ is parallel; that is, $\nabla{\tt g}=0$. 

This is equivalent to say that $\nabla_X{\tt g}=0$, for every $X\in\vf(M)$.
\end{definition}

Another interesting way to state this property is the following:
 for every $X,Y,Z\in\vf(M)$,
$$
\Lie(Z)({\tt g}X,Y)={\tt g}(\nabla_ZX,Y)+{\tt g}(X,\nabla_ZY) \ .
$$  

Using the idea of covariant derivative along a curve,
this condition is equivalent to say that, 
for every curve $\gamma\colon I\subset\Real\to M$ and 
$X,Y\in\vf(\gamma)$,
the following equation holds
  $$
  D(X\mid Y)=(\nabla_{\dot{\gamma}(t)}X\mid Y)+(X\mid\nabla_{\dot\gamma(t)}Y)  \ .
$$
Remember that, for a function $f:I\subset\Real\to\Real$, we have $Df=\d f/\d t$. 

Observe that, in particular, if $X,Y$ are parallel along 
$\gamma$, then $(X\mid Y)$ is constant, 
then we get that the parallel transport with a Riemannian connection is an isometry. 
Furthermore, if $\gamma$ is a geodesic line, then $(\gamma'\mid\gamma')$ is constant.

Finally, the {\sl Fundamental Theorem of Riemannian geometry} states that, 
if $(M,{\tt g})$ is a Riemannian manifold, 
there exists one, and only one, connection on $M$ which is Riemannian and symmetric.

\begin{definition}
The above connection, called the \textbf{Levi-Civita connection}
associated to the metric,  is defined by the {\rm Koszul formula\/}:
$$
2(\nabla_XY\mid Z)=\Lie(X)(Y\mid Z)+\Lie(Y)(Z\mid X)-\Lie(Z)(X\mid Y)
$$
$$
+([X,Y]\mid Z)+([Z,X]\mid Y)-([Y,Z]\mid X).
$$
\end{definition}

In the sequel, unless indicated, when we use a connection $\nabla$ 
in a Riemannian manifold, we will assume that it is the Levi-Civita connection.

In a local chart $(U,x^i)$ in $M$, the Christoffel symbols of the Levi-Civita connection in the corresponding local basis for the vector fields are given by:
$$
\Gamma_{ij}^k=\frac{1}{2}g^{kl}\left(\frac{\partial g_{jl}}{\partial x^i}+\frac{\partial g_{il}}{\partial x^j}-\frac{\partial g_{ij}}{\partial x^l}\right).
$$  
They are also written as $\Gamma_{ij}^k=g^{kl}[ij,l]$ 
where $[ij,l]$ are called {\sl Christoffel symbols of the first class}.

\subsection{Submanifolds of a Riemannian manifold}

Let ${\mathit j}\colon M\hookrightarrow\widetilde{M}$ 
be an embedded submanifold of a Riemannian manifold 
$(\widetilde{M},\widetilde{g})$ with the Levi-Civita connection 
$\widetilde{\nabla}$. We know that ${\tt g}=j^*\widetilde{\tt g}$ 
is a Riemannian metric on $M$. Let $\nabla$ be its Levi-Civita connection.

\begin{itemize}
  \item 
For every $p\in M$, the embedding ${\mathit j}$ allows us to identify 
$\Tan_p M$ as a subspace of $\Tan_{j(p)}\widetilde M$, 
hence we have a splitting
$$
\Tan_{j(p)}\widetilde M=\Tan_p M\oplus (\Tan_p M)^\perp \,.
$$
Then every tangent vector $u\in\Tan_{j(p)}\widetilde M$ 
can be decomposed as a sum  $u=u^\top+u^\perp$, 
the tangent part to $M$ and the orthogonal one. This splitting can be extended to vector fields on the manifold $\widetilde M$. 
  \item
Let $X,Y\in\vf(M)$; on a neighbourhood of every $p\in M$ 
we can extend the vector fields $X,Y$ to vector fields 
defined on an open set of $\widetilde M$. 
Hence, we can calculate $(\widetilde{\nabla}_XY)(p)\in\Tan_{j(p)}\widetilde M$, 
as this value is independent of the extensions. 
With this construction, we have a vector field along the embedding ${\mathit j}$; 
that is, $\widetilde{\nabla}_XY\colon M\to \Tan\widetilde M$. 
This vector field can be split into,
$$
\widetilde{\nabla}_XY=(\widetilde{\nabla}_XY)^\top+(\widetilde{\nabla}_XY)^\perp \,.
$$
   \item
It is easy to prove that, for $X,Y\in\vf(M)$, the map $(X,Y)\mapsto(\widetilde{\nabla}_XY)^\top$, 
 defines a connection on $(M,{\tt g})$. 
This connection is precisely the Levi-Civita connection. 
  \item
Now we consider the other component of the above splitting: 
The expression $\Pi(X,Y)=(\widetilde{\nabla}_XY)^\perp$ is 
$\Cinfty(M)$-bilinear and symmetric in $X,Y\in\vf(M)$. 
Symmetry holds because $[X,Y]\in\vf(M)$, for $X,Y\in\vf(M)$, hence $[X,Y]^\perp=0$. 
  \item
The above property of bilinearity implies that
the map $\Pi_p\colon\Tan_p M\times\Tan_p M\to(\Tan_p M)^\perp$, for every $p\in M$,
is well-defined and gives us a vector valued bilinear symmetric form 
taking values on $(\Tan_p M)^\perp$. 
It is called the {\sl second fundamental form} of $M$.
  
The decomposition $\widetilde{\nabla}_XY=\nabla_XY+\Pi(X,Y)$ is called {\sl Gauss formula}.
  
Furthermore, if $N\in\vf(\widetilde M)$ is orthogonal to $M$, 
then we have the so-called {\sl Weingarten equation}: 
$(\widetilde{\nabla}_XN\mid Y)=-(N\mid\Pi(X,Y))$.
  \item
({\sl \textbf{Nash Embedding Theorem\/}}): Every Riemannian manifold 
with a countable basis of open sets is isometric 
to a submanifold of some $\Real^n$.
  \end{itemize}

It is interesting to analyze the following particular cases of submanifolds of Riemannian manifolds:

\subsubsection{Hypersurfaces} 

Let $(\widetilde M,{\tt g})$ be a Riemannian manifold and 
$M\subset\widetilde M $ be a hypersurface of $\widetilde M$,
and suppose that both manifolds are oriented.
 
If $(X_1,\ldots,X_m)$ is an oriented positive basis of vector fields on $M$,
then, on the points of $M$, we can take a unique vector field $N$, 
unitary and orthogonal to $M$, such that $(X_1,\ldots,X_m, N)$ 
is an oriented positive basis for $\widetilde M$ on the points of $M$. We say that $N$ is the positive-oriented normal vector field to $M$. In this situation, the second fundamental form can be understood as 
$\Real$-valued, and we write $h(X,Y)=(\Pi(X,Y)\mid N)$. 
We have that $\Pi(X,Y)=h(X,Y)N$.

With this construction, we have a 2-covariant tensor field on $M$. 
With this quadratic form, using the metric ${\tt g}$,
at every point $p\in M$, we can define a symmetric endomorphism 
$S$ on $\Tan_pM$, called the {\sl Weingarten map}, which is defined by 
  $$
  (S(X)\mid Y)=h(X,Y):=((\widetilde{\nabla}_XY)^\perp\mid N)\ .
  $$
The eigenvalues of $S$ are called {\sl principal curvatures}
and the product of all of them is called the {\sl Gauss curvature}. 
The eigenvectors are called {\sl principal directions}.  

\subsubsection{Curves}

Let $(M,{\tt g})$ be a Riemannian manifold and $C\subset M$ a connected orientable regular curve; that is, a  submanifold of dimension one.

For every $p\in C$ there exists a neighbourhood and 
a local parametrization of $C$ with a curve $c\colon I\subset\Real\to M$ 
such that $\parallel c'(t)\parallel=1$, the {\sl arc parameter},
and with the orientation of $C$. 
If ${\mathbf t}=c'\in\vf(c)$ is the tangent vector field, for every $s\in I$,
we have that ${\mathbf t}(s)$ is a positive basis of $\Tan_{c(s)}C$. 

Suppose that $({\mathbf t},\nabla_{{\mathbf t}(s)}{\mathbf t},\ldots,\nabla_{{\mathbf t}(s)}^{m-1}{\mathbf t})$ 
are linearly independent in some $s_o\in I$. 
If we apply the Gram-Schmidt orthogonalization procedure, 
we get an orthonormal basis $({\mathbf f}_1(s)={\mathbf t},\ldots,{\mathbf f}_m(s))$ 
of $\Tan_{c(s)}C$, for $s$ on a neighbourhood of $s_o$. 
This is the so called {\sl Frenet reference}.
The covariant derivatives of these vector fields 
are linear combinations of themselves, 
and the coefficients are the curvatures of $C$. 
This linear combinations are the {\sl Frenet--Serret formulas}.

\subsection{Curvature and distance on a Riemannian manifold}

\begin{definition}
Let $(M,{\tt g})$ be a Riemannian manifold and $R$ be
the curvature tensor of the Levi-Civita connection $\nabla$ of $(M,{\tt g})$. 
The \textbf{Riemann curvature tensor} is the tensor field 
${\rm Rie}\in{\mathcal T}_4(M)$ obtained by lowering 
the contravariant index of $R$ to the last place:
$$
{\rm Rie}(X,Y,Z,W)=(R(X,Y)Z\mid W)\ .
$$
The \textbf{Ricci curvature tensor} is the tensor field 
${\rm Ric}\in{\mathcal T}_2(M)$ obtained after the contraction 
of the first covariant index with the contravariant index of $R$;
that is ${\rm Ric}=c_1^1(R)$.

\noindent The \textbf{scalar curvature} is the function 
$S\in{\mathcal C}^{\infty}(M)$ obtained by taking the trace, with respect to $\bf g$, of the Ricci tensor:
  $$
  {\rm S}={\rm tr}_{\tt g}({\rm Ric})={\rm tr}({\rm Ric}^\sharp)\ . 
  $$
\end{definition}

\begin{remark}{\rm  
These definitions can change by demanding that the contractions
are with other indexes, and then the only difference is a change in the sign.
  \begin{itemize}
  \item
A Riemannian manifold is {\sl flat} if it is locally isometric 
to the Euclidean space.
Then, a Riemannian manifold $(M,g)$ is flat if, and only if,
its curvature tensor vanishes; that is,
its Levi-Civita connection is flat.
  \item
In a local reference for vector fields, we have that
\begin{description}
\item
${\rm Rie}=R_{ijkl}\,E^i\otimes E^j\otimes E^k\otimes E^l$, where $R_{ijkl}=R_{ijk}^sg_{sl}$.
\item
${\rm Ric}=R_{ij}\,E^i\otimes E^j$ with $R_{jk}=R_{ijk}^i$.
\item
${\rm S}=R_k^k=g^{kj}R_{jk}$.
  \end{description}
\end{itemize}
}\end{remark}

Now, let $(M,{\tt g})$ be a Riemannian manifold.

\begin{itemize}
  \item
  {\sl Length of a curve\/}: Let $\gamma\colon I\to M$ be a ${\mathcal C}^1$-piecewise curve. The \emph{length} of $\gamma$ is
  $$
  \ell(\gamma)=\int_I\parallel\gamma'\parallel\ .
  $$ 
  Observe that this length is invariant by reparametrizations;
  that is, if $\varphi\colon J\to I$ is a diffeomorphism 
  between open real intervals, then $\ell(\gamma\circ\varphi)=\ell(\gamma)$.
  \item
  Supposing that the manifold $M$ is connected, any two points can be connected by a piecewise curve. Then,
  if $p,q\in M$,  the {\sl distance} between $p$ and $q$ is defined as:
  $$
  {\rm d}(p,q):=\inf\{\ell(\gamma)\mid \gamma\textrm{ is a }{\mathcal C}^1\textrm{-piecewise curve from}\,\, p\,\, \textrm{to}\,\, q\}.
  $$
  The function $ {\rm d}$ is a distance defining the topology of $M$and is called {\sl Riemannian distance}.
  \item
  Let $M$ be a connected differentiable manifold, 
  suppose it is Hausdorff but not necessarily paracompact. 
  Then the following conditions are equivalent:
  \begin{enumerate}
  \item 
  $M$ admits a Riemannian metric.
  \item 
  $M$ is metrizable as topological space.
  \item 
  The topology of $M$ has a countable basis of open sets.
  \item 
  $M$ is paracompact.
\end{enumerate}
  \item
If $(M,{\tt g})$ is a connected Riemannian manifold,
we say that it is {\sl geodesically complete} 
if the domain of definition of all the geodesic curves on $M$ is $\Real$. 
  \item
({\sl \textbf{Hopf--Rinow Theorem\/}}): 
Let $(M,{\tt g})$ be a connected Riemannian manifold and ${\rm d}$ 
its Riemannian distance. Then, the following conditions are equivalent:
  \begin{enumerate} 
   \item 
For a subset of $M$, if it is bounded by ${\rm d}$ and closed,
then it is compact.
   \item 
As a metric space, $M$ is complete.
   \item 
$(M,{\tt g})$ is geodesically complete.
\end{enumerate}
\end{itemize}

%%%%%%%%%%%%%%%%%%%%%%%%%%%%%%%%%%%%%%%%%%%%%%%%%%%%%%%%%%%%%%%%%%

\section{Newtonian dynamical systems}

Next, in this section,
we use the above geometric structures to state the description of dynamical systems 
on a (semi)Riemannian manifolds
(see, for instance, \cite{AM-78,Ar-89,CC-2005,CM-2023,CDW-87,GN-2014,Ol-02} as general references).

\subsection{Newton dynamical equations. Kinetic energy}

From a geometric approach, mechanical systems in Newtonian physics have a common background which, as in the above cases,
 is collected in the following postulates:

\begin{pos}
{\rm (First Postulate of  Newtonian mechanics\/)}:
The \textbf{configuration space} $Q$ of the dynamical system
with $n$ degrees of freedom is a $n$-dimensional differentiable manifold which is endowed with a \textsl{Riemannian metric}
\footnote{
This formulation can be extended to describe relativistic systems
and, in this case, the metric is semi-Riemannian.}. 

The \textbf{state space}, or \textbf{phase space} of coordinates--velocities, is the tangent bundle $\Tan Q$ or, alternatively, the \textbf{phase space} of coordinates--momenta
is the cotangent bundle $\Tan^*Q$ of the manifold $Q$.
\end{pos}

\begin{pos}
{\rm (Second Postulate of  Newtonian mechanics\/)}:
The \textbf{observables} or physical magnitudes of the dynamical system
are functions of $\Cinfty (\Tan Q)$ or $\Cinfty (\Tan^*Q)$.

The result of the measure of an observable
is the value of its representing function at a point of the phase space.
\end{pos}

\begin{pos}
{\rm (Third Postulate of Newtonian mechanics\/)}:
The dynamics of the system is given by a $1$-form on $Q$,
the  \textbf{work form} or \textbf{force form}, or equivalently, a vector field on $Q$, 
the \textbf{force field} or simply the \textbf{force}. 
%And this element determines the evolution of the system, that is  its dynamical trajectories.
\end{pos}

With these ideas in mind, we define:

\begin{definition}
A \textbf{Newtonian mechanical system} is a triple $(Q,{\tt g},\omega)$, where
\ben
\item
$Q$ is a differentiable manifold ($\dim\, Q=n$).
\item
${\tt g}$ is a {\sl Riemannian metric} on $Q$. Then  $(Q,{\tt g})$ is a Riemannian manifold.
\item
$\omega$ is a differential 1-form on $Q$, called the \textbf{work form}.
\een
As ${\tt g}$ is a Riemannian metric, the work form $\omega\in\df^1(Q)$
is associated to a unique vector field ${\rm F}\in\mathfrak{X} (Q)$ such that
$\inn({\rm F}) {\tt g}=\omega$. We call  ${\rm F}$ the
\textbf{the force field} of the system. 
In this case, we denote the system as $(Q,{\tt g},{\rm F})$).
\end{definition}

For a Newtonian system $(Q,{\tt g},{\rm F})$, let $\nabla$ be the {\sl Levi-Civita connection} associated to ${\tt g}$.

\begin{pos}
{\rm (Fourth Postulate of Newtonian mechanics\/)}:
The dynamical trajectories of the Newtonian dynamical system
$(Q,{\tt g},{\rm F})$ are the curves $\gamma\colon[a,b]\subset\Real\to Q$
solution to the equation
\beq
\nabla_{\dot{\gamma}}\dot\gamma ={\rm F}\circ\gamma
\label{eqdin}
\eeq
which is called the \textbf{Newton equation} of the system.
\end{pos}

\begin{remark}{\rm 
Note that if ${\rm F}=0$, then the dynamical trajectories are the geodesic curves of the metric ${\tt g}$. 
This corresponds to the first Newton Law ({\sl Inertia Law\/}).
}\end{remark}

If $(U,x^i)$is a local chart on $Q$, and $\{\Gamma^k_{ij}\}$
are the corresponding Christoffel symbols of the connection $\nabla$, then the dynamical equation is locally given by
$$
\ddot\gamma^k+\Gamma_{ij}^k\dot\gamma^i\dot\gamma^j={\rm F}^k\circ\gamma \ .
$$
Furthermore, if, $\omega=\omega_i\d x^i$ and 
\(\dst {\rm F}={\rm F}^i\derpar{}{x^i}\), 
in this chart, then we have that
$$
\omega_i=g_{ij}{\rm F}^j \quad , \quad
{\rm F}^i=g^{ij}\omega_j \ ,
$$
where $g^{ij}$ are the components of the inverse matrix of ${\tt g}$ in this local chart (the ``inverse metrics'').

This formalism has a dual counterpart as follows:

\begin{definition}
Let $(Q,{\tt g})$ be a Riemannian manifold. The musical diffeomorphism associated to the Riemannian metric ${\tt g}$, defined by
$$
\begin{array}{ccccc}
\theta & \colon & \Tan Q & \longrightarrow & \Tan^*Q  \\
& & (q,v) &\mapsto & (q,\inn (v)g)
\end{array}
$$
gives us the commutative diagram
$$
\begin{array}{ccc}
\Tan Q &
\begin{picture}(135,10)(0,0)
\put(63,6){\mbox{$\theta$}}
\put(0,3){\vector(1,0){135}}
\end{picture}
& \Tan^*Q
\\ &
\begin{picture}(135,45)(0,0)
\put(14,22){\mbox{$\tau_Q$}}
\put(110,22){\mbox{$\pi_Q$}}
\put(63,0){\mbox{$Q$}}
\put(0,45){\vector(3,-2){55}}
\put(135,45){\vector(-3,-2){55}}
\end{picture} &
\end{array}.
$$
Hence $\theta$ can be understood as a differential $1$-form on $Q$ {\sl along} $\tau_Q$
(denoted by $\theta\in\df^1(Q,\tau_Q)$),
called the \textbf{linear momentum $1$-form}
 associated with the metric.
\end{definition}

Using $\theta$, the Newton equation (\ref{eqdin}) can be equivalently written as:

\begin{prop}
({\bf Dual form of the dynamical equations}): 
Given a Newtonian dynamical system $(Q,{\tt g},\omega)$,
a curve $\gamma\colon I\subset\Real\to Q$
is a solution to the dynamical equation if, and only if, it satisfies that
\beq
\nabla_{\dot\gamma}(\theta\circ\dot\gamma)=\omega\circ\gamma
\label{eqdindual}
\eeq
\end{prop}
\begin{proof}
If $\gamma\colon [a,b]\subset\Real\to Q$ is a curve on $Q$, then
$\theta\circ\dot\gamma\in\df^1(Q,\gamma)$, and we have that $$
\nabla_{\dot\gamma}(\theta\circ\dot\gamma )=
\nabla_{\dot\gamma}(\inn(\dot\gamma){\tt g})=
\inn(\nabla_{\dot\gamma}\dot\gamma){\tt g}+\inn(\dot\gamma)\nabla_{\dot\gamma}{\tt g}=
\inn(\nabla_{\dot\gamma}\dot\gamma){\tt g}=
\theta\circ\nabla_{\dot\gamma}\dot\gamma
$$
On the other side if $\gamma$ is a dynamical trajectory, that is 
$\nabla_{\dot\gamma}\dot\gamma= {\rm F}\circ\gamma$, then:
$$
\nabla_{\dot\gamma}(\theta\circ\dot\gamma)=
\theta\circ\nabla_{\dot\gamma}\dot\gamma=
\theta\circ {\rm F}\circ\gamma=
\inn({\rm F}){\tt g}\circ\gamma=
\omega\circ\gamma \ ;
$$
which is the equation (\ref{eqdindual}).
\\ \qed \end{proof}

Note that as $\theta\circ\dot\gamma\in\df^1(Q,\gamma)$, then,
for every $X\in\vf(Q)$, we have
$$
\nabla_{\dot\gamma}((\theta\circ\dot\gamma)(X))=
(\nabla_{\dot\gamma}(\theta\circ\dot\gamma))(X)+
(\theta\circ\dot\gamma)(\nabla_{\dot\gamma}X) \ .
$$

Hence, we obtain the following:

\begin{teor}
{\rm (Linear momentum conservation)}:
If the work form (or in a equivalent way, the force field) of a Newtonian mechanical system $(Q,{\tt g},\omega)$ is zero, then the linear momentum is invariant, is a constant, along the trajectories of the system.
(In a more geometric way, for every trajectory of the system $\gamma$, the form 
$\theta\circ\dot\gamma\in\df^1(Q,\gamma)$ is parallel along the trajectory $\gamma$
\footnote{We maintain the statement as in the classical form in physics:
in $\Real^3$, with a natural Cartesian system of coordinates, we have that
$\Gamma^k_{ij}=0$, for every $ i,j,k$, hence the components of
$\theta\circ\dot\gamma$ are constants.
}).
\end{teor}
\begin{proof}
If
$\omega=0$ or, equivalently, $ F=0$, then the Newton equation is
\(\dst \nabla_{\dot\gamma}(\theta\circ\dot\gamma)=0\) ;
hence $\theta\circ\dot\gamma$ is parallel along the trajectory $\gamma$.
\\ \qed \end{proof}

Finally, with the Riemannian metric, we can associate the following function:

\begin{definition}
Given a Riemannian manifold $(Q,{\tt g})$,  the associated \textbf{kinetic energy} is the function
$K\in\Cinfty(\Tan Q)$ defined by
$$
\begin{array}{ccccc}
K & \colon & \Tan Q & \longrightarrow & \Real  \\
& & (q,v) &\mapsto & \frac{1}{2}{\tt g}(v,v)
\end{array} \ .
$$
\end{definition}

Its local expression is
$$
K(q^i,v^j)=\frac{1}{2}g_{ij}(q)v^iv^j \ .
$$

\subsection{Euler--Lagrange equations}

Now, we transform the Newton equation into a new form, 
which is easier to state when we know the elements defining the system.
First, we need a technical result:

\begin{lem}
Let $(Q,{\tt g})$ be a Riemannian manifold, $K\in\Cinfty (\Tan Q)$ its kinetic energy, $\nabla$ the Levi-Civita connection of ${\tt g}$, and
$\gamma\colon I\subset\Real\to Q$ a curve.
If $(U,\varphi=(q^i))$ is a local chart of $Q$ with $\gamma(t)\in U$, and $(\tau_Q^{-1}(U),q^i,v^i)$
is the natural chart on $\Tan Q$; then
$$
\frac{d}{d t}\left(\derpar{K}{v^j}\circ\dot\gamma\right)-
\derpar{K}{q^j}\circ\dot\gamma=
{\tt g}\left(\nabla_{\dot\gamma}\dot\gamma,\derpar{}{q^j}\right) \ .
$$
\end{lem}
\begin{proof}
We compare the local expressions of both sides of the equation.
As $K=\frac{1}{2}g_{ij}v^iv^j$, we have that
$$
\derpar{K}{v^j}=g_{ij}v^i \quad ,\quad
\derpar{K}{q^j}=\frac{1}{2}\derpar{g_{ik}}{q^j}v^iv^k \ ,
$$
thus, if $\gamma=(\gamma^1,\ldots ,\gamma^n)$,
$$
\derpar{K}{v^j}\circ\dot\gamma=(g_{ij}\circ\gamma)\circ\dot\gamma^i \quad ,\quad
\derpar{K}{q^j}\dot\gamma=
\frac{1}{2}\left(\derpar{g_{ik}}{q^j}\circ\gamma\right)\dot\gamma^i\dot\gamma^k \ .
$$
Hence we obtain that
$$
\frac{d}{d t}\left(\derpar{K}{v^j}\circ\dot\gamma\right)=
\left(\derpar{g_{ij}}{q^k}\circ\gamma\right)\dot\gamma^k\dot\gamma^i+
(g_{ij}\circ\gamma)\ddot\gamma^i
$$
and
\beq
\frac{d}{d t}\left(\derpar{K}{v^j}\circ\dot\gamma\right)-
\derpar{K}{q^j}\circ\dot\gamma=
\left(\derpar{g_{ij}}{q^k}\circ\gamma\right)\dot\gamma^k\dot\gamma^i+
(g_{ij}\circ\gamma)\ddot\gamma^i-
\frac{1}{2}\left(\derpar{g_{ik}}{q^j}\circ\gamma\right)\dot\gamma^i\dot\gamma^k \ . \label{dtK}
\eeq
Furthermore,
$$
g\left(\nabla_{\dot\gamma}\dot\gamma,\derpar{}{q^j}\right)=
g\left((\ddot\gamma^{i}+\Gamma^i_{kl}\dot\gamma^k\dot\gamma^l)\derpar{}{q^i},\derpar{}{q^j}\right)=
(g_{ij}\circ\gamma)\ddot\gamma^i+
(g_{ij}\circ\gamma)\Gamma^i_{kl}\dot\gamma^k\dot\gamma^l.
$$
But we have that
$$
[kl,j]=g_{ij}\Gamma^i_{kl}=
\frac{1}{2}\left(\derpar{g_{jk}}{q^l}+\derpar{g_{jl}}{q^k}-\derpar{g_{lk}}{q^j}\right) \ ,
$$
and by substitution in equation (\ref{dtK}), we obtain the desired result.
\\ \qed \end{proof}

\begin{teor}
{\rm (Lagrange)}:
Let $(Q,{\tt g},\omega)$ a Newtonian mechanical system, and
$\gamma\colon I\subset\Real\to Q$ a curve contained in the domain $U\subset Q$ of a chart
$(U,\varphi=(q^i))$ of $Q$. Then, $\gamma$ is a solution to the dynamical equation
(\ref{eqdin}) if, and only if, it satisfies the equations 
\beq
\frac{d}{d t}\left(\derpar{K}{v^j}\circ\dot\gamma\right)-
\derpar{K}{q^j}\circ\dot\gamma=
(\omega\circ\gamma)\left(\derpar{}{q^j}\right)=\omega_j \circ\gamma \ ,
\label{eel}
\eeq
which are called \textbf{Euler--Lagrange equations of the second kind}
of the system.
\end{teor}
\begin{proof}
If $\gamma$ is a solution to the Newton equation \eqref{eqdin}, then
$$
{\tt g}\left(\nabla_{\dot\gamma}\dot\gamma,\derpar{}{q^j}\right)=
{\tt g}\left({\rm F}\circ\gamma,\derpar{}{q^j}\right)=
(\omega\circ\gamma)\left(\derpar{}{q^j}\right).
$$
Hence, by the previous Lemma, the curve $\gamma$ satisfies equations (\ref{eel}).

Conversely, if $\gamma$ satisfies (\ref{eel}), then,
$$
{\tt g}\left(\nabla_{\dot\gamma}\dot\gamma,\derpar{}{q^j}\right)=
(\omega\circ\gamma)\left(\derpar{}{q^j}\right)={\tt g}\left(F\circ\gamma,\derpar{}{q^j}\right) \ ,
$$
Therefore $\nabla_{\dot\gamma}\dot\gamma={\rm F}\circ\gamma$,
 hence $\gamma$ is a solution to the Newton equation.
\\ \qed \end{proof}

\begin{remark}{\rm 
\bit
\item
To calculate the Newton equation in a local chart, we need to know the Christoffel symbols 
of the Levi-Civita connection $\nabla$ associated to ${\tt g}$;
but to write the Euler--Lagrange equations we do not need the local expression of the connection. 
Moreover, observe that the way to get Euler--Lagrange equations does not depend on the local chart we use. We need only to know the local expression of $K$ and $\omega$, and calculate the suitable derivatives.
\item
Another way to understand Euler--Lagrange equations is as follows:
let $(U,q^i)$ be a local chart of $Q$; 
the Newton equations of the system in this chart are written as
$$
\ddot\gamma^i+\Gamma^i_{jk}\dot\gamma^j\dot\gamma^k={\rm F}^i\circ\gamma  \ ,
$$
or, as it is usual,
$$
\ddot q^i+\Gamma^i_{jk}\dot q^j\dot q^k={\rm F}^i \ .
$$
This is a second order system of differential equations. 
To transform it into a first order system, we introduce new variables
$v^i=\dot q^i$, and we obtain the first order system on $\Tan Q$
$$
\dot q^i = v^i \quad , \quad
\dot v^i = {\rm F}^i-\Gamma^i_{jk}v^jv^k \ .
$$
The associated vector field is given by
$$
X=v^i\derpar{}{q^i}+({\rm F}^i-\Gamma^i_{jk}v^jv^k)\derpar{}{v^i} =v^i\derpar{}{q^i}-\Gamma^i_{jk}v^jv^k\derpar{}{v^i}+{\rm F}^v\ ,
$$
where ${\rm F}^v$ is the vertical lift of ${\rm F}$ from $Q$ to $\Tan Q$. Observe that we obtain the geodesic vector field plus the vertical lift of the force field. 

As $X\in\vf (\tau_Q^{-1}(U))$, for every point $(q,v)\in\tau_Q^{-1}(U)$
there exists a unique solution, with $(q_o,v_o)$ as initial condition. 
Consider now the functions we need to obtain 
the Euler--Lagrange equations and calculate the action of $X$ on them. 
Using the well known properties of Christoffel symbols, we have:
$$
X\left(\derpar{K}{v^k}\right) = X(g_{lk}v^l)=
\derpar{g_{lk}}{q^i}v^iv^l+g_{lk}({\rm F}^l-\Gamma_{ij}^lv^iv^j) \ ,
$$
hence
$$
X\left(\derpar{K}{v^k}\right) -\derpar{K}{q^k} =
\derpar{g_{lk}}{q^i}v^iv^l+g_{lk}{\rm F}^l-g_{lk}\Gamma_{ij}^lv^iv^j-
\frac{1}{2}\derpar{g_{ij}}{q^k}v^iv^j=g_{lk}{\rm F}^l \ .
$$
So, we have obtained:
\eit
}\end{remark}

\begin{prop}
Let $(Q,{\tt g},\omega)$ be a Newtonian mechanical system, $(U,q^i)$ a local chart on $Q$ and $(\tau_Q^{-1}(U), q^i,v^i)$ the corresponding natural chart on $\Tan Q$.
Then there exists a unique vector field $X\in\vf(\tau_Q^{-1}(U))$
which satisfies,
\ben
\item
$\Lie(X)q^k=v^k$.
\item
\(\dst \Lie(X)\left(\derpar{K}{v^k}\right) =
\derpar{K}{v^k}+g_{ik}{\rm F}^i\).
\een
Furthermore, its integral curves $\sigma\colon I\subset\Real\to \Tan Q$
 are canonical lifts  of curves
$\gamma\colon I\subset\Real\to Q$ to $\Tan Q$, 
which are solution to the Newton equations of the system.
\label{cvXY}
\end{prop}
\begin{proof}
The existence and unicity of the vector field  $X$ are consequences of being
\(\dst\left( q^k,\derpar{K}{v^k}\right)\) a local chart on $\Tan Q$,
because the Riemannian metric ${\tt g}$ is non-degenerate.
The properties of  $X$ have been proved in the previous discussion.

Furthermore, the integral curves of $X$ are the lifts to $\Tan Q$ of the solutions to the system of Euler--Lagrange equations, 
and hence solutions to the Newton equations.
\\ \qed \end{proof}

\subsubsection{Velocity dependent forces}
\label{vdf}

It is usual that the mechanical forces, or the work forms, depend 
not only on the position coordinates, but also on the velocities. 
This is the case of dissipative systems or electromagnetic, Lorentz, forces. 
Geometrically, this means that 
$\omega\in\df^1(Q,\tau_Q)$ and ${\rm F}\in\vf (Q,\tau_Q)$. 
In this case the only change we need to do in the above paragraphs are the followings:
\ben
\item
The Newton equations change to
$$
\nabla_{\dot\gamma}\dot\gamma={\rm F}\circ\dot\gamma \ ,
$$
or in dual form,
$$
\nabla_{\dot\gamma}(\theta\circ\dot\gamma)=\omega\circ\dot\gamma \ .
$$
\item
The Euler--Lagrange equations are
$$
\frac{d}{d t}\left(\derpar{K}{v^j}\circ\dot\gamma\right)-
\derpar{K}{q^j}\circ\dot\gamma=(g_{ij}{\rm F}^i)\circ\dot\gamma=
(\omega_j\circ\dot\gamma) \ .
$$
\een

 If the force field of the system depends on the velocities, the work form
 $\omega$ cannot be the exterior differential of a function defined on the configuration space $Q$.

As can be seen, the only changes consist in the substitution of
$\gamma$ by $\dot\gamma$, when we compose with $\omega$ or with ${\rm F}$,
to take into account the new domain of definition.

\subsection{Conservative systems:
Mechanical Lagrangians and Euler--Lagrange equations}

In particular, we are interested in a special kind of Newtonian systems:
those which are of conservative or mechanical Lagrangian type, called {\sl simple mechanical systems}.

\begin{definition}
A Newtonian mechanical system $(Q,{\tt g},\omega)$ is \textbf{conservative} if the work form is exact; that is, there exists $V\in\Cinfty (Q)$ such that $\omega =-\d V\,$
\footnote{
The negative sign is a customary tradition in Physics
in order to identify the potential function of the field with the potential energy of the system.
}.
In this case, the function $V$ is called the \textbf{potential energy} of the system and the force vector field is  ${\rm F}=-{\rm grad}\ V$.
\end{definition}

Thus, for these systems, we define:

\begin{definition}
In a conservative Newtonian mechanical system $(Q,{\tt g},\omega)$, with $\omega =-\d V\,$,
the \textbf{total energy} or \textbf{mechanical energy} of the system 
is the function $E\in\Cinfty(\Tan Q)$ defined as
$$
\begin{array}{ccccc}
E & \colon & \Tan Q & \longrightarrow & \Real  \\
& & (q,v) &\mapsto & K(q,v)+(\tau^*_QV)(q,v)
\end{array}
$$
(To simplify notation we usually write $E=K+V$).
\end{definition}

As a direct consequence of the definition, we have:

\begin{teor}
{\rm (Mechanical energy conservation)}:
Let $(Q,{\tt g},\omega)$ be a conservative Newtonian mechanical system, then
the mechanical energy $E$
is invariant (``constant'') along the trajectories of the system.
\end{teor}
\begin{proof}
If $\gamma\colon I\subset\Real\to Q$ is a solution to the Newton equations,
\beq
\nabla_{\dot\gamma}\dot\gamma={\rm F}\circ\gamma \quad , \quad \inn({\rm F}){\tt g}=
\omega=-\d V \ ,
\label{eqdinsis}
\eeq
then we have that
\beann
\frac{d (E\circ\dot\gamma)}{d t}
&=&
\nabla_{\dot\gamma}(E\circ\dot\gamma)=
\nabla_{\dot\gamma}\left(\frac{1}{2}{\tt g}(\dot\gamma,\dot\gamma)+V\circ\gamma\right)
\\ &=&
{\tt g}(\nabla_{\dot\gamma}\dot\gamma,\dot\gamma)+\nabla_{\dot\gamma}(V\circ\gamma)=
{\tt g}(\nabla_{\dot\gamma}\dot\gamma,\dot\gamma)+\d V(\dot\gamma)
\\ &=&
{\tt g}({\rm F},\dot\gamma)+\d V(\dot\gamma)=
\omega(\dot\gamma)+\d V(\dot\gamma)=0 \ .
\eeann
%because $\omega=-\d V$.
\qed \end{proof}

\begin{prop}
Let $(Q,{\tt g},\omega)$ be a conservative Newtonian mechanical system  with $\omega=-\d V$ and
let $\gamma\colon I\subset\Real\to Q$ be a curve with image on the domain $U\subset Q$ of a local chart
$(U,\varphi=(q^i))$ of $Q$. 
Then the Euler--Lagrange equations \eqref{eel} are 

\beq
\frac{d}{d t}\left(\derpar{\Lag}{v^j}\circ\dot\gamma\right)-
\derpar{\Lag}{q^j}\circ\dot\gamma=0 \ ,
\label{eel1e}
\eeq
that is, the Euler--Lagrange equations for the Lagrangian function $\Lag=K-\tau_Q^*V$.
\end{prop}
\begin{proof}
In fact we have that
$$
\frac{d}{d t}\left(\derpar{K}{v^j}\circ\dot\gamma\right)-
\derpar{K}{q^j}\circ\dot\gamma=
(-\d V\circ\gamma)\left(\derpar{}{q^j}\right)=
-\derpar{(\tau_Q^*V)}{q^j}\circ\gamma \ ,
$$
and recalling that
\(\dst\derpar{(\tau_Q^*V)}{v^j}=0\), we can write
$$
\frac{d}{d t}\left(\derpar{(K-\tau_Q^*V)}{v^j}\circ\dot\gamma\right)-
\derpar{(K-\tau_Q^*V)}{q^j}\circ\dot\gamma=0 \ .
$$
\qed \end{proof}

\begin{definition}
The function $\Lag:=K-\tau_Q^*V$ is said to be a \textbf{Lagrangian function of mechanical type}
or a \textbf{mechanical Lagrangian function},
and equations \eqref{eel1e} are the \textbf{Euler--Lagrange equations of the first class} of the conservative Newtonian system.
\end{definition}

From now on, we commit an abuse of notation and we write simply $\Lag=K-V$.

\bigskip
\begin{remark}{\rm 
\bit
\item
For conservative Newtonian mechanical systems, it is immediate to prove that the vector field  $X\in\vf(\tau_Q^{-1}(U))$ of Proposition \ref{cvXY}
satisfies the condition
$$
\Lie(X)\left(\derpar{\Lag}{v^k}\right) =
\derpar{\Lag}{v^k} \ ,
$$
(instead of the condition 2 in the aforementioned proposition).
\item
Observe that  the Lagrangian functions {\sl of mechanical type}
$$
 \Lag =K-V \equiv \frac{1}{2}g_{ij}(q)v^iv^j-V(q)
$$
are regular because
$$
\Omega_\Lag \equiv g_{ij}(q)\d q^i\wedge\d v^j \ .
$$
\eit
}\end{remark}

\subsection{Coupled systems (systems in  interaction)}
\protect\label{sdni}

Next we study the case of Newtonian mechanical systems which are
made of several dynamical systems in interaction; that is,
the so-called {\sl coupled systems}.

%\subsection{Definitions and equivalences}

Let $(Q_1,{\tt g}_1,\omega_1)$, $(Q_2,{\tt g}_2,\omega_2)$ be two Newtonian mechanical systems and ${\rm F}_1\in\vf (Q_1)$, ${\rm F}_2\in\vf (Q_2)$ the corresponding force fields. The dynamical equations for both systems are
\beq
\nabla^1_{\dot\gamma_1}\dot\gamma_1={\rm F}_1\circ\gamma_1
\quad ; \quad
\nabla^2_{\dot\gamma_2}\dot\gamma_2={\rm F}_2\circ\gamma_2
\label{eqsisdes}
\eeq
(with $\gamma_1\colon I\subset\Real\to Q_1$, $\gamma_2\colon I\subset\Real\to Q_2$). 
They constitute a non-coupled, separated, set of ordinary differential equations. 
Together, they are a system of  non-coupled ordinary differential equations.

Consider the system $(Q,{\tt g},\omega)$, where:
\bit
\item
$Q=Q_1\times Q_2$; with $\pi_1\colon Q\to Q_1$, $\pi_2\colon Q\to Q_2$.
\item
${\tt g}={\tt g}_1\oplus{\tt g}_2$.
\item
$\omega=\pi_1^*\omega_1+\pi^*_2\omega_2$.
\eit
We can see that $\nabla=\nabla^1\oplus\nabla^2$ 
is the Levi-Civita connection of the Riemannian metric ${\tt g}$, defined as follows:
 if $\gamma\colon I\subset\Real\to Q$
is given by $\gamma=(\gamma_1,\gamma_2)$, then
$$
\nabla_{\dot\gamma}\dot\gamma=
\nabla^1_{\dot\gamma_1}\dot\gamma_1+\nabla^2_{\dot\gamma_2}\dot\gamma_2 \ .
$$
Moreover, from the expression of $\omega$ we obtain that the associated force field ${\rm F}\in\vf (Q)$ is
$$
{\rm F}(q)=((q_1,q_2),{\rm F}_1(q_1),{\rm F}_2(q_2)) \ ,
$$
with $q\equiv (q_1,q_2)\in Q$.
Then the dynamical equation of the system is
$$
\nabla_{\dot\gamma}\dot\gamma={\rm F}\circ\gamma \	 ,
$$
which is equivalent to the non-coupled system (\ref{eqsisdes}).
From the physical point of view, this situation is a model of two
{\sl joint mechanical systems} but {\sl without interaction}.
With this is mind we have:

\begin{definition}
We say that $N$ Newtonian systems $(Q_\mu,g_\mu,{\rm F}_\mu)$, with
$\mu=1,\ldots ,N$, are \textbf{coupled} (or \textbf{in interaction}) if
${\rm F}_\mu\in\vf (Q_\mu,\pi_\mu)$, that is,
we have the following commutative diagram for the force field acting on the system
$$
\begin{array}{ccc}
& &\Tan Q_\mu \\
&
\begin{picture}(40,30)(0,0)
\put(7,20){\mbox{${\rm F}_\mu$}}
\put(0,0){\vector(4,3){40}}
\end{picture}
&
\Bigg\downarrow \tau_\mu
\\
Q=\prod_{\mu=1}^n Q_\mu
&
\begin{picture}(40,5)(0,0)
\put(20,5){\mbox{$\pi_\mu$}}
\put(0,0){\vector(1,0){40}}
\end{picture}
& Q_\mu
\end{array}
$$
\end{definition}

In this case, if $N=2$, we have a system of differential equation as in (\ref{eqsisdes}) but now it is coupled; that is, we have:
$$
\nabla^1_{\dot\gamma_1}\dot\gamma_1={\rm F}_1\circ\gamma={\rm F}_1\circ(\gamma_1,\gamma_2)
\quad ; \quad
\nabla^2_{\dot\gamma_2}\dot\gamma_2={\rm F}_2\circ\gamma={\rm F}_1\circ(\gamma_1,\gamma_2)
%\label{eqsisdes}
$$

We want to formulate this system as a single Newtonian mechanical system, so we need the following result:

\begin{lem}
Let $M_1,M_2$ be manifolds, $M=M_1\times M_2$ and  $\pi_i\colon M\to  M_i$ ($i=1,2$) the natural projections. Then
\ben
\item
$\vf (M)$ is canonically isomorphic to
$\vf (M_1,\pi_1)\times\vf (M_2,\pi_2)$ (as $\Cinfty(M)$-m\'odules).
\item
$\df^1(M)$ is canonically isomorphic to
$\df^1(M_1,\pi_1)\times\df^1(M_2,\pi_2)$ (as $\Cinfty(M)$-m\'odules).
\een
\end{lem}
\begin{proof}
Recall that if $q\equiv (q_1,q_2)\in M$, we have the isomorphism
$$
\begin{array}{ccccc}
\alpha_q & \colon & \Tan_qM & \longrightarrow & \Tan_{q_1}M_1\times\Tan_{q_2}M_2  \\
& & v &\mapsto & (\Tan_q\pi_1(v),\Tan_q\pi_2(v)) \ .
\end{array}
$$
Hence we obtain the sequence
$$
\vf (M) \mapping{\phi} \vf (M_1,\pi_1)\times\vf (M_2,\pi_2) \mapping{\psi}\vf (M) \ ,
$$
where the maps are defined as
\beann
\phi (X)(q)&:=&\alpha_q (X(q)) \ , \\
\psi (X_1,X_2)(q)&:=&\alpha_q^{-1}(X_1(q),X_2(q)) \ .
\eeann
If $\rho_i\colon\vf (M_1,\pi_1)\times\vf (M_2,\pi_2)\to\vf (M_i,\pi_i)$ ($i=1,2$)
are the natural projections, we have that:
\ben
\item
The map $\phi$ is well-defined and $\phi (X)$ is differentiable:

It is enough to observe that $\rho_i(\phi(X))$ is differentiable.
To obtain this consider $f\in\Cinfty (M_i)$, we have that
$$
((\rho_i(\phi(X)))f)(q)=\rho_i(\alpha_p(X(q)))f=\Tan_q\pi_i(X(q))f=
X(q)(f\circ\pi_i)=(X(f)\circ\pi_i))(q) \ ,
$$
which depends differentially on $q\in M$.
\item
The map $\phi$ is injective, because the maps $\alpha_q$ are isomorphisms.
\item
The map $\phi$ is a $\Cinfty(M)$-m\'odules isomorphism as can be seen directly.
\item
$\psi\circ\phi={\rm Id}_{\vf (M)}$;
in fact, if $X\in\vf(M)$ and $p\in M$, we have that
$$
((\psi\circ\phi)(X))(q)=\psi(\phi (X))(q)=\alpha_p^{-1}(\phi(X)(q))=
\alpha_q^{-1}(\alpha_p(X(q)))=X(q) \ .
$$
\een

The statement (2) is a direct consequence of (1).
\\ \qed \end{proof}

From this Lemma we obtain:

\begin{teor}
Two Newtonian mechanical systems in interaction,
$(Q_\mu,g_\mu,{\rm F}_\mu)$, $\mu=1,2$, are equivalent to a unique Newtonian mechanical system
 $(Q,{\tt g},{\rm F})$, where
\ben
\item
$Q=Q_1\times Q_2$; with $\pi_1\colon Q\to Q_1$, $\pi_2\colon Q\to Q_2$.
\item
${\tt g}_1\oplus {\tt g}_2$.
\item
$\omega=(\omega_1,\omega_2)$ and ${\rm F}=({\rm F}_1,{\rm F}_2)$.
\een
(we have identified
$\vf (M)$ with $\vf (M_1,\pi_1)\times\vf (M_2,\pi_2)$, and
$\df^1(Q)$ with $\df^1(Q_1,\pi_1)\times\df^1(Q_2,\pi_2)$,
in agreement with the above Lemma)
\footnote{
The meaning of this ``equivalence" is the following:
if $\gamma\colon I\subset\Real\to Q$ is a solution to the dynamical equation of $(Q,{\tt g},{\rm F})$, and $\gamma=(\gamma_1,\gamma_2)$,
then $\gamma_\mu\colon\subset\Real\to Q_\mu$, ($\mu=1,2$),
is a solution to the coupled dynamical equations of both systems $(Q_\mu,g_\mu,{\rm F}_\mu)$;
and conversely.
}.
\end{teor}
\begin{proof}
The metric ${\tt g}={\tt g}_1\oplus{\tt g}_2$ is Riemannian on $Q$. Its Levi-Civita connection $\nabla$ is $\nabla_1\oplus\nabla_2$;
then the dynamical equation associated to this Newtonian system can be split into the two components  on $Q_1$ and $Q_2$,
 according to the dynamical equations of the two different systems in interaction $Q_1$, $Q_2$; hence we obtain the result.

Observe that
$$
\inn({\rm F}){\tt g}=\inn({\rm F}_1,{\rm F}_2)({\tt g}_1\oplus {\tt g}_2)=
\inn({\rm F}_1){\tt g}_1+\inn({\rm F}_2){\tt g}_2=\omega_1+\omega_2 \ ,
$$
and the dynamical equation $\nabla_{\dot\gamma}\dot\gamma={\rm F}\circ\gamma$ 
splits in two different ones:
$$
\nabla^1_{\dot\gamma_1}\dot\gamma_1={\rm F}_1\circ\gamma={\rm F}_1\circ(\gamma_1,\gamma_2)
\quad ; \quad
\nabla^2_{\dot\gamma_2}\dot\gamma_2={\rm F}_2\circ\gamma={\rm F}_1\circ(\gamma_1,\gamma_2)
%\label{eqsisdes}
$$
\qed \end{proof}

\begin{remark}{\rm 
When studying a physical complex system, it is usual that we first consider the two component systems as independent; that is,
without interaction, $(Q_1,{\tt g}_1,{\rm F}_1)$, $(Q_2,{\tt g}_2,{\rm F}_2)$, and in a second approach we introduce the interaction
${\rm F}^{int}=({\rm F}^{int}_1,{\rm F}^{int}_2)$; that is, the force fields depending on the positions of both systems, then we obtain the system $(Q_1\times Q_2,{\tt g}_1\oplus {\tt g}_2,{\rm F}_1+{\rm F}^{int}_1,{\rm F}_2+{\rm F}^{int}_2)$.
}\end{remark}

\begin{corol}
A collection of $n$ coupled Newtonian mechanical systems $(Q_\mu,g_\mu,{\rm F}_\mu)$, $\mu=1,\ldots n$, is ``equivalent'' to a unique Newtonian mechanical system $(Q,{\tt g},{\rm F})$, where
\ben
\item
$Q=Q_1\times\ldots\times Q_N$; with $\pi_\mu\colon Q\to Q_\mu$.
\item
${\tt g}={\tt g}_1\oplus\ldots {\tt g}_N$.
\item
$\omega=(\omega_1,\ldots ,\omega_N)$ and ${\rm F}=({\rm F}_1,\ldots ,{\rm F}_N)$.
\een
(observe that we have identified
$\vf (M)$ with \(\dst\prod_{\mu=1}^N\vf (M_\mu,\pi_\mu)\), and
$\df^1(Q)$ with \(\dst\prod_{\mu=1}^N\df^1 (M_\mu,\pi_\mu)\)
in agreement with the above Lemma).
\end{corol}

The system $(Q,{\tt g},{\rm F})$ is a Newtonian mechanical system composed by different systems with interaction.

\noindent {\bf Local coordinate expressions}:
To write the dynamical equation in a local chart, consider the curve $\gamma\equiv (\moment{\gamma}{1}{n})$
(where $\gamma_\mu$ is the component of $\gamma$ on the manifold $Q_\mu$). The dynamical equation is written as
$$
\nabla^\mu_{\dot\gamma_\mu}\dot\gamma_\mu={\rm F}_\mu\circ\gamma \ .
$$
To obtain the local expression of $\nabla$, let
$(U_\mu,q_\mu^i)$ be local charts on $Q_\mu$.
We have that $(U,q_\mu^j)$, with $U=U_1\times\ldots\times U_N$
is a local chart on $Q$, and if $\{ (\Gamma_\mu)_{jk}^i\}$
are the Christoffel symbols of $\nabla_\mu$,
we have that
%\beann
%\nabla_{\derpar{}{q^j_\mu}}\derpar{}{q^k_\mu}&=&(\Gamma_\mu)_{jk}^i\derpar{}{q_\mu^i}
%\\
%\nabla_{\derpar{}{q^j_\mu}}\derpar{}{q^k_\nu}&=&0
%\qquad (\mu\not=\nu )
%\eeann
%\vspace{-1mm}
$$
\nabla_{\derpar{}{q^j_\mu}}\derpar{}{q^k_\mu}=(\Gamma_\mu)_{jk}^i\derpar{}{q_\mu^i} \quad ,\quad
\nabla_{\derpar{}{q^j_\mu}}\derpar{}{q^k_\nu}=0
\qquad (\mu\not=\nu ) \ ,
$$
because $\nabla$ is the Levi-Civita connection. From these expressions we have that
$$
\nabla_{\dot\gamma}\dot\gamma=
\sum_{\mu=1}^N\left(\ddot\gamma_\mu^j\derpar{}{q^j_\mu}+
(\Gamma_\mu)_{jk}^i\dot\gamma_\mu^j\dot\gamma_\mu^k
\derpar{}{q^i_\mu}\right) =
{\rm F}\circ\gamma \ ,
$$
and, for every $\mu=1,\ldots ,N$,
$$
\ddot\gamma_\mu^i\derpar{}{q^i_\mu}+
(\Gamma_\mu)_{jk}^i\dot\gamma_\mu^j\dot\gamma_\mu^k\derpar{}{q^i_\mu}=
{\rm F}^i_\mu\circ\gamma \ .
$$
Hence,  if $\dim\, Q_\mu=N_{\mu}$, for every $\mu=1,\ldots ,N$
and every $i=1,\ldots ,N_{\mu}$, we have 
$$
\ddot\gamma_\mu^i+
(\Gamma_\mu)_{jk}^i\dot\gamma_\mu^j\dot\gamma_\mu^k =
{\rm F}_\mu^i\circ\gamma \ .
$$

To write the Euler--Lagrange equations of the system,
we need to obtain the associated natural chart $(\tau^{-1}_\mu (U),q_\mu^j,v_\mu^j)$
in $\Tan Q$ and we have
$$
\frac{d}{d t}\left(\derpar{K}{v^j_\mu}\circ\dot\gamma\right)-
\derpar{K}{q^j_\mu}\circ\dot\gamma=
( g_\mu)_{jk}{\rm F}_\mu^k \ ,
$$
where the kinetic energy function is \(\dst K=\sum_{\mu=1}^NK_\mu\) ,
because ${\tt g}={\tt g}_1\oplus\ldots\oplus {\tt g}_N$. Hence, taking into account that
$$
\frac{d}{d t}\left(\derpar{K_\nu}{v^j_\mu}\right)= 0 \quad , \quad
\derpar{K_\nu}{q^j_\mu}=0 \quad ;\quad \mbox{if $\mu\not=\nu$} \ ,
$$
we write
$$
\frac{d}{d t}\left(\derpar{K_\mu}{v^j_\mu}\circ\dot\gamma\right)-
\derpar{K_\mu}{q^j_\mu}\circ\dot\gamma=
(g_\mu)_{jk}{\rm F}_\mu^k \ .
$$

Moreover, if the system is conservative, then $\omega=-\d V$ (with $V\in\Cinfty (Q)$);
that is,
$$
\omega=-\derpar{V}{q^j_\mu}\d q^j_\mu \ .
$$
If we consider the Lagrangian function of the system $\Lag=K-V$,
with local expression
$$
\Lag =\frac{1}{2}\sum_{\mu=1}^N(g_\mu)_{jk}v^j_\mu v^k_\mu-V=
\sum_{\mu=1}^NK_\mu-V \ ,
$$
we have
$$
\frac{d}{d t}\left(\derpar{\Lag}{v^j_\mu}\circ\dot\gamma\right)-
\derpar{\Lag}{q^j_\mu}\circ\dot\gamma=0 \ .
$$

Finally, if the force fields depend on the velocities, the expression would be the same writing ${\rm F}\circ\dot\gamma$ instead of ${\rm F}\circ\gamma$ in every equation. In this case
${\rm F}_\mu\in\vf (Q_\mu,\pi_\mu\circ\tau_\mu)$, where
$\tau_\mu\colon \Tan Q_{\mu}\to Q_\mu$ are the natural projections.

\section{Systems with holonomic and nonholonomic constraints}
\protect\label{sdnlh}

Another relevant topic is the study of dynamical systems in classical mechanics
with holonomic or nonholonomic constraints.
Nonholonomic systems are those subjected to constraints depending on positions and velocities, and
they  have been thoroughly studied 
(see, for instance,
\cite{Ar-89,Bli,CLMM-2002,LMM,LM-96,Els,Ga-70,Lewis2020,MCL-2000,Neimark-Fufaev,Ve}).

\subsection{Holonomic constraints. Holonomic d'Alembert Principle}

Let $(Q,{\tt g},\omega)$ be a Newtonian mechanical system and ${\rm F}\in\vf (Q)$
its force field. Let $S$ be a submanifold of $Q$
(usually called {\sl submanifold of holonomic constraints}) and
$j_S\colon S\hookrightarrow Q$ the natural embedding. 
The problem we study consists in describing the dynamics of the system when it is forced to evolve on the submanifold $S$.
To force this behaviour, first it is compulsory to apply a new force field ${\rm R}$, called {\sl constraint force}, which obliges the system to remain on $S$.
In general, this force depends, not only on the position, but also on the velocity; 
then ${\rm R}\in\vf (Q,\tau_Q)$ and, moreover, it is a new unknown to find.

Then we have a dynamical equation for curves
$\gamma\colon I\subset\Real\to S$, which is
\beq
\nabla_{\dot\gamma}\dot\gamma ={\rm F}\circ\gamma+{\rm R}\circ\dot\gamma \ .
\label{eqdinS}
\eeq

To solve this problem, we introduce the following:

\begin{assump}
{\rm (d'Alembert Principle}):
The constraint force ${\rm R}$ is orthogonal to the submanifold $S$;
that is, for every $q\in S$ and for every $ u,v\in\Tan_qS$, we have  ${\tt g}(u,{\rm R}(q,v))=0$.
\end{assump}

The question is how to solve the problem and obtain some information about the constraint force.
Let ${\tt g}_S:=j_S^*{\tt g}$. It is obvious that $(S,{\tt g}_S)$ is a Riemannian manifold.
Let $\nabla^S$ be the Levi-Civita connection associated to ${\tt g}_S$.
We have the following natural orthogonal splitting 
$$
\Tan_qQ=\Tan_qS\oplus (\Tan_qS)^\perp \qquad (\forall q\in S\subset Q) \ ,
$$
and the projections
$$
\pi_S(q)\colon \Tan_qQ\rightarrow\Tan_qS
\quad , \quad
\pi^\perp_S(q)\colon \Tan_qQ\rightarrow (\Tan_qS)^\perp \ ,
$$
and, as they are defined at every point, we have
$$
\pi_S\colon \Tan Q\vert_S\rightarrow\Tan S
\quad , \quad
\pi^\perp_S\colon \Tan Q\vert_S\rightarrow \Tan S^\perp \ .
$$
Then d'Alembert's Principle reduces to $\pi_S\circ{\rm R}=0$.
Hence, we have that:

\begin{prop}
$\nabla^S=\pi_S\circ\nabla$.
\end{prop}
\begin{proof}
It is a direct computation to prove that, on the vector fields tangent to $S$, the map $\pi_S\circ\nabla$ is a connection on $S$ and it is symmetrical. 
Furthermore, for every $X,Y,Z\in\vf (S)$, we have that
$$
(\pi_S\circ\nabla_Z){\tt g}(X,Y)={\tt g}((\pi_S\circ\nabla_Z)X,Y)+{\tt g}(X,(\pi_S\circ\nabla_Z)Y) \ ,
$$
and hence $\pi_S\circ\nabla$ is the Levi-Civita connection of ${\tt g}_S$.
\\ \qed \end{proof}

Now, if we take the dynamical equation (\ref{eqdinS}) and we split it into the tangent and orthogonal  components to $S$, we obtain
\bea
\pi_S(\nabla_{\dot\gamma}\dot\gamma) &=&
\pi_S\circ{\rm F}\circ\gamma+\pi_S\circ{\rm R}\circ\dot\gamma =
\pi_S\circ{\rm F}\circ\gamma \ ,
\label{eqdinsplit1}  \\
\pi_S^\perp(\nabla_{\dot\gamma}\dot\gamma) &=&
\pi_S^\perp\circ{\rm F}\circ\gamma+\pi_S^\perp\circ{\rm R}\circ\dot\gamma =
\pi_S^\perp\circ{\rm F}\circ\gamma +{\rm R}\circ\dot\gamma \ .
\label{eqdinsplit2}
\eea
Denoting ${\rm F}^S:=\pi_S\circ{\rm F}\in\vf (S)$, the projection of
 ${\rm F}$ on $S$, then equation \eqref{eqdinsplit1} is 
 \beq
 \nabla^S_{\dot\gamma}\dot\gamma = {\rm F}^S\circ\gamma \ ;
\label{eqdinresS}
\eeq
and this is the dynamical equation of the Newtonian mechanical system
$(S,{\tt g}_S,\omega_S)$, where $\omega_S=\inn({\rm F}^S){\tt g}_S$.

The solutions to equation (\ref{eqdinresS}) are curves
$\gamma\colon I\subset\Real\to S$ such that, if we introduce them into equation 
(\ref{eqdinsplit2}), it allows us to calculate the constraint force ${\rm R}$
for that trajectory, obtaining
$$
\nabla_{\dot\gamma}\dot\gamma -\nabla^S_{\dot\gamma}\dot\gamma=
{\rm F}\circ\gamma-{\rm F}^S\circ\gamma +{\rm R}\circ\dot\gamma \ .
$$
Then ${\rm R}\circ\dot\gamma\in\vf (Q,\dot\gamma)$.
Observe that we can calculate the constraint force only for each trajectory of the system, but not as a vector field depending on the velocities.

We also have:

\begin{prop}
$\omega_S=j_S^*\omega$.
\end{prop}
\begin{proof}
If $q\in S$ and $v\in\Tan_qS$ we have that
\beann
\omega_S(v)&=&(\inn({\rm F}^S){\tt g}_S)(v)={\tt g}_S({\rm F}^S,v)={\tt g}({\rm F}^S,v)={\tt g}({\rm F},v) \ ,
\\
(j_S^*\omega)(v)&=&(j_S^*\inn({\rm F})g)(v)={\tt g}({\rm F},v) \ ,
\eeann
and the result follows.
\\ \qed \end{proof}

We have obtained that the dynamical system describing the motion of the systems $(Q,{\tt g},\omega)$ constrained to move on the submanifold $S$, is the Newtonian mechanical system $(S,j_S^*{\tt g},j_S^*\omega)$.

D'Alembert's Principle tells us that the constraint force ${\rm R}\in\vf (Q,\tau_Q)$, the force that obliges the system to move on the submanifold $S$, is orthogonal to $S$, but we can set this principle in a dual way:

\begin{prop}
{\rm (Dual d'Alembert Principle)}:
Let $\rho=\inn({\rm R}){\tt g}$. Then $j_S^*\rho =0$.
\end{prop}
\begin{proof}
Let $q\in S$ and $u,v\in\Tan_qS$, then
$$(j_S^*\rho_{(q,v)})(u)={\tt g}_S({\rm R}(q,v),u))=0 \ .$$
\\ \qed \end{proof}

Consider now the linear momentum form of the constrained system,
$\theta_S\colon\Tan S\to\Tan^*S$. The equation of motion using $\theta_S$ is,
$$
\nabla^S_{\dot\gamma}(\theta_S\circ\dot\gamma)=\omega_S\circ\gamma
$$
and as $\omega_S=j_S^*\omega$, we have
$$
\nabla^S_{\dot\gamma}(\theta_S\circ\dot\gamma)=
j_S^*\nabla_{\dot\gamma}(\theta\circ\dot\gamma) \ .
$$

Next, we discuss  some particular cases:

\subsubsection{Systems with one constraint}

Let $S=\{ q\in Q \ ;\ \varphi (q)=0\}$, with $\varphi\in\Cinfty (Q)$ 
and suppose that $\d\varphi(q)\neq 0$, for every $q\in S$. This implies that $S$ is a submanifold of $Q$.
Let $X\in\vf (Q)$ such that $\inn (X){\tt g}=\d\varphi$, then
$X$ is orthogonal to $S$. Hence
$$
\pi^\perp_S({\rm F})=\frac{{\tt g}({\rm F},X)}{{\tt g}(X,X)}X=
\frac{\d\varphi ({\rm F})}{\| \d\varphi\|^2}X \ ,
$$
and we have that
$$
{\rm F}^S=\pi_S({\rm F})={\rm F}-\frac{\d\varphi ({\rm F})}{\| \d\varphi\|^2}X \ .
$$
As a consequence,
$$
\omega_S=\omega-\frac{\d\varphi ({\rm F})}{\| \d\varphi\|^2}\d\varphi \ ,
$$
and this allows us to find the trajectories of the system as solutions to the differential equation
$$
\nabla^S_{\dot\gamma}\dot\gamma=
{\rm F}\circ\gamma-\frac{\d\varphi ({\rm F})}{\| \d\varphi\|^2}X \ .
$$
We have that the constraint force along the solution $\gamma$ is given by the equation
$$
\nabla_{\dot\gamma}\dot\gamma-\nabla^S_{\dot\gamma}\dot\gamma=
\frac{\d\varphi ({\rm F})}{\| \d\varphi\|^2}\circ\gamma-{\rm R}\circ\dot\gamma \ ,
$$
where the only unknown is ${\rm R}\circ\dot\gamma$.

\subsubsection{Systems with several constraints}

Consider now
$S=\{ q\in Q \ ;\ \varphi_1(q)=0,\ldots ,\varphi_h(q)=0\}$,
with $\varphi_1\ldots ,\varphi_h\in\Cinfty (Q)$, such that
$\d\varphi_1(q),\ldots\d\varphi_h(q)$ are linearly independent at every point $q\in S$ (we assume that $S$ is not empty).
Let $\moment{Z}{1}{n-h}\in\vf (Q)$ such that:
\ben
\item
$\inn (Z_i)\d\varphi_j=0$.
\item
${\tt g}(Z_i,Z_j)=0$; $i\not= j$.
\een
To obtain these vector fields $Z_i$, it is enough to take vector fields
$\moment{X}{1}{n-h}\in\vf (Q)$ satisfying the first condition (a linear equation) and apply the well known {\sl Gram--Schmidt method}.
In this situation,  we have
$$
\pi_S({\rm F})=\sum_{i=1}^m\frac{{\tt g}({\rm F},Z_i)}{{\tt g}(Z_i,Z_i)}Z_i \ ;
$$
hence, as in the previous case, we obtain the dynamical equation and the expression of the constraint force along every trajectory.

\subsection{Euler--Lagrange equations}

In the above paragraph, we have studied the dynamics of a Newtonian mechanical system $(Q,{\tt g},\omega)$
constrained to move on the submanifold $j_S\colon S\hookrightarrow Q$. We have proved that its dynamics is given by the Newtonian mechanical system $(S,{\tt g}_S,\omega_S)$.
To write the corresponding Euler--Lagrange equation of this last system, take a local chart $(U,q^i)$ on $S$ and the corresponding natural lift $(\tau_Q^{-1}(U),q^i,v^i)$
to $\Tan S$. Then we have
\beq
\frac{d}{d t}\left(\derpar{K_S}{v^k}\circ\dot\gamma\right)-
\derpar{K_S}{q^k}\circ\dot\gamma=
({\tt g}_S)_{ik}{(\rm F^S)}^i \ ,
\label{equno}
\eeq
where $K_S\in\Cinfty(\Tan S)$ is the {\sl kinetic energy} of the system, defined by
$$
\begin{array}{ccccc}
K_S & \colon & \Tan S & \longrightarrow & \Real  \\
& & (q,v) &\mapsto & \frac{1}{2}{\tt g}_S(v,v)
\end{array}.
$$
Then we have that:

\begin{prop}
$K_S=(\Tan j_S)^*K$.
\end{prop}
\begin{proof}
If $q\in S$ and $v\in\Tan_qS$, then 
$$
K_S(q,v)=\frac{1}{2}({\tt g}_S)_{ik}(q)v^iv^k=\frac{1}{2}g_{ik}(q)v^iv^k
$$
and the result follows.
\\ \qed \end{proof}

Thus, equation (\ref{equno}) is
$$
\frac{d}{d t}\left(\derpar{(\Tan j_S)^*K}{v^k}\circ\dot\gamma\right)-
\derpar{(\Tan j_S)^*K}{q^k}\circ\dot\gamma=
(\omega_S)_k\circ\gamma =(j_S^*\omega)_k\circ\gamma \ .
$$

In the case that the dynamical system is conservative, that is $\omega=-\d V$,
then
$$
\omega_S=j_S^*\omega=-j_S^*\d V=-\d j_S^*V \ ,
$$
and the above equation takes the expression
$$
\frac{d}{d t}\left(\derpar{(\Tan j_S)^*K}{v^j_S}\circ\dot\gamma\right)-
\derpar{(\Tan j_S)^*K}{q^j}\circ\dot\gamma=
-\derpar{j_S^*V}{q^k}\circ\gamma \ ,
$$
or, equivalently,
$$
\frac{d}{d t}\left(\derpar{\Lag_S}{v^j}\circ\dot\gamma\right)-
\derpar{\Lag_S}{q^j}\circ\dot\gamma=0 \ ,
$$
where $\Lag_S:=(\Tan j_S)^*\Lag$.
Observe that the constraint force is not in these equations. This was one of the innovations developed by Lagrange.

If $(W,x^i)$ is a local chart on $Q$, and $(U,q^i)$ is another on $S$,
both adapted to the map $j\colon W\hookrightarrow U$, then we have the local expression 
$x^i=f^i(q)$, and hence
%\(\dst \frac{d x^i}{d t}=\derpar{f^i}{q^j}\frac{d q^j}{d t}\) ,
\(\dst \dot x^i=\derpar{f^i}{q^j}\dot q^j\).
This shows that  it is enough to know the Lagrangian function $\Lag$ of the unconstrained system, to introduce the expression of $x^i,\dot x^i$ using the above local expressions and, by direct derivation, to obtain the Euler--Lagrange equations of the constrained system.

\subsection{Some examples of Newtonian mechanical systems}

In the following examples the question is to identify the elements defining the different systems; that is,
the {\sl configuration manifold}, the {\sl Riemannian metric} and the \textsl{force field} or the {\sl work form}.

\subsubsection{Unconstrained particle in $\Real^3$}

Consider a particle with mass ${\rm m}$ moving on an open set
$Q\subset\Real^3$, and subjected to a force defined by a vector field
${\rm F}\in\vf (Q)$.

The {\sl geometric metric} ${\tt g}$ is the original metric on $\Real^3$.
Consider, instead, the Riemannian metric given by $\widetilde{\tt g}:={\rm m}{\tt g}$.
Observe that the corresponding Levi-Civita connections $\nabla$ and $\widetilde\nabla$ for both metrics ${\tt g}$ y $\widetilde{\tt g}$
are the same: it is enough to calculate their Christoffel symbols in any local chart, or recall the calculus of $\nabla_XY$ for the  Levi-Civita connection.
Let \(\dst\widetilde{\rm F}:=\frac{{\rm F}}{{\rm m}}\in\vf (Q)\),
and consider the Newtonian mechanical system $(Q,\widetilde{\tt g},\widetilde{\rm F})$.
The Newton equation is
$$
\nabla_{\dot\gamma}\dot\gamma =
\widetilde{\rm F}\circ\gamma =\frac{{\rm F}}{{\rm m}}\circ\gamma \ ,
$$
which, as it is usual in mechanics, can be written as
$$
{\rm m}\nabla_{\dot\gamma}\dot\gamma = {\rm F}\circ\gamma \ .
$$
Observe that the work form is given by
$$
\omega=\inn(\widetilde{\rm F})\widetilde {\tt g}=\inn ({\rm F}){\tt g} \ ,
$$
and the linear momentum form is
$$
\widetilde\theta\circ\dot\gamma=\inn(\dot\gamma)\widetilde {\tt g} \ ,
$$
whose local expression is
$$
\widetilde\theta\circ\dot\gamma={\rm m}g_{ij}\dot\gamma^i\d q^j\circ\gamma \ .
$$

\subsubsection{Particle constrained to a surface of $\Real^3$}

Consider now a particle with mass ${\rm m}$ on an open set
$Q\subset\Real^3$. We have seen above that the Newtonian system describing this situation is given by $(Q,\widetilde{\tt g},\widetilde{\rm F})$.
If $j_S\colon S\hookrightarrow Q$ is a regular surface and the particle is constrained to move on it, then the associated Newtonian system is $(S,\widetilde {\tt g}_S,\widetilde{\rm F}^S)$, where
$\widetilde{\tt g}_S=j_S^*\widetilde {\tt g}$ and $\widetilde{\rm F}^S=\pi_S\circ\widetilde{\rm F}$.
Hence, the dynamical equation is
$$
\nabla^S_{\dot\gamma}\dot\gamma=\widetilde{\rm F}^S\circ\gamma \ .
$$
If we know one solution $\gamma$, then the constraint force along $\gamma$ is given by the equation
$$
\nabla_{\dot\gamma}\dot\gamma -\nabla^S_{\dot\gamma}\dot\gamma=
\widetilde{\rm F}\circ\gamma-\widetilde{\rm F}^S\circ\gamma +\widetilde{\rm R}\circ\dot\gamma \ ,
$$
from where we can obtain $\widetilde{\rm R}\circ\dot\gamma$. 
Recall that ${\rm R}={\rm m}\widetilde{\rm R}$.

We can write the dual formulation using the work form
$$
\widetilde\omega=\inn({\rm F}){\tt g}=\inn(\widetilde{\rm F})\widetilde{\tt g} \ .
$$
with $\widetilde\omega_S=j_S^*\widetilde\omega$. 
Hence, we have that the dynamical equation is
$$
\nabla^S_{\dot\gamma}(\widetilde\theta_S\circ\dot\gamma)=\widetilde\omega_S\circ\gamma \ ,
$$
with $\widetilde\theta_S\circ\dot\gamma=\inn(\dot\gamma)\widetilde {\tt g}_S$,
and $\widetilde\theta_S=(\Tan j_S)^*\widetilde\theta$.

\subsubsection{Systems of particles in $\Real^3$}

Consider now a system made of $N$ particles, denoted  $\moment{P}{1}{N}$,
with masses $\moment{{\rm m}}{1}{N}$. We suppose that the particle $P_\mu$ moves in $Q_\mu$, open set of $\Real^3$, equipped with the usual metric ${\tt g}_\mu$ and the associated dynamical metric $\widetilde{\tt g}_\mu:={\rm m}_\mu {\tt g}_\mu$.

If the force field acting on every particle is
${\rm F}_\mu\in\vf (Q_\mu)$, that is ${\rm F}_\mu$ depends only on the position of the corresponding particle;
then we have  $N$ uncoupled mechanical systems, and their dynamical equation can be solved separately.

But, if  $\pi_\mu\colon \prod_{\nu=1}^N Q_\nu\to Q_\mu$,
and the forces on the system are a family of vector fields ${\rm F}_\mu\in\vf (Q_\mu,\pi_\mu)$;
 that is, the force ${\rm F}_\mu$ on the particle $P_\mu$ depends on the position of the other particles, then the $N$ particles are in interaction and,
as we have shown in Section  \ref{sdni}, the $N$ systems
$(Q_\mu,\widetilde{\tt g}_\mu,\widetilde{\rm F}_\mu)$ are equivalent to a unique Newtonian mechanical system $(Q,\widetilde{\tt g},\widetilde{\rm F})$ with
\ben
\item
\(\dst Q=\prod_{\mu=1}^NQ_\mu\) ,
\item
\(\dst\widetilde {\tt g}=\oplus_{\mu=1}^N\widetilde{\tt g}_\mu\) ,
\item
$\omega=(\omega_1,\ldots ,\omega_N)$ and
$\widetilde{\rm F}=(\widetilde{\rm F}_1,\ldots ,\widetilde{\rm F}_N)$,
\een
and the dynamical equation is the corresponding to this last system.

\subsubsection{Systems of particles on a submanifold}

In this case, we have a system of $N$ particles $\moment{P}{1}{N}$,
with masses $\moment{{\rm m}}{1}{N}$, moving respectively on open sets $\moment{Q}{1}{N}$ in $\Real^3$. We suppose that every one of these sets is equipped with the corresponding metric $g_\mu$.
If ${\rm F}$ is the force field acting on the system, then the system is given by $(Q,\widetilde{\tt g},\widetilde{\rm F})$,
as we have seen in the above paragraph. Observe that it is the same situation if the particles are in interaction or not.

If the dynamics is constrained to a submanifold
$j_S\colon S\hookrightarrow Q$, then there exists a constraint force ${\rm R}$,
and the dynamical equation is
$$
\nabla_{\dot\gamma}\dot\gamma=\widetilde{\rm F}\circ\gamma +\widetilde{\rm R}\circ\dot\gamma \ ;
$$
for curves $\gamma\colon I\subset\Real\to S$, where 
$\widetilde{\rm R}=(\widetilde{\rm R}_1,\ldots,\widetilde{\rm R}_N)$, ${\rm R}_{\mu}=
{\rm m}_{\mu}\widetilde{\rm R}_\mu$, $\mu=1,\ldots,N$.

To solve this equation assuming the \textsl{d'Alembert Principle};
that is, that $\widetilde{\rm R}$ is {\it $\widetilde{\tt g}$-orthogonal to $S$},
we need to decompose the equation projecting onto $S$ and on its orthogonal complement at every point, as we did in the general study, (see Section \ref{sdnlh}). Then we solve the $S$ component and use the orthogonal one to compute the constraint force along every solution.

\subsection{Nonholonomic constraints. Nonholonomic d'Alembert Principle}
\protect\label{slnh}

Let $(Q,{\tt g},\omega)$ be a Newtonian mechanical system and ${\rm F}\in\vf (Q)$ be the force field. Let $C$ be a submanifold of $\Tan Q$,
such that $\tau_Q(C)=Q$. In this situation $C$ is called the 
 {\sl submanifold of nonholonomic constraints}, and
$j_C\colon C\hookrightarrow \Tan Q$ is the natural embedding.
We want to describe the dynamics of the system when it is constrained to evolve on the submanifold $C$. The system is given by $(Q,{\tt g},F,C)$.

To solve this problem we suppose that there exists a  {\sl constraint force}
${\rm R}$ usually depending on the velocities,  that is ${\rm R}\in\vf (Q,\tau_Q)$, which forces the system to move on $C$ and it is unknown.
Then the Newton dynamical equation is given, in this case, for curves
$\gamma\colon I\subset\Real\to Q$ such that satisfy:
\ben
\item
$\dot\gamma(t)\in C$, $ t\in I$.
\item
$\nabla_{\dot\gamma}\dot\gamma ={\rm F}\circ\gamma+{\rm R}\circ\dot\gamma$.
\een
Thus, we need to state conditions allowing us to find the trajectories of the system and calculate ${\rm R}$\footnote{Arnold Sommerfeld, in \cite{Som-1952}, says that this force $\mathrm{R}$ is a ``geometric force'', versus $F$ which is an ``applied force''.}.

Let $(q,v)\in C$. The condition assumed on $C$,
 $\tau_Q (C)=Q$, tells us that the dimension of the subspace of 
${\rm V}_{(q,v)}(\Tan Q)$ which is tangent to $C$, does not depend on the point $(q,v)$. Let
$$
\Tan_{(q,v)}^VC={\rm V}_{(q,v)}(\Tan Q)\cap\Tan_{(q,v)}C=\{ w\in{\rm V}_{(q,v)}(\Tan Q)\ ;\ w\in\Tan_{(q,v)}C\}
$$
be the vertical subspace tangent to $C$. This is a vector subbundle of $\Tan Q$ and we can write $\Tan^VC={\rm V}(\Tan Q)|_{C}\cap\Tan C$ as vector bundles on the manifold $C$.
Consider the vertical lift from the point $q\in Q$ to $(q,v)$ given by
$$
\lambda_q^{(q,v)}\colon\Tan_qQ\to{\rm V}_{(q,v)}(\Tan Q)\, ,
$$
defined as:
$$
\lambda_q^{(q,v)}(u_{q}):\phi\mapsto\lim_{t\rightarrow 0}\frac{\phi(q,v+tu)-\phi(q,v)}{t}\, ,
$$
that is, the directional derivative of $\phi$ along $u_{q}$ at the point $(q,v)\in\Tan Q$.

As $\lambda_q^{(q,v)}$ is an isomorphism from $\Tan_{q}Q$ to $\Tan_{(q,v)}(\Tan Q)$, let $(\Tan_{(q,v)}^VC)_q$ the inverse image of
$\Tan_{(q,v)}^VC\subset\Tan_{(q,v)}(\Tan Q)$ by $\lambda_q^{(q,v)}$. Then $\Tan_{q}Q=(\Tan_{(q,v)}^VC)_q\oplus(\Tan_{(q,v)}^VC)^\perp_q$, being this an orthogonal decomposition with respect to ${\tt g}$.

Then we state:

\begin{assump}
{\rm (Nonholonomic d'Alembert Principle}):
The constraint force ${\rm R}\in\vf (Q,\tau_Q)$ satisfies that
$$
{\rm R}(q,v)\in(\Tan^V_{(q,v)}C)^\perp_q \ ,
$$
that is, ${\tt g}({\rm R}(q,v),w)=0$, for every $w\in (\Tan^V_{(q,v)}C)_q$.
\end{assump}

In the classical physics literature, the elements in $(\Tan^V_{(q,v)}C)_q$ are called {\sl virtual velocities}. 

\begin{remark}{\rm  If there are no constraints, that is $C=\Tan Q$, then for every $(q,v)\in C$ we have that $\Tan_{(q,v)}^VC=V_{(q,v)}(\Tan Q)$, hence $(\Tan^V_{(q,v)}C)_q=\Tan_{q}Q$ and $(\Tan^V_{(q,v)}C)^\perp_q=\{0\}$, that is ${\rm R}(q,v)=0$ and there is no constraint force.
}\end{remark}

\bigskip
This principle allows us to obtain the dynamical equations of the trajectories of the system and the constraint force along every trajectory, as we will see in the sequel. 

In order to obtain this, we need to characterize the subspace $(\Tan^V_{(q,v)}C)_q$ in relation with the constraints, that is the functions vanishing on the submanifold $C$.
First, let $\phi\in\Cinfty (\Tan Q)$, and consider the 1-form $\d^V\phi\in\df^1(Q,\tau_Q)$ defined by
$$
(\d^V\phi(q,v))(u)=\d\phi (\lambda_q^{(q,v)}(u))
 ,  \quad (q,v)\in\Tan Q, \quad u\in\Tan_qQ \ ;
$$
whose expression in a local natural chart $(q^i,v^i)$ of $\Tan Q$ can be calculated directly applying it to  $\dst \derpar{}{q^i}$ and we obtain \(\dst\d^V\phi=\derpar{\phi}{v^i}\d q^i\). We have the following result:

\begin{prop}\label{verticalvectors}
Let $(q,v)\in C$.
\begin{enumerate}
\item If $w\in\Tan_qQ$, then  \(\dst w\in(\Tan_{(q,v)}^VC)_q\) if, and only if,
$(\d^V\phi(q,v))(w)=0$,
for every $\phi\in\Cinfty (\Tan Q)$ such that $j_C^*\phi =0$.
\item Let $((\Tan_{(q,v)}^VC)_q)^o=\{\alpha\in\Tan^{*}_{q}; \alpha(w)=0, \forall w\in(\Tan_{(q,v)}^VC)_q\}\subset\Tan^{*}_{q}Q$ be the annihilator of $(\Tan_{(q,v)}^VC)_q$, then $
((\Tan_{(q,v)}^VC)_q)^o=\{\d^V\phi(q,v);\forall  \phi\in\Cinfty(\Tan Q),   j_C^*\phi=0\}.
$
\item If $w\in\Tan_qQ$, then  \(\dst w\in(\Tan_{(q,v)}^VC)_q\) if, and only if, $\inn(w){\tt g}\in((\Tan_{(q,v)}^VC)_q)^o$.
\end{enumerate}
\end{prop}
\begin{proof}
\begin{enumerate}
\item Let $w\in\Tan_qQ$, then we have
\begin{eqnarray*}
w\in(\Tan_{(q,v)}^VC)_q\ \Longleftrightarrow \
\lambda_q^{(q,v)}(w)\in\Tan_{(q,v)}^VC \ \Longleftrightarrow \\
\lambda_q^{(q,v)}(w)(\phi)=0, \forall  \phi\in\Cinfty(\Tan Q),  \mathrm{with} \,\, j_C^*\phi=0
\ \Longleftrightarrow\\
\d\phi(\lambda_q^{(q,v)}(w))=0, \forall  \phi\in\Cinfty(\Tan Q), \mathrm{with} \,\, j_C^*\phi=0 
\ \Longleftrightarrow \\
(\d^V\phi(q,v))(w)=0,\forall  \phi\in\Cinfty(\Tan Q),  \mathrm{with} \,\, j_C^*\phi=0\ . \quad
\end{eqnarray*}
\item It is a consequence of the previous item.
\item It is a direct consequence of the definitions.
\end{enumerate}
\qed \end{proof}

\begin{corol}
Let ${\rm R}\in\vf (Q,\tau_Q)$ and $(q,v)\in C$; then
${\rm R}(q,v)\in(\Tan^V_{(q,v)}C)^\perp_p$ if, and only if, 
$\inn({\rm R}(q,v))g\in((\Tan_{(q,v)}^VC)_q)^o$.
%$$
%(j_C^*(\inn({\rm R}){\tt g}))\big\vert_{(\Tan^V_{(q,v)}C)_q}=0 \ .
%$$
\end{corol}

Usually the submanifold $C$ is given by the vanishing of a finite family of constraint functions defined on $\Tan Q$. We want to characterize $((\Tan_{(q,v)}^VC)_q)^o$ using these constraints. Then, suppose that the submanifold $C$ is defined by the vanishing of $r$ functions $\{\phi^i\}$, with $r<n=\dim\, Q$, satisfying the condition
\(\dst \mathrm{rank}\,\left(\derpar{\coor{\phi}{1}{r}}{\coor{v}{1}{n}}\right)=r\). Then $\dim\, C=2n-r$.
We have:

\begin{prop}
Let $(q,v)\in C$, then
\begin{enumerate}
\item $\dim \Tan^V_{(q,v)}C=r$.
\item $(\Tan^V_{(q,v)}C)_{q}=\{w\in \Tan_{q}Q; (\d^V\phi^i(q,v))(w)=0, i=1,\ldots,r \}$.
\item If $\alpha\in\Tan_{q}^{*}Q$ satisfies $\alpha |_{(\Tan^V_{(q,v)}C)_{q}}=0$, then $\alpha$ is a linear combination of the elements $\d^V\phi^1(q,v), \ldots,\d^V\phi^r(q,v)$;
that is, the vector space $((\Tan_{(q,v)}^VC)_q)^o$ is generated by $\{\d^V\phi^1(q,v),\ldots,\d^V\phi^r(q,v)\}$.
\end{enumerate}
\end{prop}
\begin{proof}
Let $(q^{i}, v^{i})$ a natural coordinate system on $\Tan Q$.
\begin{enumerate}
\item The assumed condition
\(\dst {\rm rank}\,\left(\derpar{\coor{\phi}{1}{r}}{\coor{v}{1}{n}}\right)=r\) implies that, up to a change of order in the coordinates $\coor{q}{1}{n}$, we can suppose that
$$
\det\, \left(\derpar{\coor{\phi}{1}{r}}{\coor{v}{1}{r}}\right)\not= 0 \ .
$$
Then $(q^{1},\ldots,q^{n},\phi^1,\ldots,\phi^r,v^{r+1},\ldots,v^{n})$ is a local coordinate system of $\Tan Q$ by the Inverse Function Theorem. The vector space $V_{q,v}(\Tan Q)$ is generated by
$$
\left\{\derpar{}{\phi^1},\ldots,\derpar{}{\phi^r},\derpar{}{v^{r+1}},\ldots,\derpar{}{v^n}\right\}_{(q,v)}\, ,
$$
and the subspace  $\Tan^V_{(q,v)}C\subset V_{q,v}(\Tan Q)$ is generated by
$$
\left\{\derpar{}{\phi^1},\ldots,\derpar{}{\phi^r}\right\}_{(q,v)}\ .
$$
\item The inclusion part is proved in the first item of Proposition \ref{verticalvectors} and the equality follows from a dimensional analysis.
\item 
The previous items imply that $\{\d^V\phi^1(q,v),\ldots,\d^V\phi^r(q,v)\}$ is a basis of $
((\Tan_{(q,v)}^VC)_q)^o$. 
%Consequence of the previous items.
\end{enumerate}
\qed \end{proof}

Then, as a corollary, we obtain:

\begin{prop}
For every $(q,v)\in C$, the form $\eta\in\df^1(Q,\tau_Q)$ satisfies the condition
\(\dst j^*_C\eta\vert_{(\Tan^V_{(q,v)}C)_q}=0\)
if, and only if, there exist $\moment{f}{1}{r}\in\Cinfty (\Tan Q)$
such that $\eta=f_i\d^V\phi^i$
\end{prop}
\begin{proof}
It is a consequence that, for every $(q,v)\in C$, we have that $\eta(q,v)\in((\Tan_{(q,v)}^VC)_q)^o$.
\\ \qed \end{proof}

In the case that the constraints define only locally the submanifold $C$, then the above results are valid only in the corresponding open set.

The last proposition allows us to state the so called  
\rm{Dual d'Alembert nonholonomic Principle} which states that 
``the work form $\inn({\rm R})g$ corresponding to the constraint force ${\rm R}$ annihilates the virtual velocities of the system''.

And, as an immediate result, we have:

\begin{corol}
If ${\rm R}$ is the nonholonomic constraint force, then there exist
$\coor{f}{1}{r}\in\Cinfty (\Tan Q)$ such that
$$
\inn({\rm R}){\tt g}=f_i\d^V\phi^i=f_i\derpar{\phi^i}{v^j}\d q^j \ ,
$$
and, as a consequence,
$$
{\rm R}=f_i\derpar{\phi^i}{v^j}g^{jk}\derpar{}{q^k} \ .
$$
\end{corol}

\begin{definition}
The functions $\coor{f}{1}{r}$ are called {\sl \textbf{Lagrange multipliers}}
of the nonholonomic system.
\end{definition}

\noindent{\bf Particular case}: The submanifold $C\subset\Tan Q$ is a linear subbundle of  $\Tan Q$. 
\begin{enumerate}
\item In this case $C$ is defined by the vanishing of a family of differential forms, that is, we have $\omega^{i}\in\Omega^{1}(Q)$, $i=1,\ldots,r$, linearly independent at every point of $Q$, and
$$
C=\{(q,v)\in\Tan Q;\, \omega^{i}_{q}(v)=0,i=1,\ldots,r\}\, .
$$
In local coordinates, if $\omega^{i}=a^{i}_{j}(q)\d q^{j}$, then $\phi^i=a^{i}_{j}(q)v^{j}$, that is the constraints are linear in the velocities, and the expression of the constraint force is
$$
{\rm R}=f_ia^i_{j}g^{jk}\derpar{}{q^k} \ .
$$
\item 
Alternatively, we can suppose that the subbundle $C$ is given as a regular distribution $\mathcal{D}$, the distribution annihilated by $\{\omega^{i}, i=1,\ldots,r\}$. If $(q,v)\in\mathcal{D}$, by linearity, we have that  $\Tan_{(q,v)}^VC=\lambda_{q}^{(q,v)} (\mathcal{D}_{q})$, hence $(\Tan_{(q,v)}^VC)_{q}=\mathcal{D}_{q}$ and $(\Tan_{(q,v)}^VC)^\perp_q=\mathcal{D}^\perp_q$, then the constraint force ${\rm R}$ is orthogonal to $\mathcal{D}$.
\item 
If the distribution $\mathcal{D}$ is integrable and $(q,v)\in\mathcal{D}$ is the initial condition of the dynamical equation for the solution $\gamma$, then the image of $\gamma$ is contained in the integral submanifold of $\mathcal{D}$ passing through the point $q\in Q$, because $\dot\gamma(t)\in\mathcal{D}_{\gamma(t)}$ for every $t$. The constraint force ${\rm R}$, orthogonal to $\mathcal{D}$, forces the system to move on the integral submanifolds of the constraint distribution $\mathcal{D}$.
\end{enumerate}

\begin{remark}{\rm  
\bit
\item 
We can understand the solution as follows: if we have a constraint $\phi:\Tan Q\to\Real$, there is an associated 1-form, $\d^V\phi$, which gives a ``constraint force'' $R^{\phi}$ defined by $\inn(R^{\phi}){\tt g}=\d^V\phi$.
If we have $r$ independent constraints $\{\phi^i\}$, we have the corresponding constraint forces, $R^{\phi^i}$, and the subbundle generated by them, $\{R^{\phi^i}\}$, and then the resultant constraint force ${\rm R}$ is on this subbundle.
\item 
For these systems, d'Alembert's Principle says that, if the system moves ``along the vertical fibers'' of C,  the work realized by the constraint force along the trajectory is null. As we have said, these vertical velocities are called ``virtual'' because it is not possible to move the system with these velocities on the constraint manifold $C$ .
\eit
}\end{remark}

\subsection{Dynamical equations}

Following the above results,
if $C$ is locally defined by the annihilation of
$r$  functions $\{\phi^i\}$, and they are independent constraints,
then the dynamical equations are
$$
\nabla_{\dot\gamma}\dot\gamma ={\rm F}\circ\gamma+f_j\inn(\d^V\phi^j)g^{-1}\circ\dot\gamma \ ,
$$
or, in dual form,
$$
\nabla_{\dot\gamma}(\theta\circ\dot\gamma)=
\omega\circ\gamma+f_j\d^V\phi^j\circ\dot\gamma \ .
$$
These equations together with the constraints defining $C$, $\phi^1=0,\ldots,\phi^r=0$,
are a system of $n+r$ equations with $n+r$ unknowns:
the components of the trajectory $\gamma$ and the multipliers $f_i$. Observe that some of them are differential equations and the remaining ones are the constraint functions.

The corresponding Euler--Lagrange equations are
$$
\frac{d}{d t}\left(\derpar{K}{v^j}\circ\dot\gamma\right)-
\derpar{K}{q^j}\circ\dot\gamma=
\omega_j\circ\gamma +f_k\derpar{\phi^k}{v^j}\circ\dot\gamma
\quad , \quad (j=1,\ldots ,n) \ ,
$$
because $\omega=\omega_k\d q^k$, with $\omega_k=g_{kj}{\rm F}^j$.

If the system is conservative, then $\omega=-\d V$,
where $V\in\Cinfty (Q)$ is the potential function. In this case we can introduce the Lagrangian function $\Lag:=K-\tau_Q^*V$ and we have
$$
\frac{d}{d t}\left(\derpar{\Lag }{v^j}\circ\dot\gamma\right)-
\derpar{\Lag}{q^j}\circ\dot\gamma=
f_k\derpar{\phi^k}{v^j}\circ\dot\gamma \ .
$$
In any case, these equations, together with the constraints, defining $C$, are also a system of $n+r$ equations with $n+r$ unknowns.

If $C$ is a vector subbundle, then  $\phi^i=a^{i}_{j}(q)v^{j}$ and the Euler--Lagrange equations are
$$
\frac{d}{d t}\left(\derpar{K}{v^j}\circ\dot\gamma\right)-
\derpar{K}{q^j}\circ\dot\gamma=
\omega_j\circ\gamma +f_ka^k_j\circ\dot\gamma
\quad , \quad (j=1,\ldots ,n)\, .
$$
Or in the case of conservative systems
$$
\frac{d}{d t}\left(\derpar{\Lag }{v^j}\circ\dot\gamma\right)-
\derpar{\Lag}{q^j}\circ\dot\gamma=
f_ka^k_j\circ\dot\gamma \ .
$$

\vspace{0,5cm}

Summarizing, the fundamental data in the Lagrangian formulation
of the autonomous dynamical systems are the following:
\bit
\item
A Newtonian mechanical system (independent of time) is a  triple $(Q,{\tt g},\omega )$, where
$(Q,{\tt g})$ is a Riemannian manifold where the motion happens,
which is called the {\sl configuration space} of the system,
whose points represent the {\sl positions} of the particles of the system,
and
the form $\omega$ or, equivalently, the vector field ${\rm F}$,
are geometric elements carrying the dynamical information.
\item
The dynamical trajectories, which are the solutions to Newton equations, 
are curves on $\Tan Q$ which are canonical lifts of curves on $Q$.
The points of the manifold $\Tan Q$,
which are the possible initial conditions for these equations 
and, hence, they represent  the possible {\sl position} and {\sl velocities}
of the system) and are the {\sl physical states} of the  system, and
then $\Tan Q$ is said to be the
 {\sl state} or {\sl phase space} (of {\sl velocities\/}) of the system.
\item
In the particular case of conservative systems,
the geometric element representing the dynamics is substituted
by the  {\sl (mechanical) Lagrangian function}
$\Lag\in\Cinfty (\Tan Q)$, and  the dynamical equations
are  the {\sl Euler--Lagrange equations (of the first kind)}.
\eit

Finally, if  the degrees of freedom of the system
have some kind of restriction, we have two different situations: 
\begin{enumerate}
\item If the dynamics takes place on some submanifold
$j_S\colon S\hookrightarrow Q$, the situation is the same, but considering
$(S,j_S^*{\tt g},j_S^*\omega)$  as a Newtonian system.
Observe that, in this case, the configuration space is $S$
and the corresponding phase space (of velocities) is $\Tan S$.
\item If the configuration space is not restricted but the velocities are, then we have a nonholonomic system, and we need to introduce new unknowns, the Lagrange multipliers, to obtain the corresponding dynamical equations. 
\end{enumerate}

\section{Nonautonomous Newtonian systems}

In some interesting cases, the force field acting on a Newtonian system 
depends not only on the positions and the velocities, but also on time. 
In the following paragraphs, we extend our geometric formulation to this situation.

\subsection{Mechanical systems with time-dependent forces}

The geometric model appropriate to this case is the following:

\begin{definition}
A {\sl \textbf{nonautonomous Newtonian mechanical system}} is a triple
 $(\Real\times Q,{\tt g},{\rm F})$, where
$(Q,{\tt g})$ is a Riemannian manifold and the force field is of the form 
${\rm F}\in\vf (Q,\pi_2)$ (with $\pi_2\colon\Real\times Q\to Q$); that is,
$$
\begin{array}{ccc}
& &\Tan Q \\
&
\begin{picture}(40,30)(0,0)
\put(5,15){\mbox{${\rm F}$}}
\put(-6,-8){\vector(4,3){50}}
\end{picture}
&
\Bigg\downarrow \tau_Q
\\
\Real\times Q
&
\begin{picture}(40,5)(0,0)
\put(17,5){\mbox{$\pi_2$}}
\put(0,0){\vector(1,0){40}}
\end{picture}
& Q
\end{array} \ .
$$
Moreover, if the force field depends on the velocities, then 
${\rm F}\in\vf (Q,\tau_Q\circ\rho_2)$, with $\rho_2\colon\Real\times\Tan Q\to \Tan Q$;
that is,
$$
\begin{array}{ccc}
& &\Tan Q_\mu \\
&
\begin{picture}(60,30)(0,0)
\put(15,15){\mbox{${\rm F}$}}
\put(-6,-8){\vector(2,1){70}}
\end{picture}
&
\Bigg\downarrow \tau_Q
\\
\Real\times Q
&
\begin{picture}(20,5)(0,0)
\put(3,8){\mbox{$\rho_2$}}
\put(-6,3){\vector(1,0){25}}
\end{picture}
\Tan Q
\begin{picture}(20,5)(0,0)
\put(9,8){\mbox{$\tau_Q$}}
\put(5,3){\vector(1,0){25}}
\end{picture}
& Q
\end{array}\ . 
$$
\end{definition}

The Newton equations are written in the usual way:
\bit
\item
In the case that the force does not depend on the velocities,
$$
\nabla_{\dot\gamma}\dot\gamma ={\rm F}\circ\bar\gamma \ ,
$$
where $\bar\gamma=(t,\gamma)\colon I\subset\Real\to I\times Q$.
We can also use the dual form by means of the corresponding work form $\omega\in\df^1(Q,\pi_2)$.
\item
If the force field depends on the velocities, then
$$
\nabla_{\dot\gamma}\dot\gamma ={\rm F}\circ\bar{\dot\gamma} \ ,
$$
where $\bar{\dot\gamma}=(t,\dot\gamma)\colon I\subset\Real\to I\times\Tan Q$.
As above, we can use the corresponding work form $\omega\in\df^1(Q,\tau_Q\circ\rho_2)$ and obtain the equations in the dual form.
\eit

The  Euler--Lagrange equations are the same as usual, but the second term depends on time $t\in\Real$.
In particular, if the work form depends on time, $\omega\in\df^1(Q,\pi_2)$,
we say that the system is \textsl{conservative} if there exists
$V\colon\Real\times Q\to \Real$, such that $\omega=-\d V_t$,
where $V_t\colon Q\to Q$ is defined by $V_t(p):=V(t,p)$,
for every $p\in Q$ and $t\in\Real$.
In this situation we can define the Lagrangian function $\Lag:=K-V$, depending on time, and the Euler--Lagrange equation are as usual
$$
\frac{d}{d t}\left(\derpar{\Lag}{v^i}\circ\bar{\dot\gamma}\right)-
\derpar{\Lag}{q^i}\circ\bar{\dot\gamma}= 0 \ .
$$

Time dependent Newtonian systems with holonomic and nonholonomic constraints can be studied in the same way, as we see in the next section.

\subsection{Time-dependent holonomic and nonholonomic constraints}

\begin{definition}
A {\sl \textbf{nonautonomous holonomic Newtonian mechanical system}}
is a triple $(Q,{\tt g},\omega)$, where $(Q,{\tt g})$ is a Riemannian manifold, 
$\omega\in\df^1(Q)$), and we have an embedding 
$j_S\colon S\hookrightarrow \Real\times Q$,  where $S$ is a submanifold of $\Real\times Q$.
\end{definition}

In this situation, we need to suppose that the constraint force depends not only on positions and velocities but on time also; that is ${\rm R}\in\vf (Q,\tau_Q\circ\rho_2)$. The Newton equation is
$$
\nabla_{\dot\gamma}\dot\gamma ={\rm F}\circ\bar\gamma +{\rm R}\circ\bar{\dot\gamma} \ ,
$$
and d'Alembert's Principle can be stated as follows:

\begin{assump}
{\rm (Nonautonomous d'Alembert principle)}:
The constraint force ${\rm R}\in\vf (Q,\tau_Q\circ\rho_2)$ satisfies that,
for every $t\in\Real$, $p\in S$ and $u,v\in\Tan_qS$,
$$
{\tt g}({\rm R}(t,(\Tan_qj_S)(u)),(\Tan_qj_S)(v))=0 \ .
$$
\end{assump}

As in the autonomous case, this means that 
${\rm R}(t,(\Tan_qj_S)(u))\in(\Tan_{j_S(q)}S)^\perp$, and the results we obtain are the same as in Section \ref{sdnlh}.

\begin{definition}
A {\sl \textbf{nonautonomous nonholonomic Newtonian mechanical system}}
is a triple $(Q,{\tt g},\omega)$, where $(Q,{\tt g})$ is a Riemannian manifold, $\omega\in\df^1(Q)$),
and we have an embedding 
$j_C\colon C\hookrightarrow \Tan Q$, such that $(\tau_Q\circ j_C)(C)=Q$,
 where $C$ is a differentiable manifold.
\end{definition}

Once again we need to assume that the constraint force also depends on time, that is ${\rm R}\in\vf (Q,\tau_Q\circ\rho_2)$, and this makes the dynamical trajectories of the system, 
$\gamma\colon I\subset\Real\to Q$, to satisfy
$\dot\gamma(t)\in\Tan C$,  for every $t\in I$.
Hence, the Newton equation is 
$$
\nabla_{\dot\gamma}\dot\gamma ={\rm F}\circ\bar\gamma +{\rm R}\circ\bar{\dot\gamma} \ .
$$

To state the \textsl{nonholonomic d'Alembert Principle}, for every $t\in\Real$
and ${\rm p}\in C$, we need to define $(\Tan_{j_C({\rm p})}C)_{\tau_
Q({\rm p})}$,
as in the time independent case (see Section \ref{slnh}).

With this aim, let ${\rm p}\in C$ and $j_C({\rm p})=(q,v)\in j_C(C)\subset\Tan Q$, we take $\Tan^V_{(q,v)}C$ as the vertical subbundle tangent to $C$ and in the same way  $(\Tan^V_{(q,v)}C)_q$ using the vertical lift $\lambda_{q}^{(q,v)}$. Then

\begin{assump}
{\rm (Nonautonomous nonholonomic d'Alembert Principle)}:
The constraint force ${\rm R}\in\vf (Q,\tau_Q\circ\rho_2)$ satisfies that,
$$
{\tt g}({\rm R}(t,(q,v)),u)=0 \quad , \quad u\in(\Tan^V_{(q,v)}C)_q
$$
for every  ${\rm p}\in C$, $j_C({\rm p})=(q,v)\in j_C(C)$, and every $t\in\Real$; that is,
${\rm R}(t,j_C({\rm p}))\in(\Tan_{j_C({\rm p})}C)_{(\tau_Q \circ j_C)({\rm p})}^\perp $.
\end{assump}

The consequences of this principle on the dynamical equations and the calculus of the constraint force are similar to the autonomous case
(see Section \ref{slnh}), with the only difference that now the Lagrange multipliers $f_i$ are elements of $\Cinfty (\Real\times\Tan Q)$.

\begin{remark}{\rm  Sometimes it is necessary to work with constraints depending on time, called  \textit{rheonomic} constraints instead of \textit{scleronomic} constraints  or not depending on time (see  \cite{Sch-2005} for details). In this situation we have a submanifold $B\subset\Real\times\Tan Q$ such that corresponding manifolds, $B_{t}$ for any $t\in\Real$, are diffeomorphic. Then d'Alembert's Principle is stated at every $t$ for the submanifold $B_{t}$. 
}\end{remark}

%%%%%%%%%%%%%%%%%%%%%%%%%%%%%%%%%%%%%%%%%%%%%%%%%%%%%%%%%%%%%%%%%%%%%%%%%%%%%%%%%%%%%%%%%%%%%%%%%%%%%%%%%%%%%%%%%%

\chapter{An introduction to contact mechanics and dissipative dynamical systems}
\label{chap:contactmech}

In Chapter \ref{chap:cosym} we have studied nonautonomous or time-dependent systems;
one of the main characteristics of which is that, as discussed there, they are not conservative , but dissipative; 
that is, unlike the conservative systems studied in chapters \ref{sdl} and \ref{sdn},
energy is not a conserved quantity. 
However, in mechanics, there are many other systems that, 
being autonomous or not autonomous indistinctly, also manifest this characteristic \cite{Galley-2013,Ra-2006}. 
They are those in which forces of a nonconservative type appear, such as those described in Section \ref{vdf}

In recent years, there has been a great interest in geometrically studying these types of systems, by using techniques from {\sl contact geometry} \cite{ABKLR-2012,BHD-2016,BGG-2017,Geiges-2008}.
In fact, the {\sl contact structure} is quite similar to cosymplectic structure and allows us to give
a natural Hamiltonian description of mechanical systems with dissipation
\cite{Bravetti2017,BCT-2017,BLMP-2020,CG-2019,GG-2022,GG-2022b,LL-2018,LIU2018,Vi-2018}.
Their Lagrangian formalism \cite{CIAGLIA2018,DeLeon2019,GGMRR-2019b,GG-2022b},
and their unified Lagrangian-Hamiltonian formalism \cite{LGMMR-2020} have also been stated
in many different situations and application;
such as, damped oscillators, motion on viscous fluids, and  motion with friction in general.
It should be noted, however, that the first Lagrangian formulation for these kinds of systems
was first introduced by {\it G. Herglotz}
\cite{He-1930,Her-1985}
(see also \cite{GeGu2002,LPAF-2018,LIU2018} and \cite{LLM-2020} for a modern geometric version of this variational approach), who used a generalization of the Hamilton variational principle to obtain the
so-called {\sl Herglotz--Euler--Lagrange equations}, which are the same equations that are obtained using contact geometry.
In all these descriptions, the Lagrangian and the Hamiltonian functions depend also on an additional variable that, 
as we will see, can be identified with the ``action'' of the system;
and, for this reason, physicists often refer to them as {\sl action-dependent systems}.

%Nevertheless, in general, 
Contact geometry has also been used to describe different types of physical systems
in thermodynamics (for instance, in
\cite{Bravetti2017,Bravetti-2019,SLLM-2019,VS-2021}), 
quantum mechanics \cite{CIAGLIA2018,HW-2018,KA-2013}, circuit theory \cite{CIAGLIA2018,Goto-2016}, astrophysics \cite{GB-2019}, theoretical physics \cite{KA-2013}, control theory \cite{LLM-2020,RMS-2017}, etc.;
and even to describe other mechanical systems than just dissipative ones \cite{LR-2022}.
Finally, Herglotz's variational methods and the contact structure itself have been generalized 
in different ways for the treatment of action-dependent field theories
\cite{LGMRR-2022,GGMRR-2019,GGMRR-2020,GLMR-2022,Vi-2015}
(see also \cite{ACGL-2018,BH-2005,Bo-96,Almeida-2018,Fi2022,Mo-2008,Ri-2022,RSS-2023,TV-2008} for other contributions and less general approaches).

This whole body of doctrine, generically referred to as {\sl contact mechanics}, 
is currently a topic of active research, both in its foundation, and in its extensions and applications.

In this chapter, we review the main foundations of contact geometry 
and its application to describe dissipative autonomous dynamical systems.
The extension for the treatment of dissipative nonautonomous dynamical systems ({\sl cocontact formalism\/}) has been done in
\cite{LGGMR-2023,DeLeon2016b,RiTo-2022},
as well as to other types of systems and problems, such as 
higher-order dissipative systems \cite{LGLMR-2021},
nonholonomic systems \cite{LJL-2021,LLMR-2021}, or reduction theorems \cite{GG-2023,Wi-2002}.

We start giving the basic concepts and properties of {\sl contact manifolds} and the generic definition of {\sl contact Hamiltonian systems}. 
Next we develop the Lagrangian, Hamiltonian and unified formalisms for these kinds of systems.
We also study symmetries and the fundamental concept of {\sl dissipated quantities} geometrically;
defining, in particular, the notion of
{\sl contact Noether symmetry} for contact Hamiltonian and Lagrangian systems,
and establishing the statements of the so-called ``dissipation theorems''. 
These are analogous to the conservation theorems of conservative systems,
and show how to associate dissipated and conserved quantities to these symmetries.
Finally, the examples of the damped harmonic oscillator and the Kepler problem with dissipation
are analyzed in this context.

\section{A survey on contact geometry}
\label{scg}

(See, for instance, \cite{ABKLR-2012,
BHD-2016,Bravetti-2019,dN-2013,CIAGLIA2018,GGMRR-2019b,KA-2013,LM-sgam}
for more information).

\subsection{Contact manifolds}

\begin{definition}
\label{definition-contact-manifold}
A {\sl \textbf{contact manifold}} is a pair $({\rm M},\bmeta)$,
where ${\rm M}$ is a $(2n+1)$-dimensional manifold and
$\bmeta\in\df^1({\rm M})$ is a differential $1$-form
such that $\bmeta\wedge(\d\bmeta)^n$ is a volume form in ${\rm M}$.
Then, the form $\bmeta$ is called a {\sl \textbf{contact form}}
(or a {\sl \textbf{contact structure}}\/),
\end{definition}

As a straightforward consequence of this definition, we have that, 

\begin{teor}
\label{1stprop}
Given a contact manifold $({\rm M},\bmeta)$,
the condition that $\bmeta\wedge(\d\bmeta)^{n}$ is a volume form
is equivalent to have the splitting 
\beq
\label{splitcont}
\Tan {\rm M} =\ker\d\bmeta\oplus\ker\bmeta\equiv\mathcal{D}^{\rm R}\oplus\mathcal{D}^{\rm C} \ . 
\eeq
Then, there exists a unique vector field $\Reeb\in\vf({\rm M})$ such that
\begin{equation}
\label{eq-Reeb}
\inn(\Reeb)\d\bmeta = 0\quad ,\quad
\inn(\Reeb)\bmeta = 1 \ ,
\end{equation}
and hence it generates the distribution ${\cal D}^{\rm R}$.

Conversely, if we have two distributions ${\cal D}^{\rm R}$ and ${\cal D}^{\rm C}$
of ranks $1$ and $2n$ respectively, such that
\eqref{splitcont} holds, they define a contact structure on ${\rm M}$.
\end{teor}

\begin{definition}
The above vector field $\Reeb\in\vf({\rm M})$
is called the {\sl \textbf{Reeb vector field}},
and ${\cal D}^{\rm R}$ and $\mathcal{D}^{\rm C}$
are known as the {\sl \textbf{Reeb}} (or {\sl \textbf{horizontal distribution}}\/),
and the {\sl \textbf{contact}} (or {\sl \textbf{vertical distribution}} of $({\rm M},\bmeta)$.
\end{definition}

\begin{remark}{\rm
It is relevant to point out that, although Definition \ref{definition-contact-manifold} is the most assumed,
there are other different terminologies in relation to the concept of {\sl contact manifold}.
For instance, it can be defined by demanding the existence of two distributions
${\cal D}^{\rm R}$ and $\mathcal{D}^{\rm C}$, satisfying the properties stated in Theorem \ref{1stprop}, instead of using differential forms 
(see, for instance, \cite{Ar-89,Geiges-2008,GG-2022,KA-2013}),
and when these distributions are associated to a fixed form,
it is called a {\sl \textbf{strict contact manifold}}.
Also, in \cite{AM-78}, a more generic definition is given,
by demanding the existence of a $2$-form of maximal rank in an odd-dimensional manifold, 
and then calling {\sl \textbf{exact contact}} the case in which this form is exact
(which is the case of Definition \ref{definition-contact-manifold}).
}\end{remark}

Given a contact manifold $({\rm M},\bmeta)$, as a consequence of the splitting \eqref{splitcont}, there exists a vector bundle isomorphism
$$
    \begin{array}{rcl}
        \flat_{\bmeta}\colon \Tan {\rm M} & \to & \Tan^*{\rm M}\\ \noalign{\medskip}
       ({\rm p},X_{\rm p}) & \mapsto & 
\big({\rm p},\inn(X_{\rm p})\d\bmeta_{\rm p} + [\inn(X_{\rm p})\bmeta_{\rm p}]\bmeta_{\rm p}\big)\ .
    \end{array}
$$
and its inverse $\sharp_{\bmeta}=\flat_{\bmeta}^{-1}\colon \Tan^*{\rm M} \to \Tan {\rm M}$.
Their natural extensions are the $\Cinfty({\rm M})$-module isomorphisms
which are denoted with the same notation,
\beq
\begin{array}{rccl}
   \flat_{\bmeta}\colon & \vf({\rm M}) & \longrightarrow & \df^1({\rm M}) \\
   & X & \longmapsto & \inn(X)\d\bmeta+(\inn(X)\bmeta)\bmeta
\end{array} \ 
\label{isocont}
\eeq
and its inverse $\sharp_{\bmeta}=\flat_{\bmeta}^{-1}\colon \df^1({\rm M})\to \vf({\rm M})$.
In particular, for the Reeb vector field
we have that $\flat_{\bmeta}(\Reeb)=\bmeta$.

\begin{remark}
\label{nocontacto}{\rm
If in Definition \ref{definition-contact-manifold} the form $\bmeta$ is such that
$\bmeta\wedge(\d\bmeta)^{n}$ is not a volume form, 
but the rank of the distribution $\ker\d\bmeta\cap\ker\bmeta$ is constant 
and $\dim{\rm M}-{\rm rank}(\ker\d\bmeta\cap\ker\bmeta)$ is odd,
then we say that $\bmeta$ is a {\sl \textbf{precontact form}} on ${\rm M}$
and that $({\rm M},\bmeta)$ is a {\sl \textbf{precontact manifold}}
(moreover, ${\rm dim}\,{\rm M}$ could be arbitrary).
Under these hypotheses, there exist Reeb vector fields defined by \eqref{eq-Reeb}, but they are not uniquely defined,
and the map $\flat_{\bmeta}$ is not an isomorphism \cite{LGGMR-2023}.
(See also \cite{GG-2023} for a more general definition of precontact structure,
using {\sl contact distributions\/}).
}\end{remark}

\begin{prop}\label{prop-adapted-coord}
On a contact manifold $({\rm M},\bmeta)$, there are charts of coordinates $(U;z^I;s)$, $I=1,\ldots,2n$,
such that
$$
\bmeta\vert_U = \d s - f_I(z^J) \,\d z^I 
\quad ,\quad
\Reeb\vert_U = \frac{\partial}{\partial s}\ ,
$$
where $f_I$ are functions depending only on the~$z^J$.
(They are called $\textsl{adapted coordinates}$ of the contact structure.)
\end{prop}
\begin{proof}
On ${\rm M}$, we can take local charts of coordinates $(U;z^I,s)$, $I=1,\ldots,2n$, such that they rectify the vector field $\Reeb$; that is, $\Reeb=\displaystyle\derpar{}{s}$ on $U$.
Then $\displaystyle\bmeta\vert_U=a\,\d s-f_I(z^J,s)\d z^I$; but the conditions defining $\Reeb$ imply that  $\displaystyle\derpar{f_I}{s}=0$ and $a=1$, hence the result follows.
\\ \qed \end{proof}

Furthermore, one can prove the existence of Darboux-type coordinates:

\begin{teor}
{\rm (Darboux theorem for contact manifolds)}
Let $({\rm M},\bmeta)$ be a contact manifold. 
Then, for every point $p\in {\rm M}$ there exist a chart of coordinates 
$(U; x^i, y_i, s)$, $1\leq i\leq n$, such that
\begin{equation*}
\bmeta\vert_U= \d s - y_i\,\d x^i 
\quad ,\quad
\Reeb\vert_U = \frac{\partial}{\partial s}\ .
\end{equation*}
These are the {\sl \textbf{Darboux}} or {\sl \textbf{canonical coordinates}} of the contact manifold $({\rm M},\bmeta)$.
\end{teor}
\begin{proof}
Taking local charts of adapted coordinates,
$(U;z^I,s)$, if we do the quotient $U/{\cal D}^{\rm R}$,
the form $\d\bmeta\vert_U=\d f_I(z^J)\wedge\d z^I$
projects to the quotient and is a symplectic form on it.
Then we can take symplectic Darboux coordinates on the quotient
and pull them back to $U$, thus obtaining contact chart Darboux coordinates $(U; x^i, y_i, s)$ on ${\rm M}$.
\\ \qed \end{proof}

Relevant examples of contact manifolds are:

\noindent{\bf Canonical model}:
The canonical model for contact manifolds
is the manifold $\Tan^*Q\times\Real$.
In fact, if $\Theta\in\df^1(\Tan^*Q)$ 
is the canonical $1$-forms in $\Tan^*Q$,
$s$ is the Cartesian coordinate of $\Real$, and 
$\pi_1\colon \Tan^*Q \times\Real \to \Tan^*Q$ 
is the canonical projection, then 
$\bmeta=\d s-\pi_1^*\Theta= \d s-p_i\d q^i$ is a contact form in  $\Tan^*Q\times\Real$,
the Reeb vector field is
$\displaystyle\Reeb=\frac{\partial}{\partial s}$, and $(s,q^i,p_i)$ are Darboux coordinates on $\Tan^*Q\times\Real$.

\noindent{\bf Contactification of a symplectic manifold}.
If $(P,\omega)$ is an exact symplectic manifold such that $\omega=-\d\theta\in\df^2(P)$,
consider the manifold ${\rm M} =P\times\Real$.
If $s$ is the Cartesian coordinate of~$\Real$,
then the 1-form
$\bmeta=\d s-\theta\in\df^1({\rm M})$
(where we have denoted also by $\theta$ the pull-back of~$\theta$ to ${\rm M}$) is a contact form,
and $({\rm M},\bmeta)$ is a contact manifold which is called the
{\sl\textbf{contactified}} of~$P$.
Observe that the canonical model, $\Tan^*Q\times\Real$,
is the contactified of $\Tan^*Q$ endowed with its canonical symplectic structure.

\subsection{Hamiltonian, gradient, and evolution vector fields on a contact manifold}

As in the case of cosymplectic manifolds,
for a contact manifold $({\rm M},\bmeta)$,
the existence of the natural $\Cinfty({\rm M})$-modules isomorphism $\flat_{\bmeta}$,
introduced in the above section,
allows us to associate some characteristic vector fields to a function $f\in\mathcal{C}^\infty({\rm M})$:

\begin{definition}
\label{GraHaEv}
Let $({\rm M},\bmeta)$ be a contact manifold and $f\in\mathcal{C}^\infty({\rm M})$.

 The {\sl \textbf{Hamiltonian vector field}} associated with $f$
is the vector field $\X_f\in\vf({\rm M})$ defined by
\ $\flat_{\bmeta}(\X_f):=\d f-(\Reeb(f)+f)\bmeta$.

The {\sl \textbf{gradient vector field}} associated with $f$
is the vector field ${\bf grad\, f}\in\vf({\rm M})$ defined by
\ $\flat_{\bmeta}({\bf grad}\, f):=\d f$.

The {\sl \textbf{evolution vector field}}  associated with $f$
is the vector field \ $\evo_f\in\vf({\rm M})$ defined as
 \ $\evo_f:=f\,\Reeb+\X_f$ \ or, equivalently,
 \ $\flat_{\bmeta}(\evo_f):=\d f-\Reeb(f)\,\bmeta$.
\end{definition}

\begin{lem} \label{prop:grad-evol}
\quad ${\bf grad}\,f=\X_f+(\Reeb(f)+f)\,\Reeb$.
\end{lem}
\begin{proof}
Observe that,
    $$
\flat_{\bmeta}(\X_f)=\d f-(\Reeb(f)+f)\,\bmeta=
\flat_{\bmeta}({\bf grad}\,f)-(\Reeb(f)+f)\bmeta \ ,
$$
then, being $\flat_{\bmeta}$ a diffeomorphism and bearing in mind that $\flat_{\bmeta}(\Reeb)=\bmeta$, the result follows.
\\ \qed \end{proof}

Taking this into account, these vector fields can be equivalently characterized as follows:

\begin{prop}
\label{GraHaEvProp}
The Hamiltonian vector field associated with $f$
is determined equivalently by the equations:
\beq
\inn(\X_f)\bmeta=-f \quad,\quad 
\inn(\X_f)\d\bmeta = \d f-\Reeb(f)\,\bmeta\ .
\label{evolvf1}
\eeq

The gradient vector field associated with $f$ is determined equivalently by the equations:
\beq
 \inn({\bf grad}\,f)\bmeta=\Reeb(f) 
 \quad ,\quad 
\inn({\bf grad}\,f)\d\bmeta=\d f-\Reeb(f)\,\bmeta\ .
\label{evolvf2}
\eeq

The evolution vector field  associated with $f$
is determined equivalently by the equations:
\beq
\inn(\evo_f)\bmeta=0 \quad ,\quad \inn(\evo_f)\d\bmeta=\d f-\Reeb(f)\,\bmeta\ .
\label{evolvf3}
\eeq
\end{prop}
\begin{proof}
For every $f\in\Cinfty({\rm M})$,
for the Hamiltonian vector field $\X_f$,
using the definitions of the isomorphism $\flat_{\bmeta}$ and of $\X_f$, first we have that,
\beq
\label{aux1}
\inn(\X_f)\d\bmeta+(\inn(\X_f)\bmeta)\bmeta=
\flat_{\bmeta}(\X_f)=\d f-(\Reeb(f)+f)\bmeta \ ,
\eeq
therefore, contracting both members with the Reeb vector field and using \eqref{eq-Reeb}, we get
$$
(\inn(\X_f)\bmeta)\inn(\Reeb)\bmeta=
\inn(\Reeb)\d f-(\Reeb(f)+f)\inn(\Reeb)\bmeta=
-f\,\inn(\Reeb)\bmeta
\ \Longleftrightarrow\ \inn(\X_f)\bmeta=-f \ .
$$
Now, going to \eqref{aux1}, we obtain that
$$
\inn(\X_f)\d\bmeta=\d f-\Reeb(f)\,\bmeta \ .
$$
Conversely, using \eqref{evolvf1}
in the Definition \eqref{isocont} of $\flat_{\bmeta}$,
we get $\flat_{\bmeta}(\X_f):=\d f-(\Reeb(f)+f)\bmeta$.

For the gradient vector field ${\bf grad}\,f$,
from the definition of ${\bf grad}\,f$ and the above Lemma \ref{prop:grad-evol},
and using the above results and \eqref{eq-Reeb}, we have
\begin{align*}
\inn({\bf grad}\,f)\d\bmeta &=
\inn(\X_f)\d\bmeta+(\Reeb(f)+f)\,\inn(\Reeb)\d\bmeta=
\inn(\X_f)\d\bmeta=\d f-\Reeb(f)\,\bmeta \ , 
\\
\inn({\bf grad}\,f)\bmeta&= 
\inn(\X_f)\bmeta+(\Reeb(f)+f)\,\inn(\Reeb)\bmeta=
-f+(\Reeb(f)+f)=\Reeb(f) \ .
\end{align*}
Conversely, using \eqref{evolvf2}
in the Definition \eqref{isocont} of $\flat_{\bmeta}$,
we obtain that $\flat_{\bmeta}({\bf grad}\, f):=\d f$.

Finally, for the evolution vector field $\evo_f$, from the definition of $\evo_f$ and using the above results,
 \begin{align*}
\inn(\evo_f)\d\bmeta&=\inn(\X_f)\d\bmeta+f\inn(\Reeb)\d\bmeta=\inn(\X_f)\d\bmeta=
\d f-(\Reeb(f))\,\bmeta \ ,
   \\
\inn(\evo_f)\bmeta &= \inn(\X_f)\bmeta+f\inn(\Reeb)\bmeta=-f+f=0\ .
 \end{align*}
Once again, using \eqref{evolvf3}
in the Definition \eqref{isocont} of $\flat_{\bmeta}$,
we obtain that $\flat_{\bmeta}(\evo_f):=\d f-\Reeb(f)\,\bmeta$.
\qed \end{proof}

As an immediate consequence of Proposition \ref{GraHaEvProp} we obtain that:

\begin{prop}
The equations for the integral curves $c\colon I\subset\Real\to {\rm M}$
of the Hamiltonian, the gradient, and the evolution vector fields associated with
$f\in\mathcal{C}^\infty({\rm M})$ are, respectively,
\bea
\label{hamilton-contactc-curves-eqs}
\inn(\widetilde c)(\eta\circ c)=-f\circ c \ & ,& \
\inn(\widetilde c)(\omega\circ c) = (\d f- R(f)\eta)\circ c\ ,  \\
 \inn(\widetilde c)(\eta\circ c) = R(f)\circ c \ & , &\
\inn(\widetilde c)(\omega\circ c)=(\d f-R(f)\eta)\circ c\ , \nonumber \\
\inn(\widetilde c)(\eta\circ c)=0 \ & ,& \ 
\inn(\widetilde c)(\omega\circ c))=(\d f-R(f)\eta)\circ c\ , \nonumber
%\label{evoic01}
\eea
\end{prop}

\noindent {\bf Local expressions}:
In Darboux coordinates $(x^i, y_i,s)$ on ${\rm M}$, we have that $\dst \Reeb(f)=\derpar{f}{s}$, then
$$
\d f-\Reeb(f)\bmeta=\left(\frac{\partial f}{\partial x^i} + 
y_i\frac{\partial f}{\partial s}\right)\d x^i+\derpar{f}{y_i}\,\d y_i \ ,
$$
and, from \eqref{evolvf1}, \eqref{evolvf2}, and \eqref{evolvf3} (or also using \eqref{evolvf1}, Lemma \ref{prop:grad-evol} and the definition of $\evo_f$), we obtain,
\begin{align}
\noalign{\medskip}
\X_f= \frac{\partial f}{\partial y_i}\frac{\partial}{\partial x^i} - 
\left(\frac{\partial f}{\partial x^i} + 
y_i\frac{\partial f}{\partial s}\right)\frac{\partial}{\partial y_i} + 
\left(y_i\frac{\partial f}{\partial y_i} - f\right)\frac{\partial}{\partial s}\ , 
\label{evocoor01} \\ 
\noalign{\medskip} {\bf grad}\,f = \frac{\partial f}{\partial y_i}\frac{\partial}{\partial x^i} - 
\left(\frac{\partial f}{\partial x^i} + 
y_i\frac{\partial f}{\partial s}\right)\frac{\partial}{\partial y_i} + 
\left(y_i\frac{\partial f}{\partial y_i} + \derpar{f}{s}\right)\frac{\partial}{\partial s}\ , \label{evocoor02} \\
\noalign{\medskip}
\evo_f = \frac{\partial f}{\partial y_i}\frac{\partial}{\partial x^i} - 
\left(\frac{\partial f}{\partial x^i} + 
y_i\frac{\partial f}{\partial s}\right)\frac{\partial}{\partial y_i} + 
y_i\frac{\partial f}{\partial y_i}\,\frac{\partial}{\partial s}\ .
\label{evocoor03}
\end{align}
Therefore, if $c(t)=(x^i(t), y_i(t),s(t))$ is an integral curve of any of these vector fields, this implies that $c(t)$ should satisfy,
respectively, the following systems of differential equations:
\begin{align}
\noalign{\medskip}
\label{hel-eqs01}
\dst\frac{dx^i}{dt} = \frac{\partial f}{\partial y_i}\quad ,\quad
\dst\frac{dy_i}{dt} = -\left(\frac{\partial f}{\partial x^i} + y_i\frac{\partial f}{\partial s}\right)\quad , \quad
\dst\frac{ds}{dt} = y_i\frac{\partial f}{\partial y_i} - f \ ,  \\
\noalign{\medskip}\nonumber
\dst\frac{dx^i}{dt} = \frac{\partial f}{\partial y_i} \quad  ,\quad
\dst\frac{dy_i}{dt} = -\left(\frac{\partial f}{\partial x^i} + y_i\frac{\partial f}{\partial s}\right)\quad ,
\quad \dst\frac{ds}{dt} = y_i\frac{\partial f}{\partial y_i} + \derpar{f}{s}\ , \\
\noalign{\medskip}\nonumber
\dst\frac{dx^i}{dt} = \frac{\partial f}{\partial y_i}\quad ,\quad
\dst\frac{dy_i}{dt} = -\left(\frac{\partial f}{\partial x^i} + y_i\frac{\partial f}{\partial s}\right)\quad ,\quad
\dst\frac{ds}{dt} = y_i\frac{\partial f}{\partial y_i} \ . 
\end{align}

\section{Contact Hamiltonian dynamical systems}
\label{chs}

As we shall see in the following sections, the contact structures and their underlying tools constitute an ideal framework for the geometric description of dissipative dynamic systems.
(See, for instance, \cite{BCT-2017,CIAGLIA2018,DeLeon2019,GGMRR-2019b,Go-69,LL-2018}  for more details).

\subsection{Contact Hamiltonian systems}
\label{nasposcon}

Following the same guidelines as in all previous chapters,
the geometric study of dissipative (nonconservative or action-dependent) Hamiltonian dynamical systems
in general, is based on the following set of postulates:

\begin{pos}
{\rm (First Postulate of contact Hamiltonian mechanics\/)}:
The phase space of a regular (resp. singular)
dissipative dynamical system 
is a differentiable manifold ${\rm M}$ endowed with a contact 
(resp. precontact) structure $\bmeta\in\df^1({\rm M})$:
\label{axicon1}
\end{pos}

\begin{pos}
{\rm (Second Postulate of contact Hamiltonian mechanics\/)}:
The observables or physical magnitudes of a  dissipative dynamical system
are functions of $\Cinfty ({\rm M})$.
\label{axicon2}
\end{pos}

\begin{pos}
{\rm (Third Postulate of contact Hamiltonian mechanics\/)}:
The dynamics of a dissipative dynamical system
is given by a function $h\in\Cinfty({\rm M})$
(or, in general, a closed $1$-form $\alpha \in Z^1({\rm M})$, such that $\alpha=\d h$, locally)
which is called the {\sl \textbf{Hamiltonian function}}
(or  the {\sl \textbf{Hamiltonian $1$-form}}\/) of the system.
This function represents the energy of the system.
\label{axicon3}
\end{pos}

\begin{pos}
{\rm (Fourth Postulate of contact Hamiltonian mechanics\/)}:
The dynamical trajectories of the dissipative system are the integral curves
of the Hamiltonian vector field $\X_h\in\vf({\rm M})$ associated with $h$;
that is, of the vector field solution to equations \eqref{evolvf1}, that are now written as,
\beq
\inn(\X_h)\bmeta=-h \quad,\quad 
\inn(\X_h)\d\bmeta = \d h-\Reeb(h)\,\bmeta\ .
\label{evolvf111}
\eeq
Thus, these trajectories are the solutions to equations \eqref{hel-eqs01}, which read as,
\beq
\label{hel-eqs011}
\dst\frac{dx^i}{dt} = \frac{\partial h}{\partial y_i}\quad ,\quad
\dst\frac{dy_i}{dt} = -\left(\frac{\partial h}{\partial x^i} + y_i\frac{\partial h}{\partial s}\right)\quad , \quad
\dst\frac{ds}{dt} = y_i\frac{\partial h}{\partial y_i} - h \ , 
\eeq
\label{axicon4}
\end{pos}

It is interesting to point out the different roles of the Hamiltonian and the evolution vector fields in cosymplectic and in contact mechanics:
in cosymplectic mechanics the dynamic is given by  the evolution vector field, but in contact mechanics is given by the Hamiltonian vector field
\footnote{
Although the terminology could be changed to make these aspects coherent, it has not been done because this is the usual nomenclature in the bibliography.
}.

\begin{definition}
A {\sl \textbf{regular dissipative}} or {\sl \textbf{contact Hamiltonian dynamical system}}
is a set $({\rm M},\bmeta,h)$,
where $({\rm M},\bmeta)$ is a contact manifold 
and $h\in\Cinfty({\rm M})$ is the Hamiltonian function of the system.
If $({\rm M},\bmeta)$ is a precontact manifold, then $({\rm M},\bmeta,h)$
is said to be a {\sl \textbf{singular dissipative}} or {\sl \textbf{precontact Hamiltonian dynamical system}}
\footnote{
We write ‘(pre)contact’ to refer interchangeably to both situations (contact and precontact), and write ‘contact’ or ‘precontact’ to distinguish each of them in particular.
In any case, we will refer to this formalism as {\sl contact mechanics}, and we will talk about {\sl contact dynamical systems}, in general.
}.

The equations 
\eqref{evolvf111} and \eqref{hel-eqs011}
are the {\sl \textbf{(pre)contact Hamiltonian equations}} 
for $\X_h$ and its integral curves, respectively.
\label{stdhrbis}
\end{definition}

\begin{definition}
Given a dissipative Hamiltonian dynamical system $({\rm M},\bmeta,h)$, 
the \textbf{Hamiltonian problem} posed by  the system
consists in finding the Hamiltonian vector field $\X_h \in \vf ({\rm M})$
associated with $h$ (if it exists).
\end{definition}

\subsection{Properties of contact Hamiltonian systems}

\begin{prop}
\label{teo-evoeqs1}
If $({\rm M},\bmeta,h)$ is a regular dissipative Hamiltonian system, 
then there exists a unique Hamiltonian vector field $\X_h\in\vf({\rm M})$;
that is, a unique vector field which is the solution to equations \eqref{evolvf111}.
\end{prop}
\begin{proof}
It is a straightforward consequence of the existence of the isomorphism $\sharp_{\bmeta}$ in the regular case.
\\ \qed \end{proof}

\begin{remark}{\rm 
As it happens with singular dynamical systems,
if$({\rm M},\bmeta;h)$ is a precontact Hamiltonian system,
equations \eqref{evolvf111} are not necessarily compatible everywhere on ${\rm M}$ 
and the corresponding constraint algorithm must be implemented in order to find 
a {\sl final constraint submanifold} $P_f\hookrightarrow {\rm M}$
(if it exists) where there are contact Hamiltonian vector fields $\X_h\in\vf({\rm M})$,
tangent to $P_f$, which are solutions (not necessarily unique) to equations \eqref{evolvf111} on $P_f$.
}\end{remark}

\begin{prop}
Let $({\rm M},\bmeta,h)$ be a contact Hamiltonian system.
Then, the contact Hamiltonian equations \eqref{evolvf111} can be equivalently written as
\begin{equation}
 \Lie(\X_h)\bmeta= -\Reeb(h)\,\bmeta \quad ,\quad  \inn(\X_h)\bmeta = -h \ .
 \label{equivmanner}
\end{equation}
and its integral curves $c\colon I\subset\Real\to {\rm M}$  are solutions to the equations
\begin{equation*}
 \Lie(\widetilde c)(\bmeta\circ c))= -\big((\Reeb(h))\,\bmeta\big)\circ c \quad ,\quad  \inn(\widetilde c)(\bmeta\circ c)) = -h\circ c  \ .
\end{equation*}
\end{prop}
\begin{proof}
As $\inn(\X_h)\bmeta=-h$, from equations \eqref{evolvf111} we obtain,
$$
\Lie(\X_h)\bmeta=
\inn(\X_h)\d\bmeta+\d\inn(\X_h)\bmeta=
\d h-\Reeb(h)\,\bmeta-\d h=-\Reeb(h)\,\bmeta \ .
$$
Conversely, from \eqref{equivmanner},
$$
-\Reeb(h)\,\bmeta=\Lie(\X_h)\bmeta= 
\inn(\X_h)\d\bmeta+\d\inn(\X_h)\bmeta=
\inn(\X_h)\d\bmeta-\d h \ \Longleftrightarrow \
\inn(\X_h)\d\bmeta=\d h-\Reeb(h)\,\bmeta \ .
$$

From here, the equation for the integral curves is immediate. 
\\ \qed\end{proof}

The first result shows that,
unlike in symplectic Hamiltonian systems, the geometric structure of contact Hamiltonian systems is not conserved by the dynamics,
and the same happens with the Hamiltonian function. Indeed,

\begin{prop} {\rm (Dissipation of energy):}
Let $({\rm M},\bmeta,h)$ be a contact Hamiltonian system. If $\X_h\in\Cinfty({\rm M})$ is solution to the equations \eqref{evolvf111} then,
\beq
\Lie(\X_h)h=-(\Reeb(h))\,h \ .
\label{eq:disipenerg}
\eeq
\label{disipenerg}
\end{prop}
\vspace{-1cm}
\begin{proof}
As a consequence of the above proposition, we have,
$$
\Lie(\X_h)h=-\Lie(\X_h)(\inn(\X_h)\bmeta=
-\inn(\X_h)\Lie(\X_h)\bmeta=
\inn(\X_h)\big((\Lie(\Reeb)h)\,\bmeta\big)=-(\Lie(\Reeb)h)\,h\ .
$$
\qed\end{proof}

As a final result, there is a new way of writing equations \eqref{evolvf111} 
without using of the Reeb vector field $\Reeb$. 

\begin{prop}
Let $({\rm M},\bmeta,h)$ be a contact Hamiltonian system.
If $U=\{p\in {\rm M};h(p)\not= 0\}$ and $\Omega = -h\,\d\bmeta+\d h\wedge\bmeta$ on~$U$,
equations \eqref{evolvf111} can also be written as
\beq
\label{altevolvf}
\inn(\X_h)\Omega=0 \quad,\quad \inn(\X_h)\bmeta=-h\quad ;
\quad \mbox{\rm (on ${\rm M}$)} \ , 
\eeq
and its integral curves ${\bf c}\colon I\subset\Real\to {\rm M}$ are solutions to
$$
\inn(\widetilde c)(\Omega\circ c)) = 0 \quad , \quad\inn(\widetilde c)(\bmeta\circ c)) = - h\circ{\bf c} \:. 
$$
\end{prop}
\begin{proof}
If $\X_h$ satisfies equations \eqref{altevolvf}, then,
$$
0=\inn(\X_h)\Omega=-h\,\inn(\X_h)\d\bmeta+\big(\inn(\X_h)\d h\big)\bmeta-\d h\,\inn(\X_h)\bmeta=-h\,\inn(\X_h)\d\bmeta+\big(\inn(\X_h)\d h\big)\bmeta+h\,\d h
$$
and hence
\beq
\label{no-reeb-1}
h\,\inn(\X_h)\d\bmeta=\big(\inn(\X_h)\d h\big)\bmeta+h\,\d h \ .
\eeq
Contracting with the Reeb vector field,
\begin{equation*}
0=h\,\inn(\Reeb)\inn(\X_h)\d\bmeta =
(\inn(\X_h)\d\h)\inn(\Reeb)\eta+h\,\inn(\Reeb)\d h\ ,
\end{equation*}
from which $\inn(\X_h)\d h=-h\,\inn(\Reeb)\d h$,
and using this in \eqref{no-reeb-1}, we get
\begin{equation*}
h\inn(\X_h)\d\bmeta=
h\big(\d h-(\inn(\Reeb)\d h)\bmeta\big)=
h\big(\d h-(\Lie(\Reeb)h)\,\bmeta\big)\ .
\end{equation*}
and hence $\inn(\X_h)\d\bmeta=\d h-(\Lie(\Reeb)h)\,\bmeta$.
    
Conversely, if $\X_h$ satisfies equations  \eqref{evolvf111}; then, taking into account \eqref{eq:disipenerg},
\beann
\inn(\X_h)\Omega&=&\inn(\X_h)(-h\d\bmeta+\d h\wedge\bmeta)=
-h\inn(\X_h)\d\bmeta+(\inn(\X_h)\d h)\,\bmeta+h\d h
\\ &=& 
h(\Lie(\Reeb)h)\,\bmeta+(\Lie(\X_h)h)\bmeta=
\big(h(\Lie(\Reeb)h)+(\Lie(\X_h)h)\big)\,\bmeta=0 \ ,
\eeann 

From here, the equation for the integral curves is immediate.
\\ \qed\end{proof}

This form of the dynamical equations is especially interesting in the case of singular systems because they do not depend on the Reeb vector field and, as we pointed out, 
in a precontact manifold Reeb vector fields are not uniquely determined.

\section{Contact Lagrangian dynamical systems}

In this section, we review the Lagrangian formulation for contact systems \cite{CIAGLIA2018,DeLeon2019,GGMRR-2019b}.

\subsection{Contact Lagrangian systems}
\label{sec-conLagsys}

If $Q$ is a $n$-dimensional manifold, consider the product manifold 
$\Tan Q\times\Real$ equipped with adapted coordinates $(q^i,v^i, s)$.
Now, the canonical projections are denoted
$$ 
s\colon \Tan Q\times\Real\to\Real \,, \qquad \tau_1\colon \Tan Q\times\Real\to\Tan Q\,, \qquad \tau_0\colon \Tan Q\times\Real\to Q\times\Real\ . 
$$
As it was discussed in Section \ref{geomRxTQ},
the canonical geometric structures of the tangent bundle $\Tan Q$, 
namely the canonical endomorphism and the Liouville vector field, 
are extended naturally to $\Tan Q\times\Real$.
We use the same notation as in that section
to denote them: ${\cal J}\in{\cal T}^1_1(\Tan Q\times\Real)$ and $\Delta\in\vf(\Tan Q\times\Real)$, respectively;
and they have the same coordinate expressions as usual,
$\dst{\cal J} = \frac{\partial}{\partial v^i}\otimes \d q^i$ and $\dst\Delta = v^i\, \frac{\partial}{\partial v^i}$.

Similarly, the canonical lift of a curve ${\bf c}\colon\Real \rightarrow Q\times\Real$ to $\Tan Q\times\Real$, 
with ${\bf c} = (\mathbf{c}^i(t),s(t))$,
is defined as in Section \ref{geomRxTQ}, and it is
$\dst{\bf\widehat c}(t)=\Big(c^i(t),\frac{\d c^i}{\d t}(t), s(t)\Big)$.
Finally, the definitions of {\sl holonomic curves} and {\sc sode} vector fields in $\Tan Q\times\Real$
are also as in Section \ref{geomRxTQ}.

Then, the foundations of the Lagrangian formalism of the contact formulation for autonomous dissipative systems
are established through the reformulation of the generic postulates stated in Section \ref{nasposcon} for contact Hamiltonian systems:

\begin{pos}
{\rm (First Postulate of contact Lagrangian mechanics\/)}:
The configuration space of a dissipative dynamical system with $n$ degrees of freedom 
is $Q\times\Real$, where $Q$ is a $n$-dimensional differentiable manifold, and $n$ are the degrees of freedom of the system.
The phase space is the bundle $\Tan Q\times\Real$.
\end{pos}

\begin{pos}
{\rm (Second Postulate of contact Lagrangian mechanics\/)}:
The observables or physical magnitudes of a 
dissipative dynamical system are functions of $\Cinfty(\Tan Q\times\Real)$.
\end{pos}

\begin{pos}
{\rm (Third Postulate of contact Lagrangian mechanics\/)}:
There is a function $\Lag\in\Cinfty (\Tan Q\times\Real)$, 
called the {\sl \textbf{contact Lagrangian function}},
which contains the dynamical information of the system.
\end{pos}

\begin{remark}{\rm
In many cases, the contact Lagrangian function is of the following type:
$\Lag=L+\phi$, 
where $L=\tau_1^{\,*}L_o$ for a function
$L_o\in\Cinfty(\Tan Q)$, and $\phi=\tau_0^{\,*}\phi_o$, for $\phi_o\in\Cinfty(Q\times\Real)$;
that is, in coordinates, $\Lag(q^i,v^i,s)=L(q^i,v^i)+\phi(q^i,s)$.
In particular, the case where $\phi(q^i,s)=\gamma s$, with $\gamma\in\Real$,
appears very frequently in physical applications.   
}\end{remark}

Starting from a Lagrangian function the following magnitudes and structures are defined:

\begin{definition}
\label{definition:lagrangian-function}
Let $\Lag\in\Cinfty(\Tan Q\times\Real)$ be a Lagrangian function.

\noindent The {\sl\textbf{Lagrangian energy}}
associated with $\Lag$ is the function $E_\Lag:= \Delta(\Lag)-\Lag\in\Cinfty(\Tan Q\times\Real)$. 

\noindent The {\sl\textbf{Cartan Lagrangian forms}}
associated with $\Lag$ are defined as
\begin{equation}
\label{eq:thetaL}
\theta_\Lag:={\cal J}(\d\Lag)\in\df^1(\Tan Q\times\Real) 
\quad ,\quad
\omega_\Lag:=-\d\theta_\Lag\in\df^2(\Tan Q\times\Real)  \ .
\end{equation}
The {\sl\textbf{(pre)contact Lagrangian form}} is
$$
\bmeta_L=\d s-\theta_\Lag\in\df^1(\Tan Q\times\Real)
\ ,
$$
and it satisfies that $\d\bmeta_L=\omega_L$.
\\
The pair $(\Tan Q\times\Real,\Lag)$ is a {\sl\textbf{contact Lagrangian system}}.
\end{definition}

In natural coordinates $(q^i, v^i, s)$ on $\Tan Q\times\Real$, we have
\beann
%\label{eq:etaL}
\bmeta_L&=&\d s-\frac{\partial L}{\partial v^i} \,\d q^i \ ,
\\ %\nonumber
\d\bmeta_L&=& 
-\frac{\partial^2L}{\partial s\partial v^i}\d s\wedge\d q^i 
-\frac{\partial^2L}{\partial q^j\partial v^i}\d q^j\wedge\d q^i 
-\frac{\partial^2L}{\partial v^j\partial v^i}\d v^j\wedge\d q^i\ .
\eeann

\begin{definition}
\label{legmap1}
Given a Lagrangian 
$\Lag\in\Cinfty(\Tan Q\times\Real)$, the {\sl \textbf{Legendre map}}
associated with $\Lag$ is the fiber derivative of~$\Lag$,
considered as a function on the vector bundle
$\pi_0 \colon\Tan Q\times\Real\to Q\times\Real$
that is, the map
$\mathfrak{F}\Lag\colon\Tan Q\times\Real \to\Tan^*Q\times\Real$ 
given by
$$
\mathfrak{F}\Lag (q,v_q,s) = \left(\mathfrak{F}\Lag_s (q,v_q),s\right)
\ ,
$$
where $\Lag_s\colon\Tan Q\to\Real$ denotes the restriction of $\Lag$ to each fiber of the bundle $\Tan Q\times\Real\to\Real$
(that is, the Lagrangian $\Lag$ with $s$ ``freezed'').
\end{definition}

In natural coordinates, we have:
$$
q^i \circ \mathfrak{F}\Lag = q^i \quad , \quad p_i \circ \mathfrak{F}\Lag =\derpar{\Lag}{v^i} \quad , \quad s \circ \mathfrak{F}\Lag = s \ .
$$

Observe that, if $\theta\in\df^1(\Tan^*Q\times\Real)$ 
is the canonical form and $\omega=-\d\theta$, 
$$
\theta_\Lag=\mathfrak{F}\Lag^{\;*}\theta\quad , \quad\omega_\Lag=\mathfrak{F}\Lag^{\;*}\omega\ .
$$

\begin{prop}
\label{Prop-regLag}
For a Lagrangian function, $\Lag$ the following conditions are equivalent:
\begin{enumerate}
\item
The Legendre map
$\mathfrak{F}\Lag$ is a local diffeomorphism.
\item 
The pair $(\Tan Q\times\Real,\bmeta_\Lag)$ is a contact manifold.
\item 
The Hessian matrix \ $\dst W_{ij}\equiv\left(\frac{\partial^2\Lag}{\partial v^i\partial v^j}\right)$ is nondegenerate everywhere.
\end{enumerate}
\end{prop}
\begin{proof}
The proof can be done using natural coordinates.
\\ \qed\end{proof}

\begin{definition}
A Lagrangian function $\Lag$ is said to be {\sl\textbf{regular}} if the equivalent
conditions in the above Proposition \ref{Prop-regLag} hold.
Otherwise, $\Lag$ is called a {\sl\textbf{singular}} Lagrangian.
In particular, 
$\Lag$ is said to be {\sl\textbf{hyperregular}} 
if $\mathfrak{F}\Lag$ is a global diffeomorphism.
\end{definition}

\begin{remark}
\label{nocontact}{\rm
If $\Lag$ is a regular Lagrangian, then
$(\Tan Q\times\Real,\bmeta_\Lag,E_\Lag)$
is a contact Hamiltonian system.
When $\Lag$ is not regular, it can induce a precontact structure, but also a structure which is neither contact nor precontact.
For example, 
the Lagrangian $\displaystyle \Lag=\sum_{i=1}^n v^i s$ in $\Tan Q\times\Real$
gives a form $\bmeta_\Lag\in\df^1(\Tan Q\times\Real)$
for which the condition \eqref{splitcont} does not hold
and that has no Reeb vector fields associated to it
(see Remark \ref{nocontacto}).
}\end{remark}

\begin{definition}
If $\Lag \in \Cinfty (\Tan Q\times\Real)$ is a Lagrangian function and $(\Tan Q\times\Real,\bmeta_\Lag)$ is
a contact or a precontact manifold, then
the pair $(\Tan Q\times\Real,\Lag)$ is said to be a {\sl\textbf{(pre)contact Lagrangian dynamical system}}.
\label{stdlr1}
\end{definition}

Given a contact Lagrangian system $(\Tan Q\times\Real,\Lag)$, the Reeb vector field $\Reeb_\Lag\in\vf(\Tan Q\times\Real)$ is uniquely determined by the conditions
\begin{equation}\label{eq:reeb}
    \inn(\Reeb_\Lag)\d\bmeta_\Lag=0 \quad ,\quad \inn(\Reeb_\Lag)\bmeta_\Lag=1 \ ,
\end{equation}
and its local expression is
\begin{equation}
\label{coorReeb}
\Reeb_\Lag=\frac{\partial}{\partial s}-W^{ji}\frac{\partial^2\Lag}{\partial s \partial v^j}\,\frac{\partial}{\partial v^i} \,,
\end{equation}
where $(W^{ij})$ is the inverse of the partial Hessian matrix,
namely 
$W^{ij} W_{jk} = \delta^i_{k}$.
Observe that the Reeb vector field does not appear
in the simplest form $\dst\derpar{}{s}$,
since the natural coordinates in $\Tan Q \times \Real$
are neither adapted nor Darboux coordinates for $\bmeta_\Lag$.
If $(\Tan Q\times\Real,\Lag)$ is a precontact Lagrangian system then Reeb vector fields are not uniquely defined.

Finally, we state:

\begin{pos}
{\rm (Fourth Postulate of contact Lagrangian mechanics\/)}:
The dynamical trajectories of the system are the integral curves
of a vector field $\X_\Lag\in \vf(\Tan Q\times\Real)$ satisfying the conditions:
\ben
\item
$\X_\Lag$ is the Hamiltonian vector field associated with $E_\Lag$;
that is, it is a solution to the equations
\beq
\label{eq-E-L-contact1}
\inn(X_\Lag)\bmeta_\Lag=-E_\Lag \quad,\quad
\inn(\X_\Lag)\d \bmeta_\Lag=\d E_\Lag-(\Reeb_\Lag(E_\Lag))\,\bmeta_\Lag \ ,
\eeq
where $\Reeb_\Lag$ is a Reeb vector field, which is determined by the equations \eqref{eq:reeb}.
\item
$\X_\Lag$ is a {\sc sode}:
\ ${\cal J}(\X_\Lag)=\Delta$.
\een
Therefore, these trajectories are the holonomic curves 
${\bf c}\colon I\subset\Real\to\Tan Q\times\Real$
which are the solutions to the equations
\beq
\inn({\bf\widetilde c})(\bmeta_\Lag\circ {\bf c}) = - E_\Lag\circ{\bf c} \quad  , \quad
\inn({\bf\widetilde c})(\d\bmeta_\Lag\circ {\bf  c})=\big(\d E_\Lag - (\Reeb_\Lag(E_\Lag))\,\bmeta_\Lag\big)\circ{\bf c} .
\label{hec}
\eeq
Equations \eqref{hec} are called the
{\sl \textbf{(pre)contact Euler--Lagrange equations for curves}}.
Equations \eqref{eq-E-L-contact1} are called the
{\sl \textbf{(pre)contact Lagrangian equations for vector fields}} and
a vector field solution to them (if it exists) is a
{\sl \textbf{(pre)contact Lagrangian dynamical vector field}}.
If, in addition, $\X_\Lag$ is a {\sc sode}, then
it is called a {\sl \textbf{(pre)contact Euler--Lagrange vector field}} of the system, 
\end{pos}

\begin{definition}
Given a (pre)contact Lagrangian dynamical system $(\Tan Q\times\Real,\Lag)$, 
the {\sl \textbf{Lagrangian problem}} posed by this system
consists in finding a {\sc sode} vector field $\X_\Lag\in\vf(\Tan Q\times\Real)$
solutions to \eqref{eq-E-L-contact1}.
\end{definition}

\subsection{The contact Euler--Lagrange equations}

First, using \eqref{coorReeb}, a simple computation in coordinates shows that
\begin{equation}
\label{ReebLag}
\Reeb_\Lag(E_\Lag)=-\derpar{\Lag}{s} \ .
\end{equation}
Taking this into account, in natural coordinates, for a holonomic curve
${\bf c}(t)=(q^i(t),\dot q^i(t),s(t))$ on $\Tan Q\times\Real$,
equations \eqref{hec} are the {\sl\textbf{ Herglotz--Euler--Lagrange equations}},
\begin{align}\label{ELeqs2}
\frac{ds}{dt}&=\Lag \ ,
 \\
\label{ELeqs3}
\frac{\d}{\d t}\left(\derpar{\Lag}{v^i}\right)-
\displaystyle\frac{\partial\Lag}{\partial q^i}=
\displaystyle\frac{\partial^2\Lag}{\partial v^j \partial v^i}\,
\frac{d^2q^j}{dt^2} +\displaystyle\frac{\partial^2\Lag}{\partial q^j\partial v^i}\,\frac{dq^j}{dt}+
\displaystyle \frac{\partial^2L}{\partial s \partial v^i}\, \frac{ds}{dt}
-\displaystyle\frac{\partial\Lag}{ \partial q^i}&=
\displaystyle\frac{\partial\Lag}{\partial s}\displaystyle\frac{\partial\Lag}{\partial v^i} \ .
\end{align}
Furthermore, for a vector field
$\X_\Lag=\displaystyle f^i\,\frac{\partial}{\partial q^i}+F^i\,\frac{\partial}{\partial v^i}+
g\,\frac{\partial}{\partial s}$,
equations \eqref{eq-E-L-contact1} lead to
\begin{align}
\displaystyle
\Lag+\frac{\partial\Lag}{\partial v^i} (f^i-v^i) -g 
&= 0\ ,
\label{A-E-L-eqs4}
\\
\left( f^j-v^j \right)
\frac{\partial^2\Lag}{\partial q^i \partial v^j}
+\frac{\partial\Lag}{\partial q^i}
-\frac{\partial^2\Lag}{\partial s \partial v^i}g
-\frac{\partial^2\Lag}{\partial q^j \partial v^i}f^j
-\frac{\partial^2\Lag}{\partial v^j \partial v^i}F^j
+\frac{\partial\Lag}{\partial s}
\frac{\partial\Lag}{\partial v^i}
&= 0\ ,
\label{A-E-L-eqs3}
\\
\left( f^j-v^j \right)
\frac{\partial^2\Lag}{\partial v^i \partial v^j}
&=0
\label{A-E-L-eqs1} \ ,
\\
\left( f^j-v^j \right)
\frac{\partial^2\Lag}{\partial v^j \partial s}
&=0 \ ,
\label{A-E-L-eqs2}
\end{align}
and then it is immediate to prove that:

\begin{prop}
\label{ELeq-teor}
If $\Lag$ is a regular Lagrangian, then $\X_\Lag$ is a {\sc sode} and
the equations \eqref{A-E-L-eqs4} and \eqref{A-E-L-eqs3} become
\begin{align}
\label{ELeqs0}
g&=\Lag \,,
\\
\label{ELeqs1}
\displaystyle\frac{\partial^2\Lag}{\partial v^j \partial v^i}\,
F^j +\displaystyle\frac{\partial^2\Lag}{\partial q^j \partial v^i} \,v^j  +
\displaystyle \frac{\partial^2\Lag}{\partial s \partial v^i}\,\Lag
-\displaystyle\frac{\partial\Lag}{ \partial q^i}&=
\displaystyle\frac{\partial\Lag}{\partial s}\displaystyle\frac{\partial\Lag}{\partial v^i} \ .
\end{align}
which, for the integral curves of $\X_\Lag$, give 
the equations \eqref{ELeqs2} and \eqref{ELeqs3}.
\end{prop}
\begin{proof}
In fact, if $\Lag$ is a regular Lagrangian, equations \eqref{A-E-L-eqs1} lead to $v^i=f^i$, 
which are the {\sc sode} condition for the vector field $\X_\Lag$.
Then, \eqref{A-E-L-eqs2} holds identically, and
\eqref{A-E-L-eqs4} and \eqref{A-E-L-eqs3} give the equations \eqref{ELeqs0} and \eqref{ELeqs1}
and then, for the integral curves of $X_L$, we get equations \eqref{ELeqs2} and \eqref{ELeqs3}.
\\ \qed \end{proof}

Thus, the local expression of this Euler--Lagrange vector field is
\begin{equation}
\label{sode-coor}
\X_\Lag=
v^i\,\frac{\partial}{\partial q^i}
+W^{ik}
\left(
\frac{\partial\Lag}{ \partial q^k} 
-\frac{\partial^2\Lag}{\partial q^j \partial v^k} \,v^j
-\Lag\frac{\partial^2\Lag}{\partial s \partial v^k} 
+\frac{\partial\Lag}{\partial s}
 \frac{\partial\Lag}{\partial v^k} 
\right)
\frac{\partial}{\partial v^i}+
\Lag\,\frac{\partial}{\partial s}\ .
\end{equation}

\begin{remark}{\rm
It is important to point out that
the expression in coordinates \eqref{ELeqs2} of
the second Lagrangian equation \eqref{eq-E-L-contact1} relates the variation 
of the ``dissipation coordinate'' $s$ to the Lagrangian function and, from here, 
we can identify this coordinate with the Lagrangian action, $\displaystyle s=\int{\cal L}\,dt$.
}\end{remark}

\begin{remark}{\rm
If $\Lag$ is a singular Lagrangian, but $(\Tan Q\times\Real,\bmeta_\Lag)$ is a precontact manifold, although the Reeb vector fields
are not uniquely defined, the Lagrangian equations \eqref{eq-E-L-contact1} 
are independent on the Reeb vector field used (see \cite{DeLeon2019}).
Alternatively, the Reeb independent equations \eqref{altevolvf}
for the precontact Hamiltonian system $(\Tan Q\times\Real,\bmeta_\Lag, E_\Lag)$
can be used instead.
In any case, solutions to the Lagrangian equations
are not necessarily {\sc sode} and
the condition ${\cal J}(\X_\Lag)=\Delta$ must be added to the Lagrangian equations.
Furthermore, these equations are not necessarily compatible everywhere on $\Tan Q\times\Real$ 
and a constraint algorithm must be implemented in order to find 
a final constraint submanifold $S_f\hookrightarrow\Tan Q\times\Real$
(if it exists) where there are {\sc sode} vector fields $\X_\Lag\in\vf(\Tan Q\times\Real)$,
tangent to $S_f$, which are solutions to the above equations on $S_f$ \cite{DeLeon2019}.
}\end{remark}

\subsection{Canonical Hamiltonian formalism}

Consider a hyperregular Lagrangian system $(\Tan Q\times\Real,\Lag)$
(the regular case is the same, taking $\mathfrak{F}\Lag (\Tan Q\times\Real)\subset\Tan^*Q\times\Real$ instead of $\Tan^*Q\times\Real$, or locally, at least).
Then, $\mathfrak{F}\Lag$ is a diffeomorphism
between the contact manifolds $(\Tan Q\times\Real,\bmeta_\Lag)$ and
the canonical contact manifold
$(\Tan^*Q\times\Real,\bmeta)$,
with 
$$
\mathfrak{F}\Lag^{\;*}\bmeta=\bmeta_\Lag\quad , \quad
\mathfrak{F}\Lag_*{\Reeb}_\Lag={\Reeb} \ .
$$
Furthermore, there exists a function ${\rm h}\in\Cinfty(\Tan^*Q\times\Real)$ 
such that $\mathfrak{F}\Lag^{\;*}{\rm h}=E_\Lag$;
so we have the contact Hamiltonian system $(\Tan^*Q\times\Real,\bmeta,{\rm h})$
and the contact Hamiltonian equations \eqref{evolvf111}
(or their equivalent expressions \eqref{equivmanner} or
\eqref{altevolvf}),
which read as,
\beq
\inn(\X_{\rm h})\bmeta=-{\rm h} \quad,\quad 
\inn(\X_{\rm h})\d\bmeta = \d{\rm h}-\Reeb({\rm h})\,\bmeta\ ,
\label{evolvf11b}
\eeq
and have a unique solution $\X_{\rm h}$.
The dynamical trajectories are the integral curves
of this contact Hamiltonian vector field and, then, they are the solutions to equations \eqref{hel-eqs011}, which read as,
\beq
\label{hel-eqs01b}
\dst\frac{dx^i}{dt} = \frac{\partial{\rm h}}{\partial p_i}\quad ,\quad
\dst\frac{dp_i}{dt} = -\left(\frac{\partial{\rm h}}{\partial x^i} + p_i\frac{\partial{\rm h}}{\partial s}\right)\quad , \quad
\dst\frac{ds}{dt} = p_i\frac{\partial{\rm h}}{\partial p_i} - h \ , 
\eeq

Then, if $\X_{\rm h}\in\vf(\Tan^*Q\times\Real)$ is the
contact Hamiltonian vector field associated with ${\rm h}$,
we have that $\mathfrak{F}\Lag_*\X_\Lag=\X_{\rm h}$.

For singular Lagrangians, as in the above chapters, we define:

\begin{definition}
A singular Lagrangian function $\Lag\in\Cinfty(\Tan Q\times\Real)$
 is called {\sl \textbf{almost-regular}} if $\mathcal{P}:=\mathfrak{F}\Lag(\Tan Q\times\Real)$ is
a closed submanifold of \ $\Tan^*Q\times\Real$ (
the natural embedding is denoted $\jmath_{\cal P}\colon\mathcal{P}\hookrightarrow\Tan^*Q\times\Real$), 
$\mathfrak{F}\Lag$ is a submersion onto its image, and
the fibers $\mathfrak{F}\Lag^{-1}(FL({\rm p}))$, for every ${\rm p}\in\Tan Q\times\Real$, are
connected submanifolds of $\Tan Q\times\Real$.
\end{definition}

In these cases, we have $({\cal P},\bmeta_{\cal P})$, where
$\bmeta_{\cal P}=j_{\cal P}^*\bmeta\in\Omega^1({\cal P})$,
Furthermore, the Lagrangian energy function $E_\Lag$ is $\mathfrak{F}\Lag$-projectable; 
i.e., there is a unique ${\rm h}_{\cal P}\in\Cinfty({\cal P})$
which is $\mathfrak{F}\Lag$-related with $E_\Lag$.
Then, if $({\cal P},\bmeta_{\cal P})$
is a contact or a precontact manifold,
then $({\cal P},\bmeta_{\cal P},{\rm h}_{\cal P})$
is a (pre)contact Hamiltonian system.
whose (pre)contact Hamiltonian equations are \eqref{evolvf111} or, alternatively, \eqref{altevolvf},
adapted to this situation.
As in the Lagrangian formalism, these equations are not necessarily compatible everywhere on ${\cal P}$ 
and a constraint algorithm must be implemented in order to find 
a final constraint submanifold $P_f\hookrightarrow{\cal P}$
(if it exists) where there are vector fields $\X_{{\rm h}_{\cal P}}\in\vf({\cal P})$,
tangent to $P_f$, which are solutions to the precontact Hamiltonian equations on $P_f$ (see \cite{DeLeon2019} for a detailed analysis on all these topics).

\section{Unified Lagrangian-Hamiltonian formalism for contact systems}
\label{ulhcs}

The Lagrangian-Hamiltonian unified formalism 
of contact Lagrangian systems has been developed in \cite{LGMMR-2020}.
Next we present its main features.

\subsection{Extended precontact unified bundle. Unified formalism}

\begin{definition}
Let $Q$ be a $n$-dimensional differentiable manifold.
The {\sl \textbf{extended precontact unified bundle}}
(or {\sl \textbf{extended precontact Pontryagin bundle}}) is
$$
\mathfrak{M}=\Tan Q\times_Q\Tan^*Q\times\Real \ ,
$$
and it is endowed with the natural projections
$$
\rho_1\colon\mathfrak{M}\to\Tan Q\times\Real \ ,\
\rho_2\colon\mathfrak{M}\to\Tan^*Q\times\Real \ ,\
\rho_0\colon\mathfrak{M}\to Q\times\Real \ ,\ 
s\colon\mathfrak{M}\to \Real \ .
\label{project}
$$
\end{definition}

Natural coordinates in $\mathfrak{M}$ are $(q^i,v^i,p_i,s)$.

\begin{definition}
A curve $\mbox{\boldmath $\sigma$}\colon\Real\rightarrow\mathfrak{M}$
is a \textbf{holonomic} in $\mathfrak{M}$ if
$\rho_1\circ\mbox{\boldmath $\sigma$}\colon\Real\to\Tan Q\times\Real$ is holonomic curve.
A vector field $\Gamma\in\vf(\mathfrak{M})$ 
satisfies the \textbf{second-order condition} in $\mathfrak{M}$
(for short: it is a {\sc sode} in $\mathfrak{M}$) if all of its integral curves 
are holonomic in $\mathfrak{M}$. 
\end{definition}

In coordinates, a holonomic curve and a {\sc sode} in $\mathfrak{M}$ are expressed as 
\beann
\mbox{\boldmath $\sigma$}&=&
\left(\sigma_1^i(t),\frac{d\sigma_1^i}{d t}(t),\sigma_{2\,i}(t),\sigma_0(t) \right) \ ,
\\
\Gamma&=&v^i \frac{\partial}{\partial q^i}+F^i \frac{\partial}{\partial v^i}+G_i \frac{\partial}{\partial p_i}+f\,\frac{\partial}{\partial s} \ .
\eeann

\begin{definition}
The bundle $\mathfrak{M}$ is endowed with the following canonical structures:
\ben
\item
The \textbf{coupling function} in $\mathfrak{M}$ is the
map $\mathfrak{C}\colon\mathfrak{M}\to\Real$  defined as follows: for every
$w=({\rm v}_q,{\rm p}_q,s)\in\mathfrak{M}$, where $q\in Q$,
${\rm p}_q\in\Tan^*Q$, and ${\rm v}_q\in\Tan Q$, then
$\mathfrak{C}(w):=\langle{\rm p}_q,{\rm v}_q\rangle$.
\item
The \textbf{canonical precontact $1$-form} on $\mathfrak{M}$
is the $\rho_1$-semibasic form
$$\bmeta_{\mathfrak{M}}:=\rho_2^*\bmeta=
d s-\pi_1^*\rho_2^*\pi_1^*\Theta\in\df^1(\mathfrak{M})\ ;$$
where $\bmeta$ is the canonical contact form on $\Tan^*Q\times\Real$
(and $\Theta$ is the canonical $1$-form on
$\Tan^*Q$).
\een
\label{coupling}
\end{definition}

In natural coordinates of $\mathfrak{M}$ we have that
$\bmeta_{\mathfrak{M}}=\d s-p_i\,\d q^i$.

\begin{definition}
Given a Lagrangian function $\Lag\in\Cinfty(\Tan Q\times\Real)$, let
$\mathfrak{L}=\rho_1^*\Lag\in\Cinfty(\mathfrak{M})$.
The \textbf{Hamiltonian function} is defined as
\beq
{\bf H}:=\mathfrak{C}-\mathfrak{L} =p_i v^i-\Lag (q^j,v^j,s)\in\Cinfty(\mathfrak{M}) \ .
\label{Hamf}
\eeq
\end{definition}

\begin{remark}{\rm
Observe that $(\mathfrak{M},\bmeta_{\mathfrak{M}})$ is a precontact manifold.
As a consequence, the Reeb vector fields are not uniquely defined and
in natural coordinates of $\mathfrak{M}$ the general solution to \eqref{eq-Reeb} are the vector fields
$\displaystyle \Reeb=\derpar{}{z}+F^i\derpar{}{v^i}$
for arbitrary functions $F^i$.
Nevertheless, as we have pointed out, 
the dynamics obtained from the formalism is independent of the choice of the Reeb vector fields;
therefore, as $\mathfrak{M}=\Tan Q\times_Q\Tan^*Q\times\Real$
is a trivial bundle over $\Real$,
the canonical vector field $\displaystyle \derpar{}{s}$ on $\Real$
can be lifted canonically to a vector field in $\mathfrak{M}$,
which can be taken as a representative of the family of
these Reeb vector fields.
}\end{remark}

Thus, we have that $(\mathfrak{M},\eta_\mathfrak{M},{\bf H})$
is a precontact Hamiltonian system and then,
we can pose the dynamic problem for this system which consists of finding
a Hamiltonian vector field which is a solution to the precontact Hamiltonian equations
\begin{equation}
\label{Whamilton-contact-eqs}
\inn(\X_{\bf H})\d\bmeta_{\mathfrak{M}}=\d{\bf H}-({\Reeb}({\bf H}))\bmeta_{\mathfrak{M}}
\quad ,\quad
\inn(\X_{\bf H})\bmeta_{\mathfrak{M}}=-{\bf H}\ .
\end{equation}
Then, the integral curves $\mbox{\boldmath $\sigma$}\colon I\subset\Real\to\mathfrak{M}$ of $X_{\bf H}$,
are the solutions to the equations
\begin{equation}
\label{Whamilton-contactc-curves-eqs}
\inn(\widetilde{\mbox{\boldmath $\sigma$}})(\d\bmeta_{\mathfrak{M}}\circ\mbox{\boldmath $\sigma$})=\left(\d{\bf H}-({\Reeb}({\bf H}))\bmeta_{\mathfrak{M}}\right)\circ\mbox{\boldmath $\sigma$}
\quad ,\quad
 \inn(\widetilde{\mbox{\boldmath $\sigma$}})(\bmeta_{\mathfrak{M}}\circ\mbox{\boldmath $\sigma$})=-{\bf H}\circ\mbox{\boldmath $\sigma$} \ .
    \end{equation}

As we shall see next,
these equations are not compatible on $\mathfrak{M}$,
and we have to implement the constraint algorithm
in order to find the final constraint submanifold of $\mathfrak{M}$
where there are consistent solutions
to the equations. 
In fact, in natural coordinates of $\mathfrak{M}$ we have that
$$
\d{\bf H}=v^i\d p_i+\left(p_i -\derpar{\Lag}{v^i}\right)\d v^i-\derpar{\Lag}{q^i}\,\d q^i-\derpar{\Lag}{s}\,\d s \ ,
$$
and, if the local expression of
a vector field $\X_{\bf H}\in\vf(\mathfrak{M})$ is
$$
\X_{\bf H} = f^i\derpar{}{q^i}+F^i\derpar{}{v^i}+G_i\derpar{}{p_i}+f\derpar{}{s} \ ,
%\label{coorvf}
$$
we obtain that
$$
\inn(\X_{\bf H})\bmeta_{\mathfrak{M}}=f-f^ip_i \quad , \quad
\inn(\X_{\bf H})\d\bmeta_{\mathfrak{M}}=f^i\,\d p_i-G_i\,\d q^i 
\quad , \quad
({\Reeb}({\bf H}))\bmeta_{\mathfrak{M}}=-\derpar{\Lag}{s}(\d s-p_i\d q^i) \ .
$$
Then, equations \eqref{Whamilton-contact-eqs} give
\bea
f&=&(f^i-v^i)\,p_i+\Lag \ ,
\label{zero} \\
f^i&=&v^i \ , \label{one} \\
p_i&=&\derpar{\Lag}{v^i} \ ,
\label{two} \\
G_i&=&\derpar{\Lag}{q^i}
+p_i\derpar{\Lag}{s} \ ,
\label{three}
\eea
The equations (\ref{one}) are the holonomy conditions
(i.e., $\X_{\bf H}$ must be a {\sc sode}),
which arise straightforwardly from the unified formalism.
The algebraic equations (\ref{two}) are compatibility conditions
defining the submanifold $\mathfrak{M}_0\hookrightarrow\mathfrak{M}$,
where the equations are compatible.
This $\mathfrak{M}_0$ is the {\sl first constraint submanifold} of the Hamiltonian precontact system $(\mathfrak{M},\bmeta_{\mathfrak{M}},{\bf H})$, and is the graph of $\mathfrak{F}\Lag$; that is,
$$
\mathfrak{M}_0=\{ ({\rm v}_q,\mathfrak{F}\Lag({\rm v}_q))\in\mathfrak{M} \ ,\ \mbox{\rm for ${\rm v}_q\in\Tan Q$}\} \ .
$$
In this way, as it is usual, the unified formalism includes the definition of
the Legendre map as a consequence of the constraint algorithm.

Thus, vector fields which are solutions to \eqref{Whamilton-contact-eqs} are,
$$
\X_{\bf H}=
v^i\derpar{}{q^i}+F^i\derpar{}{v^i}+
\left(\derpar{\Lag}{q^i}
+p_i\derpar{\Lag}{z}\right)\derpar{}{p_i}
+\Lag\derpar{}{s} 
\quad \mbox{\rm (on $\mathfrak{M}_0$)}\ .
$$
The constraint algorithm continues by demanding that
$\X_{\bf H}$ must be tangent to $\mathfrak{M}_0$. As 
$\displaystyle \xi^1_j=p_j-\derpar{\Lag}{v^j}\in\Cinfty(\mathfrak{M})$
are the constraints defining $\mathfrak{M}_0$, this condition is
\beq
\X_{\bf H}\left(p_j-\derpar{\Lag}{v^j}\right)=
-\frac{\partial^2\Lag}{\partial q^i\partial v^j}v^i
-\frac{\partial^2\Lag}{\partial v^i\partial  v^j}F^i-
\Lag\frac{\partial^2\Lag}{\partial z\partial v^j}
+\derpar{\Lag}{q^j}+p_j\derpar{\Lag}{s}=0
\quad \mbox{\rm (on $\mathfrak{M}_0$)} \ ,
\label{tangcond}
\eeq
and there are two options:
\bit
\item
If $\Lag$ is a regular Lagrangian, these equations determine all the
functions \(F^i\).
Then the solution is unique and the algorithm ends.
\item
If $\Lag$ is singular, then these equations establish relations
among the arbitrary functions $F^i$;
some of them remain undetermined and the solutions are not unique.
Furthermore, new constraints $\xi^2_\mu\in\Cinfty(\mathfrak{M})$
could appear, then defining a new submanifold 
$\mathfrak{M}_1\hookrightarrow\mathfrak{M}_0\hookrightarrow\mathfrak{M}$.
Then, the algorithm continues by demanding the tangency of
$\X_{\bf H}$ to $\mathfrak{M}_1$, and so on,
until we obtain a final constraint submanifold $\mathfrak{M}_f$ (if it exists)
where tangent solutions $X_{\bf H}$ exist.
\eit
If $\mbox{\boldmath $\sigma$}(t)=(q^i(t),v^i(t),p_i(t),s(t))$ 
is an integral curve of $\X_{\bf H}$, we have that 
$\displaystyle f^i=\frac{d q^i}{d t}$,
$\displaystyle F^i=\frac{d v^i}{d t}$, 
$\displaystyle G_i=\frac{d p_i}{d t}$,
$\displaystyle f=\frac{d s}{d t}$, 
and equations (\ref{zero}), (\ref{one}), (\ref{two}), and (\ref{three}) 
give to the coordinate expression of the equations
\eqref{Whamilton-contactc-curves-eqs}; in particular:
\bit
\item
From (\ref{one}), we have that $\dst v^i=\frac{d q^i}{d t}$;
that is, the holonomy condition.
\item
Using (\ref{one}) again, the equation (\ref{zero}) gives the equation \eqref{ELeqs2} again.
\item
The equations \eqref{three} read,
$$
\frac{dp_i}{dt}=\derpar{\Lag}{q^i}+p_i\derpar{\Lag}{s}=
-\left(\derpar{{\bf H}}{q^i}+p_i\derpar{{\bf H}}{s}\right) \ ,
$$
which are the second group of Hamilton's equations \eqref{hel-eqs01}.
Then, using \eqref{two},
these equations are (on $\mathfrak{M}_0$),
$$
\frac{d}{dt}\left(\derpar{\Lag}{v^i}\right)=\derpar{\Lag}{q^i}
+\derpar{\Lag}{v^i} \derpar{\Lag}{s} \ ,
$$
which are the Herglotz--Euler-Lagrange equations (\ref{ELeqs3}).
The first group of Hamilton's equations \eqref{hel-eqs01} arises
from the definition of the Hamiltonian function
\eqref{Hamf}, using the holonomy condition.
\item
Using \eqref{two} and \eqref{ELeqs2},
the tangency condition \eqref{tangcond} gives
again the Herglotz--Euler-Lagrange equations \eqref{ELeqs3}.
If $\Lag$ is singular, these equations can be incompatible.
\eit

\subsection{Recovering the Lagrangian and Hamiltonian formalisms and equivalence}
\label{recovering}

Now we  study the equivalence of the unified formalism 
with the Lagrangian and Hamiltonian formalisms
for the hyperregular case
(the regular case is the same, at least locally).
See \cite{LGMMR-2020} more details about the singular case.

First, observe that, in this case, denoting by 
$\jmath_0\colon\mathfrak{M}_0\hookrightarrow\mathfrak{M}$
the natural embedding, we have that
$$
(\rho_1\circ\jmath_0)(\mathfrak{M}_0)=\Tan Q\times\Real
\quad , \quad
(\rho_2\circ\jmath_0)(\mathfrak{M}_0)=\Tan^*Q\times\Real \ .
$$
Furthermore, as $\mathfrak{M}_0$ is the graph of the Legendre map
$\mathfrak{F}\Lag$, it is diffeomorphic to $\Tan Q\times\Real$,
being the restricted projection $\rho_1\circ\jmath_0$
this diffeomorphism.
So, we have the diagram
$$
\xymatrix{
\ & \ & \mathfrak{M} \ar@/_1.3pc/[ddll]_{\rho_1} \ar@/^1.3pc/[ddrr]^{\rho_2} & \ & \ \\
\ & \ & \mathfrak{M}_0 \ar[dll] \ar[drr] \ar@{^{(}->}[u]^{\jmath_0} & \ & \ \\
\Tan Q\times\Real \ar[rrrr]^<(0.40){\mathfrak{F}\Lag}
& \ & \ & \ & \Tan^*Q\times\Real 
}
$$
As in the other unified formalisms,
functions, differential forms on $\mathfrak{M}$,
and vector fields on $\mathfrak{M}$ tangent to $\mathfrak{M}_0$
can be restricted to $\mathfrak{M}_0$. 
Then, they can be translated to the Lagrangian or the Hamiltonian side 
by using that $\mathfrak{M}_0$ is diffeomorphic to $\Tan Q\times\Real$,
or projecting to $\Tan^* Q\times\Real$.
In particular, a simple calculation in coordinates shows that
$$
\rho_1^*E_L={\bf H}=\rho_2^*{\rm h} \ ,
$$
Furthermore,
every curve $\mbox{\boldmath $\sigma$}\colon I\subseteq\Real\to\mathfrak{M}$
which takes values in $\mathfrak{M}_0$, can be split as
$\mbox{\boldmath $\sigma$}=(\mbox{\boldmath $\sigma$}_L,\mbox{\boldmath $\sigma$}_H)$, where
$\mbox{\boldmath $\sigma$}_L=\rho_1\circ\mbox{\boldmath $\sigma$}\colon I\subseteq\Real \to\Tan Q\times\Real$
and $\mbox{\boldmath $\sigma$}_H=
\mathfrak{F}\Lag\circ\mbox{\boldmath $\sigma$}_L\colon I\subseteq\Real\to\Tan^*Q\times\Real$.
Therefore, taking all of this into account, the results, and the discussion in the above section
lead to state:

\begin{teor}
Let $\mbox{\boldmath $\sigma$}\colon\Real\to\mathfrak{M}$,
with  ${\rm Im}\,(\mbox{\boldmath $\sigma$})\subset\mathfrak{M}_0$,
be a curve solution to the equations \eqref{Whamilton-contactc-curves-eqs}
on $\mathfrak{M}_0$.
Then
$\mbox{\boldmath $\sigma$}_L$ is the lift of the projected curve
${\bf c}=\rho_0\circ\mbox{\boldmath $\sigma$}\colon\Real\to Q\times\Real$
to $\Tan Q\times\Real$ (that is, $\mbox{\boldmath $\sigma$}_L$ is
a holonomic curve),
and it is a solution to the equations \eqref{hec}.
Moreover, the curve 
$\mbox{\boldmath $\sigma$}_H=
\mathfrak{F}\Lag\circ\widetilde{\bf c}$
is a solution to the equations \eqref{hamilton-contactc-curves-eqs} on $\Tan^*Q\times\Real$.

Conversely, for every curve ${\bf c}\colon\Real\to Q\times\Real$ such that
$\widetilde{\bf c}$ is a solution to the equations \eqref{hec}, we have that the curve 
$\mbox{\boldmath $\sigma$}=(\widetilde{\bf c},\mathfrak{F}\Lag\circ\widetilde{\bf c})$
is a solution to the equations \eqref{Whamilton-contactc-curves-eqs} on $\mathfrak{M}_0$,
and the curve $\mathfrak{F}\Lag\circ\widetilde{\bf c}$
is a solution to the equations \eqref{hamilton-contactc-curves-eqs} on $\Tan^*Q\times\Real$.
 \label{mainteor01}
\end{teor}

The curves $\mbox{\boldmath $\sigma$}\colon\Real\to\mathfrak{M}$ 
solution to the equations \eqref{Whamilton-contactc-curves-eqs}
are the integral curves of {\sc sode} vector fields $\X_{\bf H}\in\vf(\mathfrak{M})$ 
solution to \eqref{Whamilton-contact-eqs} on $\mathfrak{M}_0$.
The curves $\mbox{\boldmath $\sigma$}_\Lag\colon\Real\to\Tan Q\times\Real$
are the integral curves of {\sc sode} vector fields $\X_\Lag\in\vf(\Tan Q\times\Real)$ 
solution to  \eqref{hec}.
Finally, the curves
$\mbox{\boldmath $\sigma$}_{\rm h}\colon\Real\to\Tan^*Q\times\Real$
are the integral curves of the contact Hamiltonian vector field $\X_{\rm h}\in\vf(\Tan^*Q\times\Real)$ 
solution to \eqref{evolvf11b} on $\Tan^*Q\times\Real$.
Then, an immediate corollary of the above theorem is the following:

\begin{teor}
Let $\X_{\bf H}\in\vf(\mathfrak{M})$ be the contact Hamiltonian vector field
solution to the equations \eqref{Whamilton-contact-eqs}
on $\mathfrak{M}_0$, and tangent to 
$\mathfrak{M}_0$.

Then the vector field $\X_\Lag\in\vf(\Tan Q\times\Real)$, defined by
$\X_\Lag\circ\rho_1=\Tan\rho_1\circ\X_{\bf H}$,
is a {\sc sode} vector field which is a
solution to the equations \eqref{eq-E-L-contact1}; that is, the contact Euler-Lagrange vector field,
on $\Tan Q\times\Real$.

Furthermore, the vector field $\X_{\rm h}\in\vf(\Tan^*Q\times\Real)$, defined by
$\X_{\rm h}\circ\rho_2=\Tan\rho_2\circ X_{\bf H}$,
is the contact Hamiltonian vector field solution to the equations \eqref{evolvf11b} on $\Tan^*Q\times\Real$.
\end{teor}

These results are analogous to those obtained for the unified
formalism of nonautonomous dynamical systems.

\begin{remark}{\rm
For singular almost-regular Lagrangians, the equivalence between the constraint algorithms 
in the unified and in the Lagrangian formalism
only holds when the second-order condition is imposed,
since, unlike in the unified formalism, this condition does not hold in the Lagrangian case (see \cite{MR-92,SR-83}).
}\end{remark}

\section{Symmetries of contact dynamical systems}
\label{sims}

(See \cite{GGMRR-2019b} and also \cite{DeLeon2019b} for another complementary approach to these topics. See also \cite{GLR-2023} for a complete classification of symmetries of autonomous and nonautonomous contact systems, and \cite{ACLMMP-23}
for a reduction scheme).

\subsection{Symmetries of contact Hamiltonian systems. Dissipated and conserved quantities}

Let $(M,\bmeta,h)$ be a contact Hamiltonian system with Reeb vector field $\Reeb$,
and let $\X_h$ be the contact Hamiltonian vector field for this system;
that is, the solution to the Hamilton equations \eqref{evolvf1}.

As, for other kinds of dynamical systems,
a \textsl{dynamical symmetry} of a contact Hamiltonian system  is a diffeomorphism 
$\Phi\colon M\longrightarrow M$ such that $\Phi_*\X_h=X_h$.
Similarly, an \textsl{infinitesimal dynamical symmetry} of a contact Hamiltonian system  is a vector field 
$Y\in\vf(M)$ whose local flux is a dynamical symmetry; that is,
$\Lie(Y)\X_h=[Y,\X_h]=0$.
This means that dynamical symmetries map solutions into solutions.

Nevertheless, we are mainly interested in those symmetries that let
the geometric structures invariant:

\begin{definition}
\label{defsym}
A \textbf{contact Noether symmetry} of a contact Hamiltonian system $(M,\bmeta,h)$ is a diffeomorphism 
$\Phi\colon M\rightarrow M$ such that
$$
\Phi^*\bmeta=\bmeta
\quad ,\quad 
\Phi^*h=h \ .
$$
An \textbf{infinitesimal contact Noether symmetry} of a contact Hamiltonian system $(M,\bmeta,h)$ is a vector field 
$Y\in \vf(M)$ whose local flux is a contact Noether symmetry; that is,
$$
\Lie(Y)\bmeta=0
\quad ,\quad 
\Lie(Y)h=0 \ .
$$
\end{definition}

\begin{prop}
\label{symreeb}
Every (infinitesimal) contact Noether symmetry preserves the Reeb vector field; that is, $\Phi_*\Reeb=\Reeb$
(or $[Y,\Reeb]=0$).
\end{prop}
\begin{proof}
We obtain that,
    \beann
\inn(\Phi^{-1}_*\Reeb)(\Phi^*\d\bmeta)&=&
\Phi^*(\inn(\Reeb)\d\bmeta)=0 \ , \\
\inn(\Phi^{-1}_*\Reeb)(\Phi^*\bmeta)&=&
\Phi^*(\inn(\Reeb)\bmeta)=1 \ ,
    \eeann
and, as $\Phi^*\bmeta=\bmeta$ and the Reeb vector field is unique, 
from these equalities we get $\Phi^{-1}_*\Reeb=\Reeb$ or, equivalently,
$\Phi_*\Reeb=\Reeb$.

For the infinitesimal case, if $Y\in\vf(M)$ is an infinitesimal Noether symmetry, we have:
\bea
\label{aux2}
\Lie([Y,\Reeb])\bmeta&=& 
\Lie(Y)\Lie(\Reeb)\bmeta-\Lie(\Reeb)\Lie(Y)\bmeta=
\Lie(Y)\inn(\Reeb)\d\bmeta+\d\inn(\Reeb)\eta=0 \ ,
\\
\Lie([Y,\Reeb])\d\bmeta&=& \d\Lie([Y,\Reeb])\bmeta=0 \ , \nonumber
\eea
hence $[Y,\Reeb]=\Lie(Y)\Reeb\in \ker\bmeta\cap\ker\d\bmeta=\{ 0\}$.
\\ \qed\end{proof}

As in the dynamical systems studied in the previous chapters,
 we have:

\begin{prop}
\label{csis}
Every (infinitesimal) contact Noether symmetry is a (infinitesimal) dynamical symmetry.
\end{prop}
\begin{proof}
In fact, using that $\Phi$ is a contact Noether symmetry and that, by the above proposition, $\Phi_*\Reeb=\Reeb$, we obtain:
\beann
\inn(\Phi^{-1}_*\X_h)\d\bmeta&=&
\inn(\Phi^{-1}_*\X_h)(\Phi^*\d\bmeta)=
\Phi^*(\inn(\X_h)\d\bmeta)=
\Phi^*(\d h-(\Reeb(h))\,\bmeta)=\d h-(\Reeb(h) )\,\bmeta \ ,
\\
\inn(\Phi^{-1}_*\X_h)\bmeta&=&
\inn(\Phi^{-1}_*\X_h)(\Phi^*\bmeta)=
\Phi^*(\inn(\X_h)\bmeta)=
\Phi^*(-h)=-h \ ,
\eeann
then $\Phi^{-1}_*\X_h$ is a solution to the dynamical equation \eqref{evolvf11b} and,
as the solution is unique, we conclude that $\Phi_*\X_h=\X_h$, and
hence $\Phi$ is a symmetry.

For the infinitesimal case we have that,
\beann
\inn([Y,\X_h])\bmeta&=& 
\Lie(Y)\inn(\X_h)\bmeta-\inn(\X_h)\Lie(Y)\bmeta=
-\Lie(Y)h=0 \ ,
\\
\inn([Y,\X_h])\d\bmeta&=&
\d\inn([Y,\X_h])\bmeta=0\ ,
\eeann
then $[Y,\X_h]\in \ker\bmeta\cap\ker\d\bmeta=\{ 0\}$
and it is a dynamical symmetry.
\\ \qed\end{proof}

Now, the result stated in Proposition \ref{disipenerg},
$\Lie(\X_h)h=-(\Reeb(h))\,h$,
induces us to define:

\begin{definition}\label{definition:diss-quan}
A \textbf{dissipated quantity} of the contact Hamiltonian system $(M,\bmeta,h)$ is a function 
$F\in\Cinfty(M)$ satisfying that 
\beq\label{eq:dissipation}
\Lie(\X_h)F=-(\Reeb(h))\,F \ .
\eeq
A \textbf{conserved quantity} of the contact Hamiltonian system $(M,\bmeta,h)$ is a function 
$G\in\Cinfty(M)$ satisfying that 
$$
\Lie(\X_h)G=0 \ .
$$
\end{definition}

For contact Hamiltonian systems, 
as we are dealing with dissipative systems,
symmetries are associated with dissipated
quantities:

\begin{teor}
\label{th:dissipation}
{\rm (Dissipation theorem).} 
If $Y\in\vf(M)$ is an infinitesimal dynamical symmetry,
then the function $F=-\inn(Y)\bmeta$
is a dissipated quantity.
\end{teor} 
\begin{proof}
In fact, Proposition \ref{csis} says that $[Y,\X_h]=0$; then,
\begin{align*}
\Lie(\X_h) F &=
-\Lie(\X_h) \inn(Y) \bmeta =
- \inn(Y) \Lie(\X_h)\bmeta - \inn(\Lie(\X_h)Y)\, \bmeta =
\\
&=
(\Reeb(h)) \inn(Y)\bmeta + \inn([Y,\X_h]) \bmeta =
 -(\Reeb(h))\, F + \inn([Y,\X_h]) \bmeta =
 -(\Reeb(h))\, F\ ,
\end{align*}
where we have applied that $\Lie(\X_h)\bmeta=-\Reeb(h))\,\bmeta$ (by \eqref{equivmanner}).
\\ \qed\end{proof}

\begin{remark}{\rm
In particular, the Hamiltonian vector field $\X_h$ is
trivially a symmetry, since $[\X_h,\X_h]=0$.
Then, Proposition \eqref{eq:disipenerg}
establishes that its associated dissipated quantity is the energy, 
$h=-\inn(X_h)\bmeta$.
}\end{remark}

These are ``non conservation theorems''
which account for the non-conservation of these quantities associated with the symmetries.
In particular, every dissipated quantity changes with the same rate, $-\Reeb(h)$, 
which suggests that the quotient of two dissipated quantities should be a conserved quantity. 
Indeed:

\begin{prop}
\label{prop:disscon} 
\begin{enumerate}
\item 
If $F_1$ and $F_2$ are dissipated quantities and $F_2\not=0$, then $\dst\frac{F_1}{F_2}$ is a conserved quantity.
\item 
If $F$ is a dissipated quantity and $G$ is a conserved quantity, then $FG$ is a dissipated quantity.
\end{enumerate}
\end{prop}
\begin{proof}
In fact, we have
\begin{align*}
\Lie(\X_h)\left(\frac{F_1}{F_2}\right)&=
\frac{1}{F_2}\Lie(\X_h)F_1-\frac{F_1}{{F_2}^2}\Lie(\X_h){F_2}=-
\frac{1}{F_2}(\Reeb(h))F_1+\frac{F_1}{{F_2}^2}(\Reeb(h))F_2=0\ .
\\
\Lie(\X_h)(FG)&=G\,\Lie(\X_h)F+F\,\Lie(\X_h)G=-(\Reeb(h))\,FG.
\end{align*}
\qed\end{proof}

A straightforward consequence of this proposition is that symmetries can also have associated conserved quantities:

\begin{corol}
For every infinitesimal symmetry $Y\in\vf(M)$,
if $h\not=0$, then $\dst -\frac{\inn(Y)\bmeta}{h}$ is a conserved quantity.
\label{consquanquot}
\end{corol}

Contact symmetries also allow us to generate new dissipated quantities 
from another one.

\begin{prop}
If $\Phi\colon M\rightarrow M$ is a contact Noether symmetry and 
$F\in\Cinfty(M)$ is a dissipated quantity, then so is $\Phi^*F$.
Similarly, if $Y\in\vf(M)$ is an infinitesimal contact Noether symmetry, then
$\Lie(Y)F$ is a dissipated quantity.
\end{prop}
\begin{proof}
In fact, we have
$$
\Lie(\X_h)(\Phi^*F)=
\Phi^*\Lie(\Phi_*\X_h)F=
\Phi^*\Lie(\X_h)F=
\Phi^*(-\Reeb(h))F=
-(\Reeb(h))(\Phi^*F) \ .
$$
For the infinitesimal case, taking into account that $Y$ is a contact infinitesimal $\Lie(\X_h)F=0$, that $[Y,\X_h]=\X_h$ (Proposition \ref{csis}) and that $[Y,\Reeb]=0$ (Proposition \ref{symreeb}), we have:
\beann
\Lie(\X_h)\Lie(Y)F&=&
\Lie([\X_h,Y])F-\Lie(Y)\Lie(\X_h)F+
-\Lie(Y)\big((\Reeb(h))\,F\big)
\\ &=&
-\Lie(\X_h)F-F\,\Lie(Y)(\Reeb(h))-(\Reeb(h))(\Lie(Y)F)
\\ &=&
-F\,\Lie([Y,\Reeb])h-F\,\Lie(\Reeb)\Lie(Y)h-(\Reeb(h))(\Lie(Y)F)
=-(\Reeb(h))(\Lie(Y)F) \ .
\eeann
\qed\end{proof}

\subsection{Symmetries for contact Lagrangian and canonical Hamiltonian systems}

Next, we particularize the results on symmetries and dissipated quantities 
to the case of Lagrangian dissipative systems and their canonical Hamiltonian formalism.

Consider a regular contact Lagrangian system $(\Tan Q\times\Real,\Lag)$, 
with Reeb vector field $\Reeb_\Lag$,
and let $\X_\Lag$ be the contact Euler--Lagrange vector field for this system;
that is, the solution to the Lagrangian equations \eqref{eq-E-L-contact1}.

All the definitions and results about symmetries and dissipated quantities stated in the preceding section 
hold for the contact system $(\Tan Q\times\Real,\bmeta_\Lag,E_\Lag)$.
In particular, the dissipation theorem states that
$-\inn(Y)\bmeta_\Lag$ is a dissipated quantity, for every infinitesimal dynamical symmetry $Y\in\vf(\Tan Q\times\Real)$,
the energy dissipation theorem states that
\beq
\label{energdis}
\Lie(\X_\Lag)E_\Lag=-(\Reeb_\Lag(E_\Lag))\,E_\Lag\ .
\eeq

Furthermore,
if $\varphi\colon Q\to Q$ is a diffeomorphism,
we can construct the {\sl canonical lift} of $\varphi$ to $\Tan Q\times\Real$ as the diffeomorphism 
$\Phi:=(\Tan\varphi, {\rm {\rm Id}_{\mathbb R}}) \colon
\Tan Q\times\Real
\longrightarrow 
\Tan Q\times\Real$.
These kinds of diffeomorfisms $\Phi$ are usually called {\sl natural transformations} on $\Tan Q\times\Real$.
Similarly, for every vector field $Z\in \vf(Q)$
we can define its {\sl complete lift}
to $\Tan Q\times\Real$ as the vector field
$Y\in\vf(\Tan Q\times\Real)$
whose local flux is the canonical lift of 
the local flux of $Z$ to $\Tan Q\times\Real$.
The infinitesimal transformation generated by $Y$ are called {\sl infinitesimal natural transformation} 
on $\Tan Q\times\Real$.

Then, taking into account the definitions
of the canonical endomorphism ${\cal J}$
and the Liouville vector field $\Delta$ in $\Tan Q\times\Real$,
it can be proved that canonical lifts of diffeomorphisms and vector fields
from $Q$ to $\Tan Q\times\Real$ preserve these canonical structures.
Furthermore, if these lifts leave the Lagrangian function invariant,
they also preserve the Reeb vector field $\Reeb_\Lag$.
Therefore, as an immediate consequence, 
we obtain the following relation:

\begin{prop}
If $\Phi\in{\rm Diff}(\Tan Q\times\Real)$ (resp. $Y\in\vf(\Tan Q\times\Real)$) is a canonical lift
to $\Tan Q\times\Real$ of a diffeomorphism $\varphi\in{\rm Diff}(Q)$
(resp.\ of a vector field $Z\in\vf(Q)$)
that leaves the Lagrangian invariant, then
it is a (infinitesimal) contact Noether symmetry, {\it i.e.},
$$
\Phi^*\bmeta_\Lag=\bmeta_\Lag \ ,\
\Phi^*E_\Lag=E_\Lag \qquad 
({\rm resp.}\ \Lie(Y)\bmeta_\Lag=0 \ ,\
\Lie(Y)E_\Lag=0 \:)\ .
$$
As a further consequence, it is a (infinitesimal) dynamical symmetry.
\end{prop}

As a direct consequence, if $\displaystyle\frac{\partial \Lag}{\partial q^i}=0$, 
then $\displaystyle\frac{\partial}{\partial q^i}$ 
is an infinitesimal contact Noether symmetry, and its associated dissipated quantity 
is the momentum 
$\displaystyle\frac{\partial \Lag}{\partial v^i}$;
that is,
$$
\Lie(\X_\Lag)\left(\frac{\partial\Lag}{\partial v^i}\right)=
-(\Reeb_\Lag(E_\Lag))\frac{\partial\Lag}{\partial v^i}=
\derpar{\Lag}{s}\frac{\partial\Lag}{\partial v^i}\ .
$$
(See also \cite{GeGu2002}for a similar description).

For the canonical Hamiltonian formalism, consider
the canonical contact Hamiltonian system $(\Tan^*Q\times\Real,\bmeta,{\rm h})$.
As in the Lagrangian formalism, if $\varphi\colon Q\to Q$ is a diffeomorphism,
we can construct the {\sl canonical lift} of $\varphi$ to $\Tan^*Q\times\Real$,
which is the diffeomorphism 
$\Phi:=(\Tan^*\varphi, {\rm {\rm Id}_{\mathbb R}}) \colon
\Tan^*Q\times\Real
\longrightarrow 
\Tan^*Q\times\Real$.
and is called a {\sl natural transformation} on $\Tan^*Q\times\Real$.
In the same way, for every vector field $Z\in \vf(Q)$
we define its {\sl complete lift}
to $\Tan^*Q\times\Real$ as the vector field
$Y\in\X(\Tan^*Q\times\Real)$
whose local flux is the canonical lift of 
the local flux of $Z$ to $\Tan^*Q\times\Real$.
It is called a {\sl natural infinitesimal transformations} on $\Tan^*Q\times\Real$.

The canonical forms $\Theta$ and $\Omega=-\d\Theta$
in $\Tan^*Q$ and their extensions to $\Tan^*Q\times\Real$ are invariant under the action
of canonical lifts of diffeomorphisms and vector fields from $Q$ to $\Tan^*Q$ and $\Tan^*Q\times\Real$.
Then, taking into account the definition
of the contact form $\bmeta$ in $\Tan^*Q\times\Real$, we have:

\begin{prop}
If
$\Phi\in{\rm Diff}(\Tan^*Q\times\Real)$ (resp. $Y\in\vf(\Tan^*Q\times\Real)$) is a canonical lift to $\Tan^*Q\times\Real$ of a diffeomorphism $\varphi\in{\rm Diff}(Q)$
(resp. of $Z\in\vf(Q)$), then
\begin{enumerate}
    \item
$\Phi^*\bmeta=\bmeta$ (resp $\Lie(Y)\bmeta=0$).
    \item
If, in addition,
$\Phi^*{\rm h}={\rm h}$ (resp. $\Lie(Y){\rm h}=0$),
then it is a  (infinitesimal) contact Noether symmetry.
\end{enumerate}
\end{prop}

In particular, we have the following:

\begin{teor}{\rm (Momentum dissipation theorem).} 
If $\displaystyle\frac{\partial {\rm h}}{\partial q^i}=0$, 
then $\displaystyle\frac{\partial}{\partial q^i}$ 
is an infinitesimal contact Noether symmetry, 
and its associated dissipated quantity  
is the corresponding momentum $p_i$;
that is,
$\Lie(\X_{\rm h})p_i=-(\Reeb({\rm h}))\,p_i$.
\label{disipe3}
\end{teor}
\begin{proof}
A simple computation in local coordinates shows that 
$\displaystyle
\Lie\Big(\frac{\partial}{\partial q^i}\Big)\bmeta=0$ 
and 
$\displaystyle
\Lie\Big(\frac{\partial}{\partial q^i}\Big)=0$.
Therefore, $\dst\frac{\partial}{\partial q^i}$
is a contact Noether symmetry and, hence, a dynamical symmetry. 
The other results are a consequence of the dissipation theorem. 
\end{proof}

\section{Examples}

Finally, we will consider the systems studied as examples in the previous chapters, 
incorporating a standard dissipation term that accounts for dissipative forces proportional to velocity.

\subsection{The damped harmonic oscillator}

Consider a harmonic oscillator in a medium with friction.
As in Section \ref{sho},
the configuration bundle of the system is $Q=\Real$ but, in order to develop the contact formulation, we take the manifold $Q\times\Real$. with coordinates $(t,q)$.

\subsubsection{Lagrangian formalism}

The Lagrangian description of the one-dimensional damped harmonic oscillator is done in the phase space
$\Tan Q\times\Real\simeq=\Real^2\times\Real$.
It is described by the hyperregular contact Lagrangian function,
$$
\Lag=\frac{1}{2}mv^{2}-\frac{1}{2}kq^{2}-\gamma s \ ,
$$
where $\gamma\in\Real^+$ is the dissipation parameter.

We have the contact Lagrangian form,
$$
\bmeta_\Lag=\d s-v\d q \ ,
$$
and the energy Lagrangian function is
$$ 
E_\Lag = \frac{1}{2}mv^2 + \frac{1}{2}kq^2 + \gamma s \ .
$$
For $\displaystyle X_\Lag=f\derpar{}{q}+F\derpar{}{v}+g\derpar{}{s}\in\vf(\Tan Q\times\Real)$,
the contact Lagrangian equations \eqref{ELeqs0}
and \eqref{ELeqs1} are
$$
g=\frac{1}{2}mv^{2}-\frac{1}{2}kq^{2}-\gamma s \quad , \quad f=v \quad , \quad mF=-kq-\gamma mv \ .
$$
whose solution is the {\sc sode} vector field,
$$
\X_\Lag=v\frac{\partial}{\partial q}-
\left(\frac{k}{m}\,q+\gamma v\right) \frac{\partial}{\partial v}+\left(\frac{1}{2}mv^{2}-\frac{1}{2}kq^{2}-\gamma s\right)\frac{\partial}{\partial s}\ .
$$
Its integral curves are the solutions to equations \eqref{ELeqs2} and \eqref{ELeqs3},
which for this system give the Herglotz--Euler--Lagrange equations
\beq
\label{HELdho}
\frac{ds}{dt}=\frac{1}{2}m\left(\frac{dq}{dt}\right)^2-\frac{1}{2}kq^{2}-\gamma s \quad , \quad
\frac{d^2q}{dt^2}=-\frac{k}{m}\,q-\gamma\,\frac{dq}{dt} \ ,
\eeq
and the last equation correspond to the well-known dynamical equation of the damped harmonic oscillator. 

The dissipation of the energy is given by 
\eqref{eq:disipenerg}, which in the Lagrangian formalism reads
$$
\Lie(\X_\Lag)E_\Lag=-\gamma\left(\frac{1}{2}mv^2 + \frac{1}{2}kq^2 + \gamma s\right)\ .
$$

\subsubsection{Hamiltonian formalism}

For the canonical Hamiltonian formalism of the system, 
$\Tan^*Q\times\Real\simeq=\Real^2\times\Real$,
and the Legendre map is 
$$
\mathfrak{F}\Lag^*q=q  \quad , \quad \mathfrak{F}\Lag^*p=mv  \quad , \quad \mathfrak{F}\Lag^*s=s  \ .
$$
The canonical contact form is $\bmeta=\d s-p\,\d q$, and 
the canonical Hamiltonian function reads
$$ 
{\rm h}=\frac{p^2}{2m^2} + \frac{1}{2}kq^2 + \gamma s \ .
$$
For $\displaystyle X_{\rm h}=f\derpar{}{q}+F\derpar{}{p}+g\derpar{}{s}$,
the contact Hamiltonian equations \eqref{evolvf11b} give
$$
f=\frac{p}{m} \quad , \quad F=-kq-\gamma p \quad , \quad g=\frac{p^2}{2m^2}-\frac{1}{2}kq^2-\gamma s  \ ,
$$
and its solution is the Hamiltonian vector field,
$$
\X_{\rm h}=\frac{p}{m}\frac{\partial}{\partial q}-
\left(kq+\gamma p\right) \frac{\partial}{\partial p}+\left(\frac{p^2}{2m^2}-\frac{1}{2}kq^2-\gamma s\right)\frac{\partial}{\partial s}\ ,
$$
whose integral curves are the solutions to equations \eqref{hel-eqs01b},
which for this system give the Hamilton equations,
\beq
\label{Hdho}
\dst\frac{dq}{dt}=\frac{p}{m}\quad ,\quad
\dst\frac{dp}{dt}=-kq-\gamma p\quad , \quad
\dst\frac{ds}{dt}=\frac{p^2}{2m^2}-\frac{1}{2}kq^2-\gamma s \ .
\eeq

Finally, the dissipation of the energy is given by \eqref{energdis},
$$
\Lie(\X_{\rm h}){\rm h}=-\gamma\left(\frac{p^2}{2m^2} + \frac{1}{2}kq^2 + \gamma s\right)\ .
$$

\subsubsection{Unified Lagrangian-Hamiltonian formalism}

The extended unified bundle is $\mathfrak{M}=\Tan Q\times_Q\Tan^*Q\times\Real\simeq\Real^4$
with coordinates $(q,v,p,s)$.
Then we have the precontact Hamiltonian system
$(\mathfrak{M},\bmeta_{\mathfrak{M}},{\bf H})$,
where the canonical contact form is
$$
\bmeta_{\mathfrak{M}}=\d s-p\,\d q \ ,
$$
the Hamiltonian function is
$$
{\bf H}=pv-\frac{1}{2}mv^{2}+\frac{1}{2}kq^{2}+\gamma s \ .
$$
and we can take $\dst \Reeb=\derpar{}{s}$ as Reeb vector field.
For this system, the compatibility condition for equations \eqref{Whamilton-contact-eqs} leads to define the submanifold $\mathfrak{M}_0\hookrightarrow\mathfrak{M}$
defined by
$$
\mathfrak{M}_0=\{ (q,v,p,s)\in\mathfrak{M}\,\mid\,  p-mv=0 \} \ ,
$$
and the Hamiltonian vector field solution to \eqref{Whamilton-contact-eqs} on $\mathfrak{M}_0$ is
$$
\X_{\bf H}\vert_{\mathfrak{M}_0}=
v\derpar{}{q}+F\derpar{}{v}-
\left(kq+\gamma p\right)\derpar{}{p}
+\left(\frac{p^2}{2m^2}-\frac{1}{2}kq^2-\gamma s\right)\derpar{}{s}\ .
$$
The tangency condition for $\X_{\bf H}$ on $\mathfrak{M}_0$ gives
$\dst F=-\frac{k}{m}q-\gamma\frac{p}{m}$,
and thus finally,
$$
\X_{\bf H}\vert_{\mathfrak{M}_0}=
v\derpar{}{q}-\left(\frac{k}{m}q+\gamma\frac{p}{m}\right)\derpar{}{v}-
\left(kq+\gamma p\right)\derpar{}{p}
+\left(\frac{p^2}{2m^2}-\frac{1}{2}kq^2-\gamma s\right)\derpar{}{s}\ .
$$
Therefore, the integral curves of $\X_{\bf H}$ are the solutions to
$$
\frac{dq}{dt}=v \quad , \quad
m\frac{dv}{dt}=-mkq-\gamma p \quad , \quad
\frac{dp}{dt}=-kq-\gamma p \quad , 
\quad
\frac{ds}{dt}=\frac{p^2}{2m^2}-\frac{1}{2}kq^2-\gamma s \ .
$$
Te first two equations of this system are equivalent to the
second Herglotz--Euler--Lagrange equation
\eqref{HELdho}
and, using the constraint $p=mv$ (the Legendre map),
the first and third equations are the first pair of the Hamilton equations \eqref{Hdho}.

\subsection{The Kepler problem with friction}

Consider the motion of a massive particle (of mass $m$)
under the action of Newtonian central force
in a stellar media with friction.
As in the previous cases, the motion is on a plane; hence
$Q=\Real^2$, and we take $(r,\phi,s)$ as coordinates in $Q\times\Real$,
(with the origin of $r$ at the center of the force).

\subsubsection{Lagrangian formalism}

The Lagrangian formalism takes place in $\Tan Q\times\Real\simeq\Real^5$, with local coordinates 
$(r,\phi,v_r,v_\phi,s)$
The contact Lagrangian function that describes the dynamics is
$$
\Lag=\frac{1}{2}m(v_r^2+r^2v_\phi^2)-\frac{K}{r}-\gamma s \quad , \quad K\not=0 \ ;
$$
which is regular as in the above situations.
The energy Lagrangian function and the contact Lagrangian form are
\beann
E_\Lag&=&\frac{1}{2}m(v_r^2+r^2v_\phi^2)+\frac{K}{r}+\gamma s \ , \\
\bmeta_\Lag&=&\d s-m(v_r\,\d r+r^2v_\phi\,\d\phi) \ , 
\eeann
and the Lagrangian is regular.
For $\displaystyle \X_\Lag=f_r\derpar{}{r}+f_\phi\derpar{}{\phi}+
g_r\derpar{}{v_r}+g_\phi\derpar{}{v_\phi}+g\derpar{}{s}\in\vf(\Tan^*Q\times \Real)$, 
the contact Lagrangian equations \eqref{ELeqs0}
and \eqref{ELeqs1} give
\beann
g=\frac{1}{2}m(v_r^2+r^2v_\phi^2)-\frac{K}{r}-\gamma s \quad , \quad
f_r=v_r \quad , \quad f_\phi=v_\phi \ , \\ 
mg_r=2mrv_\phi f_\phi-mrv_\phi^2+\frac{K}{r^2}-\gamma mv_r \quad , \quad 
g_\phi=-\frac{2v_\phi f_r}{r}-\gamma r^2v_\phi   \ ,
\eeann
and the contact Euler--Lagrange vector field is
\beann
\X_\Lag&=&v_r\derpar{}{r}+v_\phi\derpar{}{\phi}+\left(rv_\phi^2+\dst\frac{K}{mr^2}-\gamma v_r\right)\derpar{}{v_r}
+\left(-\frac{2v_\phi v_r}{r}-\gamma v_\phi\right)\derpar{}{v_\phi} \\
& & +
\left(\frac{1}{2}m(v_r^2+r^2v_\phi^2)-\frac{K}{r}-\gamma s\right)\derpar{}{s}\ .
\eeann
Then, its integral curves are the solutions to equations \eqref{ELeqs2} and \eqref{ELeqs3} which read as
\bea
\frac{ds}{dt}=\frac{1}{2}m\left(\Big(\frac{dr}{dt}\Big)^2+r^2\Big(\frac{d\phi}{dt}\Big)^2\right)-\frac{K}{r}-\gamma s \quad , \quad
\nonumber \\
 m\frac{d^2r}{dt^2}=mr\Big(\frac{d\phi}{dt}\Big)^2+\frac{K}{r^2}-m\gamma v_r  
\quad , \quad  \frac{d^2\phi}{dt^2}=
-\frac{2}{r}\,\frac{d\phi}{dt}\,\frac{dr}{dt}-\gamma \frac{d\phi}{dt} \ ,
\label{eqEL-M}
\eea
and are the Herglotz--Euler--Lagrange equations for this system.

Similarly to the above example,
the dissipation of the energy reads
$$
\Lie(\X_\Lag)E_\Lag=-\gamma \left(\frac{1}{2}m(v_r^2+r^2v_\phi^2)+\frac{K}{r}+\gamma s\right)\ .
$$
Furthermore, there exists a Lagrangian contact Noether symmetry 
which is generated by the vector field $\dst Y=\derpar{}{\phi}$; in fact,
bearing in mind \eqref{NsymKep1} and \eqref{NsymKep2}, we obtain
\beann
\Lie(Y)\bmeta_\Lag&=& 
\Lie\left(\derpar{}{\phi}\right)\left(\d s-m(v_r\,\d r+r^2v_\phi\,\d\phi)\right) \\ &=&
\Lie\left(\derpar{}{\phi}\right)\d s-
\Lie\left(\derpar{}{\phi}\right)\left(m(v_r\,\d r+r^2v_\phi\,\d\phi)\right)=0 \ , 
\\
\Lie(Y)E_\Lag&=& 
\Lie\left(\derpar{}{\phi}\right)\left(\frac{1}{2}m(v_r^2+r^2v_\phi^2)+\frac{K}{r}+\gamma s\right)=0 \ ,
\eeann
and the corresponding dissipated quantity is the
angular momentum map
$$
F=-\inn\left(\derpar{}{\phi}\right)\bmeta_\Lag=mr^2v_\phi \ .
$$
Finally, using Corollary \ref{consquanquot},
we obtain that $\dst\frac{F}{E_\Lag}$ is a conserved quantity associated with $Y$.

\subsubsection{Hamiltonian formalism}

For the Hamiltonian formalism, $\Tan^*Q\times\Real\simeq\Real^5$,
with local coordinates $(r,\phi,p_r,p_\phi,s)$. The Legendre transformation is,
$$
\mathfrak{F}\Lag^*r=r  \quad , \quad \mathfrak{F}\Lag^*\phi=\phi \quad , \quad
\mathfrak{F}\Lag^*p_r=mv_r  \quad , \quad \mathfrak{F}\Lag^*p_\phi=mr^2v_\phi \quad , \quad \mathfrak{F}\Lag^*s=s  \ ,
$$
which is a diffeomorphism, and the Lagrangian is hyperregular.
The canonical Hamiltonian function and the canonical contact form are
\beann
{\rm h}&=&\frac{p_r^2}{2m}+\frac{p_\phi^2}{2mr^2}+\frac{K}{r}+\gamma s \\
\bmeta&=&\d s-p_r\,\d r-p_\phi\,\d\phi \ .
\eeann
For $\displaystyle \X_{\rm h}=F_r\derpar{}{r}+F_\phi\derpar{}{\phi}+
G_r\derpar{}{p_r}+G_\phi\derpar{}{p_\phi}+g\derpar{}{s}\in\vf(\Tan^*Q\times\Real)$,
the contact Hamiltonian equations \eqref{evolvf11b} lead to
$$
F_r=\frac{p_r}{m} \ , \ F_\phi=\frac{p_\phi}{mr^2} \ , \
G_r=\frac{p_\phi^2}{mr^3}+\frac{K}{r^2}-\gamma p_r \ , \ G_\phi=-\gamma p_\phi  \ , \ g=\frac{p_r^2}{2m}+\frac{p_\phi^2}{2mr^2}-\frac{K}{r}-\gamma s \ .
$$
and then, the contact Hamiltonian vector field is
$$
\X_{\rm h}=\frac{p_r}{m}\derpar{}{r}+\frac{p_\phi}{mr^2}\derpar{}{\phi}+\left(\frac{p_\phi^2}{mr^3}+\frac{K}{r^2}-\gamma p_r\right)\derpar{}{p_r}-\gamma p_\phi\derpar{}{p_\phi}+\left(\frac{p_r^2}{2m}+\frac{p_\phi^2}{2mr^2}-\frac{K}{r}-\gamma s\right)\derpar{}{s} \ .
$$
Hence, its integral curves are the solutions to the equations \eqref{hel-eqs01b} which are
\bea
m\frac{dr}{dt} =p_r \quad , \quad mr^2\frac{d\phi}{dt} =p_\phi \quad , \quad
\frac{dp_r}{dt}=\frac{p_\phi^2}{mr^3}+\frac{K}{r^2}-\gamma p_r & , &
\frac{dp_\phi}{dt}=-\gamma p_\phi \quad ,
\label{dislawKeppler} \\
\frac{ds}{dt}=\frac{p_r^2}{2m}+\frac{p_\phi^2}{2mr^2}-\frac{K}{r}-\gamma s & , & \nonumber
\eea
and are the Hamiltonian equations for this system.

As it is usual, using the Legendre map
one can easily check that $\mathfrak{F}\Lag_*X_\Lag=X_{\rm h}$.

Similarly to the above example,
the dissipation of the energy reads
$$
\Lie(\X_{\rm h}){\rm h}=-\gamma \left(\frac{p_r^2}{2m}+\frac{p_\phi^2}{2mr^2}+\frac{K}{r}+\gamma s\right)\ .
$$
In addition, we have a Hamiltonian contact Noether symmetry which is again the
vector field $\dst Y=\derpar{}{\phi}$, since
\beann
\Lie(Y)\bmeta&=& 
\Lie\left(\derpar{}{\phi}\right)\left(\d s-p_r\,\d r-p_\phi\,\d\phi\right) =0 \ , 
\\
\Lie(Y){\rm h}&=& 
\Lie\left(\derpar{}{\phi}\right)\left(\frac{p_r^2}{2m}+\frac{p_\phi^2}{2mr^2}+\frac{K}{r}+\gamma s\right)=0 \ .
\eeann
The associated dissipated quantity is again the angular momentum,
$$
F=-\inn\left(\derpar{}{\phi}\right)\bmeta=p_\phi \ .
$$
and the corresponding dissipation law is given directly by the last Hamilton equation in \eqref{dislawKeppler}.
As above, $\dst\frac{F}{{\rm h}}$ is a conserved quantity associated with $Y$.

\subsubsection{Unified Lagrangian-Hamiltonian formalism}

In the precontact extended unified bundle
$\mathfrak{M}=\Tan Q\times_Q\Tan^*Q\times\Real\simeq\Real^7$,
with coordinates $(r,\phi,v_r,v_\phi,p_r,p_\phi,s)$, 
the canonical precontact form and the Hamiltonian function are
\beann
\bmeta_{\mathfrak{M}}&=&\d s-p_r\,\d r-p_\phi\,\d\phi \ , \\
{\bf H}&=&p_rv_r+p_\phi v_\phi-\frac{1}{2}m(v_r^2+r^2v_\phi^2)+\frac{K}{r}+\gamma s \ ,
\eeann
and we can take $\dst \Reeb=\derpar{}{s}$ as Reeb vector field.

For $\displaystyle \X_{\bf H}=f_r\derpar{}{r}+
f_\phi\derpar{}{\phi}+g_r\derpar{}{v_r}+g_\phi\derpar{}{v_\phi}+G_r\derpar{}{p_r}+G_\phi\derpar{}{p_\phi}+g\derpar{}{s}\in\vf(\mathfrak{M})$,
equations \eqref{Whamilton-contact-eqs} lead to
\bea
f_r=v_r \quad , \quad f_\phi=v_\phi \quad , \quad 
G_r=\frac{K}{r^2}+mrv_\phi^2-\gamma p_r \quad , \quad G_\phi=-\gamma p_\phi & , & 
\nonumber \\
p_r=mv_r \quad , \quad p_\phi=mr^2v_\phi \quad , \quad
g=\frac{p_r^2}{2m}+\frac{p_\phi^2}{2mr^2}-\frac{K}{r}-\gamma s & . & \label{eqsM}
\eea
The first two equations in \eqref{eqsM} are constraints defining the submanifold
$\mathfrak{M}_0\hookrightarrow\mathfrak{M}$ which give the Legendre map.
Thus, the Hamiltonian vector field is
\beann
\X_{\bf H}\vert_{\mathfrak{M}_0}&=&
v_r\derpar{}{r}+v_\phi\derpar{}{\phi}+g_r\derpar{}{v_r}+g_\phi\derpar{}{v_\phi}
+\left(\frac{K}{r^2}+mrv_\phi^2-\gamma p_ r\right)\derpar{}{p_r}-\gamma p_\phi\derpar{}{p_\phi} \\ & &
+\left(\frac{p_r^2}{2m}+\frac{p_\phi^2}{2mr^2}-\frac{K}{r}-\gamma s\right)\derpar{}{s} \ ,
\eeann
and the tangency condition of $\X_{\bf H}$ on $\mathfrak{M}_0$ leads to obtain that, finally,
\beann
\X_{\rm H}\vert_{\mathfrak{M}_0}&=&
v_r\derpar{}{r}+v_\phi\derpar{}{\phi}+
\left(rv_\phi^2+\dst\frac{K}{mr^2}-\gamma\frac{p_r}{m}\right)\derpar{}{v_r}-\left(\frac{2v_rv_\phi}{r}+\gamma\frac{p_\phi}{m}\right)\derpar{}{v_\phi}
\\ & &
+\left(\frac{K}{r^2}+mrv_\phi^2-\gamma p_ r\right)\derpar{}{p_r}-\gamma p_\phi\derpar{}{p_\phi}+\left(\frac{p_r^2}{2m}+\frac{p_\phi^2}{2mr^2}-\frac{K}{r}-\gamma s\right)\derpar{}{s} \ .
\eeann
Therefore, the integral curves of $\X_{\bf H}$,
on $\mathfrak{M}_0$, are the solutions to
\bea
\frac{dr}{dt}=v_r \quad , \quad
\frac{d\phi}{dt}=v_\phi \quad , \quad
\frac{dv_r}{dt}=\frac{K}{mr^2}+rv_\phi^2-\gamma\frac{p_r}{m} \quad , \quad
\frac{dv_\phi}{dt}=-\frac{2v_rv_\phi}{r}-\gamma\frac{p_\phi}{m} & , &
\label{unieqs1}
\\
\frac{dp_r}{dt}=\frac{K}{r^2}+mrv_\phi^2-\gamma p_r \quad , \quad
\frac{dp_\phi}{dt}=-\gamma p_\phi \quad , \quad
\frac{ds}{dt}=\frac{p_r^2}{2m}+\frac{p_\phi^2}{2mr^2}-\frac{K}{r}-\gamma s & . &
\label{unieqs2}
\eea
Equations \eqref{unieqs1} are equivalent to the Herglotz--Euler--Lagrange equation \eqref{eqEL-M},
and using the constraints $p_r=mv_r$ 
and $p_\phi=mr^2v_\phi$; that is, the Legendre map,
the first and second equations in \eqref{unieqs1} and \eqref{unieqs2} are
the Hamiltonian equations \eqref{dislawKeppler} of the system.

%%%%%%%%%%%%%%%%%%%%%%%%%%%%%%%%%%%%%%%%%%%%%%%%%%%

\begin{appendix}
\chapter{Additional contents}

\section{Tangent and cotangent bundles}
\label{sec:tangentb}

(For a detailed account of these subjects and the proof of the results, see for example, \cite{AM-78,Ar-89,Con2001,KN-96,Lee2013,Ok-87,PmQ-69,St-64,Wa-fsmlg}).

\subsection{The tangent bundle of a manifold. Canonical lifts}
\label{canliftq}

Let $Q$ be a differentiable manifold and  $q\in Q$.
Every differentiable curve
$\gamma\colon (-\epsilon ,\epsilon )\subset \Real \to Q$
passing through $q$; that is, 
$\gamma (0)=q$ ,
induces a derivation ${\rm D}_\gamma$ in $\Cinfty (Q)$,
in the following way
$$
\begin{array}{cccc}
{\rm D}_\gamma \colon & \Cinfty (Q) & \to & \Real
\\
& f & \mapsto &\displaystyle \lim_{t \mapsto 0}\frac{(f\circ\gamma )(t)-(f\circ\gamma )(0)}{t}
\end{array} \quad .
$$
In the set of such differentiable curves we can define the following equivalence relation
$$
\gamma_1 \sim \gamma_2 \quad \Longleftrightarrow
\quad
{\rm D}_{\gamma_1}={\rm D}_{\gamma_2} \ .
$$
Then:
\begin{definition}
\ben
\item
A \textbf{tangent vector} to $Q$ at $q$
is every equivalence class defined  by this relation.
\item
The \textbf{tangent space} of $Q$ at $q$,
denoted by $\Tan_qQ$, is the vector space of all the tangent vectors
to $Q$ at $q$.
\item
The \textbf{tangent bundle} of $Q$ is defined as
 \(\dst \Tan Q := \bigcup_{q \in Q}\Tan_qQ\).
We denote $\tau_Q \colon \Tan Q \to Q$ its natural projection.
\een
\end{definition}

Every point ${\rm p}\in\Tan Q$ is a couple
$(q,v)\equiv v_q$, where $q=\tau_Q({\rm p})\in Q$
and $v\in\Tan_qQ$ (it is a tangent vector).

\begin{prop}
The tangent bundle $\Tan Q$ is a $2n$-dimensional differentiable manifold whose differentiable structure is inherited from $Q$.
In addition, the natural projection $\tau_Q$ is a submersion.
\end{prop}
\begin{proof}
If ${\cal A}=\{ (U_\alpha;\phi_\alpha)\}$, with $\phi_\alpha\equiv (\coor{x}{1}{n})$,
is an atlas of local charts on $Q$,
then the induced atlas on $\Tan^*Q$ is
$\Tan{\cal A}=\{ (\tau_Q^{-1}U_\alpha;\psi_\alpha) \}$,
with the coordinate functions $\psi_\alpha$ defined as follows:
$$
\begin{array}{ccccc}
\psi_\alpha&\colon&\tau^{-1}_Q(U_\alpha)&\longrightarrow&
\phi_\alpha(U_\alpha)\times\Real^n\subset\Real^{2n}\\
& & {\rm p}=(q,v) & \mapsto & (q^i(p),v^i({\rm p}))
\end{array} \quad ;
$$
where $i=1,\ldots,n$, and:
\ben
\item
$q^i({\rm p})=(x^i\circ\tau_Q)({\rm p})$
\footnote{
It is usual to commit an abuse of notation denoting by $q^i$
both the coordinates in the base manifold $Q$
and in the tangent bundle $\Tan Q$, and we will do in the sequel.}.
\item
If \(\dst v=\lambda^j\derpar{}{x^j}\Big\vert_q\), then
$v^i(p)=\lambda^i=v(x^i)$;
that is, $v^i(p)$ are the components of the tangent vector
$v\in\Tan_qQ$ in the basis \(\dst\left\{ \derpar{}{x^i}\Big\vert_q\right\}\).
\een
It is obvious that \(\dst\Tan Q=\bigcup_\alpha\tau^{-1}_Q(U_\alpha)\)
and that $\psi_\alpha$ are diffeomorphisms.
Then, it is immediate to prove that $\Tan{\cal A}$ endows $\Tan Q$
with the structure of a differentiable manifold.
If $(U;q^i)$ and $(\bar U;\bar q^i)$ are two local charts
in $Q$ such that 
$U\cap\bar U\not=\buit$, and $\bar q^j =\varphi^j (q^i)$ in $U\cap\bar U$,
then for the induced charts $(\tau^{-1}_Q(U);q^i,v^i)$ and 
$(\tau^{-1}_Q\bar U);\bar q^i,\bar v^i)$
in $\Tan Q$ we have that, on $\Tan U\cap\Tan\bar U$,
the relation between the coordinates $v^i$ and $\bar v^i$ is
\(\dst\bar v^j=\derpar{\varphi^j}{q^i}v^i\), since
$$
v=v^i\derpar{}{q^i}\Big\vert_q=v^i\derpar{\varphi^j}{q^i}\derpar{}{\bar q^j}\Big\vert_q=
\bar v^j\derpar{}{\bar q^j}\Big\vert_q \ ,
$$
and hence \(\dst\bar v^j=\derpar{\varphi^j}{q^i}v^i\).

Bearing in mind this local description of the tangent bundle,
it is evident that $\dim\,\Tan Q=2n$.

Finally, it is immediate to prove that $\tau_Q$ is an surjective submersion,
since the canonical projection 
$\tau_Q \colon \Tan Q \to Q$,
is a surjective map given, in natural coordinates, by
$\tau_Q (q^i,v^i)=q^i$, and hence its tangent map
$\Tan\tau_Q\colon\Tan\Tan Q\to\Tan Q$ is defined as follows:
if $(q,v) \in \Tan Q$ and $X \in \Tan_{(q,v)}(\Tan Q)$
we have that
\(\dst X=\lambda^i\derpar{}{q^i}\Big\vert_{(q,v)}
+\mu^i\derpar{}{v^i}\Big\vert_{(q,v)}\), and
$$
[\Tan_{(q,v)}\tau_Q (X)](q^j)=X(q^j\circ\tau_Q )=X(q^j)=\lambda^j\ ,
$$
therefore
$$
(\Tan\tau_Q )((q,v),X) =
(\Tan\tau_Q )\left( (q,v),\lambda^i\derpar{}{q^i}\Big\vert_{(q,v)}+
                  \mu^i\derpar{}{v^i}\Big\vert_{(q,v)} \right)
= \left(q,\lambda^i\derpar{}{q^i}\Big\vert_q \right) \ ,
$$
and the associated matrix is
\beq
 (I_{n\times n},0_{n\times n}) =
\left(\begin{matrix}
1 & \ldots & 0 & 0 & \ldots & 0 \\
              \vdots &  & \vdots & \vdots & & \vdots \\
              0 & \ldots & 1 & 0 & \ldots & 0\end{matrix}\right) \ .
\label{matrix1}
\eeq
In this way, we conclude that $\tau_Q$ is a surjective submersion.
\\ \qed \end{proof}

\begin{definition}
The above charts of the tangent bundle
are called \textbf{natural charts} and their coordinates \textbf{ natural coordinates}
($q^i$ are the \textbf{base coordinates}
and $v^i$ the \textbf{fiber coordinates}).
\end{definition}

Observe that the coordinate change in $\Tan Q$,
from $(q^i,v^i)$ to $(\bar q^i,\bar v^i)$ has the Jacobian matrix$$
\left(\begin{matrix}
\displaystyle\left(\derpar{\varphi^j}{q^i}\right) & 0 \\
\displaystyle\left(\frac{\partial^2\varphi^j}{\partial q^k\partial q^i}\right)v^i &
\displaystyle\left(\derpar{\varphi^j}{q^i}\right)
\end{matrix}\right) \ ,
$$
and, taking into account that its determinant is positive at every point, we conclude that:

\begin{corol}
$\Tan Q$ is an orientable manifold.
\end{corol}

The tangent bundle of a manifold is an example of a  {\sl vector bundle}.
This structure is defined as follows:

\begin{definition}
A \textbf{vector bundle} is a triple $(E,B,\pi)$,
where $E,B$ are differentiable manifolds
(with $\dim\, B=m$, $\dim\, E=m+n$) and $\pi\colon E\to B$
is a surjective submersion such that, for every $p\in B$,
there exists a local chart $(U,\phi)$, $p\in U$, and a diffeomorphism
$\psi\colon\pi^{-1}(U)\to \phi(U)\times\Real^n$
satisfying that: 
\ben
\item
If $\pi_1\colon\phi(U)\times\Real^n\to\phi(U)$
is the natural projection, then $\pi_1\circ\psi=\pi$.
\item
$E_p=\pi^{-1}(p)$ is a vector space and,
if $\pi_2\colon\phi(U)\times\Real^n\to\Real^n$
is the natural projection, then the maps
$$
\begin{array}{ccccc}
\psi_p & \colon & E_p & \longrightarrow & \Real^n \\
 & & v &\mapsto & (\pi_2\circ\psi)(p,v)
\end{array}
$$
are vector space morphisms.
\een
$E$ is called the \textbf{total manifold} of the vector bundle,
$B$ is the \textbf{base manifold}, $\pi$ the \textbf{projection of the bundle}, 
$\Real^n$ is the \textbf{typical fiber},
$n$ is the \textbf{rank of the bundle} and, for every $ p\in B$,
$E_p$ is the \textbf{fiber} over $p$.
The pair $(U,\phi)$ is said to be a \textbf{trivializing open set}
and $\psi$ is the associated \textbf{coordinate map}.
\end{definition}

Observe that the family$(\pi^{-1}(U),\psi)$ is a differentiable atlas of $E$ which is said to be {\sl \textbf{adapted}} to the projection $\pi:E\to B$.

There is a natural way to lift diffeomorphisms, curves and vector fields
 on a manifold to its tangent bundle.

\begin{definition}
Let $Q$ be a differentiable manifold and a diffeomorphism
$$
\begin{array}{ccccc}
\varphi&\colon&Q&\longrightarrow&Q
\\
& & q & \mapsto & \varphi(q)
\end{array} \ .
$$
The \textbf{canonical lift} of $\varphi$ to $\Tan Q$
is the diffeomorphism
$$
\begin{array}{ccccc}
\Tan\varphi&\colon&\Tan Q&\longrightarrow&\Tan Q
\\
& &(q,v ) & \mapsto & (q,\Tan_q\varphi (v))
\end{array} \ .
$$
\end{definition}

This definition also holds for every map $\varphi \colon Q\longrightarrow Q$.

The following properties are immediate from the definition:

\begin{prop}
For every $\varphi,\phi\in{\rm Diff}\,Q$,
\ben
\item
$\varphi^{-1}\circ\tau_Q\circ\Tan \varphi =\tau_Q$.
\item
$\Tan (\varphi\circ\phi )=\Tan\varphi\circ\Tan\phi$.
\een
\label{exerc1}
\end{prop}

Using the definition of canonical lift of diffeomorphisms, we set:

\begin{definition}
Let $Z\in\vf(Q)$ be a vector field.
The  \textbf{total}, \textbf{complete}  or \textbf{canonical lift} of $Z$ to $\Tan Q$
is the vector field $Z^C\in\vf (\Tan Q)$
whose local uniparametric groups of diffeomorphisms are the
canonical lifts $\{\Tan F_t\}$ of the
local uniparametric groups of diffeomorphisms
$\{F_t\}$ of $Z$.
\end{definition}

As a direct consequence of the definition and of the theorem of existence and unicity of local uniparametric groups of diffeomorphisms,
we obtain the following result:

\begin{prop}
Let $Z\in\vf (Q)$.
Then $Z^C$ is $\tau_Q$-projectable and $\tau_{Q*}Z^C=Z$;
that is, 
$\Tan\tau_Q(Z^C_{(q,v)})=Z_x$, for every $(q,v)\in\Tan Q$.
\label{prolevcan0}
\end{prop}

\noindent{\bf Local expression}:
In a chart of coordinates $(U;q^i)$ of $Q$, if
$$
Z\vert_U=f^i(q^j)\derpar{}{q^i} \ ,
$$
then we have that, in the induced chart of natural coordinates $(\tau_Q^{-1}(U);q^i,v^i)$ of $\Tan Q$,
$$
Z^C\vert_{\tau_Q^{-1}(U)}=f^i(q^j)\derpar{}{q^i}+v^k\derpar{f^i}{q^k}(q^j)\derpar{}{v^i} \ .
$$
In order to prove it, remember that, if $(q,u)\in\tau_Q^{-1}(U)$, then
$$
Z^C(q,u)=\frac{d}{d t}\Big\vert_{t=0}\Tan F_t(q,u) \ ,
$$
where \(\dst\Tan F_t(q^i,v^i)=\left(F_t(q^i),\derpar{F_t}{q^i}v^i\right)\) and,
as a consequence,
$$
\frac{d}{d t}\Big\vert_{t=0}\Tan F_t(q,u)=
\left(\frac{d}{d t}\Big\vert_{t=0}F_t(q),
\frac{d}{d t}\Big\vert_{t=0}\left(\derpar{F_t}{q^i}v^i\right)(q,u)\right)=
\left(f^i(q),v^j\derpar{f^i}{q^j}(q^k)\right) \ ,
$$
and the result follows.

Let $\gamma\colon (a,b)\subseteq\Real\to Q$ be a curve.
If $t_0 \in (a,b)$ and $x=\gamma (t_0)$,
then $\dot\gamma (t_0)$ is the tangent vector to the curve
at the point $\gamma (t_0)$; that is,
$\dot\gamma (t_0)\in\Tan_{\gamma (t_0)}Q$.
As $\dot\gamma (t)$ is well-defined for every $t\in (a,b)$,
we can define:

\begin{definition}
\label{canlifcurv}
The \textbf{canonical lift} of a curve $\gamma$ to $\Tan Q$ is the curve
$$
\begin{array}{cccc}
\widetilde  \gamma \colon &(a,b)\subset\Real&\to&\Tan Q
\\
&t&\mapsto&(\gamma (t),\dot\gamma(t))
\end{array}  \ ,
$$
which is defined as
$$
(\gamma (t_0),\dot\gamma (t_0))=\Tan_{t_0}\gamma\frac{d}{d t} \ ;
$$
that is, for every $f\colon Q \to \Real$, then
$$
(\gamma (t_0),\dot\gamma (t_0))f=
(\Tan_{t_0}\gamma\frac{d}{d t})f=
\frac{d}{d t}\Big\vert_{t_0}(f\circ\gamma )
=\lim_{h\mapsto 0}\frac{f(\gamma (t_0+h))-f(\gamma (t_0))}{h} \ .
$$
Observe that $\tau_Q\circ\widetilde  \gamma=\gamma$.
\end{definition}

\noindent{\bf Local expression}:
If $(q^i)$ are local coordinates of $Q$ in a neighbourhood of
$\gamma (t_0)$, then $\gamma =(\gamma^1,\ldots ,\gamma^n)$,
with $\gamma^i=q^i\circ\gamma$.
If $(q^i,v^i)$ are the natural induced coordinates in $\Tan Q$,
then $\widetilde  \gamma$ is given by
$\widetilde  \gamma =(\gamma^1,\ldots ,\gamma^n,\dot\gamma^1,\ldots ,\dot\gamma^n)$,
where \(\dst\dot\gamma^i=\frac{d\gamma^i}{d t}\).

The tangent vector to $\widetilde  \gamma$ at $\widetilde  \gamma (t_0)$
is \(\dst \Tan_{t_0}\widetilde  \gamma\frac{d}{d t}\)
and, if $f\colon\Tan Q\to \Real$, we have that
$$
\left( \Tan_{t_0}\widetilde  \gamma\frac{d}{d t}\right) f=
\frac{d}{d t}\Big\vert_{t_0}f\circ\widetilde  \gamma \ ,
$$
and hence
\begin{eqnarray*}
\left(\Tan_{t_0}\widetilde  \gamma\frac{d}{d t}\right) q^i&=&
\frac{d}{d t}\Big\vert_{t_0}q^i\circ\widetilde  \gamma=\dot\gamma^i(t_0) \ ,
\\
\left(\Tan_{t_0}\widetilde  \gamma\frac{d}{d t}\right) v^i&=&
\frac{d}{d t}\Big\vert_{t_0}v^i\circ\widetilde  \gamma=
\frac{d}{d t}\Big\vert_{t_0}\dot\gamma^i=\ddot\gamma^i(t_0) \ ;
\end{eqnarray*}
that is
$$
\Tan_{t_0}\widetilde\gamma\Big(\frac{d}{d t}\Big)=
\dot\gamma^i(t_0)\derpar{}{q^i}\Big\vert_{\widetilde  \gamma (t_0)}
+\ddot\gamma^i(t_0)\derpar{}{v^i}\Big\vert_{\widetilde  \gamma (t_0)} \ .
$$

\begin{remark}{\rm
\label{intcurve}
Remember that, given a vector field $X\in\vf(Q)$,
a curve $\gamma\colon (a,b)\subseteq\Real\to Q$
is an {\sl integral curve} of $X$ at $\gamma(t)\in Q$, for $t\in (a,b)$, if
$$
X(\gamma(t))=\dot\gamma(t)\ .
$$
Taking this into account, if $\beta\in\df^k(Q)$,
the contraction of $\beta$ with $X$, denoted by $\inn(X)\beta$, allows us to define a new contraction
$\big(\inn(\widetilde \gamma)(\beta\circ\gamma)$  in a natural way, as follows:
\beq
\inn(\widetilde \gamma(t))\big(\beta(\gamma(t))\big):=
\inn\big(\dot\gamma(t)\big)(\beta\vert_{\gamma(t)}) \quad, \quad \mbox{\rm for $t\in (a,b)$} \ .
\label{contr2}
\eeq
}\end{remark}

\subsection{The cotangent bundle of a manifold. Canonical lifts}
\label{sec:cotbun}

\begin{definition}
Let $Q$ be a differentiable manifold.
The \textbf{cotangent bundle} of $Q$ is the dual bundle
$\Tan ^*Q$ of the tangent bundle $\Tan Q$;
that is,  \(\dst \Tan^*Q := \bigcup_{q \in Q}\Tan^*_qQ\)~.

We denote $\pi_Q \colon \Tan^*Q \to Q$ the natural projection.
\end{definition}

Thus, every point ${\rm p} \in \Tan^*Q$ is a pair
$(q,\xi)\equiv\xi_q$, where $q=\pi_Q ({\rm p}) \in Q$
and $\xi \in \Tan^*_qQ$ (it is a linear form on $\Tan_qQ$).

\begin{prop}
\label{prop:cotbun}
Let $Q$ be a $n$-dimensional differentiable manifold.
The cotangent bundle $\Tan^* Q$ is a $2n$-dimensional differentiable manifold, 
whose differentiable structure is induced by the one of $Q$ in a natural way.
Moreover, the natural projection $\pi_Q$ is a submersion.
\end{prop}
\begin{proof}
If ${\cal A}=\{ (U_\alpha;\phi_\alpha=(x^i) \}$ is an atlas of local charts in $Q$,
then the induced atlas
$\Tan^*{\cal A}=\{ (\Tan^*U_\alpha;\eta_\alpha=(q^i,p_i) \}$
is obtained as follows:

Given one of its elements $(U;x^i)$, we take 
\begin{enumerate}
\item
$\Tan^* U := \pi_Q^{-1}(U)$.
\item
For every ${\rm p} \in \Tan^*Q$, let $q^i({\rm p}) := x^i \circ \pi_Q ({\rm p})$\footnote
{We commit an abuse of notation, denoting
by $q^i$ both the coordinates $q^i$ and $x^i$.}.
\item
Let ${\rm p} \in \Tan^*Q$ with ${\rm p}=(q,\xi)$. Then \
$\displaystyle p_i({\rm p}) := \Big\langle\xi ,\derpar{}{q^i}\Big\vert_q\Big\rangle$;\
that is, $p_i({\rm p})$ are the components of the linear form
$\xi$ in the natural basis $\d q^i\mid_q$ of $\Tan^*_qM$: 
$\xi = p_i({\rm p})\,\d q^i\mid_q$.
\end{enumerate}
If $(q^i)$ and $(\bar q^i)$ are two local systems of coordinates
in $U\subset Q$, and $\bar q^j =\phi^j (q^i)$, then
the relation between $(p_i)$ and $(\bar p_i)$ is the following,
$$
p_i(q,\xi )=\Big\langle\derpar{}{q^i},\xi\Big\rangle=
\Big\langle\derpar{\phi^j}{q^i}\Big\vert_q \derpar{}{\bar q^j},\xi\Big\rangle
=\derpar{\phi^j}{q^i}\Big\vert_q\Big\langle\derpar{}{\bar q^j},\xi\Big\rangle
=\derpar{\phi^j}{q^i}\Big\vert_q \bar p_j(q,\xi ) \ ;
$$
that is, \(\dst p_i=\derpar{\phi^j}{q^i}\bar p_j\).
Then, the changes of coordinates in $\Tan^*Q$,
from $(q^i,p_i)$ to $(\bar q^i,\bar p_i)$, have as Jacobian matrix
$$
\left(\begin{matrix}\displaystyle\left(\derpar{\phi^j}{q^i}\right) & 0 \\
\displaystyle\left(\frac{\partial^2\phi^j}{\partial q^k\partial q^i}\right)^{-1}p_i &
\displaystyle\left(\derpar{\phi^j}{q^i}\right)^{-1}\end{matrix}\right) \ .
$$

From this local description of the cotangent bundle
it is evident that its dimension is $2n$ as we know previously being the dual bundle of $\Tan Q$.
\\ \qed  \end{proof}

\begin{definition}
The above coordinate charts are called
\textbf{natural charts} of the cotangent bundle
and their elements \textbf{natural coordinates}
($q^i$ are the \textbf{ base coordinates}
and $p_i$ the \textbf{fiber coordinates}).
\end{definition}

Observe that the determinant of the above last matrix is the unity and this is the same for any change of natural coordinates, hence the manifold $\Tan^*Q$ is also an orientable manifold.

In an analogous way as for the tangent bundle, there are
natural (canonical) ways to lift
diffeomorphisms and vector fields from a differentiable manifold $Q$
to its cotangent bundle $\Tan^*Q$.

\begin{definition}
Consider a diffeomorphism
$$
\begin{array}{ccccc}
\varphi&\colon&Q&\longrightarrow&Q
\\
& & x & \mapsto & \varphi(x)
\end{array} \ .
$$
The \textbf{canonical lift} of $\varphi$ to $\Tan^*Q$ is the diffeomorphism
$$
\begin{array}{ccccc}
\Tan^*\varphi&\colon&\Tan^*Q&\longrightarrow&\Tan^*Q
\\
& &(\varphi(x),\xi ) & \mapsto & (x,\Tan^*\varphi (\xi))
\end{array}
$$
where $\Tan^*\varphi (\xi)$ is defined by duality as follows: if
$$
\begin{array}{ccccc}
\Tan\varphi&\colon&\Tan Q&\longrightarrow&\Tan Q
\\
& &(q,v) & \mapsto & (x,\Tan\varphi (v))
\end{array}
$$
then, for every $v\in\Tan_qQ$,
$$
(\Tan^*\varphi(\xi))(v) :=\xi (\Tan\varphi(v)) \ .
$$
\end{definition}

The following properties follow straightforwardly from the definition:

\begin{prop}
\label{exerc2}
If $\varphi,\phi\colon Q\to Q$ are diffeomorphisms, then:
\ben
\item
$\varphi\circ\pi_Q\circ\Tan^*\varphi =\pi_Q$.
\item
$\Tan^*(\varphi\circ\phi )=\Tan^*\phi\circ\Tan^*\varphi$,
\een
\end{prop}

\begin{definition}
Consider a vector field $Z\in\vf (Q)$.
The \textbf{canonical lift} of $Z$ to $\Tan^*Q$
is the vector field $Z^*\in\vf (\Tan^*Q)$
whose local uniparametric groups of diffeomorphisms
are the canonical lifts $\Tan^* F_t$
of the local uniparametric groups of diffeomorphisms
$ F_t$ of $Z$.
\end{definition}

As a straightforward consequence of the definition and of the
theorem of existence and unicity of uniparametric groups
of diffeomorphisms, we have that:

\begin{prop}
If $Z\in\vf (Q)$, then $Z^*$ is $\pi_Q$-projectable and, 
for every ${\rm p}\equiv (q,\xi)\in\Tan^*Q$,
we have that $\Tan\pi_Q (Z^*_{\rm p})=Z_q$.
\label{prolevcan}
\end{prop}

The coordinate expression of the canonical lift of a vector field to $\Tan^*Q$ is given in \eqref{canlifcot}.

%%%%%%%%%%%%%%%%%%%%%%%%%%%%%%%%%%%%%%%%%%%%%%%%%%%%%%%%%%%%%%%%%%%%%%%%%%%%

\section{Lie groups and Lie algebras}
\label{Liega}

(For more information about definitions and properties of Lie groups and Lie algebras, see,
for instance, \cite{Ar-89,CP-adg,Ok-87,To-65,Wa-fsmlg}).

\begin{definition}
A (finite-dimensional) \textbf{Lie group}
is a (finite-dimensional) differentiable manifold $G$ such that
{\rm (i)}
$G$ is a group and
{\rm (ii)}
the following group operations are smooth,
$$
\begin{array}{ccccccc}
G \times G &\to& G & \quad , \quad & G &\to& G
\\
(g_1,g_2) &\mapsto& g_1g_2 & \quad , \quad & g &\mapsto& g^{-1}
\end{array} \ .
$$
\end{definition}

\begin{definition}
Let $G$ be a Lie group. A subset $H \subset G$.
is a \textbf{Lie subgroup} of $G$ if
{\rm (i)} $H$ is a subgroup of $G$ and
{\rm (ii)} $H$ has a differentiable structure with respect
to which it is a submanifold of $G$ and is a Lie group.
\end{definition}

Given a Lie group $G$ and an element $g \in G$,
we can define the diffeomorphisms
$$
\begin{array}{ccccccccc}
L_g \colon & G & \to & G & \ , \ & R_g \colon & G & \to & G
\\
& g' & \mapsto & gg' & \ , \ & &  g' & \mapsto & g'g
\end{array} \ ,
$$
which are the so-called {\sl \textbf{left}} and {\sl \textbf{right translations}}
respectively, and satisfy the following properties:
\begin{enumerate}
\item
For every $g,g' \in G$;\quad
$L_g \circ L_{g'} = L_gL_{g'}$ , \ 
$R_g \circ R_{g'} = R_gR_{g'}$.
\item
For every $g \in G$;\quad
$L_g^{-1} = L_{g^{-1}}$ , \
$R_g^{-1} = R_{g^{-1}}$.
\item
For every $g,g' \in G$;
$L_g \circ R_{g'} = R_{g'} \circ L_g$.
\end{enumerate}
These translations also induce translations on the set of
vector fields $\vf(G)$, which are denoted by
$L_{g_*}$ and $R_{g_*}$, respectively. Then, 

\begin{definition}
A vector field $X \in\vf(G)$ is
\textbf{left} (resp. \textbf{right}) \textbf{invariant}
if, for every $g \in G$,
$$
L_{g_*}X = X \quad  (resp. \ R_{g_*}X = X) \ .
$$
The sets of these vector fields are denoted
$\vf_L(G)$ and $\vf_R(G)$ respectively.
\end{definition}

\begin{prop}
The sets $\vf_L(G)$ and $\vf_R(G)$ are Lie subalgebras of $\vf(G)$.
\end{prop}
\begin{proof}
The sum of invariant vector fields is also invariant, 
and the same thing holds for the Lie bracket, 
since, for every $g \in G$ and $X_1,X_2 \in\vf_L(G)$ (or $X_1,X_2 \in\vf_R(G)$),
$$
L_{g_*}[X_1,X_2] = [L_{g_*}X_1,L_{g_*}X_2]=[X_1,X_2] \ .
$$
\qed  \end{proof}

\begin{definition}
The set $\vf_L(G)$ is the
\textbf{Lie algebra} of $G$ and is denoted by ${\bf g}$.
\end{definition}

\begin{prop}
If $e\in G$ is the neutral element of $G$, then there is a
canonical vector space isomorphism between 
$\vf_L(G)$ and $\Tan_eG$.
\end{prop}
\begin{proof}
To every $X\in\vf_L(G)$ corresponds a vector of
$\Tan_eG$ by means of the isomorphism
$$
\begin{array}{cccc}
\rho_1 \colon &\vf_L(G)&\to&\Tan_eG
\\
&X&\mapsto&X(e):=X_e
\end{array} \ .
$$
Conversely, from every $X_e \in \Tan_eG$,
we can obtain a unique vector field $X \in\vf_L(G)$
using the map
$$
\begin{array}{ccccc}
X&\colon&G&\to&\Tan G
\\
& &g&\mapsto&(g,\Tan L_gX_e)
\end{array} \ .
$$
which is a vector field since, for every $f \in \Cinfty(G)$, 
the map $X(f) \colon G \to \Real$
is smooth and, hence, so is $X$.
Moreover, $X$ is left-invariant, as you can easily see
taking into account its construction from $X_e$.
The map giving $X$ from $X_e$ is denoted $\rho_2 \colon \Tan_eG \to\vf_L(G)$.
\\ \qed  \end{proof}

\begin{itemize}
\item
Thus, it is usual to identify ${\bf g}$ with $\Tan_eG$, and then
the Lie bracket between two vectors of $\Tan_eG$ is defined as
$$
[X_e,X_e'] :=[\rho_2X_e,\rho_2X_e'](e) \ .
$$
Then, if $\{ v_i \}$ is a basis of $\Tan_eG$,
there exist $c_{ij}^k\in\Real$ such that
$$
[v_i,v_j] = c_{ij}^kv_k \ ,
$$
which are called {\sl \textbf{structure constants}}
of the group $G$ (related to the chosen basis of ${\bf g}$).
It can be proved that
$c_{ij}^k = 0$ (for all $i,j,k$)
if, and only if, $G$ is locally isomorphic to 
$\Real^n$, that is, $G$ is an {\sl Abelian Lie group}.
\item
As it is obvious, we can take $\vf_R(G)$ instead of $\vf_L(G)$ 
in order to define ${\bf g}$.
Then, the identification with $\Tan_eG$
would be established in an analogous way,
although the Lie algebra structure so obtained in $\Tan_eG$
is the opposite one to the above.
Nevertheless, both structures are isomorphic
since the map $-{\rm Id}_{\Tan_eG}$ exchanges them.
\end{itemize}

\begin{definition}
Let ${\bf g}$ be a Lie algebra.
Then, we have the following sequence of ideals:
$$
{\bf g} \supseteq {\bf g'} \supseteq {\bf g''} \supseteq \ldots 
\supseteq {\bf g^k} \supseteq \ldots
$$
where
$$
{\bf g'} := [{\bf g},{\bf g}] \ , \ 
{\bf g''} := [{\bf g'},{\bf g'}] \ , \ 
\ldots \ , \
{\bf g^k} := [{\bf g^{k-1}},{\bf g^{k-1}}] \ .
$$
Then, ${\bf g^k}$ is called the
{\rm $k$th-derived algebra} of ${\bf g}$.
\end{definition}

Observe that, if  ${\bf g^k} = 0$, then ${\bf g^{k+1}} = 0$.
Then we define:

\begin{definition}
Let ${\bf g}$ be a Lie algebra.
\begin{enumerate}
\item
${\bf g}$ is \textbf{Abelian} if \
${\bf g'} = 0$.
\item
${\bf g}$ is \textbf{solvable} if \
${\bf g^k} = 0$, for some $k>0$.
\item
Consider now the sequence \ 
${\bf g} = {\bf g_1} \supseteq {\bf g_2} \supseteq \ldots 
\supseteq {\bf g_k} \supseteq \ldots$
where
$$
{\bf g_1} := {\bf g} \ , \ 
{\bf g_2} := [{\bf g},{\bf g_1}] \ , \ 
\ldots \ , \
{\bf g_k} := [{\bf g},{\bf g_{k-1}}]
$$
${\bf g}$ is \textbf{nilpotent} if \
${\bf g_k} = \{ 0 \}$, for some $k>0$.
\item
${\bf g}$ is \textbf{simple} if \
it has no ideals other than $\{ 0 \}$ and ${\bf g}$.
\item
${\bf g}$ is \textbf{semisimple} if \
its largest solvable ideal is $\{ 0 \}$;
that is, if ${\bf g^k}={\bf g}$, for every $k$.
\end{enumerate}
\end{definition}

It is easy to prove that
$$
{\bf g} \ {\rm abelian} \ \Rightarrow
{\bf g} \ {\rm nilpotent} \ \Rightarrow
{\bf g} \ {\rm solvable} \ ,
$$
and that every non-Abelian simple Lie algebra
is semisimple.

\begin{definition}
Let ${\bf g}$ be a Lie algebra.
The \textbf{center} of ${\bf g}$ is the set
$$
\{ X \in {\bf g} \ | \ [X,Y]=0 \ ,  Y \in {\bf g} \} \ ,
$$
which is a commutative ideal of ${\bf g}$.
\end{definition}

%In an analogous way as for the invariant vector fields, we can define:

\begin{definition}
A $p$-form $\beta \in {\mit\Omega}^p(G)$ is 
\textbf{left} (resp. \textbf{right}) \textbf{invariant}
if, for every $g \in G$,
$$
L_g^*\beta = \beta \quad (resp. \ R_g^*\beta = \beta ) \ .
$$
The set of the left-invariant $1$-forms  
is the dual vector space of $\vf_L(G) := {\bf g}$ 
and is denoted by ${\bf g}^*$.
\end{definition}

\begin{prop}
${\bf g}^*$ is isomorphic to $\Tan_e^*G$.
\end{prop}
\begin{proof}
It is obvious from the definition.
\\ \qed  \end{proof}

As a consequence, from a covector  $\beta_e \in \Tan_e^*G$, 
we can obtain another one at every point $g \in G$
doing $L_{g^{-1}}^*\beta_e := \beta_g$.
Then, the map
$$ 
\begin{array}{ccccc}
{\rm B}&\colon&G&\to&\Tan^*G
\\
& &g&\mapsto&(g,\beta_g )
\end{array} 
$$
is a left-invariant $1$-form.
Therefore, it is easy to prove the following properties:

\begin{prop}
\label{A30}
\begin{enumerate}
\item
For every $g \in G$, $\beta \in {\bf g}^*$,
and $X,Y \in {\bf g}$, we have that:
\begin{enumerate}
\item
$\d \beta$ is left-invariant.
\item
$\langle X_g,\beta_g\rangle = \langle X_e,\beta_e\rangle$;
that is,
$\langle X,\beta\rangle = ctn.$
\item
$\d \beta (X,Y) = -\langle[X,Y],\beta\rangle$.
\end{enumerate}
\item
If $\{ v_i \}$ is a basis of $\Tan_eG$ and
$\{ \gamma^i \}$ is the dual basis in $\Tan_e^*G$, then
\begin{enumerate}
\item
$\d \gamma^k (v_i,v_j) = -c_{ij}^l<v_l,\gamma^k> = - c_{ij}^k$
\item
The following equations holds
$$
\d \gamma^k = -\frac{1}{2} c^k_{ij}\gamma^i \wedge \gamma^j
$$
The numbers $c^k_{ij}\in\Real$ are called the
\textbf{Maurer--Cartan structure equations}
of the group $G$ (related to the chosen basis of ${\bf g}^*$).
\end{enumerate}
\end{enumerate}
\end{prop}

\begin{definition}
The \textbf{canonical} or  \textbf{Maurer--Cartan form}
of the Lie group $G$ is the $1$-form with values on ${\bf g}$,
$\omega \in {\mit\Omega}^1(G,{\bf g})$,
defined as follows: 
for every $g \in G$ and $X \in\vf(G)$,
$$
\omega_g(X_g) := \Tan L_g^{-1}X_g
$$
\end{definition}

Obviously, $\omega$ is left-invariant.

Observe that, in the definition, $X \in {\cal X}(G)$.
In particular, if $X \in {\bf g}$, 
(that is, $X$ is left-invariant) then we have that
$\Tan L_g^{-1}X_g = X_e$.
Then an alternative definition of $\omega$ is the following:

\begin{definition}
The \textbf{Maurer--Cartan form}
of the Lie group $G$ is the only left-invariant $1$-form
with values on ${\bf g}$,
$\omega \in {\mit\Omega}^1(G,{\bf g})$,
such that $\omega(X) = X$, for every $X \in {\bf g}$;
that is, it gives the identity one on ${\bf g}$.
\end{definition}

\begin{definition}
The \textbf{uniparametric subgroups} of $G$ 
are the integral curves passing through $e$
of the left (resp. right) invariant vector fields,
that is, the maps 
$$ 
\begin{array}{ccccc}
\alpha &\colon&\Real&\to&G
\\
& &t&\mapsto&g_t:=\alpha(t)
\end{array} 
$$
such that, if $X \in \vf_L(G)$, then
{\rm (i)} \ 
$\displaystyle X_{g_{t_0}} = \frac{\d}{\d t} \alpha (t)\Big\vert_{t_0}$ and \ 
{\rm (ii)}
$\displaystyle X_e = \frac{\d}{\d t} \alpha (t)\Big\vert_0$.
\end{definition}

These uniparametric subgroups are complete and,
there exists a bijective correspondence between
${\bf g}$ and the set of uniparametric subgroups of $G$.
The map which implements this correspondence is defined as follows:

\begin{definition}
The \textbf{exponential map}
is defined as,
$$ 
\begin{array}{ccccc}
{\rm exp}&\colon&\Tan_eG&\to&G
\\
& &tX_e&\mapsto&g_t
\end{array} 
$$
for $t \in \Real$, and such that
${\rm exp}(0)=e$; that is, if $\alpha$ is the uniparametric subgroup
corresponding to $\rho_2X_e$, then \ 
{\rm (i)} \ 
${\rm exp}(X_e) := \alpha(1)$ and \
{\rm (ii)}\ 
${\rm exp}(0) := e$.
\end{definition}

\begin{itemize}
\item
Notice that the exponential map maps
a point of the line $tX_e$ on $\Tan_eG$ into
a point of the integral curve of the vector field
$\rho_2X_e$ on $G$.
\item
Taking into account that $\Real$ is a Lie group
with respect to the sum and that
$\alpha \colon \Real \to G$ 
is a Lie group homomorphism, we have that
\beann
{\rm exp}(t_1X_e){\rm exp}(t_2X_e) 
&=& g_{t_1}g_{t_2} = \alpha(t_1)\alpha(t_2) = \alpha(t_1+t_2) 
\\
&=& g_{(t_1+t_2)} = {\rm exp}((t_1+t_2)X_e) \ .
\eeann
Furthermore,
$$
\frac{\d}{\d t} {\rm exp}(tX_e)\Big\vert_{t_0} =
\frac{\d}{\d t}\alpha(t)\Big\vert_{t_0} = X_{g_{t_0}} \ ,
$$
and, in particular,
$$
\frac{\d}{\d t} {\rm exp}(tX_e)\Big\vert_{t=0} = X_e \ .
$$
All these properties justify the name given to this map
\footnote{
Another justification is that, 
if $V$ is a vector space, and we consider
$G={\rm Aut}(V)$, then
${\rm exp} \colon {\rm End}(V) \to {\rm Aut}(V)$ and, for every $a \in {\rm End}(V)$, we obtain that
$$
{\rm exp}(a) =e^a := {\bf 1} + a + a^2/2! + a^3/3! + \ldots + a^n/n! + \ldots
$$
where ${\bf 1}$ is the identity and
$a^n := a \circ \ldots \circ a$ ($n$ times).}.
Observe that the last two equalities are natural
from the definition, since
$\alpha(t) = {\rm exp}(tX_e)$ 
is the integral curve of $X$. 
Hence, the flux of $X$ is
$$ 
\begin{array}{ccccc}
\tau&\colon&\Real \times G&\to&G
\\
& &(t,g)&\mapsto&g\,{\rm exp}(tX_e) \ ;
\end{array} 
$$
that is,
$\tau_t(g) = g\,{\rm exp}(tX_e) = R_{{\rm exp}(tX_e)}g$.
\item
The exponential map allows constructing locally the group $G$
from the algebra ${\bf g}$.
\end{itemize}

\end{appendix}

%%%%%%%%%%%%%%%%%%%%%%%%%%%%%%%%%%%%%%%%%%%%%%%%%%%%%%%%%%%%%%%%%%%%%%%%%%%%%%%%%%%%%%%%%%%%%%%%%%%%%%%%%%%%%%%%%

\bibliographystyle{abbrv}
\small
%\phantomsection
\addcontentsline{toc}{chapter}{References}
\label{Chap:References}
%\bibliography{Biblio_gms.bib}
\bibliographystyle{AMS_Mod}
%\nocite{*}

%%%%%%%%%%%%%%%%%%%%%%%%%%%%%%%%%%%

\end{document}